\documentclass[%
 reprint,
 superscriptaddress,
 amsmath,amssymb,
 aps,
 pra,
floatfix,
nofootinbib,
]{revtex4-2}
\usepackage{float}
\makeatletter
\let\newfloat\newfloat@ltx
\makeatother
\usepackage{graphicx}
\usepackage{dcolumn}
\usepackage{bm}
\usepackage[english]{babel}
\usepackage[utf8]{inputenc}
\usepackage{amsthm}
\usepackage{mathtools}
\usepackage{physics}
\usepackage{algorithm}
\usepackage{algpseudocode}
\usepackage{amsmath}
\usepackage{xcolor}
\usepackage{graphicx}
\usepackage[left=23mm,right=13mm,top=35mm,columnsep=15pt]{geometry} 
\usepackage{adjustbox}
\usepackage{placeins}
\usepackage[T1]{fontenc}
\usepackage{lipsum}
\usepackage{csquotes}
\usepackage{float}
\usepackage{hyperref}
\usepackage[title]{appendix}
\usepackage{multirow}
\usepackage{booktabs}
\usepackage{array}
\usepackage{caption}
\usepackage{comment}
\usepackage{physics}
\usepackage{braket}
\usepackage{esvect}
\usepackage[color=yellow]{todonotes}
\usepackage{soul}
\usepackage{caption}
\usepackage[normalem]{ulem}
\usepackage{enumitem}
\newtheorem{theorem}{Theorem} 

\theoremstyle{definition}
\newtheorem{definition}{Definition}[section]
\newtheorem{claim}{Claim} 
\newtheorem{conjecture}{Conjecture} 

\theoremstyle{remark}
\newtheorem*{remark}{Remark}

\newtheorem{example}{Example}

\makeatletter
\renewcommand{\paragraph}[1]{%
  \vspace{1.5ex}%
  \begin{center}
    \textit{#1}
  \end{center}
  \vspace{0.5ex}%
}
\makeatother

\makeatletter
\def\l@subsubsection#1#2{} 
\makeatother 

\makeatletter
\newcommand{\appsubsection}[1]{%
  \refstepcounter{subsection}%
  \bigskip
  \noindent\textbf{\thesubsection\quad #1}\par
  \smallskip
}
\makeatother

\begin{document} 
\title{Do quantum linear solvers offer advantage for networks-based system of linear equations? }
\author{Disha Shetty}
\email{dishag.shetty12@gmail.com}
\affiliation{Centre for Quantum Engineering, Research and Education, TCG CREST, Sector V, Salt Lake, Kolkata 700091, India} 
 
\author{Supriyo Dutta} 
\affiliation{Department of Mathematics,
National Institute of Technology Agartala,
Jirania, West Tripura, India - 799046}
\author{Palak Chawla}
\affiliation{Centre for Quantum Engineering, Research and Education, TCG CREST, Sector V, Salt Lake, Kolkata 700091, India}
\author{Akshaya Jayashankar}
\affiliation{Centre for Quantum Engineering, Research and Education, TCG CREST, Sector V, Salt Lake, Kolkata 700091, India}
\author{Jordi Riu}
\affiliation{Qilimanjaro Quantum Tech, Carrer de Veneçuela, 74, Sant Martí, 08019, Barcelona, Spain}
\affiliation{Universitat Politècnica de Catalunya, Carrer de Jordi Girona, 3, 08034 Barcelona, Spain}
\author{Jan Nogué}
\affiliation{Qilimanjaro Quantum Tech, Carrer de Veneçuela, 74, Sant Martí, 08019, Barcelona, Spain}
\affiliation{Universitat Politècnica de Catalunya, Carrer de Jordi Girona, 3, 08034 Barcelona, Spain}
\author{K. Sugisaki}
\affiliation{Centre for Quantum Engineering, Research and Education, TCG CREST, Sector V, Salt Lake, Kolkata 700091, India}
\affiliation{Deloitte Tohmatsu Financial Advisory LLC, 3-2-3 Marunouchi, Chiyoda-ku, Tokyo 100-8363, Japan}
\author{V. S. Prasannaa}
\email{srinivasaprasannaa@gmail.com}
\affiliation{Centre for Quantum Engineering, Research and Education, TCG CREST, Sector V, Salt Lake, Kolkata 700091, India} 
\affiliation{Academy of Scientific and Innovative Research (AcSIR), Ghaziabad 201002, India} 

\begin{abstract} 
In this exploratory numerical study, we assess the suitability of Quantum Linear Solvers (QLSs) toward providing a quantum advantage for Networks-based Linear System Problems (NLSPs). NLSPs naturally arise from graphs, and are of importance as they are connected to real-world applications. The achievable advantage with a QLS for an NLSP depends on the interplay between the scaling of condition number and sparsity of matrices associated with the graph family. We analyze 50 graph families and identify that within the scope of our study, only 21 of them exhibit prospects for an exponential advantage with the Harrow-Hassidim-Lloyd (HHL) algorithm relative to an efficient classical solver. We call graph families that offer advantage with HHL as good graph families. We also compare the performance of the considered 50 graph families with 7 other QLSs. Furthermore, we report that some graph families graduate from offering no advantage with HHL to promising an exponential advantage with improved algorithms such as the Childs-Kothari-Somma algorithm. We also introduce a unified graph superfamily and show the existence of infinite good graph families in it. Since the runtime expressions for linear solvers involve condition number, which in itself is not easy to compute, ascertaining advantage prospects with quantum linear solvers itself is not an easy problem. Thus, we conjecture the conditions under which one may visually examine a graph family and guess the prospects for an advantage. Finally, we very briefly touch upon some practical issues that may arise even if the aforementioned graph theoretic requirements are satisfied, including quantum hardware challenges. 
\end{abstract} 

\maketitle 

\tableofcontents 

\section{Introduction}\label{sec:intro} 

Quantum algorithms promise speed-up for certain problems relative to their best known classical counterparts, thus motivating the ongoing second quantum revolution, which involves building reliable quantum computers to eventually advance toward commercial realization \cite{scholten2024assessing, bova2021commercial, gill2025quantum, macquarrie2020emerging, cumming2022using, sinno2023performance, ruane2025quantum}. Quantum linear solvers (QLSs) are particularly significant in this context, as they are among the few classes of known quantum algorithms that can, in principle, offer an exponential advantage  \cite{morales}. A QLS prescribes a protocol towards evaluating systems of linear equations, $A\vec{x}=\vec{b}$, where the $(\mathcal{N} \times \mathcal{N})$ matrix $A$ and the vector $\vec{b}$ are known quantities. Since such systems are ubiquitous in natural sciences and engineering \cite{anton2013elementary, machol1961linear, hamid2019balancing, wilson1958solution}, solving them efficiently in view of the aforementioned quantum speed-up is of importance. An example of a QLS is the well-known Harrow-Hassidim-Lloyd (HHL) algorithm ~\cite{Harrow2009QuantumEquations}, whose runtime complexity goes as $\mathrm{poly}(\mathrm{log}(\mathcal{N}), s, \kappa, 1/\epsilon)$ where $\mathcal{N}$, $s$, and $\kappa$ refer to the system size, sparsity, and the condition number of the $A$ matrix respectively, while $\epsilon$ refers to the additive error in the output state after performing HHL. The algorithm can outperform their best known classical counterparts such as the conjugate gradient algorithm \cite{cg1994}, when these parameters grow in a certain way with respect to the system size. However, identifying such settings where QLSs offer such an advantage is challenging in spite of a lot of efforts in literature in this direction. Our exploratory work is aimed toward addressing this timely problem. 

\textit{Analysis of quantum advantage from HHL for specific problems-- }Prior numerical studies in this direction indicate that the prospects of speed-up are limited at best due to the condition number scaling. In Ref. \cite{coffrin}, the authors consider the DC power flow problem and analyze the end-to-end complexity of the HHL algorithm, including obtaining the scaling behaviours of the condition number and sparsity with system size for this application. Their numerical simulations demonstrate that since $\kappa$ grows polynomially in system size (where system size is the number of buses), practical advantage from HHL for this application is unlikely. The authors of Ref. \cite{mlhhl} study the suitability of the HHL algorithm along these lines in the context of labeling problems using machine learning classifiers, and reach the conclusion that the condition number has a critical impact on the problem. The authors of Ref. \cite{hydro} explore modeling hydrological fracture networks, and en route, carry out an analysis of $\kappa$ with system size. Their results point to $\kappa$ growing unfavourably in system size, unless preconditioning, which itself is classically resource intensive, is employed. A work by the authors of Ref. \cite{yalovetzky2021hybrid} applies the HHL algorithm to finance (portfolio optimization), and their data indicates that $\kappa$ scales quadratically with system size(number of assets) even in the best case scenario. A relatively recent work that comments on $\kappa$ scaling is Ref. \cite{gopalakrishnan2024solving}, where the authors solve Hele-Shaw flow in fluid dynamics using HHL, and find that $\kappa$ scales exponentially for their example case with system size (domain grid points). A recent analysis on limited molecular systems indicated that HHL-based quantum chemistry application \cite{psihhl} seems promising as it shows polylogarithmic scaling in $\kappa$. 

\textit{Analysis of quantum advantage from HHL in the most general setting-- }The exact opposite viewpoint would involve comparing the runtime complexities of the HHL and some suitably chosen efficient classical algorithm based on their respective complexity expressions. The authors of Ref. \cite{tu2025towards} do exactly this in their work, with the end goal of carrying out a resource estimation analysis in terms of space, time, and energy for the HHL algorithm. 

\begin{figure*}[t]
\includegraphics[width=17cm]{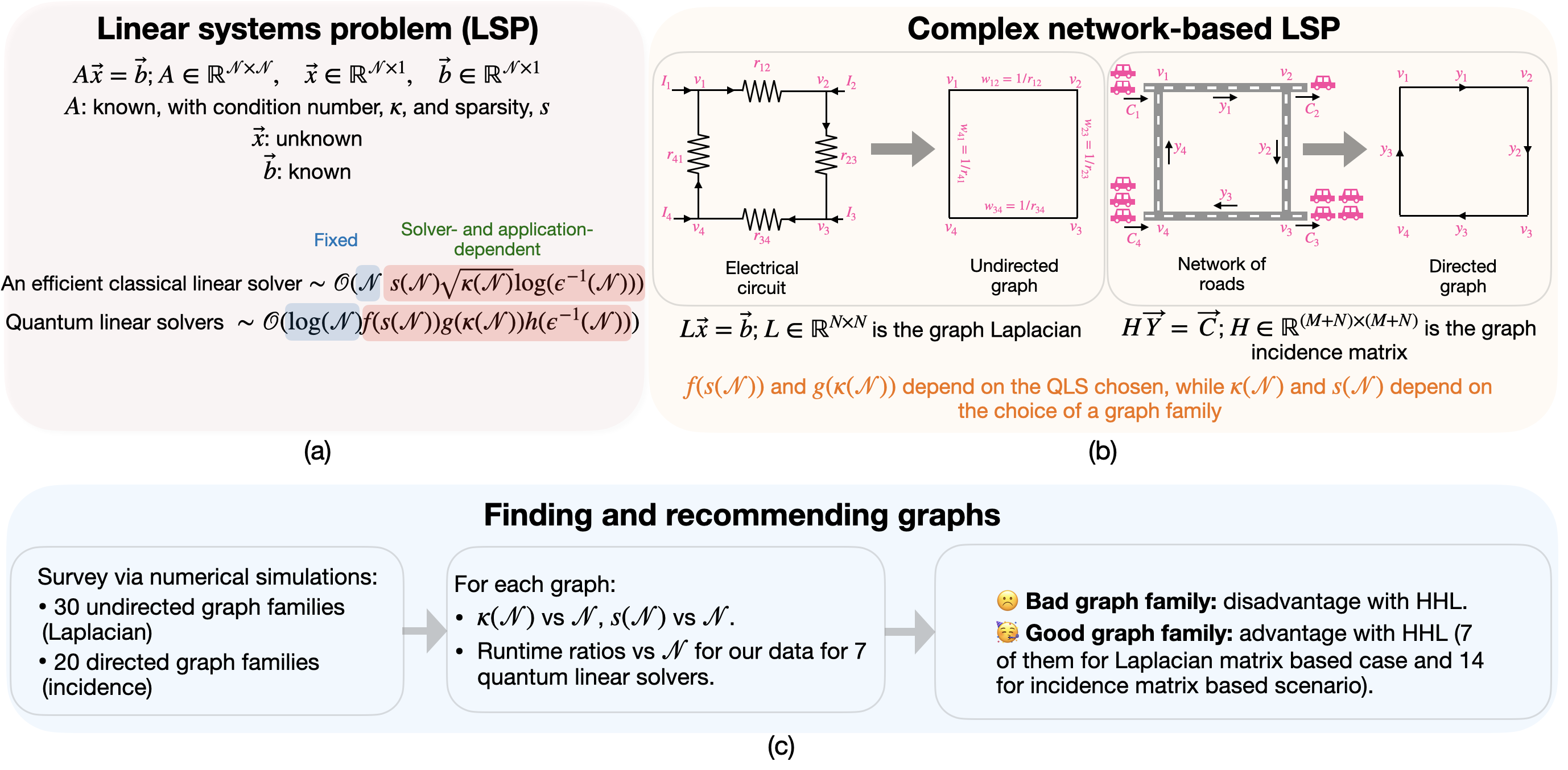}
\caption{An overview of the current study. (a) Illustration of the linear systems problem (where for simplicity, we assume real-valued entries for $A$, $\vec{b}$, and $\vec{x}$) and the runtime complexity scaling of quantum linear solvers and an efficient classical linear solver, for which we happen to borrow the runtime expression of the otherwise limited conjugate gradient method. Here, ${f}(s(\mathcal{N}))$ denotes a function of sparsity, $s(\mathcal{N})$, which in turn depends on the system size, $\mathcal{N}$. (b) Depiction of the connection between real-world applications, such as effective resistance determination and traffic flow congestion detection with graph Laplacian and graph incidence matrices respectively, and linear equations. (c) Schematic of our numerical survey on 50 graph families, where for each of them, we study $\kappa(\mathcal{N})$ and $s(\mathcal{N})$ behaviour with system size, $\mathcal{N}$ ($N$ for Laplacian matrix and $N+M$, for the incidence matrix), to infer within the scope of our calculations the potential for quantum advantage (good graph family) or no advantage (bad graph family), all with the HHL algorithm and compared relative to the efficient classical linear solver. }\label{fig:gblsp}
\end{figure*} 

\textit{Aurea mediocritas?: Analysis of quantum advantage from quantum linear solvers in networks-based linear system problems-- }A third route, a road not taken, could be to chart a middle course by neither pursuing application-specific studies nor a fully general analysis, but rather group together diverse applications under a common mathematical framework, and then perform extensive numerical analysis to analyze the potential for advantage. Such a framework should offer flexible scaling in $\kappa$ and $s$, and system size. 

For this work, we consider for our numerical analyses the highly flexible Networks-based Linear System Problems (NLSP) framework, where one begins with a complex network and by applying a set of rules arrives at a system of linear equations. The NLSP framework accommodates a panoply of graphs/complex networks, and thus admits not only a huge variety of functions for $\kappa$, $s$, and system size, but also many potential applications. We pick two types of NLSPs, one where the $A$ matrix is the graph Laplacian, and the other where it is the graph incidence matrix. The former finds its applications in problems such as determining effective resistances in electrical circuits and finding voltages in power flow problems \cite{vishnoi2013lx, spielman2010algorithms}, while the latter can be used in applications that involve finding flows in branches of networks, such as traffic flow congestion detection \cite{anton2013elementary}. 

Our study investigates different graph families to identify those that offer an advantage. Through detailed numerical analyses, we study the scaling of $\kappa$, which has garnered much attention in literature as discussed in the earlier paragraphs, and $s$ across candidate graph families. It is especially important to note that finding analytical expressions for $\kappa$ for a graph family is extremely rare, and is hard in general. This necessitates and motivates our numerical study. Using our numerical data for $\kappa$ and $s$, we compare the runtime of a QLS relative to a fictitious efficient classical linear solver (CLS), enabling us to estimate the crossover points where quantum advantage can be realized. We identify graph families that offer advantage with HHL (good graph families), and also discuss a new graph superfamily that subsumes in it an infinity of good graph families. 

As performing numerical analyses to categorize the graph families are computationally costly, we ask whether one can \textit{qualitatively} assess $\kappa$ and $s$, and thus by extension the possibility of assessing prospects of an advantage by only looking at small instances of a graph family. We conjecture that it is possible, based on conclusions from our data. 

Lastly, we briefly discuss challenges outside graph theoretic considerations, and as an illustration of how the ideas from NLSP can be applied to a specific problem, we switch gears and consider the calculation of effective resistances in electrical circuits using the HHL algorithm on trapped ion quantum hardware. Fig. \ref{fig:gblsp} presents a summary of the topics covered in our study. 

The rest of the work is structured as follows: Given that the topics discussed in this study lie at the intersection of quantum algorithms, complexity, and graph theory and thus can be of interest to readers from all three communities, we attempt to make the article self-contained by devoting Sections \ref{sec:qls} and \ref{sec:cn-lsp} for introducing quantum linear solvers and networks-based linear system problems. Our results and subsequent discussions form the remaining sections: In Section \ref{sec:survey}, we discuss the results from our survey on candidate complex networks and their suitability for achieving quantum advantage. We begin with the results for Laplacian matrix (Section \ref{sec:LG}), and then move to incidence matrix (Section \ref{sec:IG}). This is followed by Section \ref{sec:graph_superfamily}, where we discuss a general graph superfamily construction from which we identify new good graphs. We introduce a conjecture to guess the possibility of advantage from a graph construction in Section \ref{sec:lessons}. In Section \ref{sec:hardware}, we briefly comment on our quantum hardware computations carried out on toy matrices. Finally, we conclude with a summary of the work and future prospects in Section \ref{sec:conclusion}. 

\section{Quantum linear solvers} \label{sec:qls} 

Given a linear system of equations, $A \vec{x}=\vec{b}$, where the coefficient matrix $A$ and the vector $\vec{b}$ are known, we `find' the vector $\vec{x}$ as $A^{-1} \vec{b}$ using a quantum algorithm, preferably in a time $\mathcal{O}(\mathrm{log}(\mathcal{N}))$, where $\mathcal{N}$ is the system size. We write `find' within quotation marks, since in practice, we calculate a feature of $\vec{x}$, as reading the elements of the solution vector takes away the advantage that the algorithm offers. 

In this sub-section, we mostly focus on the prototypical quantum linear solver, the HHL algorithm, as it typically conveys the core ideas that a QLS is built on. This is followed by a brief introduction to the other QLSs that we consider for this work. We note that the list of QLSs we consider here is not exhaustive. 

We begin with some definitions that are relevant to this section. 

\subsection{Definitions} 

\begin{definition} \label{def:kappa}
The finite condition number, $\kappa$, of a matrix is defined by the ratio of the absolute value of its largest to the absolute value of the smallest non-zero eigenvalues. We refer to the quantity simply as condition number for brevity in this work. 
\end{definition} 

\begin{definition} \label{def:sparsity}
The sparsity, $s$, of a matrix is defined by the number of non-zero entries in the row that contains the maximum number of non-zero entries. 
\end{definition} 

\begin{definition} \label{def:syssize} 
System size, $\mathcal{N}$, is defined as the number of rows of the matrix, $A$. 
\end{definition} 

In this work, we consider the following functions for growth of $\kappa$ and $s$ with system size: 

\begin{itemize}
\item Constant: $c$. 
\item Polylogarithmic: $a_p\ \mathrm{log}(\mathcal{N})^p + a_{p-1}\ \mathrm{log}(\mathcal{N})^{p-1} + \cdots +a_1\ \mathrm{log}(\mathcal{N}) + a_0$. We abbreviate this function as `polylog'. 
\item Polynomial: $a_p\mathcal{N}^p + a_{p-1}\mathcal{N}^{p-1} + \cdots +a_1\mathcal{N} + a_0$. 
\item Exponential: $a_2e^{a_1\mathcal{N}}+a_0$. 
\end{itemize} 

In the above equations, $a_i \in \mathbb{R}$. 

\subsection{The HHL algorithm} \label{sec:hhl} 

The HHL algorithm `finds' the solution $\ket{x} = A^{-1}\ket{b}$ by starting with the state, $\ket{b}$, and using a combination of quantum phase estimation (QPE) and controlled rotation gates, followed by measuring an ancillary qubit and a post-selection step. The steps involved in the algorithm, including an example of extracting a feature of the solution vector, is presented in Section S.1 of the Supplemental Material. 

The runtime complexity of the HHL algorithm goes as 
$$\mathcal{O}\left(\frac{\mathrm{log}(\mathcal{N}) \times (s(\mathcal{N}))^2 \times (\kappa(\mathcal{N}))^3}{\epsilon(\mathcal{N})}\right).$$ As discussed earlier, $\mathcal{N}$ is the system size, $s(\mathcal{N})$ is the sparsity of $A$, $\kappa(\mathcal{N})$ the condition number of $A$, and $\epsilon(\mathcal{N})$ is the additive error in the output state that we incur in the algorithm. In deriving the complexity of HHL, this error is assumed to be solely from inadequacy in the number of clock register qubits \cite{Harrow2009QuantumEquations}. In the above expression, we have explicitly shown the dependence of condition number, sparsity, and error on $\mathcal{N}$ to stress its importance in the context of our study. The scaling in $\kappa(\mathcal{N})$ and $\epsilon(\mathcal{N})$ are usually considered as typical drawbacks, and variants of HHL and subsequent QLSs improve on one or both of these aspects. 

\subsection{Variants of HHL and other QLSs} \label{sec:hhlvar} 

We now list four variants of HHL that we consider for this work, all of which focus on reducing $\kappa$ scaling: 

\begin{itemize}
\item 
\textbf{HHL with Amplitude Amplification (HHL-AA)} \cite{Harrow2009QuantumEquations, Brassard}: 
This is often assumed when one discusses HHL, but since the subroutine adds significant depth to the HHL quantum circuit (for example, see Ref. \cite{morales}), we keep it distinct from the original HHL circuit. The benefit that the variant offers is a reduction of complexity in $\kappa(\mathcal{N})$, from $\kappa(\mathcal{N})^3$ to $\kappa(\mathcal{N})^2$. Practically, this has the effect of reducing the number of shots in an HHL calculation. 

\item 
\textbf{HHL with Variable Time Amplitude Amplification (HHL-VTAA)} \cite{ambainis}: 
This approach improves over HHL-AA and can be thought of as its generalization. The method reduces the complexity further to $\kappa(\mathcal{N}) \mathrm{log}^3(\kappa(\mathcal{N}))$, but trades off in $\epsilon(\mathcal{N})$ scaling (see Table \ref{tab:complexities}). 

\item 
\textbf{Psi-HHL} \cite{psihhl}: 
This recently introduced approach reduces the complexity in $\kappa$ to its optimal scaling, that is, $\kappa(\mathcal{N})$, and with little increase in circuit depth, for cases with large $\kappa$ values. 
\item \textbf{CKS algorithm} \cite{Childs}: 
In this landmark work, the authors introduced two QLS algorithms (the Fourier approach and the Chebyshev approach) that are based on combining ideas such as the linear combination of unitaries, gapped phase estimation (to reduce the $1/\epsilon(\mathcal{N})$ scaling), and the VTAA technique (to reduce the scaling in $\kappa(\mathcal{N})$), which we together club under the term CKS algorithm, as both of them, though different in terms of their applicability in terms of sparsity of the $A$ matrix, scale near-linearly in $\kappa(\mathcal{N})$ and as polylog($1/\epsilon(\mathcal{N})$). The net scaling of the algorithm is 
$$\mathcal{O}\left(\mathrm{log}(\mathcal{N}) \times s(\mathcal{N}) \times \kappa(\mathcal{N})\times \mathrm{polylog}\left(\frac{s(\mathcal{N}) \times \kappa(\mathcal{N})}{\epsilon(\mathcal{N})}\right)\right).$$ 
Thus, this approach is near-optimal in scaling of both condition number and precision. 
\end{itemize}

Despite the promised speedup, these algorithms inherently suffer from the problem of preparing the input state, $\ket{b}$, from the classical data in $\vec{b}$. Unless $\vec{b}$ has specific structure that lowers the cost, preparing $\ket{b}$ can be exponentially costly in the number of gates, going as $\mathcal{O}(2^q)$, where $q$ is the number of qubits in the state register. Furthermore, loading an input matrix of size $2^q \times 2^q$ to the QPE module requires at most $O(2^{2q})$ operations, requiring resources that may grow exponentially with the input. We also note that reading the entire solution vector at the output of HHL or its variants using, for instance, quantum state tomography, is exponentially costly in its sample complexity, thus ruining the advantage offered by these algorithms. Instead, one needs to focus on extracting a \textit{feature} of the solution vector that is relevant to the problem considered \cite{Aaronson, morales}. 

We now list the other QLSs that we consider besides HHL and its variants: 

\begin{itemize}
\item 
\textbf{Phase randomisation method} \cite{prm}: 
The method is inspired by the adiabatic quantum computing model, and employs evolution randomisation. The method scales as 
$$\mathcal{O}\left(\frac{\mathrm{log}(\mathcal{N}) \times s(\mathcal{N}) \times \kappa(\mathcal{N}) \times \mathrm{log}(\kappa(\mathcal{N}))}{\epsilon(\mathcal{N})}\right),$$
and thus is near-optimal in $\kappa(\mathcal{N})$ without the need for the expensive VTAA procedure. However, the scaling in $\epsilon$ is still $1/\epsilon(\mathcal{N})$, as in HHL. 

\item \textbf{AQC(exp) method} \cite{toa}: 
The work demonstrates solving system of linear equations within the adiabatic quantum computing framework. This method too scales as 
$$\mathcal{O}\left(\mathrm{log}(N) \times s(\mathcal{N})  \times \kappa(\mathcal{N})\times \mathrm{polylog}\left(\frac{s(\mathcal{N}) \times \kappa(\mathcal{N})}{\epsilon(\mathcal{N})}\right) \right)$$ 
and thus is near-optimal in scaling of both condition number and precision, but unlike the CKS algorithm, the use of the expensive VTAA step is avoided. However, the minimum gap between the ground state and excited state of the time-dependent Hamiltonian that is varied adiabatically plays a strong role in determining the runtime of the algorithm \cite{toa}. Thus, only problems where the spectral gap shrinks polynomially in system size retain the speedup offered by this QLS. 

\item \textbf{Dream QLS: }
This is a fictitious QLS, which scales ideally in all of its variables. We assume that such a solver would go in its runtime complexity as 
$$\mathcal{O}\left(\mathrm{log}(\mathcal{N}) \times \sqrt{s(\mathcal{N})} \times \kappa(\mathcal{N}) \times \mathrm{log}\left(\frac{1}{\epsilon(\mathcal{N})}\right) \right).$$
The dream QLS serves as a benchmark to how much of an advantage we can get in the best case scenario, and thus subsumes all other QLSs that we do not consider. We assume in defining this solver that one cannot go below $\kappa(\mathcal{N})$ \cite{Harrow2009QuantumEquations} , $\sqrt{s(\mathcal{N})}$ \cite{mori2026sparsity}, and $\mathrm{log}(1/\epsilon(\mathcal{N}))$ \cite{morales} in its complexity expression. 
\end{itemize} 

\begin{table*}[t]
\centering
\caption{Table presenting the runtime complexities of the QLSs that we consider in this work for our survey. We note that all of the QLSs we consider are fault-tolerant era algorithms, and all of them offer a $\mathrm{log}(\mathcal{N})$ factor in their runtime as opposed to an efficient classical linear solver (CLS), which offers $\mathcal{N}$. } 
\label{tab:complexities} 
\begin{tabular}{c@{\hspace{40pt}}c} 
\hline \hline 
Algorithm & Runtime complexity \\
\hline
CLS &$\mathcal{O}\textbf{\bigg(}\mathcal{N}s(\mathcal{N}) \sqrt{\kappa(\mathcal{N})}\mathrm{log}\bigg(\frac{1}{\epsilon(\mathcal{N})}\bigg)\textbf{\bigg)}$ \\ 
HHL \cite{Harrow2009QuantumEquations}&$\mathcal{O}\textbf{\bigg(}\mathrm{log}(\mathcal{N})s(\mathcal{N})^2 \kappa(\mathcal{N})^3\frac{1}{\epsilon(\mathcal{N})}\textbf{\bigg)}$ \\ 
HHL-AA \cite{Harrow2009QuantumEquations, Brassard}&$\mathcal{O}\textbf{\bigg(}\mathrm{log}(\mathcal{N})s(\mathcal{N})^2 \kappa(\mathcal{N})^2\frac{1}{\epsilon(\mathcal{N})}\textbf{\bigg)}$ \\ 
HHL-VTAA \cite{ambainis}&$\mathcal{O}\textbf{\bigg(}\mathrm{log}(\mathcal{N})s(\mathcal{N})^2 \kappa(\mathcal{N}) \mathrm{log}^3\bigg(\frac{\kappa(\mathcal{N})}{\epsilon(\mathcal{N})}\bigg)\frac{1}{\epsilon(\mathcal{N})^3} \mathrm{log}^2 \bigg(\frac{1}{\epsilon(\mathcal{N})} \bigg) \textbf{\bigg)}$ \\ 
Psi-HHL \cite{psihhl}&$\mathcal{O}\textbf{\bigg(}\mathrm{log}(\mathcal{N})s(\mathcal{N})^2 \kappa(\mathcal{N})\frac{1}{\epsilon(\mathcal{N})}\textbf{\bigg)}$ \\ 
Phase randomisation method \cite{prm}&$\mathcal{O}\textbf{\bigg(}\mathrm{log}(\mathcal{N})s(\mathcal{N}) \kappa(\mathcal{N}) \mathrm{log}(\kappa(\mathcal{N}))\frac{1}{\epsilon(\mathcal{N})}\textbf{\bigg)}$ \\ 
CKS algorithm \cite{Childs} &$\mathcal{O}\textbf{\bigg(}\mathrm{log}(\mathcal{N})s (\mathcal{N}) \kappa(\mathcal{N}) \mathrm{polylog}\bigg(s(\mathcal{N}) \kappa(\mathcal{N}) \frac{1}{\epsilon(\mathcal{N})}\bigg)\textbf{\bigg)}$ \\ 
AQC(exp) method \cite{toa}&  $\mathcal{O}\textbf{\bigg(}\mathrm{log}(\mathcal{N})s (\mathcal{N}) \kappa(\mathcal{N}) \mathrm{polylog}\bigg(s(\mathcal{N}) \kappa(\mathcal{N}) \frac{1}{\epsilon(\mathcal{N})}\bigg)\textbf{\bigg)}$\\ 
Dream QLS&$\mathcal{O}\bigg(\mathrm{log}(\mathcal{N})\sqrt{s(\mathcal{N})}\kappa(\mathcal{N})\mathrm{log}\bigg(\frac{1}{\epsilon(\mathcal{N})}\bigg)\bigg)$\\ 
\hline \hline 
\end{tabular} 
\end{table*} 

Table \ref{tab:complexities} presents the expression for the runtime complexities of each of the aforementioned quantum algorithms in its second column. We assume throughout hereafter that $\epsilon^{-1}\sim \mathrm{log}(\mathcal{N})$, so that the desired polylog runtime complexity of a QLS is retained. We also drop the $\mathcal{O}$ hereafter when we discuss complexity expressions, for brevity. 

We now introduce a useful quantity, the runtime complexity ratio, $R(\mathcal{N})$. 

\begin{definition} \label{def:R}
$$R(\mathcal{N})=\frac{t_{\mathrm{CLS}}(\mathcal{N})}{t_{\mathrm{QLS}}(\mathcal{N})}=\frac{\mathcal{N} \mathrm{log}(\mathrm{log}(\mathcal{N}))}{\mathrm{log}^2(\mathcal{N})}\frac{s(\mathcal{N}) \sqrt{\kappa(\mathcal{N})}}{{f}(s(\mathcal{N})){g}(\kappa(\mathcal{N}))}$$ 
is defined as the ratio of runtime complexities of CLS to a QLS. 
\end{definition} 

The specific functional forms for ${f}(s(\mathcal{N}))$ and ${g}(\kappa(\mathcal{N}))$ depend on the QLS chosen. On the other hand, $\kappa(\mathcal{N})$ and $s(\mathcal{N})$ depend on the choice of application, and they determine the degree of quantum advantage that one obtains for an application. 

We now make some general observations. When $R(\mathcal{N})>1$, a QLS under consideration performs better than CLS. That is, when 
$$\frac{f(s(\mathcal{N}))g(\kappa(\mathcal{N}))}{s(\mathcal{N}) \sqrt{\kappa(\mathcal{N})}}<\frac{\mathcal{N} \mathrm{log}(\mathrm{log}(\mathcal{N}))}{\mathrm{log}^2(\mathcal{N})},$$ 
the QLS outperforms CLS. For example, in the case of HHL, the condition can be evaluated to be 
$$\kappa(\mathcal{N})^{5/2} s(\mathcal{N})<\frac{\mathcal{N}\mathrm{log}(\mathrm{log}(\mathcal{N}))}{\mathrm{log}^2(\mathcal{N})}.$$

That is, if an $(\mathcal{N} \times \mathcal{N})$ matrix $A$ is ill-conditioned enough and/or is sufficiently dense such that the product $\kappa(\mathcal{N})^{5/2} s(\mathcal{N})$ is equal to or greater than the right hand side, then we cannot expect an advantage from HHL. If we go to an improved variant such as Psi-HHL, the product $\kappa(\mathcal{N})^{1/2} s(\mathcal{N})$ can be less than  the right hand side for a matrix $A$ that is relatively denser and/or relatively more ill-conditioned. The condition for the phase randomisation method is particularly interesting, as its sparsity scaling is the same as that of CLS, and thus it is only the condition number scaling that matters when we compare the two approaches. Similarly, for the dream QLS case, as long as $\kappa$ and $s$ were to scale the same way, independent of their behaviours with $\mathcal{N}$, the runtime ratio is $\mathcal{N}/\mathrm{log}(\mathcal{N})$ thus resulting in an exponential separation in the runtimes. We also note that for the cases of the CKS and the AQC(exp) algorithms, it is not possible to separate out $\kappa(\mathcal{N})$ and $s(\mathcal{N})$ to one side due to the functional form for their complexities. Lastly, it is worth adding that since the complexity expressions contain a polylog in them, we will, for the purposes of our numerical analyses in the subsequent sections, consider three cases for the CKS and AQC(exp) algorithms for polylog: log, log$^2$, and log$^3$ behaviours. In the subsequent figures as well as main text, we use a shorthand notation: for example, the AQC(exp) algorithm with the polylog chosen to be log is denoted as AQC(1), and so on. 

We now list the expressions for $R(\mathcal{N})$ for HHL, CKS(1)/AQC(1), and the dream QLS: 

\begin{enumerate}
    \item HHL: 
    \begin{equation}
        R(\mathcal{N})_{1}=\frac{\mathcal{N}}{\mathrm{log}(\mathcal{N})} \frac{\mathrm{log(log(\mathcal{N})})}{\mathrm{log(}\mathcal{N})\kappa^{5/2}s}, 
        \label{Rhhl}
    \end{equation}
    \item : CKS(1)/AQC(1): 
    \begin{equation}
        R(\mathcal{N})_2=\frac{\mathcal{N}}{\mathrm{log}(\mathcal{N})} \frac{\mathrm{log(log(}\mathcal{N}))}{\sqrt{\kappa} \mathrm{log(log(}\mathcal{N})s\kappa)}, 
        \label{Rcks}
    \end{equation} and 
    \item Dream QLS: 
    \begin{equation}
        R(\mathcal{N})_3=\frac{\mathcal{N}}{\mathrm{log}(\mathcal{N})} \frac{\sqrt{s}}{\sqrt{\kappa}}. 
        \label{Rdqls}
    \end{equation}
\end{enumerate}

We remark here that $\frac{\mathcal{N}}{\mathrm{log}(\mathcal{N})}$ is an exponential separation in $\mathrm{log}(\mathcal{N})$. To retain the exponential behaviour in $R(\mathcal{N})_{1}$ under the assumption that $\epsilon^{-1} \sim \mathrm{log}(\mathcal{N})$, the term $\kappa^{5/2}s$ must strictly grow less than $\mathcal{N}$. With $R(\mathcal{N})_{2}$, it is sufficient for $\kappa$ to grow less than $\mathcal{N}^2$. However, with $R(\mathcal{N})_3$, to observe exponential advantage, the growth of $\kappa/s$ should be less than $\mathcal{N}^2$. We also note at this juncture that our analysis is purely based on graph theoretic considerations and neglects implementation overheads such as the number of physical qubits given a quantum error correcting code, cost of magic state distillation, etc. In view of the possibility of such overheads overriding advantages less than quadratic (for example, see \cite{prx2021}), we consider, for example, logarithmic speedup, as no advantage in this study.

\section{Networks-based linear system problems} \label{sec:cn-lsp}

\subsection{Graphs and their associated matrices} \label{sec:graphtheory}

Combinatorial graphs build up the mathematical foundations for complex networks. Here, we only define a graph, the graph Laplacian matrix, and the graph incidence matrix. The reader can refer to Section S.2 of the Supplemental Material for other definitions in graph theory that are essential for this work. 

\begin{definition} \label{def:graph}
A combinatorial graph, which we shall hereafter refer to as a graph $G = (V(G), E(G))$ for brevity, is a set of vertices, $V(G)$, and a set of edges, $E(G) \subset V(G) \times V(G)$. 
\end{definition} 

In this article, all the graphs have finite number of vertices. The vertex set of a graph $G$ can be written as $V(G) = \{v_i: i = 1, 2, \cdots, N\}$. Throughout the article, $N$ and $M$ denote the number of vertices and edges in the graph $G$, respectively. In the case of undirected graphs, with all of the edges being undirected, we consider only simple graphs. We define the Laplacian matrix of an undirected graph as follows: 

\begin{definition} \label{def:lapmat}
The Laplacian matrix $L(G)$ of a simple graph, or a weighted graph $G$, with $N$ vertices $v_i: i = 1, 2, \cdots, N$ is an $N \times N$ matrix $L(G) = D(G) - Q(G)$, where $D(G)$ is the degree matrix and $Q(G)$ is the adjacency matrix. 
\end{definition} 

A directed edge from $v_i$ to $v_j$ is denoted by $\vec{e}= \overrightarrow{(v_i, v_j)}$. Here $v_i$ and $v_j$ are the initial and terminal vertex of directed edge $\vec{e}$. For a directed graph, with all the edges being directed, we define source and sink vertices as follows: if the edges incident the vertex are all outgoing, then it considered as a source vertex. On the contrary, if all the edges incident on the vertex are incoming, it is a sink vertex. For the purposes of our analysis, we study the incidence matrix of directed graphs, which is defined below: 

\begin{definition} \label{def:incmat}
Let $\overrightarrow{G}$ be a directed graph with vertices $v_i: i = 1, 2, \cdots, N$ and edges $e_j: j = 1, 2, \cdots, M$. The vertex-edge incidence matrix $B(\overrightarrow{G}) = (b_{ij})_{N \times M}$ is defined by
\begin{eqnarray}
b_{ij} = 
\begin{cases}
-1 & \text{if $v_i$ is the initial vertex of $e_j$};\\ 
+1 & \text{if $v_i$ is the terminal vertex of $e_j$}; \\ 
0 & \text{otherwise}.
\end{cases}
\end{eqnarray}
\end{definition}

\subsection{Systems of linear equations arising from graphs}

We consider two systems of linear equations associated to the graphs due to their physical significance. We mention them below: 

\begin{enumerate}
    \item 
        \textbf{System of linear equations based on the Laplacian matrix: }\\
        Let $L$ be a Laplacian matrix of an undirected graph $G$ and $\vec{b}$ be a known vector. The problem statement is to find the vector $\vec{x}$, such that, 
        \begin{equation} \label{eqn:lx=b}
            L \vec{x} = \vec{b}. 
        \end{equation}
        As $G$ has $N$ vertices, the system size $\mathcal{N} = N$. $L$ is a square Hermitian matrix and it is positive semi-definite. At least one eigenvalue of $L$ is $0$. As $L^{-1}$ does not exist, we compute $\vec{x} = L^{+}\vec{b}$, where $L^{+}$ is the pseudo-inverse of the matrix $L$ \cite{penrose1956best}.
    \item 
        \textbf{System of linear equations based on the incidence matrix: }\\
        Let $B$ be an incidence matrix of a directed graph $G$ and $\vec{b}$ be a known vector. The problem statement is to find the vector $\vec{x}$, such that, 
        \begin{equation} \label{eqn:Bx=b}
            B \vec{x} = \vec{b}. 
        \end{equation}
        The sum of all the elements in any column of $B$ is $0$. Since $B$ is, in general, a non-Hermitian matrix, we transform the system of equations to 
\begin{equation} \label{eqn:HX=C}
   \begin{pmatrix}
        0 & B \\
        B^{\dagger} & 0
    \end{pmatrix}
{\begin{pmatrix}
        0 \\
        \vec{x}
    \end{pmatrix}}
    = {\begin{pmatrix}
        \vec{b} \\
        0
    \end{pmatrix}}.
\end{equation} 
    Note that the matrix $ \begin{pmatrix}
        0 & B \\
        B^{\dagger} & 0
    \end{pmatrix}$ is a Hermitian matrix of order $\mathcal{N} = N'= (M + N)$, which is the system size.
\end{enumerate} 

The system of linear equations arising from a graph Laplacian matrix is useful in electrical networks to find the effective resistance, $r_\text{eff}$, between two vertices \cite{spielman2010algorithms}. Example A of Section S.3 in the Supplemental Material illustrates it. Also, Example B of the same section discusses a system of linear equations related to the incidence matrix of a graph that is used to model the traffic flow congestion. 

\subsection{Categories of graph families} 

The complexity of solving a NLSP using a QLS depends on the growth of $\kappa(\mathcal{N})$ and $s(\mathcal{N})$. To understand the influence of these quantities on the runtime of a QLS, we numerically analyze them on different graph families. As HHL algorithm is the prototypical QLS, with the other QLSs considered in this study improving upon its runtime complexity, we classify the graph families into two categories based on whether HHL offers advantage or not. The categories are defined below: 

\begin{definition} \label{def:good} 
A good graph family is one for which the HHL algorithm gives advantage in its runtime relative to CLS. Here, advantage could be an exponential or a polynomial separation between the runtimes of HHL algorithm and CLS.
\end{definition}

\begin{definition} \label{def:bad}
A bad graph family is one for which the HHL algorithm yields no advantage with respect to CLS. 
\end{definition} 

If the HHL algorithm offers an advantage for some graph family, other QLSs considered in the study will too. However, there could be graph families where HHL does not offer an advantage, but improvements over it could. 

Before we proceed to the results, we summarize our typical workflow for each graph family that we pick. We start with a graph and a particular construction for generating a graph family. We plot $\kappa(\mathcal{N})$ of Laplacian or incidence matrix as well as $s(\mathcal{N})$ against system size, $\mathcal{N}$. We then fit the data points to appropriate functions, and plot $R(\mathcal{N})$ versus $\mathcal{N}$ to check if a curve goes above $1$ within our dataset. We recall that $R(\mathcal{N})>1$ is the regime of quantum advantage. We assume that our function fits for $\kappa(\mathcal{N})$ and $s(\mathcal{N})$ hold for much larger system sizes. We check if advantage can be expected from the graph family, that is, assess to which of the aforementioned set of categories our graph family belongs: good or bad. All the graph generation steps were carried out in NetworkX (version 3.4.2)~\cite{networkx}. For brevity, we hereafter call $\kappa(\mathcal{N})$, for example, as just $\kappa$, and so on. 

We add that while the interplay between $\kappa$ and $s$ scaling as functions of $\mathcal{N}$ decide the amount of advantage, the allowable values of $\mathcal{N}$ for a given graph family decides its classical hardness. Depending on the graph construction, $\mathcal{N}=f(n)$, where $n \in \mathbb{N}$. 

\newcommand{\col}[7]{#1 &#2 &#3 &#4 &#6  &#7 &#5 \\ } 
\begin{table*}[t] 
\caption{Table presenting our data for the Laplacian matrix-based good graph families that offer exponential advantage with HHL, CKS(1) (or AQC(1)) and dream QLS algorithms (see Eqns. \ref{Rhhl}, \ref{Rcks}, and \ref{Rdqls}). $t_{\mathrm{HHL}}$, $t_{\mathrm{CKS(1)}}$, $t_{\mathrm{AQC(1)}}$ and $t_{\mathrm{dream \ QLS}}$ and $t_{\mathrm{CLS}}$ denote the fitted runtime complexity expressions (excluding pre-factors) for HHL, CKS(1), AQC(1) and CLS algorithms. We set $1/\epsilon = \log{(N)}$ for our analysis. $ll(N)$ is a shorthand notation for $\log(\log(N))$. } 
\label{tab:yaya} 
\centering 
\hspace*{-1.3cm}
\begin{tabular}{|c|c|c|c|c|c|c|} 
\hline 
\col{Graph name}{$\kappa$}{$s$}{$t_{\mathrm{HHL}}$}{$t_{\mathrm{CLS}}$}{$t_{\mathrm{CKS(1)}}(\text{or }t_{\mathrm{AQC(1)}}) $}{$t_{\mathrm{dream\ QLS}}$}
\col{(Random (Yes/No))}{}{}{}{}{}{}
\col{}{}{}{}{}{}{}

\hline \hline 

\col{Hypercube graph (No)}{$\mathrm{log}(N)$}{$\mathrm{log}(N)$}{$\mathrm{log}^7(N)$}{$N \mathrm{log}^{3/2}(N) ll(N)$}{$\mathrm{log}^3(N)ll(N)$}{$\mathrm{log}^{5/2}(N)ll(N)$}
\hline 

\col{Modified Margulis-}{$\mathrm{log}^2(N)$}{$c$} {$\mathrm{log}^{8}(N)$}{$N \mathrm{log}(N)ll(N)$}{$\mathrm{log}^3Nll(N)$}{$\mathrm{log}^{3}(N)ll(N)$}
\col{Gabber-Galil (No)}{}{}{}{}{}{}
\hline 

\col{Sudoku (No)}{$\mathrm{log}^3(N)$}{$\sqrt{N}$} {$N\mathrm{log}^{11}(N)$}{$N^{3/2}\mathrm{log}^{3/2}(N) ll(N)$}{$\sqrt{N}\mathrm{log}^4(N)\mathrm{log}(\sqrt{N}\mathrm{log}^4(N))$}{$N^{1/4}\mathrm{log}^{4}(N)ll(N)$}
\hline 

\col{Barabási-Albert (Yes)}{$\mathrm{log}^3(N)$}{$\mathrm{log}^3(N)$} {$\mathrm{log}^{17}(N)$}{$N \mathrm{log}^{9/2}(N)ll(N)$}{$\mathrm{log}^7(N)ll(N)$,}{$\mathrm{log}^{11/2}(N)ll(N)$}
\hline 

\col{Newman-Watts}{$\mathrm{log}^3(N)$}{$\mathrm{log}^3(N)$} {$\mathrm{log}^{17}(N)$}{$N \mathrm{log}^{9/2}(N)ll(N)$}{$\mathrm{log}^7(N)ll(N)$}{$\mathrm{log}^{11/2}(N)ll(N)$}
\col{-Strogatz (Yes)}{}{}{}{}{}{}
\hline 

\col{Random regular}{$\mathrm{log}^3(N)$}{$c$} {$\mathrm{log}^{11}(N)$}{$N \mathrm{log}^{3/2}(N)ll(N)$}{$\mathrm{log}^4(N)ll(N)$}{$\mathrm{log}^{4}(N)ll(N)$}
\col{expander (Yes)}{}{}{}{}{}{}
\hline 

\col{Random regular (Yes)}{$\mathrm{log}^2(N)$}{$c$} {$\mathrm{log}^{8}(N)$}{$N \mathrm{log}(N)ll(N)$}{$\mathrm{log}^3(N)ll(N)$}{$\mathrm{log}^{3}(N)ll(N)$}
\hline

\hline \hline 
\end{tabular} 
\end{table*}

\section{Results from our survey of networks} \label{sec:survey} 

\newcommand{\columnn}[8]{#1 &#2 &#3 &#4 &#6  &#8 &#5 &#7  \\ }
\begin{table*}[t] 
\caption{Table presenting our data for Laplacian matrix-based bad graph families that offer no advantage with the HHL algorithm. We also give the advantage offered by these graph families with CKS(1) (or AQC(1)) (see Eq. \ref{Rcks}) and dream QLS algorithms (see Eq. \ref{Rdqls}). `Nil' represents no advantage and `Exp' refers to exponential advantage in the table. } 
\label{tab:yaya2} 
\centering 
\hspace*{-1.3cm}
\begin{tabular}{|c|c|c|c|c|c|c|c|}

\hline 
\columnn{Graph name}{$\kappa$}{$s$}{$t_{\mathrm{HHL}}$}{$t_{\mathrm{CLS}}$}{$t_{\mathrm{CKS(1)}}(\text{or }t_{\mathrm{AQC(1)}}), $}{Advantage with} {$t_{\mathrm{dream\ QLS}}$}
\columnn{(Random (Yes/No))}{}{}{}{}{}{HHL, }{}
\columnn{}{}{}{}{}{}{CKS(1)(or AQC(1)),}{}
\columnn{}{}{}{}{}{}{dream QLS}{}

\hline 
\hline 

\columnn{Grid 2d (No)}{$N^2$}{$c$} {$N^6 \mathrm{log}^{2}(N)$}{$N^2ll(N)$}{$N^2\mathrm{log}(N)\mathrm{log}(N^2\mathrm{log}(N))$}{Nil, Nil, Nil}{$N^2\mathrm{log}(N)ll(N)$}
\hline 

\columnn{Hexagonal }{$N^2$}{$c$} {$N^6 \mathrm{log}^{2}(N)$}{$N^2ll(N)$}{$N^2\mathrm{log}(N)\mathrm{log}(N^2\mathrm{log}(N))$,}{Nil, Nil, Nil}{$N^2\mathrm{log}(N)ll(N)$}
\columnn{lattice (No)}{}{}{}{}{}{}{}
\hline 

\columnn{Triangular}{$N$}{$c$} {$N^3 \mathrm{log}^{2}(N)$}{$N^{3/2}ll(N)$}{$N\mathrm{log}(N)\mathrm{log}(N\mathrm{log}(N))$}{Nil, Exp, Exp}{$N\mathrm{log}(N)ll(N)$}
\columnn{lattice (No)}{}{}{}{}{}{}{}
\hline 

\columnn{Complete (No)}{$c$}{$N$} {$N^2 \mathrm{log}^{2}(N)$}{$N^{2}ll(N)$}{$N\mathrm{log}(N)\mathrm{log}(N\mathrm{log}(N))$}{Nil, Exp, Exp}{$\sqrt{N}\mathrm{log}(N)ll(N)$}
\hline 

\columnn{Turan (No)}{$\mathrm{log}^2(N)$}{$N$} {$N^2 \mathrm{log}^{8}(N)$}{$N^{2}\mathrm{log}(N)ll(N)$}{$N\mathrm{log}^3(N)\mathrm{log}(N\mathrm{log}^3(N))$}{Nil, Exp, Exp }{$\sqrt{N}\mathrm{log}^3(N)ll(N)$}
\hline 

\columnn{Gaussian random}{$\mathrm{log}^3(N)$}{$N$} {$N^2 \mathrm{log}^{11}(N)$}{$N^{2}\mathrm{log}^{3/2}(N)ll(N)$}{$N\mathrm{log}^4(N)\mathrm{log}(N\mathrm{log}^4(N))$}{Nil, Exp, Exp}{$\sqrt{N}\mathrm{log}^4(N)ll(N)$}
\columnn{partition (Yes)}{}{}{}{}{}{}{}
\hline 

\columnn{Geographical}{$\mathrm{log}^3(N)$}{$N$} {$N^2 \mathrm{log}^{11}(N)$}{$N^{2}\mathrm{log}^{3/2}(N)ll(N)$}{$N\mathrm{log}^4(N)\mathrm{log}(N\mathrm{log}^4(N))$}{Nil, Exp, Exp}{$\sqrt{N}\mathrm{log}^4(N)ll(N)$}
\columnn{threshold (Yes)}{}{}{}{}{}{}{}
\hline 

\columnn{Soft random }{$\mathrm{log}^3(N)$}{$N$} {$N^2 \mathrm{log}^{11}(N)$}{$N^{2}\mathrm{log}^{3/2}(N)ll(N)$}{$N\mathrm{log}^4(N)\mathrm{log}(N\mathrm{log}^4(N))$}{Nil, Exp, Exp}{$\sqrt{N}\mathrm{log}^4(N)ll(N)$}
\columnn{geometric (Yes)}{}{}{}{}{}{}{}
\hline 

\columnn{Thresholded random}{$\mathrm{log}^3(N)$}{$N$} {$N^2 \mathrm{log}^{11}(N)$}{$N^{2}\mathrm{log}^{3/2}(N)ll(N)$}{$N\mathrm{log}^4(N)\mathrm{log}(N\mathrm{log}^4(N))$}{Nil, Exp, Exp}{$\sqrt{N}\mathrm{log}^4(N)ll(N)$}
\columnn{geometric (Yes)}{}{}{}{}{}{}{}
\hline 

\columnn{Planted }{$\mathrm{log}^3(N)$}{$N$} {$N^2 \mathrm{log}^{11}(N)$}{$N^{2}\mathrm{log}^{3/2}(N)ll(N)$}{$N\mathrm{log}^4(N)\mathrm{log}(N\mathrm{log}^4(N))$}{Nil, Exp, Exp}{$\sqrt{N}\mathrm{log}^4(N)ll(N)$}
\columnn{partition (Yes)}{}{}{}{}{}{}{}
\hline 

\columnn{Random}{$\mathrm{log}^3(N)$}{$N$} {$N^2 \mathrm{log}^{11}(N)$}{$N^{2}\mathrm{log}^{3/2}(N)ll(N)$}{$N\mathrm{log}^4(N)\mathrm{log}(N\mathrm{log}^4(N))$}{Nil, Exp, Exp}{$\sqrt{N}\mathrm{log}^4(N)ll(N)$}
\columnn{geometric (Yes)}{}{}{}{}{}{}{}
\hline 

\columnn{Uniform random}{$c$}{$N$} {$N^2 \mathrm{log}^{2}(N)$}{$N^{2}ll(N)$}{$N\mathrm{log}(N)\mathrm{log}(N\mathrm{log}(N))$}{Nil, Exp, Exp}{$\sqrt{N}\mathrm{log}(N)ll(N)$}
\columnn{intersection (Yes)}{}{}{}{}{}{}{}
\hline 

\columnn{H$_{k,n}$ Harary (No)}{$N^2$}{$c$} {$N^6 \mathrm{log}^{2}(N)$}{$N^{2}ll(N)$}{$N^2\mathrm{log}(N)\mathrm{log}(N^2\mathrm{log}(N))$}{Nil, Nil, Nil}{$N^2\mathrm{log}(N)ll(N)$}
\hline 

\columnn{H$_{n,m}$ Harary (No)}{$N^2$}{$c$} {$N^6 \mathrm{log}^{2}(N)$}{$N^{2}ll(N)$}{$N^2\mathrm{log}(N)\mathrm{log}(N^2\mathrm{log}(N))$}{Nil, Nil, Nil}{$N^2\mathrm{log}(N)ll(N)$}
\hline 

\columnn{Circular ladder (No)}{$N^2$}{$c$} {$N^6 \mathrm{log}^{2}(N)$}{$N^{2}ll(N)$}{$N^2\mathrm{log}(N)\mathrm{log}(N^2\mathrm{log}(N))$}{Nil, Nil, Nil}{$N^2\mathrm{log}(N)ll(N)$}
\hline 

\columnn{Ladder (No)}{$N^2$}{$c$} {$N^6 \mathrm{log}^{2}(N)$}{$N^{2}ll(N)$}{$N^2\mathrm{log}(N)\mathrm{log}(N^2\mathrm{log}(N))$}{Nil, Nil, Nil}{$N^2\mathrm{log}(N)ll(N)$}
\hline 

\columnn{Ring of cliques (No)}{$N^2$}{$c$} {$N^6 \mathrm{log}^{2}(N)$}{$N^{2}ll(N)$}{$N^2\mathrm{log}(N)\mathrm{log}(N^2\mathrm{log}(N))$}{Nil, Nil, Nil}{$N^2\mathrm{log}(N)ll(N)$}
\hline 

\columnn{Balanced binary}{$N$}{$c$} {$N^3 \mathrm{log}^{2}(N)$}{$N^{3/2}ll(N)$}{$N\mathrm{log}(N)\mathrm{log}(N\mathrm{log}(N))$}{Nil, Exp, Exp}{$N\mathrm{log}(N)ll(N)$}
\columnn{tree (No)}{}{}{}{}{}{}{}
\hline 

\columnn{Balanced ternary}{$N$}{$c$} {$N^3 \mathrm{log}^{2}(N)$}{$N^{3/2}ll(N)$}{$N\mathrm{log}(N)\mathrm{log}(N\mathrm{log}(N))$}{Nil, Exp, Exp}{$N\mathrm{log}(N)ll(N)$}
\columnn{tree (No)}{}{}{}{}{}{}{}
\hline 

\columnn{Binomial tree (No)}{$N$}{$\mathrm{log}(N)$} {$N^3 \mathrm{log}^{4}(N)$}{$N^{3/2} \mathrm{log}(N)ll(N)$}{$N\mathrm{log}^2(N)\mathrm{log}(N\mathrm{log}^2(N))$}{Nil, Exp, Exp}{$N\mathrm{log}^{3/2}(N)ll(N)$}
\hline 

\columnn{Grid 2d }{$N$}{$c$} {$N^3 \mathrm{log}^{2}(N)$}{$N^{3/2}ll(N)$}{$N\mathrm{log}(N)\mathrm{log}(N\mathrm{log}(N))$}{Nil, Exp, Exp}{$N\mathrm{log}(N)ll(N)$}
\columnn{graph ($r=c$) (No)}{}{}{}{}{}{}{}
\hline 

\columnn{Random }{$N^2$}{$N^2$} {$N^{10} \mathrm{log}^{2}(N)$}{$N^{4}ll(N)$}{$N^4\mathrm{log}(N)\mathrm{log}(N^4\mathrm{log}(N))$}{Nil, Nil, Exp}{$N^3\mathrm{log}(N)ll(N)$}
\columnn{lobster (Yes)}{}{}{}{}{}{}{}
\hline 

\columnn{G$_{n,p}$ random (Yes)}{$\mathrm{log}^2(N)$}{$N$} {$N^{2} \mathrm{log}^{8}(N)$}{$N^{2}\mathrm{log}(N)ll(N)$}{$N\mathrm{log}^3(N)\mathrm{log}(N\mathrm{log}^3(N))$}{Nil, Exp, Exp
}{$\sqrt{N}\mathrm{log}^3(N)ll(N)$}

\hline \hline 
\end{tabular} 
\end{table*} 

In this section, we present the results for Laplacian (Sections \ref{sec:LG}) and incidence matrix-based graph families (Section \ref{sec:IG}). We note that our numerical study on graph families involves a number of assumptions, which in turn set the scope of our study. Section S.4 of the Supplemental Material extensively discusses the assumptions and scope of this study, and the discussion is supported by error analysis on certain graph families. Section S.5 of the Supplemental Material provides the details on all of the input parameters along with the construction for the graph families considered. We try to ensure diversity in terms of $\kappa$, $s$ and system size growth when we pick the graph families. Hence, we have considered non-random graphs, random graphs, trees, and expanders in our graph families. We note at this juncture that to the best of our knowledge, among the 50 graph families that we pick, only the hypercube graph and the complete graph have known behaviours for their $\kappa$ scaling. In the subsequent sub-sections, we provide for each graph family a figure that contains two sub-figures: 

\begin{itemize} 
\item The first sub-figure shows $\kappa$ and $s$ versus $\mathcal{N}$ with suitably fitted curves, along with inset visualizations of two representative graphs from the family, one for a small $\mathcal{N}$ and the other for a slightly larger $\mathcal{N}$. 
\item The second sub-figure presents the runtime ratio, $R(\mathcal{N})$, between each of the considered QLSs (listed in Table \ref{tab:complexities} and the CLS. It also highlights the crossover at $R(\mathcal{N})=1$, above which the QLS outperforms CLS, even within the limited range of $\mathcal{N}$ values considered in our datasets. 
\end{itemize} 

Furthermore, we present the runtime expressions for HHL, CKS(1) (whose runtime complexity is the same as AQC(1)), dream QLS and CLS algorithms for all graph families considered in our study in Tables \ref{tab:yaya}, \ref{tab:yaya2}, \ref{tab:incidence_complexity} and \ref{tab:incidence2}. In these tables, we also compare runtime advantage from HHL which is the prototypical QLS, CKS(1) (or AQC(1)) which has near optimal runtime complexity among the QLSs considered in the study, and the dream QLS which is optimal in its runtime complexity.

\begin{figure*}
\centering 
\includegraphics[width=13.5cm]{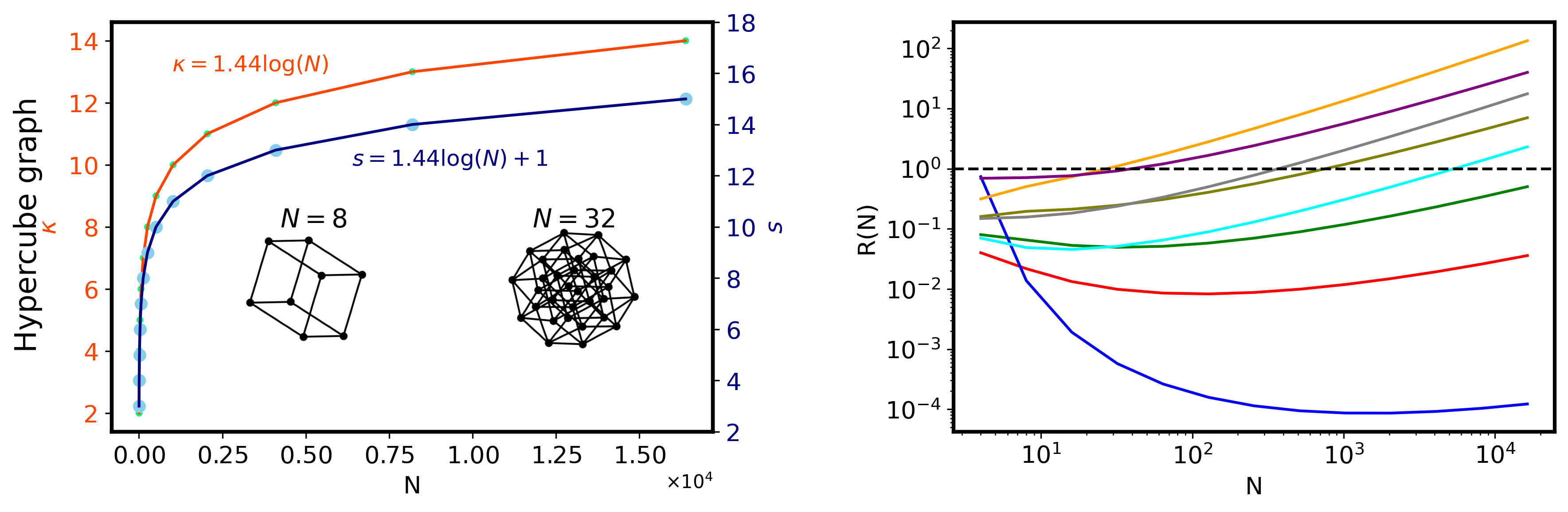} \\
\vspace{0.5em}
\makebox[0.45\textwidth]{(a)}
\makebox[0.36\textwidth]{(b)} \\
\includegraphics[width=13.5cm]{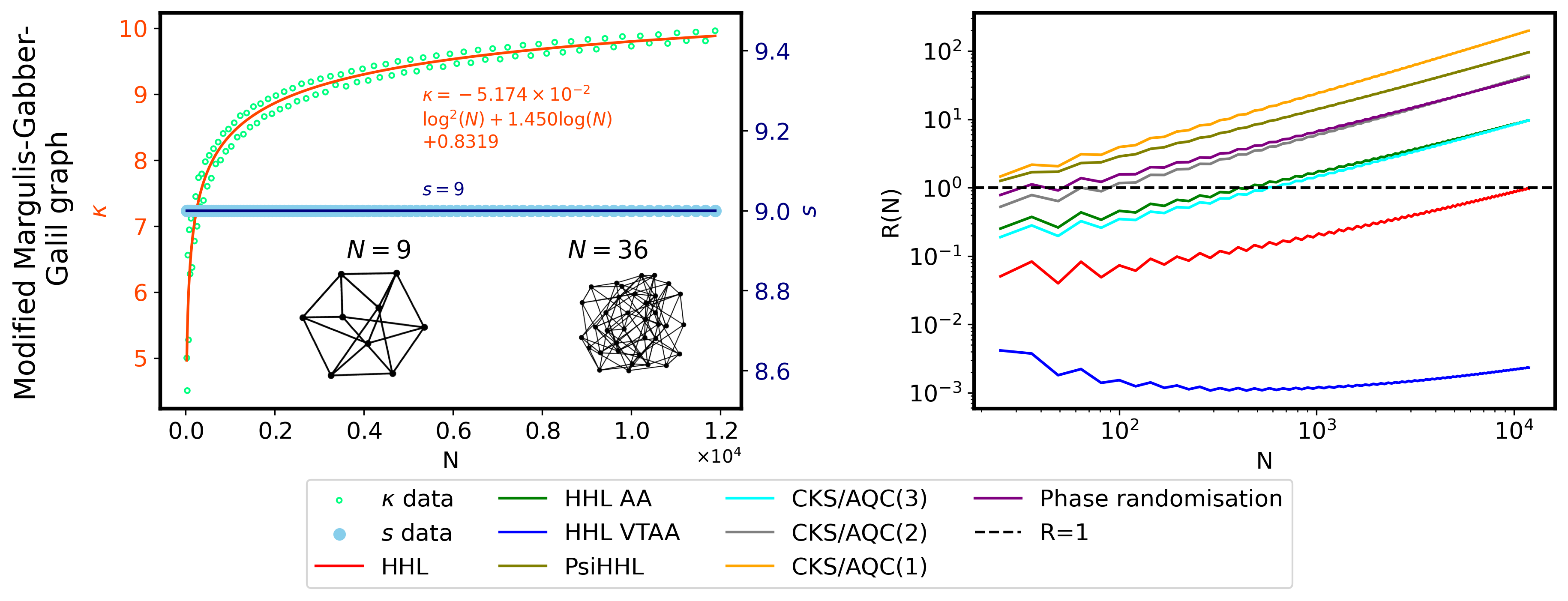} \\
\vspace{0.5em}
\makebox[0.45\textwidth]{(c)}
\makebox[0.36\textwidth]{(d)} \\

\caption{The figure presents two instances of good (advantage with HHL) graph families for Laplacian matrix-based systems. The first sub-figure of each panel shows condition number, $\kappa$, and sparsity, $s$, versus the number of vertices, $N$, and second sub-figure displays the runtime ratio, $R(N)$, versus $N$ for every QLS considered in the study. }\label{fig:bestgraph}
\end{figure*} 

\begin{figure*}
\centering 

\includegraphics[width=13.5cm]
{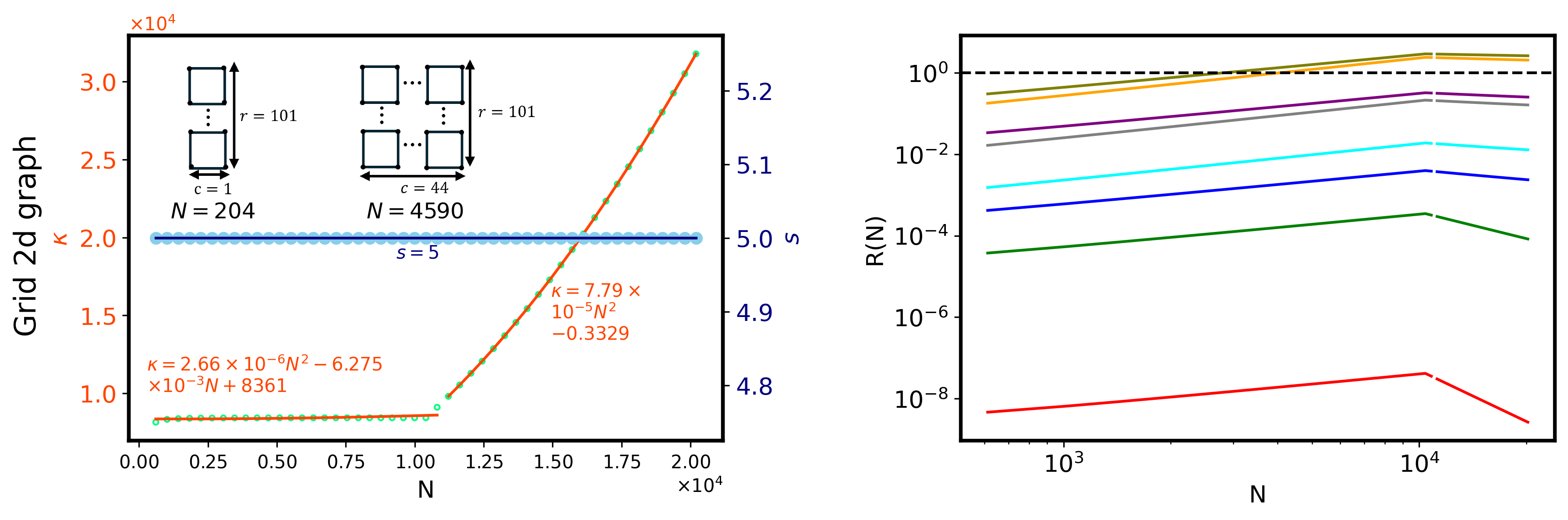}\\
\vspace{0.5em}
\makebox[0.45\textwidth]{(a)}
\makebox[0.36\textwidth]{(b)} \\

\includegraphics[width=13.5cm]
{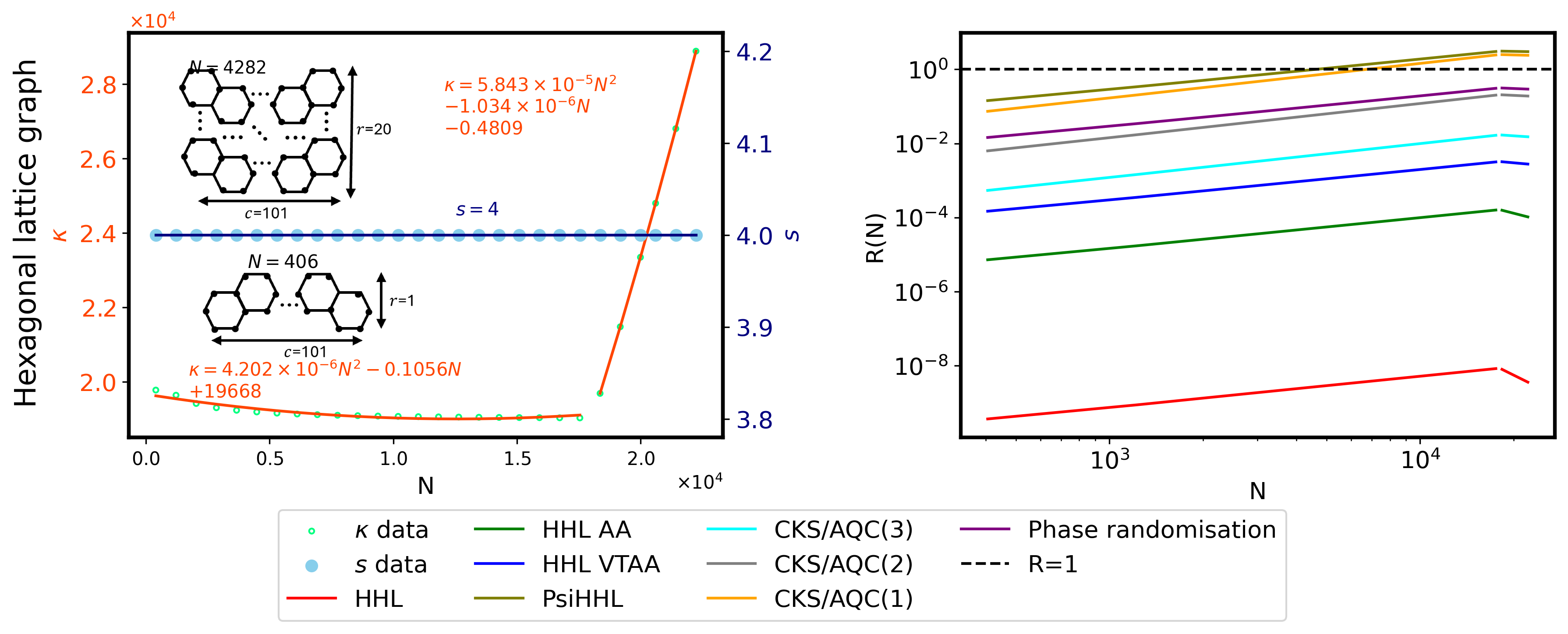}
\vspace{0.5em}
\makebox[0.45\textwidth]{(c)}
\makebox[0.36\textwidth]{(d)} \\ 

\caption{The figure illustrates two instances of bad (no advantage with HHL) graph families for Laplacian matrix-based systems. The first sub-figure of each panel represents growth of condition number, $\kappa$, and sparsity, $s$, with number of vertices $N$ while the second sub-figure shows the runtime ratio $R(N)$ with $N$ for every QLS considered in the study. }\label{fig:badgraph_lap}
\end{figure*}

\begin{figure*}[t] 
\centering 
\includegraphics[width=14cm]{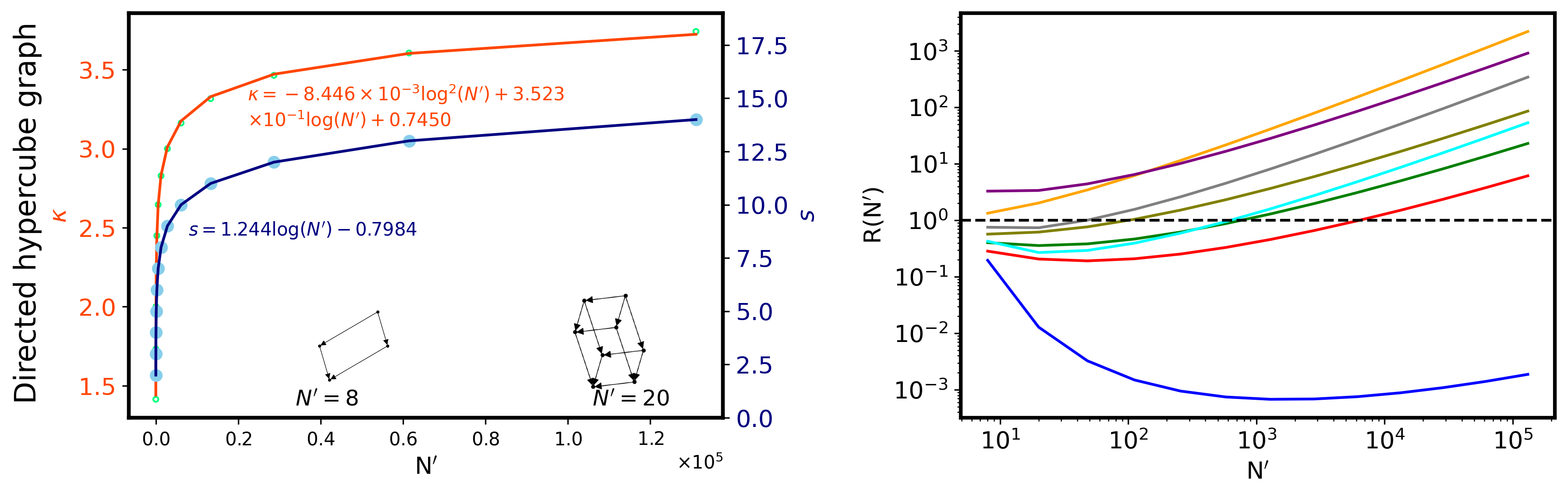} 
\vspace{0.5em}
\makebox[0.45\textwidth]{(a)}
\makebox[0.36\textwidth]{(b)} \\

\includegraphics[width=14cm]{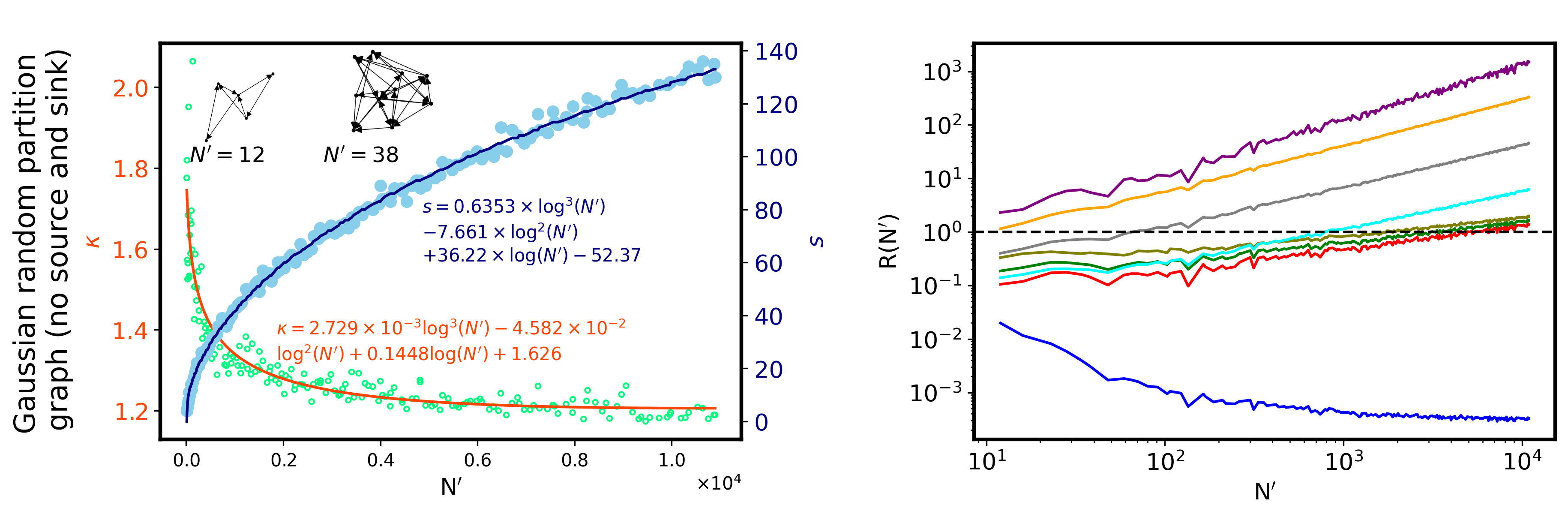} \\ 
\vspace{0.5em}
\makebox[0.45\textwidth]{(c)}
\makebox[0.36\textwidth]{(d)} \\

\includegraphics[width=14cm]{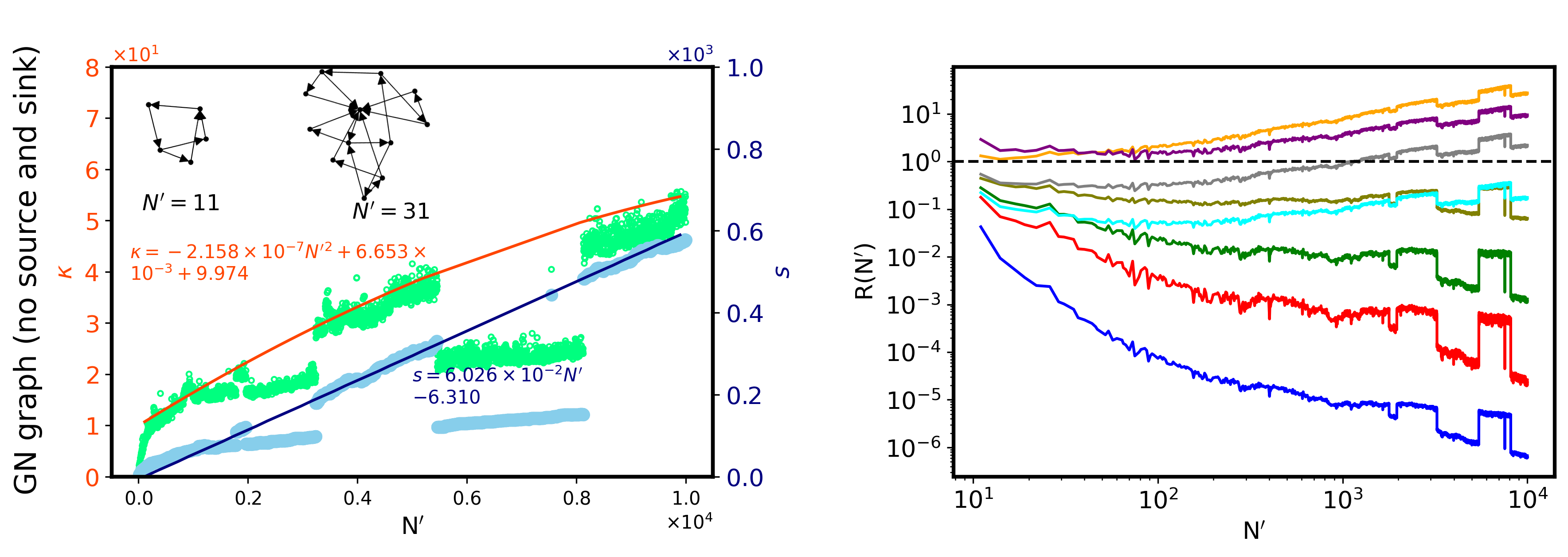} \\
\vspace{0.5em}
\makebox[0.45\textwidth]{(e)}
\makebox[0.36\textwidth]{(f)} \\

\includegraphics[width=14cm]{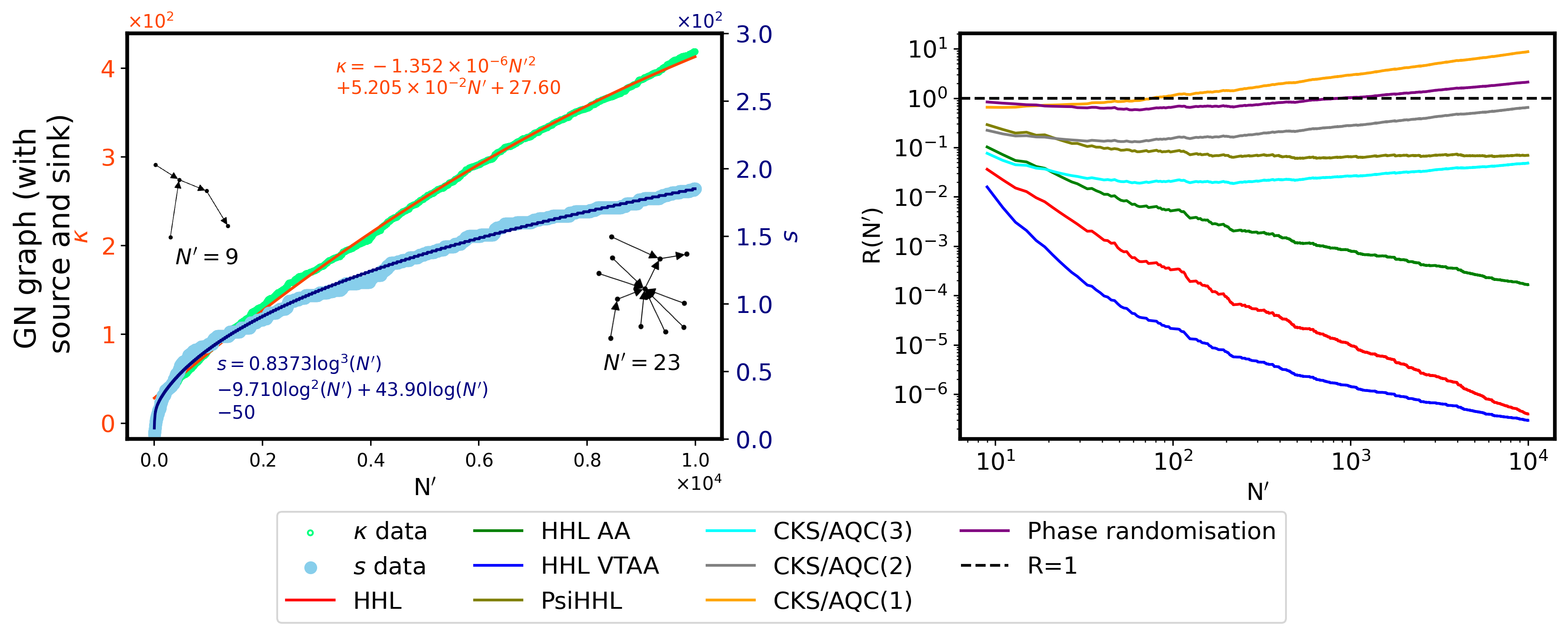} \\
\vspace{0.5em}
\makebox[0.45\textwidth]{(g)}
\makebox[0.36\textwidth]{(h)} \\

\caption{The figure presents two instances each for good and bad categories of graph families for incidence matrix-based systems. Sub-figures (a)-(d) belong to good graph category, whereas (e)-(h) belong to bad graph category.  }\label{fig:incidence_bestgraph}
\end{figure*} 

\subsection{Graph families for Laplacian matrix-based systems} \label{sec:LG} 

We survey 30 graph families and their Laplacian matrices as a part of our study, with the details provided in Tables \ref{tab:yaya} and \ref{tab:yaya2}. We find 7 good and 23 bad graph families. We present two instances each for good and bad graph categories in Figs. \ref{fig:bestgraph} and \ref{fig:badgraph_lap}. We refer the reader to Figs. S.8 and S.9 of the Supplemental Material for the rest of the good and bad graph families respectively. 

\subsubsection{Good graph families} 

As Table \ref{tab:yaya} indicates, all the graph families in this category show at most polylogarithmic growth in $\kappa$ and $s$ (with the exception of Sudoku graph, whose sparsity function grows as $\sqrt{N}$), leading to an exponential advantage with HHL. The exponential advantage is retained with CKS(1), AQC(1) and dream QLS algorithms. Furthermore, with CKS(1) and AQC(1) algorithms, good graph families offer advantage well before the other algorithms do in $N$. This is witnessed in the plot of $R(N)$ vs $N$, where the regime of advantage starts when the curve crosses $R({N})=1$ line. Among the good graph families that we identify, the admissible $N$ values in the hypercube graph family grow as $N=2^n$, and thus the problem is classically hard. As mentioned earlier, our analyses does not account for practical overheads in estimating quantum advantage. Thus, our numerical results for the crossover point cannot be taken quite literally but rather as qualitative runtime performance indicators even at small system sizes. We adopt this perspective in discussing the crossover points for the subsequent graph families. 

\subsubsection{Bad graph families} 

From Table \ref{tab:yaya2}, we see that either $\kappa$ or $s$ grow at least polynomially resulting in no advantage with the HHL algorithm. It is interesting to note that with CKS(1) and AQC(1) algorithms, 15 graph families graduate to good graph family and yield exponential advantage, while with dream QLS, 16 graph families offer exponential advantage. 

Of special interest among the families in this category is the transition that we observe in the growth of condition number for grid 2d, hexagonal lattice and triangular lattice graph families. For all the three cases, $\kappa$ grows slowly up to a certain system size and beyond it, grows polynomially. Since we are concerned with $\kappa$ scaling at large system sizes to calculate runtime ratios, we consider the functional form of $\kappa$ to be polynomial in $N$. Although sparsity remains constant for these graph families, they provide no advantage with HHL because of their polynomial growth in $\kappa$. 

\newcommand{\coll}[7]{#1 &#2 &#3 &#4 &#5  &#6 &#7 \\ } 
\begin{table*}[t] 
\caption{Table presenting our data for the incidence matrix-based good graph families that provide exponential speedup with HHL, CKS(1) (or AQC(1)) and the dream QLS. Here, we set $1/\epsilon=\log{(N')}$ for the runtime expressions of HHL($t_{\mathrm{HHL}}$), CKS(1)($t_{\mathrm{CKS(1)}}$), AQC(1)($t_{\mathrm{AQC(1)}}$), dream QLS($t_{\mathrm{dream \ QLS}}$) and CLS($t_{\mathrm{CLS}}$) algorithms. $\log(\log(N'))$ is expressed as $ll(N')$ in the table. } 
\label{tab:incidence_complexity} 
\centering 
\hspace*{-1.3cm}
\begin{tabular}{|c|c|c|c|c|c|c|} 
\hline 
\coll{Graph name}{$\kappa$}{$s$}{$t_{\mathrm{HHL}}$}{$t_{\mathrm{CKS(1)}}(\text{or }t_{\mathrm{AQC(1)}}) $}
{$t_{\mathrm{dream\ solver}}$}
{$t_{\mathrm{CLS}}$}
\coll{(Source and sink vertices (Yes/No))}{}{}{}{}{}{}
\coll{}{}{}{}{}{}{}

\hline \hline 

\coll{Directed Hypercube}{$\mathrm{log}^2(N')$}{$\mathrm{log}(N')$}{$\mathrm{log}^{10}(N')$}{$\mathrm{log}^4(N')ll(N')$}{$\mathrm{log}^{7/2}(N')ll(N')$}
{$N' \mathrm{log}^{2}(N') ll(N')$}
\coll{graph (Yes)}{}{}{}{}{}{}
\hline 

\coll{Gaussian random }{$\mathrm{log}^3(N')$}{$\mathrm{log}^3(N')$}{$\mathrm{log}^{17}(N')$}{$\mathrm{log}^7(N')ll(N')$}{$\mathrm{log}^{11/2}(N')ll(N')$}
{$N' \mathrm{log}^{9/2}(N') ll(N')$}
\coll{partition (No)}{}{}{}{}{}{}
\hline 

\coll{Gaussian random }{$\mathrm{log}^3(N')$}{$\mathrm{log}^3(N')$}{$\mathrm{log}^{17}(N')$}{$\mathrm{log}^7(N')ll(N')$}{$\mathrm{log}^{11/2}(N')ll(N')$}
{$N' \mathrm{log}^{9/2}(N') ll(N')$}
\coll{partition (Yes)}{}{}{}{}{}{}
\hline 

\coll{Planted partition}{$\mathrm{log}^3(N')$}{$\mathrm{log}^3(N')$}{$\mathrm{log}^{17}(N')$}{$\mathrm{log}^7(N')ll(N')$}{$\mathrm{log}^{11/2}(N')ll(N')$}
{$N' \mathrm{log}^{9/2}(N') ll(N')$}
\coll{graph (No)}{}{}{}{}{}{}
\hline 

\coll{Planted partition}{$\mathrm{log}^3(N')$}{$\mathrm{log}^3(N')$}{$\mathrm{log}^{17}(N')$}{$\mathrm{log}^7(N')ll(N')$}{$\mathrm{log}^{11/2}(N')ll(N')$}
{$N' \mathrm{log}^{9/2}(N') ll(N')$}
\coll{graph (Yes)}{}{}{}{}{}{}
\hline 

\coll{Navigable small }{$\mathrm{log}^3(N')$}{$\mathrm{log}(N')$}{$\mathrm{log}^{13}(N')$}{$\mathrm{log}^5(N')ll(N')$}{$\mathrm{log}^{9/2}(N')ll(N')$}
{$N' \mathrm{log}^{5/2}(N') ll(N')$}
\coll{world graph (No)}{}{}{}{}{}{}
\hline 

\coll{Navigable small }{$\mathrm{log}^3(N')$}{$\mathrm{log}(N')$}{$\mathrm{log}^{13}(N')$}{$\mathrm{log}^5(N')ll(N')$}{$\mathrm{log}^{9/2}(N')ll(N')$}
{$N' \mathrm{log}^{5/2}(N') ll(N')$}
\coll{world graph (Yes)}{}{}{}{}{}{}
\hline 

\coll{G$_{n, p}$ graph (No)}{$\mathrm{log}(N')$}{$\mathrm{log}^3(N')$}{$\mathrm{log}^{11}(N')$}{$\mathrm{log}^5(N')ll(N')$}{$\mathrm{log}^{7/2}(N')ll(N')$}
{$N' \mathrm{log}^{7/2}(N') ll(N')$}
\hline 

\coll{G$_{n, p}$ graph (Yes)}{$\mathrm{log}(N')$}{$\mathrm{log}^3(N')$}{$\mathrm{log}^{11}(N')$}{$\mathrm{log}^5(N')ll(N')$}{$\mathrm{log}^{7/2}(N')ll(N')$}
{$N' \mathrm{log}^{7/2}(N') ll(N')$}
\hline 

\coll{Paley graph (No)}{$c$}{$\mathrm{log}^3(N')$}{$\mathrm{log}^{8}(N')$}{$\mathrm{log}^4(N')ll(N')$}{$\mathrm{log}^{5/2}(N')ll(N')$}
{$N' \mathrm{log}^{3}(N') ll(N')$}
\hline 

\coll{Random uniform}{$\mathrm{log}^3(N')$}{$\mathrm{log}^3(N')$}{$\mathrm{log}^{17}(N')$}{$\mathrm{log}^7(N')ll(N')$}{$\mathrm{log}^{11/2}(N')ll(N')$}
{$N' \mathrm{log}^{9/2}(N') ll(N')$}
\coll{k-out graph (No)}{}{}{}{}{}{}
\hline 

\coll{Random uniform}{$\mathrm{log}^3(N')$}{$\mathrm{log}^3(N')$}{$\mathrm{log}^{17}(N')$}{$\mathrm{log}^7(N')ll(N')$}{$\mathrm{log}^{11/2}(N')ll(N')$}
{$N' \mathrm{log}^{9/2}(N') ll(N')$}
\coll{k-out graph (Yes)}{}{}{}{}{}{}
\hline

\coll{Scale-free graph (No)}{$\mathrm{log}^3(N')$}{$\mathrm{log}^3(N')$}{$\mathrm{log}^{17}(N')$}{$\mathrm{log}^7(N')ll(N')$}{$\mathrm{log}^{11/2}(N')ll(N')$}
{$N' \mathrm{log}^{9/2}(N') ll(N')$}
\hline 

\coll{Scale-free graph (Yes)}{$\mathrm{log}^3(N')$}{$\mathrm{log}^3(N')$}{$\mathrm{log}^{17}(N')$}{$\mathrm{log}^7(N')ll(N')$}{$\mathrm{log}^{11/2}(N')ll(N')$}
{$N' \mathrm{log}^{9/2}(N') ll(N')$}
\hline

\hline \hline 
\end{tabular} 
\end{table*} 

\newcommand{\colm}[8]{#1 &#2 &#3 &#4 &#5  &#6 &#7 &#8  \\ }
\begin{table*}[t] 
\caption{Table presenting our data for incidence matrix-based graph families that offer no advantage with the HHL algorithm. We also give the advantage offered by these graph families with CKS(1) (or AQC(1)) and dream QLS algorithms (see Eqns. \ref{Rcks} and \ref{Rdqls} respectively). Here,  `Nil' represents no advantage and `Exp' represents exponential advantage in the table. } 
\label{tab:incidence2} 
\centering 
\hspace*{-1.3cm}
\begin{tabular}{|c|c|c|c|c|c|c|c|}

\hline 
\colm{Graph name}{$\kappa$}{$s$}{$t_{\mathrm{HHL}}$}
{$t_{\mathrm{CKS(1)}}(\text{or }t_{\mathrm{AQC(1)}}), $}
{$t_{\mathrm{dream\ solver}}$}
{$t_{\mathrm{CLS}}$}{Advantage with} 
\colm{(Random (Yes/No))}{}{}{}{}{}{}{HHL, }
\colm{}{}{}{}{}{}{}{CKS(1)(or AQC(1)),}
\colm{}{}{}{}{}{}{}{dream solver}

\hline 
\hline 

\colm{GN graph(No)}{$N'^2$}{$N'$} {$N'^8\mathrm{log}^{2}(N')$}{$N'^3\mathrm{log}(N')$}{$N'^{5/2}\mathrm{log}(N')ll(N')$}{$N'^3  ll(N')$}{Nil, Nil, Exp}
\colm{}{}{}{}{$\mathrm{log}(N'^3\mathrm{log}(N'))$}{}{}{}
\hline 

\colm{GN graph(Yes)}{$N'^2$}{$\mathrm{log}^3(N')$} {$N'^6\mathrm{log}^{8}(N')$}{$N'^2\mathrm{log}^4(N')$}{$N'^{2}\mathrm{log}^{5/2}(N')ll(N')$}{$N'^2 \mathrm{log}^3(N') ll(N')$}{Nil, Nil, Nil}
\colm{}{}{}{}{$\mathrm{log}(N'^2\mathrm{log}^4(N'))$}{}{}{}
\hline 

\colm{GNC graph(No)}{$\mathrm{log}^3(N')$} {$N'^2$} {$N'^4\mathrm{log}^{11}(N')$}{$N'^2\mathrm{log}^4(N')$}{$N'\mathrm{log}^{4}(N')ll(N')$}{$N'^3 \mathrm{log}^{3/2}(N') ll(N')$}{Nil, Exp, Exp}
\colm{}{}{}{}{$\mathrm{log}(N'^2\mathrm{log}^4(N'))$}{}{}{}
\hline

\colm{GNC graph(Yes)}{$\mathrm{log}^3(N')$} {$N'^2$} {$N'^4\mathrm{log}^{11}(N')$}{$N'^2\mathrm{log}^4(N')$}{$N'\mathrm{log}^{4}(N')ll(N')$}{$N'^3 \mathrm{log}^{3/2}(N') ll(N')$}{Nil, Exp, Exp}
\colm{}{}{}{}{$\mathrm{log}(N'^2\mathrm{log}^4(N'))$}{}{}{}
\hline

\colm{GNR graph(No)}{$\mathrm{log}^3(N')$} {$N'$} {$N'^2\mathrm{log}^{11}(N')$}{$N'\mathrm{log}^4(N')$}{$\sqrt{N'}\mathrm{log}^{4}(N')ll(N')$}{$N'^2 \mathrm{log}^{3/2}(N') ll(N')$}{Nil, Exp, Exp}
\colm{}{}{}{}{$\mathrm{log}(N'\mathrm{log}^4(N'))$}{}{}{}
\hline

\colm{GNR graph(Yes)}{$N'^2$} {$\mathrm{log}^3(N')$} {$N'^6\mathrm{log}^{8}(N')$}{$N'^2\mathrm{log}^4(N')$}{$N'^2\mathrm{log}^{5/2}(N')ll(N')$}{$N'^2 \mathrm{log}^{3}(N') ll(N')$}{Nil, Nil, Nil}
\colm{}{}{}{}{$\mathrm{log}(N'^2\mathrm{log}^4(N'))$}{}{}{}
\hline

\hline \hline 
\end{tabular} 
\end{table*} 

\subsection{Graph families for incidence matrix-based systems} \label{sec:IG}

In this sub-section, we turn our attention to the system of linear equations associated with the incidence matrix of directed graphs. We recall that a sink refers to a vertex whose out-degree is zero and a source to a vertex whose in-degree is zero. Sources and sinks in directed graphs may not be allowed depending on the application of interest. Therefore, for those considered graph families that admit sources and sinks, we consider them as well as their modified versions that do not admit sources and sinks. Our algorithm to convert source and/or sink vertices to non-source and/or sink vertices is presented in Algorithm 1 of the Supplemental Material. Tables \ref{tab:incidence_complexity} and \ref{tab:incidence2} present the data for all directed graphs considered in this study. We identified 14 good and 6 bad graph families from the survey. 

\subsubsection{Good graph families} 

From Table \ref{tab:incidence_complexity}, we observe that the growth of $\kappa$ and $s$ is polylogarithmic in $N'=N+M$, which is in agreement with the observations from good graph families of Laplacian matrix-based systems. Figs. \ref{fig:incidence_bestgraph}(a) - \ref{fig:incidence_bestgraph}(d) illustrate 2 of the good graphs of this system. The results for rest of the good graphs can be found in Fig. S.10 of the Supplemental Material. Along with the HHL algorithm, these graphs also show exponential advantage with CKS(1), AQC(1), and the dream QLS. Similar to the Laplacian-matrix based systems, the phase randomisation method, CKS(1) and AQC(1) approaches start showing advantage at smaller system sizes compared to the HHL algorithm. We also note that among the good graph families that we identify, the admissible $N'$ values in the directed hypercube graph family grow as $N' \sim 2^n$, and thus the problem is classically hard. 

\subsubsection{Bad graph families} 

Table \ref{tab:incidence2} lists all the graph families that offer no advantage with the HHL algorithm. We also present two instances of the bad graph family in Figs. \ref{fig:incidence_bestgraph}(e)-(h). The rest of the bad graph families can be found in Fig. S.11 of Supplemental Material. Since either $\kappa$ or $s$ or both grow polynomially, we observe that CLS fares well over the HHL algorithm. However, with CKS(1) and AQC(1) methods, 3 bad graph families offer exponential advantage, while with the dream QLS, 4 graph families show exponential advantage. 

\subsection{A big-picture view} 

\begin{itemize} 
    \item Of the 30 graph families considered for the Laplacian matrix case, 7 show advantage, whereas for the incidence matrix case, 14 out of 20 considered graphs show an advantage. 
    \item For all of the good graphs that we have identified, the CKS(1) and the AQC(1) algorithms show crossover into the regime of advantage even within the relatively small system sizes that we consider, whereas HHL crosses over only much later or not at all. 
    \item We recall that VTAA-HHL offers a near-optimal improvement in the condition number scaling relative to HHL. However, as Table \ref{tab:complexities} indicates, the protocol trades-off favourable scaling in $\kappa$ with costlier scaling in $\epsilon$. This is reflected in our limited data regime where HHL outperforms VTAA algorithm for most of the graphs. However, as the expression for complexity is defined for large values of $\mathcal{N}$ and we evaluate runtime ratio $R(\mathcal{N})$ within the regime of our numerical data, we anticipate VTAA to outperform HHL in the asymptotic limit. 
\end{itemize} 

\section{Combining graph families: generalized hypercube graphs} \label{sec:graph_superfamily}

\begin{figure}[t]
\includegraphics[width=8.5cm]{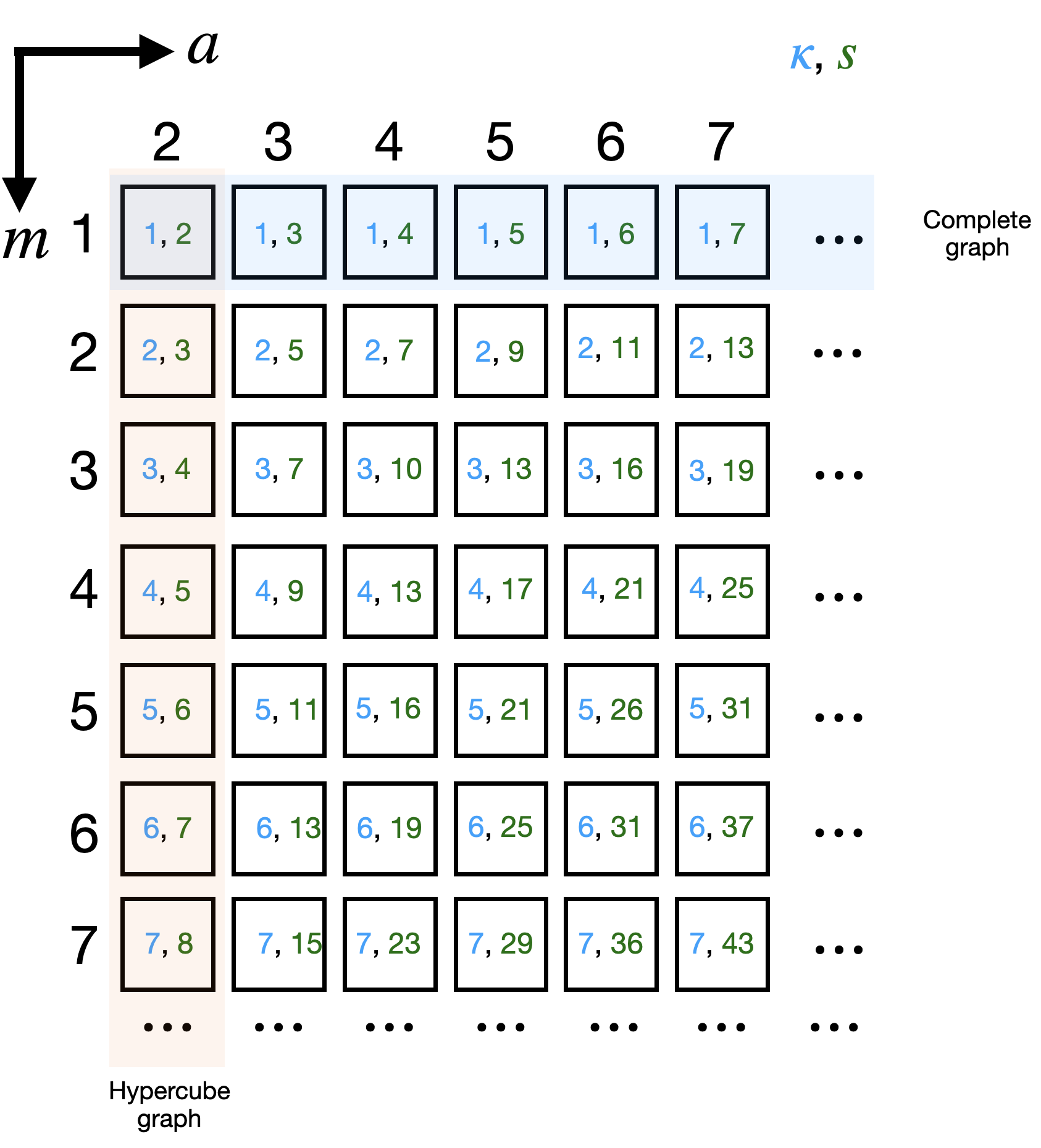}
\caption{A tableau of several families of the generalized hypercube graph superfamily, parametrized by $a$ and $m$ (the horizontal and vertical axis respectively). The blue band is the complete graph family, for which $m$ has been fixed and $a$ is varied, while the orange band, which is the hypercube graph family, constructed by fixing $a$ and varying $m$. The values of $\kappa$ and $s$ of each graph are indicated in blue and green font respectively. }\label{fig:hcg}
\end{figure} 

Our survey underscores the need for finding more graph families with slowly growing $\kappa$ and $s$. Motivated by this requirement, we generalize two of the graph families we considered for the Laplacian case to construct the generalized hypercube superfamily. 

The idea of hypercubes can be generalized in multiple ways \cite{ghcg}. In this work, we define the vertex set of a generalized hypercube $G_a^m$ as $V(G_a^m) = \{1, 2, 3, \cdots, a\}^{\times m}$ which contains $N = a^m$ vertices.  There is an edge between vertices $x = (x_1, x_2, \cdots, x_m)$ and $y = (y_1, y_2, \cdots, y_m)$ if the Hamming distance  $\mathfrak{H}(x, y) = 1$. Note that, $x_i$ and $y_i$ are any of the numbers $1, 2, \cdots, a$, for $i = 1, 2, \cdots, m$. We consider $m$ and $a$ as the parameters of generalized hypercube graphs. The generalized hypercube graphs become the hypercube graphs when $a = 2$. These graphs are important as their growth with respect to the parameter $m$ is faster than that of the hypercube graphs. Fig. \ref{fig:hcg} shows a tableau where we keep $m$ and $a$ on the Y- and X-axes, respectively. Each row or the column in the tableau is a graph family. We term the entire tableau of infinite cells as a generalized hypercube superfamily. The figure has each cell populated with two numbers, one denoting $\kappa$ and the other $s$, which we arrive at by computing the quantities. 

For a graph family with fixed $m$ and varying $a$, which corresponds any row from Fig. \ref{fig:hcg}, our numerical analysis infers the following:

\begin{enumerate}
\item The number of vertices, $N$ grows as $a^m$ (polynomial in $a$). 
\item $\kappa$ is constant, at $m$. 
\item $s$ grows as $am-m+1$ (polynomial in $a$). In terms of $N$, $s$ grows as $N^{1/m}$. 
\end{enumerate} 

\begin{remark} \label{rem:ghcg1}
The generalized hypercube $G_a^1$ gives the complete graph family ($K_a$). 
\end{remark} 

Although the complete graph is a bad graph, this is not the case for graph families in other rows of Fig. \ref{fig:hcg}. The runtime complexity ratio for HHL is, ${R}(N) \sim \frac{N^{\alpha} \mathrm{log}(\mathrm{log}(N))}{\mathrm{log(}N)} $ where $0<\alpha<1$. This ensures that there is an exponential separation between the runtimes of HHL algorithm and CLS for every row, which corresponds to a unique graph family of the tableau. 

\begin{claim}
There exist infinite families of good graphs in the generalized hypercube superfamily. Each row of the tableau, except the first, corresponds to one such family. \hfill $\square$ 
\end{claim}

On the other hand, if $a$ is fixed, we infer from our numerical analysis that 

\begin{enumerate}
\item The number of vertices, $N$ grows as $a^m$ (exponential in $m$). 
\item $\kappa$ grows as $m$, but since the system size $N = a^m$, we find that $\kappa \sim \mathrm{log}_a(N)$. 
\item $s$ grows as $am-m+1$ (polynomial in $m$), and in terms of $N$, we get $s \sim \mathrm{log}_a(N)$, where $a > 1$ ($a=1$ corresponds to the trivial case of each graph in the family having only one vertex, since $N = 1^m,\ \forall m$). 
\end{enumerate} 

\begin{remark} \label{rem:ghcg2}
$G_2^m$ gives the hypercube graph family. 
\end{remark} 

The runtime complexity ratio with HHL is, $R(N) \sim N\frac{\mathrm{log}(\log(N))}{\log^{11/2}(N)} $, which is exponential.

\begin{claim}
There exist infinite families of good graphs in the generalized hypercube superfamily, with each column of the tableau corresponding to one such family. \hfill $\square$ 
\end{claim} 

\section{Assessing advantage from graph family constructions} \label{sec:lessons} 

We summarize the workflow of our numerical survey below:

\begin{enumerate}
\item For a given graph family, defined by its construction of edges and vertices, we compute either the Laplacian or the incidence matrix.
\item We then evaluate $\kappa$ and $s$ for the given system and study how these quantities scale with system size $\mathcal{N}$, which is the number of vertices, $N$, for Laplacian matrix-based systems, and the sum of the number of vertices and edges, $N'$, for incidence matrix-based systems. Calculating $\kappa$ is computationally not easy, as it involves finding the largest and smallest eigenvalues of a matrix. 
\item Finally, we calculate the runtime ratio, which depends on the growth of $\kappa$ and $s$ with $\mathcal{N}$, and classify the graph family as good or bad. 
\end{enumerate} 

Our survey on 50 graph families led to us finding 21 good graph families, where we could realize an (exponential) advantage with HHL, and 29 bad graph families, where HHL offers no advantage. 

A natural follow-up question is whether one can bypass step 2, where computing $\kappa$ for larger system sizes is costly, and instead \textit{qualitatively} assess the prospects of an advantage with HHL directly from the construction of graph family. We reiterate that analytical expressions for $\kappa$ growth is rare; of the 50 graph families considered, we know $\kappa$ behaviour only for the hypercube graph and the complete graph. In this section, we guess the growth of $\kappa$ and $s$, guided by the insights from graph properties, as opposed to quantitatively calculating $\kappa$. 

For the purposes of addressing this question, we only consider analyses of Laplacian matrix based graph families. Thus, the runtime gain/loss is governed by the growth of $\kappa$ and $s$ of the Laplacian matrix belonging to a graph family. Now, we elaborate on the properties of a graph family that determine the growth of these parameters which in turn influence the prospects of quantum advantage:

\begin{figure*}
\centering 

\includegraphics[width=17cm]{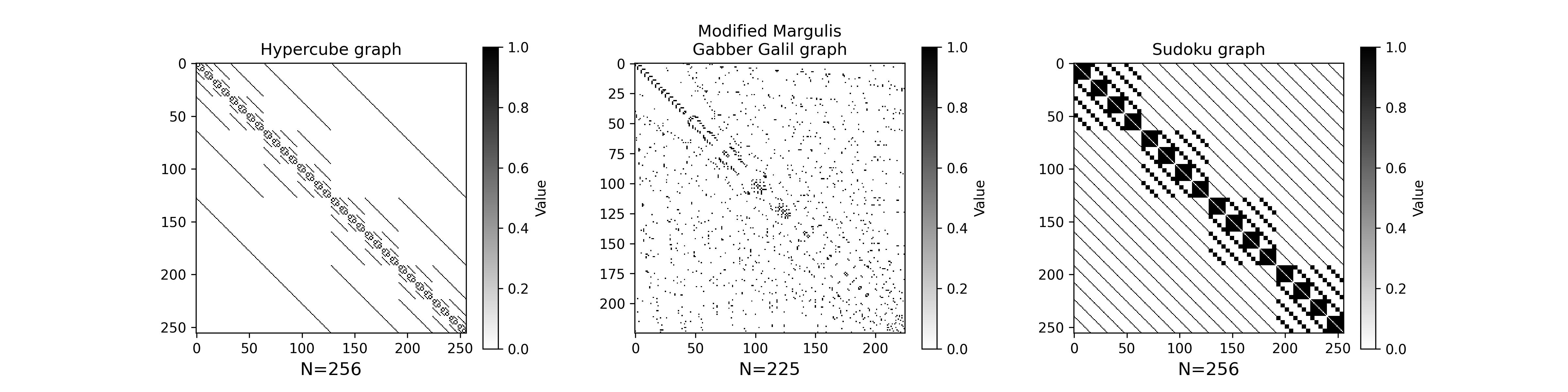} \\
    \vspace{0.5em}
    \makebox[0.31\textwidth]{(a)}
    \makebox[0.31\textwidth]{(b)}
    \makebox[0.31\textwidth]{(c)}\\ 
    
\includegraphics[width=17cm]{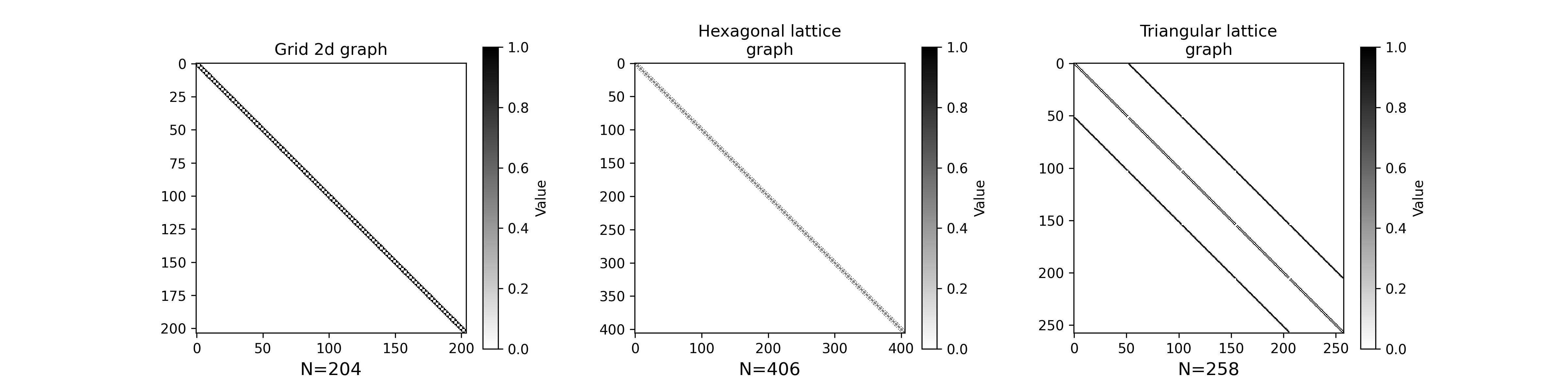} \\
    \vspace{0.5em}
    \makebox[0.31\textwidth]{(d)}
    \makebox[0.31\textwidth]{(e)}
    \makebox[0.31\textwidth]{(f)}\\ 

\caption{Sub-figures (a), (b) and (c) present the adjacency matrices of representative graphs: hypercube graph, modified Margulis-Gabber-Galil graph and Sudoku graph respectively. These graph families exhibit `diffuse' structure in their adjacency matrix elements. Similarly, sub-figures (d), (e) and (f) that show the matrix element map for grid 2d graph, hexagonal lattice and triangular lattice graphs respectively exhibit `sharp' structure.  }\label{fig:adj_conjecture} 
\end{figure*} 

\begin{itemize} 
    \item \textbf{Sparsity $s$}: The scaling of $s$ can be determined by the growth of maximum degree, $d_{max}$, defined by the maximum number of edges incident on any vertex of the graph. The quantity $d_{max}$ is related to $s$ of a Lapalcian matrix as $s= d_{max}+ 1$. The growth of $s$ can be estimated if one has access to growth of $d_{max}$ either through the definition of graph family or based on numerical analysis. For example, in a complete graph family (see S.5 A 1 of the Supplemental Material), the regularity of the graph grows linearly with increase in $N$, which is reflected in the growth of $s$ as well. As the data in Table \ref{tab:yaya} indicates, to realize runtime advantage with HHL algorithm, $s$ of a graph family must grow sufficiently slowly in $N$. 

    \item\textbf{Condition number $\kappa$}: We recall that $\kappa$ is the ratio of largest eigenvalue to the smallest non-zero eigenvalue of the Laplacian matrix. Mathematically, it is hard to extract the spectrum for a matrix as the system size grows. Thus, we intuitively understand the behaviour of $\kappa$, by examining the structure of Laplacian matrices of the surveyed graph families with the goal of finding a pattern through visual inspection, and correlate them with the growth of $\kappa$. Typically, a non-zero off-diagonal matrix element is either $-$1 or 0. These matrix elements often decide the growth of $\kappa$, and this information can be represented as a binary map of the adjacency matrix for all the Laplacian based graph families. Based on our visual inference, we could classify these graph families into two categories:
    \begin{enumerate}
        \item The non-zero off-diagonal matrix elements either occur as a scattered distribution (e.g., \ref{fig:adj_conjecture}(b)) or as a number of banded structures (for example, Fig. \ref{fig:adj_conjecture}(a) and \ref{fig:adj_conjecture}(c)), with increase in system size. We term these kind of matrices as `diffuse'. We pick a representative system size and present binary maps of graph families that occur as a diffuse pattern in their adjacency matrix structure for rest of the graph families in Fig. S.12 of Supplemental Material. 

        \item If the number of banded structures remain fixed with the increase in $N$ (for example, Fig. \ref{fig:adj_conjecture}(f) retains 3 bands with increase in $N$), then we term the pattern as `sharp'. The binary maps for the rest of the graph families with sharp structures are presented in Fig. S.13 of Supplemental Material. 
    \end{enumerate} 
    
       We observe that among 30 graph families considered, 16 graph families which exhibit polylog growth in $\kappa$ always have adjacency matrices with diffuse pattern, whereas all the 11 graph families with fast $\kappa$ growth correspond to having a sharp pattern in their adjacency matrices. Furthermore, we also note that bands are parallel for polynomially growing $\kappa$ cases and non-parallel bands are found for graph families with exponential growth in  $\kappa$. If the distribution of $-1$s is spread out in the matrix, it tends to stabilize the growth of $\kappa$, thus leading to polylog behaviour. Geometrically, this may be interpreted as spawning of new edges from most of the existing vertices of a smaller system size to new vertices of a larger system size of a graph family. We term this property of a graph as \emph{edge spawning}. In contrast, the sharp pattern is witnessed in those graph families where new vertices of larger system size are connected to only a limited number of existing vertices of smaller system size.
    
\end{itemize}

Based on our observations, we conjecture on the properties of graph that enable us to directly evaluate the prospects of an advantage in the Laplacian matrix case: 

\begin{conjecture}\label{conj:patterns}
Given a graph family defined by a construction, one may expect the family to exhibit a quantum advantage with HHL for the Laplacian-based NLSP if the following conditions are simultaneously satisfied: 
\begin{itemize} 
\item $d_{max}$ grows sufficiently slowly in $N$, and 
\item the construction exhibits edge spawning.
\end{itemize}
\end{conjecture} 

For illustration, we consider the example of the complete graph (see Fig. S.12(m) of Supplemental Material). The details of its construction can be found in Section S.5 A 1. At any system size, $N$, each vertex is connected to all other vertices. Thus, $d_{max}$ and hence $s$ grow linearly in $N$. However, since new edges keep spawning from every one of the old vertices for successively increasing system sizes, the family represents the most extreme instance of a diffuse pattern, and thus we also observe a constant value of $\kappa$. Therefore, although $\kappa$ is a constant as inferred from the extreme diffuse pattern, it is unlikely that one may realize an advantage with HHL for a complete graph given $s$ scales linearly in $N$. Our numerical data in Table \ref{tab:yaya2} shows that it is indeed a bad graph. 

On the other hand, when we examine graph families whose constructions involve attaching repeated structural units, such as grid 2d, hexagonal lattice, triangular lattice (plots of adjacency matrices can be found in \ref{fig:adj_conjecture}), ladder, circular ladder, and ring of cliques, and balanced binary tree families (plots of adjacency matrices can be found in Fig. S.13 of Supplemental Material), we find the opposite trend. In these cases, since a new structural unit is connected to a constant number of old vertices with the increase in number of vertices, $d_{max}$, and hence $s$ is constant. For the same reason, edge-spawning is not allowed in its construction and thus, $\kappa$ is expected to grow fast according to our conjecture. Therefore, in view of fast $\kappa$ growth in all the graphs, in spite of constant $s$ scaling, achieving advantage with HHL is unlikely. This is indeed the case as our data in Table \ref{tab:yaya2} indicates. 

\section{Quantum hardware demonstrations} \label{sec:hardware} 

In this section, we discuss our quantum hardware results, where we execute toy networks with classically optimized quantum circuits. We also identify the conditions on matrix structure that results in a 1-qubit HHL circuit. These analyses are restricted to Laplacian matrix-based graphs, as incidence matrix-based calculations are not presently feasible on current-day quantum hardware. 

\subsubsection{Effective resistance calculation with $4 \times 4\ A$ matrices} \label{app:ionq} 

In this sub-section, we describe our quantum hardware computations on toy networks whose $A$ matrix size is $(4 \times 4)$, where we compute effective resistance of electrical circuits shown in Fig. \ref{fig:hardware-runs} using the HHL algorithm as a representative example, on IonQ Forte-1 quantum computer (one of the current best commercially available computers). We use the HHL algorithm since it is the simplest and most established QLS. The effective resistance is extracted as an overlap between the HHL output vector and a vector with two non-zero entries, 1 and -1, at the vertices of interest. We set $n_r=3$ throughout. Due to current-day hardware limitations in terms of their gate fidelities, the use of significant resource reduction methods to reduce circuit depth is required. We borrow resource reduction strategies from Ref. \cite{nishanth2023adapt} for this purpose, and when further resource reduction is necessary, we apply reinforcement learning (RL)-based ZX calculus routine and a causal flow-preserving ZX calculus module discussed and used in Refs. \cite{riu2024reinforcement} and \cite{Palak2025relvqe}. We call this module as RLZX. Hence, our optimization strategy includes the following steps in that sequence: (i) Multi-qubit fixing (introduced in Ref. \cite{nishanth2023adapt}), (ii) Qiskit (version 0.39.5) \cite{Qiskit2021} optimization level 3, which we term L3, (iii) Pytket \cite{Sivarajah2021tket}, (iv) Qiskit L3 (version 0.39.5), (v)  Qiskit L3 (version 2.0.1), (vi) Approximating unitaries using Qiskit (version 2.0.1), (vii) RLZX, and (viii) Qiskit L3 (version 0.39.5). 

In addition to these optimization routines, we also implement the debiasing error mitigation strategy~\cite{maksymov2023symmetry} to improve our results. We could only go up to graphs containing $4$ vertices resulting in a $(4 \times 4)$ Laplacian matrix, as the resulting HHL quantum circuit for larger problem sizes was too deep to accommodate on the IonQ Forte-1 NISQ era quantum computer even after aggressive resource reduction. All of our circuits are compiled in the $\{RX,\ RY,\ RZ,\ RZZ\}$ gate set. Furthermore, even among the $(4 \times 4)$ matrices, we only consider those for which the gate counts post-resource reduction was sufficient for obtaining reasonable results. Table \ref{tab:hardware} presents the quantum hardware settings as well as the final quality of our result for effective resistance as percentage fraction difference (PFD), that is, our result relative to a classical calculation. \\ 

\begin{table*}[t]
    \centering
    \begin{ruledtabular}
    \caption{\label{tab:hardware} Table presenting the hardware settings for the quantum circuits executed on IonQ quantum computers. The $1q$ and $2q$ gates refer to the number of native one-qubit (GPI and GPI2) and two-qubit (ZZ) gates respectively. $p_{th}^{MQF}$ is the probability threshold for multi-qubit fixing. We recall that $n_r=3$ for all our computations. The percentage fraction difference (PFD) relative to a classical calculation is presented in the last column. } 
    \begin{tabular}{c|ccccccccc}
         \textrm{Circuit}&
          \textrm{$1q$-gate fidelity(\%))} &
          \textrm{$2q$-gate fidelity(\%) } &
          \textrm{Readout fidelity(\%)} &
          \textrm{T1 ($\mu s$)} & 
          \textrm{T2 ($\mu s$)} &
          
          \textrm{$1q$ gates} & 
          \textrm{$2q$ gates} &
          \textrm{$p_{th}^{MQF}$} & 
          \textrm{PFD(\%)}\\
         \colrule
         Circuit 1 & $99.98$ & $99.26$ & $99.6$ & $10^8$ & $10^6$ & $169$ & $28$ & $0.8$ & $8.82$ \\
         Circuit 2 & $99.98$ & $99.32$ & $99.59$ & $10^8$ & $10^6$ & $209$ & $30$ &$0.8$& $13.56$ \\
         Circuit 3 & $99.95$ & $99.33$ & $96.58$ & $1.88\times 10^8$ & $95 \times 10^6$ & $45$ & $7$ &$0.8$& $3.805$ \\
         Circuit 4 & $99.98$ & $99.32$ & $99.59$ & $10^8$ & $10^6$ & $223$ & $28$ &$0.8$& $10.98$ \\
         Hexagon & $99.97$ & $99.34$ & $99.26$ & $10^8$ & $10^6$ & $3$ & $0$ & $-$ & $2.72$ \\
         Octagon & $99.97$ & $99.35$ & $99.07$ & $10^8$ & $10^6$ & $3$ & $0$ &$-$& $0.2$ \\
     \end{tabular}
    \end{ruledtabular}
\end{table*} 

\noindent \textbf{Circuit 1: } 

We first consider a square graph with 4 vertices and 4 edges (see Fig. \ref{fig:hardware-runs}(a)). The resulting HHL circuit has $774$ one-qubit  and $46$ two-qubit gates. After applying optimization routines up to step (iv) of the routines listed earlier in this section, we reduce the one- and two-qubit gate counts to $78$ and $28$ respectively. We execute the circuit on the IonQ Forte-1 quantum computer with $5000$ shots over $5$ repeats. The average PFD is obtained to be $8.82$ percent. \\

\noindent \textbf{Circuit 2: }

Fig. \ref{fig:hardware-runs}(b) presents the next circuit considered. The HHL circuit initially has $3286$ one-qubit and $206$ two-qubit gates. We execute our full optimization pipeline to arrive at $89$ one-qubit gates and $30$ two-qubit gates. We execute this circuit too on the IonQ Forte-1 quantum computer with $4000$ shots and $5$ repetitions to get an average PFD of $13.56$ percent. Although we had a slightly better two-qubit gate fidelity available for this execution, the PFD is worse due to a combination of lower number of shots (limited by an upper limit of $1000000$ gate shots, which is decided by the product of the number of shots and total number of gates in a task) as well as increased gate count. \\ 

\begin{figure*}[t]
\begin{tabular}{cc}
\includegraphics[width=8cm]{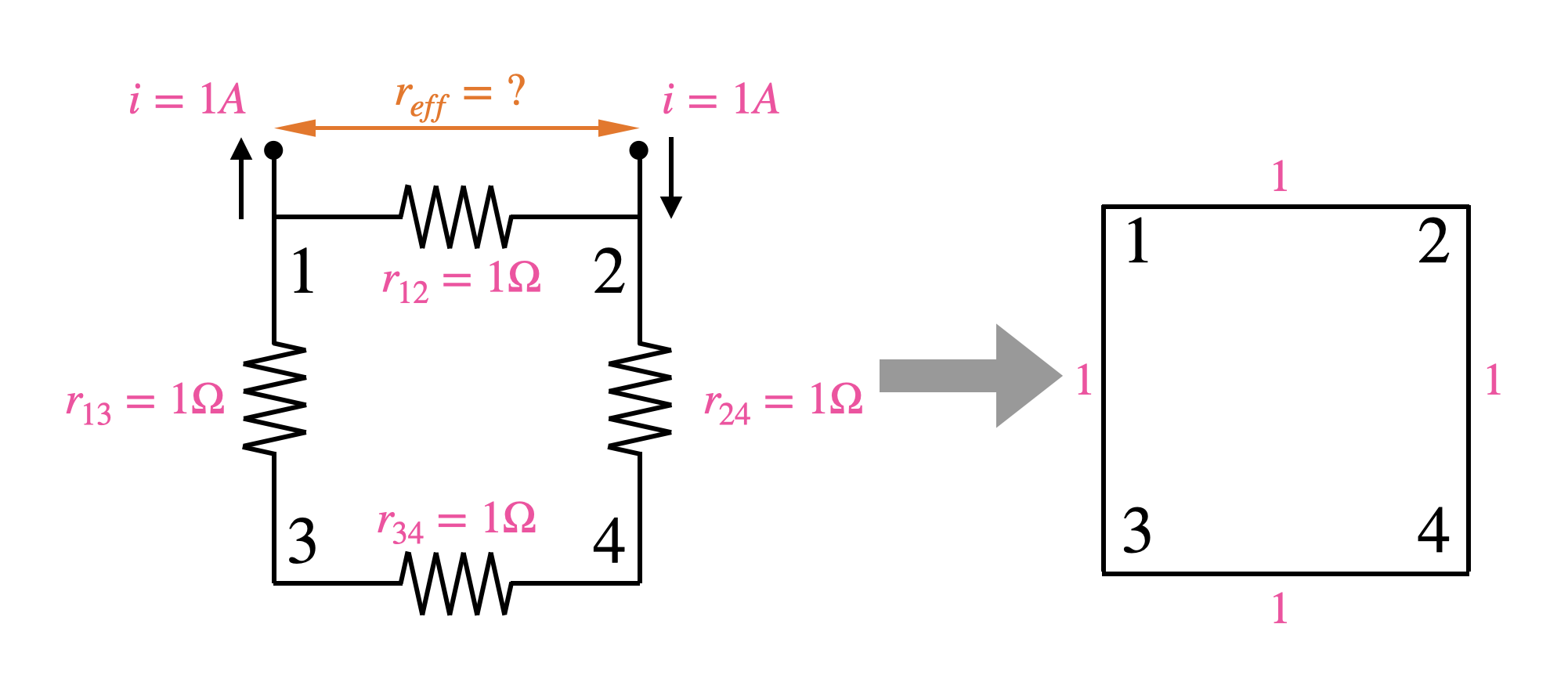} & \includegraphics[width=8cm]{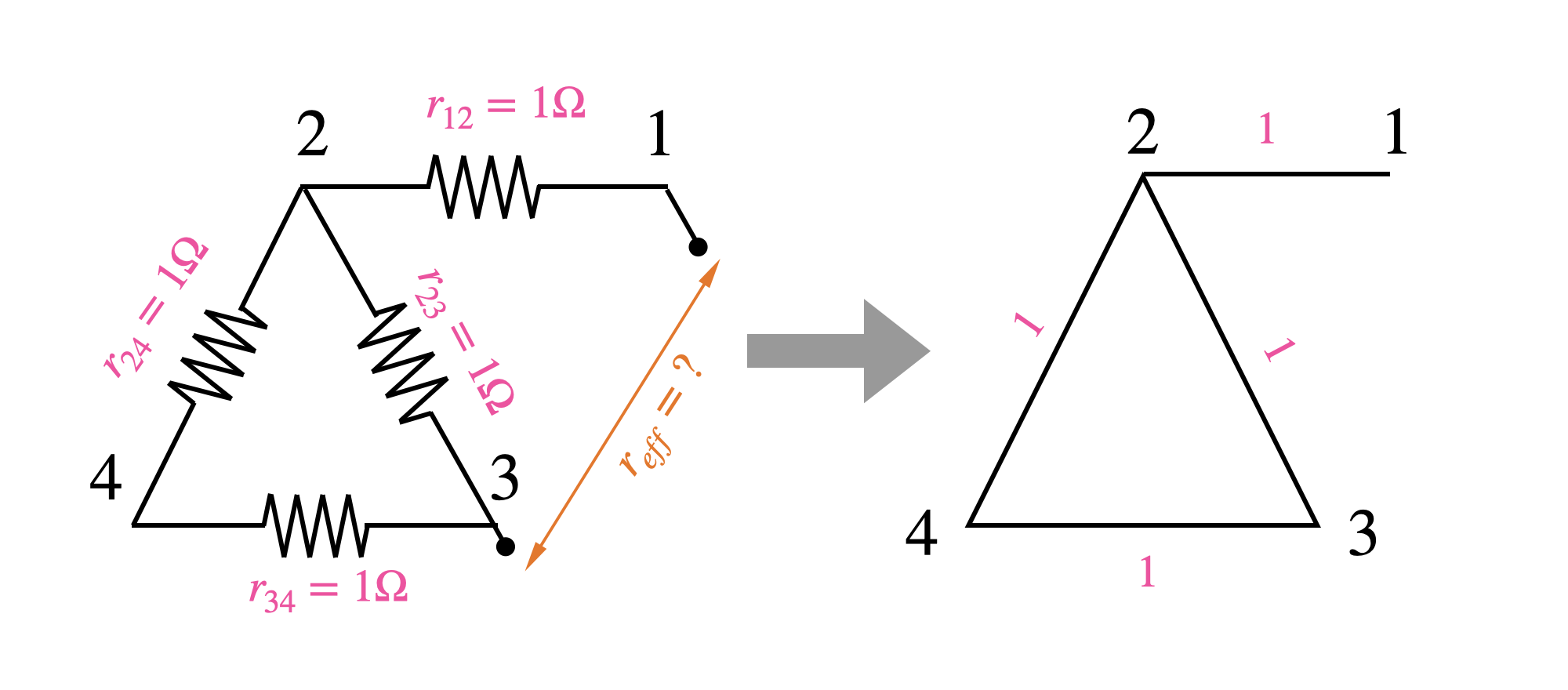}
\\ 
(a)&(b) \\ 
\includegraphics[width=8cm]{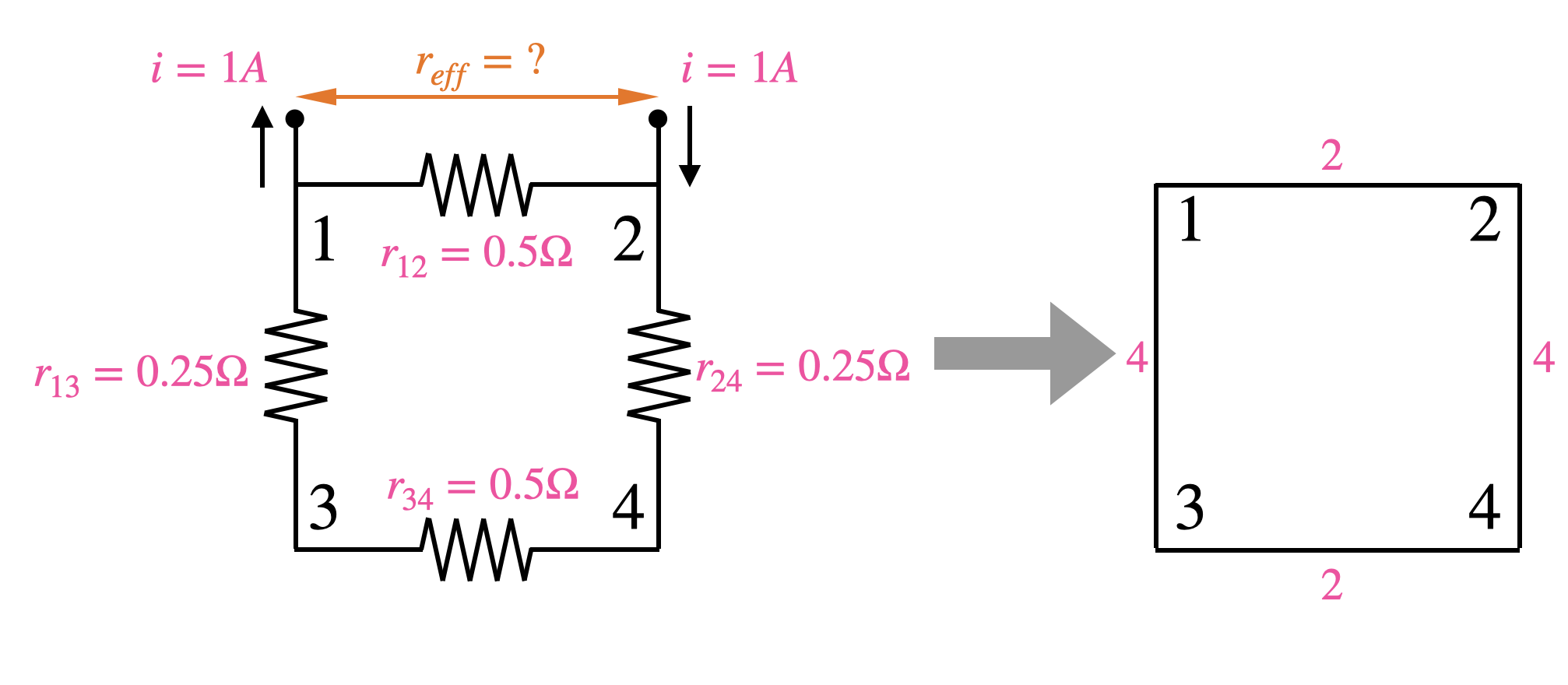}&\includegraphics[width=8cm]{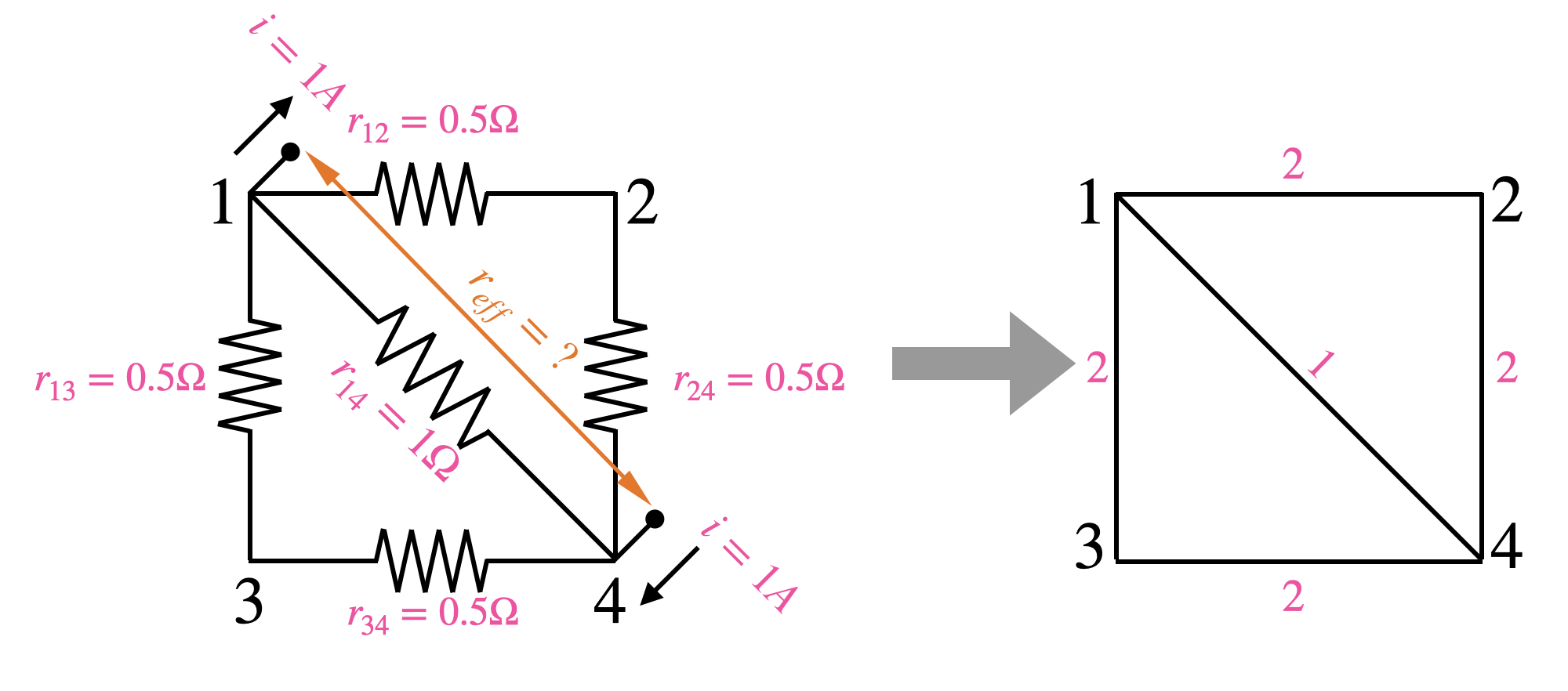} \\
(c)&(d) 
\end{tabular}
\caption{Figures showing the electrical circuits and the corresponding graphs that we considered for our quantum hardware computations. }\label{fig:hardware-runs}
\end{figure*} 

\noindent \textbf{Circuit 3: } 

For the network considered in Fig. \ref{fig:hardware-runs}(c), we allow for non-uniform edge weights, since we find that the circuit depth depends on edge weights and vertex labeling. Unlike the previous sub-figure, we did not need the additional RLZX, as we already arrived at $25$ and $7$ one-qubit ($\{RX,\ RY,\ RZ\}$) and two-qubit ($RZZ$) gates respectively, starting from $774$ one-qubit ($\{RX,\ RY,\ RZ\}$) and $46$ two-qubit gates. Due to availability constraints, we executed the tasks on IonQ Forte Enterprise-1, where the gate fidelities were comparable to the IonQ Forte-1 device, but the readout fidelity was substantially lower at only $96.58$ percent. We obtained an average PFD of $3.805$ percent with $5$ repetitions and $5000$ shots (limited once again by gate shots). \\ 

\noindent \textbf{Circuit 4: } 

We next move to Fig. \ref{fig:hardware-runs}(d), where we consider a Wheatstone bridge-like electrical circuit. Even though we assign edge weights non-uniformly, we still find the gate count to be significant. We begin with $3602$ one- and $226$ two-qubit gates, and after our resource reduction pipeline, we end up with $95$ one-qubit ($\{RX,\ RY,\ RZ\}$) and $28 (RZZ)$ two-qubit gates. We executed the tasks on IonQ Forte-1, and obtained an average PFD of $10.98$ percent with $5$ repeats and $3984$ shots. \\ 

\subsubsection{Computations on large $A$ matrices: all-qubit fixing in special graphs} \label{app:symm}

We intend to build on the idea of multi-qubit fixing (NISQ era variant) introduced in Ref. \cite{nishanth2023adapt} and ask if we can identify Laplacian graph families that accommodate the extreme case of all-qubit fixing. The multi-qubit fixing algorithm involves the execution of a QPE circuit prior to performing the HHL algorithm. Based on the probability distribution obtained from QPE, we fix the state of the clock register qubits in the HHL algorithm to either $\ket{0}$ or $\ket{1}$ provided their probability of occurrence is greater than a set threshold value. Hence, in the multi-qubit fixed HHL circuit (refer S.1 of the Supplemental Material), we can remove controlled unitary operations if their corresponding control qubit is in state $\ket{0}$, and replace controlled unitary operations with only the unitaries if their corresponding control qubit is in state $\ket{1}$. In the event that all the qubits can be fixed (all-qubit fixing), the control rotation module reduces to a series of single qubit $RY$ gates on HHL ancillary qubit, while QPE and QPE$^{\dagger}$ reset the clock register and state register to their respective initial states. In such a case, independent of $N$, we can carry out only a 1-qubit HHL calculation. We present below some relevant theorems in the context of effective resistance determination in electrical circuits involving graph Laplacians, whose proofs can be found in section S.7 of the Supplemental Material. 

\begin{theorem} \label{th:lapmat} 
Given a Laplacian matrix, $L$, and a vector, $\vec{b}$, such that it has for its entries exactly one $1$ and one $-1$ with the rest of its entries being $0$, any problem to be solved using the HHL algorithm reduces to a one-qubit calculation via all-qubit fixing as long as the input $\vec{b}$ is an eigenvector of $L$. 
\end{theorem} 

\begin{theorem} \label{th:condonlapmat}
The conditions on $L$ so that $\vec{b}=\vec{\delta_i}-\vec{\delta_j}$ is an eigenvector of the matrix are given by $l_{pi}=l_{pj}, \forall p \not \in \{i, j\}, \ \mathrm{and} \ l_{ii}=l_{jj}$. 
\end{theorem} 

\begin{theorem}\label{th:new_graph}
Let $G = (V(G), E(G))$ be any graph with the set of vertices $V(G) = \{v_i:i=1,\ 2,\ 3, \cdots, N\}$ and set of edges $E(G)$. Let $H = (V(H), E(H))$ be a graph such that $V(H) = V(G) \cup \{v_{(N + 1)}, v_{(N + 2)}\}$ and set of edges $E(H) = E(G) \cup \{(v_{(N + 1)}, u_1), (v_{(N + 1)}, u_2), \cdots, (v_{(N + 1)}, u_k)\} \cup \{(v_{(N + 2)}, u_1), (v_{(N + 2)}, u_2), \cdots, (v_{(N + 2)}, u_k)\}$,
where $u_1, u_2, \dots u_k \in V(G)$ in some order. Then, the Laplacian matrix $L(H)$ has eigenvalue $k$ with an eigenvector $\vec{b} = (\underbrace{0,0,\ldots,0}_{N \text{ times}}, 1, -1)^{\mathsf{T}}$.
\end{theorem}

\begin{theorem} \label{th:aqf}
A complete graph accommodates a one-qubit HHL via all-qubit fixing. 
\end{theorem} 

The first theorem shows that for effective resistance computations using only a 1-qubit all-qubit fixed HHL calculation, the input state has to be an eigenvector of the graph Laplacian. The second theorem provides the conditions on $L$ for possessing such an eigenvector, whereas the third provides a graph construction so that the graph Laplacian has such an eigenvector. The fourth theorem shows that a complete graph allows for all-qubit fixing. 

We consider two instances of the complete graph family to compute effective resistance on quantum hardware: all-to-all connected hexagon and all-to-all connected octagon. For complete graphs,  $\vec{I} = \delta_i - \delta_j$ (refer section S.3 of Supplemental Material) happens to be one of the eigenvectors of the Laplacian matrix. After QPE measurement, one would obtain a single peak in the histogram corresponding to the eigenvector. This, effectively, would reduce the circuit to one-qubit circuit consisting of HHL ancillary qubit. For all-to-all connected hexagon graph, we execute a one-qubit circuit containing $RY(0.3349)$ on IonQ Forte-1. The average PFD over 5 repeats with $1000$ shots was $2.72$ percent. For an all-to-all connected octagon, we execute a one- qubit circuit with gate $RY(\pi)$ on IonQ Forte-1 quantum device, to obtain an average PFD of $0.200$ with $1000$ shots and $5$ repeats. \\ 

\section{Conclusions} \label{sec:conclusion}

In summary, we investigate the prospects of quantum advantage using various quantum linear solvers (QLSs; see Section \ref{sec:qls}) for the networks-based linear systems problem (NLSP; introduced in Section \ref{sec:cn-lsp}) by surveying 50 graph families (random, non-random, trees, as well as those with/without sources and sinks), 30 based on Laplacian and 20 on incidence matrices. Our choice of picking NLSPs is motivated by both the flexibility that they offer in runtime complexity as well as their relevance to real-world problems. As Section \ref{sec:survey} describes, we find the scaling of condition number ($\kappa$) and sparsity ($s$) for each graph family, and within the scope of our survey (detailed in Section S.4 of Supplemental Material), we classify each family as good or bad, corresponding to an advantage or no advantage respectively with HHL. The rationale behind picking HHL is that other QLSs from our list provided in Table \ref{tab:complexities} build on HHL. For example, a graph family which offers polynomial advantage with HHL is guaranteed to do at least as well with other considered QLSs. 

Our findings show that 21 of the 50 graph families considered are good, that is, only they offer potential for (exponential) advantage with HHL. Realizing such an advantage almost always requires that $\kappa$ and $s$ scale at most polylogarithmically in system size (with the only exception from our dataset being the Sudoku graph family, where $\kappa \sim \sqrt{N}$). Among the 21, 7 are from Laplacian matrix based families (we considered a total of 30 of them for our study) and 14 from incidence matrix-based families (where we considered a total of 20 graph families). 

We now summarize our observations on the performance of other QLSs considered in our study: 

\begin{enumerate}
\item Improved QLSs such as the CKS(1) and the AQC(1) approaches consistently show crossover to the advantage regime at much smaller system sizes than HHL does. Thus, one can realize advantage at smaller problem instances with improved QLS algorithms. 
\item Furthermore, the CKS(1) and the AQC(1) algorithms elevate several graph families to classes with more advantage. In fact, 15 of the 23 bad graph families show exponential advantage with CKS(1)/AQC(1). We find one graph family where the dream QLS offers an advantage whereas the CKS(1) and the AQC(1) approaches do not, due to a peculiar cancellation between the growth of $\kappa$ and $s$. 
\end{enumerate} 

Section \ref{sec:graph_superfamily} discusses a generalization of the hypercube and complete graph families to the generalized hypercube superfamily, which we find accommodates an infinity of good graph families in its tableau (see Fig. \ref{fig:hcg}). We anticipate that further work in the direction of finding such generalized superfamilies can aid in finding many such new graphs.

Lastly, we briefly recap our findings from Section \ref{sec:lessons}, where we ask if one can directly look at a graph family properties and assess prospects of advantage for Laplacian matrix- based NLSPs. We examine the $A$ matrices of the 30 graph families considered in our survey. For \textit{all} of them, we find the matrix elements occurring as two distinct types of patterns for two types of $\kappa$ behaviour: `diffuse' and `sharp' corresponding to slow (at most polylogarithmic $\kappa$ scaling) and fast (polynomial/exponential). We thus conjecture that if one finds new edges spawning from existing vertices as system size increases, one could expect slow $\kappa$ growth, and hence better prospects for advantage. One could additionally look for a graph construction that could result in slow growing $d_{max}$ for favourable $s$ growth to comment on possible advantage expected from that graph family. 

Even if a graph family were to theoretically offer exponential advantage, there are other significant practical challenges that lie on the way to truly realizing such a speed-up. We list some of the fine print in this regard (see S.6 of Supplemental Material), and as proof-of-concept, perform effective resistance computations involving $(4\times 4)$ matrices on the newer generations of the commercially available IonQ quantum computers in Section \ref{app:ionq}. We find that even executing problems of this size is challenging, thus shedding light on the massive gap between dreams of achieving advantage in such classes of problems and current-day ground reality capabilities of quantum computers. 

Our work constitutes a preliminary survey that analyzes prospects for advantage, under suitable assumptions. Future directions could involve not only checking for suitability of more graph families, but also going beyond some of the assumptions that we make, such as accessing larger system sizes to arrive at better quality fits, modifying edge weights to be non-uniform, etc. 

Our work also opens new avenues to a notable offshoot; our survey involves calculating smallest and largest eigenvalues and thus arriving at a functional form via a fit for the condition number scaling with system size of Laplacians and incidence matrices for different graph families, thus naturally motivating future research on semi-empirical expressions for the spectra of such matrices. The issue of determining spectra of graphs is known to be a challenging problem in spectral graph theory for a majority of graph families, and our contribution would prove valuable not only for exploring quantum advantage in QLSs for NLSPs, but also for broader applications where graph spectra are of independent interest. 

\begin{acknowledgments} 
The work was carried out as a part of the Meity Quantum Applications Lab (QCAL) Cohort 2 projects. VSP acknowledges support from CRG grant (CRG/2023/002558). Our computations were carried out on the IIT Kharagpur HPC cluster, Param Shakti, and on TCG CREST's local HPC cluster. VSP thanks Prof. Sayan Chakraborty, IAI, TCG CREST, for useful discussions in the initial stages of the project on graphs. DS and VSP acknowledge Dr. Subimal Deb, Dr. Fazil, Mr. Peniel Bertrand Tsemo, Ms. Aashna Anil Zade, Ms. Tushti Patel and Mr. Nishanth Bhaskaran (and in no particular order) for their insightful comments and careful reading of the manuscript. DS and VSP thank Dr. Pedro C. S. Costa for discussions on complexity of block-encoding technique. 
\end{acknowledgments} 

\clearpage 

\begin{center}
     \Large \textbf{Supplemental Material} 
 \end{center} 

\begin{appendix} 

 \appendix
 \renewcommand{\thesection}{S.\arabic{section}}
 \renewcommand{\thefigure}{S.\arabic{figure}}
 \renewcommand{\thetable}{S.\arabic{table}}
 \renewcommand{\theequation}{S.\arabic{equation}} 
 \renewcommand{\thetheorem}{\thesection.\arabic{theorem}} 
 \renewcommand{\thedefinition}{\thesection. \arabic{definition}}

 \setcounter{section}{0}
 \setcounter{figure}{0}
 \setcounter{equation}{0} 
 \setcounter{theorem}{0}
 \setcounter{definition}{0}

 \makeatletter
 \def\l@subsection#1#2{} 
 \makeatother

\section{Steps involved in the HHL algorithm (including feature extraction)}\label{sm:hhl}

The HHL algorithm \cite{Harrow2009QuantumEquations} broadly encompasses the following steps.  Detailed reviews describing the HHL algorithm in general can be found in Refs. \cite{zaman2023step, lin2022lecture,morales}: 

\begin{itemize} 
\item We initialize three qubit registers to the state $\ket{b}_s\ket{0}^{\otimes n_r}_c\ket{0}_a$. The subscript,`$s$', denotes the state register, and it is assumed that the $n_b$-qubit state, $\ket{b}_s$, can be efficiently prepared from $\ket{0}_s$. $\ket{b}_s$ is obtained from normalized $\vec{b}$ via amplitude encoding. Furthermore, $\ket{b}_s$ can be expanded in the eigenbasis of $A$ as $\ket{b}_s=\sum_i b_i\ket{v_i}$. The subscripts `$c$' and `$a$' refer to the clock register (which contains $n_r$ qubits) and the HHL ancillary qubit register respectively. 
\item Perform QPE on the clock and the state registers, in order to obtain the eigenvalues ${\tilde{\lambda}_i}$ of $A$ captured using $n_r$ bits. The state after QPE is $\sum_i b_i\ket{v_i}_s\ket{\tilde{\lambda}_i}_c\ket{0}_a$. 
\item Carry out a controlled-rotation module between the clock and the HHL ancilla registers, so that the eigenvalues are inverted. This can be done in more than one way, including the use of a suitably chosen uniformly controlled rotation circuit. The state after this step is $\sum_i b_i\ket{v_i}_s\ket{\tilde{\lambda}_i}_c\bigg(\sqrt{1-\frac{C^2}{\tilde{\lambda}_i^2}}\ket{0}_a+\frac{C}{\tilde{\lambda}_i}\ket{1}_a\bigg)$, where $C$ is a suitably chosen constant. 
\item Undo QPE via a QPE$^\dag$ so that the HHL ancilla register is no longer entangled with the clock register. The state at the end of this step is $\sum_i b_i\ket{v_i}_s\ket{0}_c^{\otimes n_r}\bigg(\sqrt{1-\frac{C^2}{\tilde{\lambda}_i^2}}\ket{0}_a+\frac{C}{\tilde{\lambda}_i}\ket{1}_a\bigg)$.  
\item Measure the HHL ancilla qubit and post-select the outcome $1$. We obtain $\ket{x}=\sum_i\frac{b_i}{\tilde{\lambda}_i}\ket{v_i}$ on the state register. 
\item Extract the feature of interest via an appropriate circuit. One could add an additional state register, $\ket{b'}$, initialized to a suitable state, so that this register along with the state register after HHL serve as inputs to a feature extraction circuit module. For example, this register could be initialized to $\ket{b}$, so that the output of a SWAP test \cite{cswap} (or the Hong-Ou-Mandel (HOM) test, which is an ancilla-free version of the SWAP test but leads to destruction of state register qubits as a trade-off \cite{HOM}) module yields an overlap between the solution vector and $\vec{b'}$. 
\end{itemize} 

\begin{figure}[t]
\includegraphics[width=9cm]{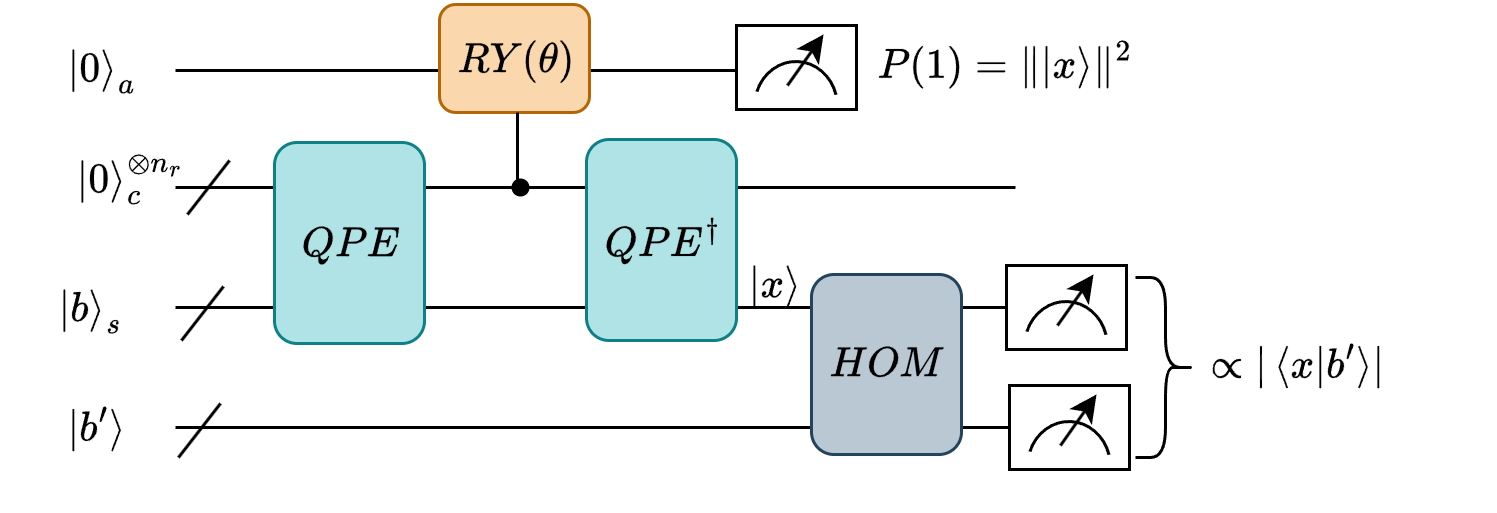}
\caption{Schematic of the HHL algorithm containing the HOM module to extract a feature from the solution vector. }\label{hhl}
\end{figure} 

A circuit that carries out the steps mentioned above (and for the specific example of extracting overlap between the input and output vectors) is illustrated in Fig. \ref{hhl}. 

\section{Some relevant graph theory definitions}\label{sm:gt}

\begin{definition}\label{def:degree}
    The degree of a vertex, $d(v)$, in a undirected graph $G$ is defined by the number of edges incident on the vertex $v$.
\end{definition}

\begin{definition}\label{def:edge_weight}
    An edge weight, $w_{ij}$, is a real number, which we assume to be positive throughout the article, which is assigned to an edge $(v_i, v_j)$. 
\end{definition}

\begin{definition}\label{def:simple}
    A simple graph is an undirected graph without any self loop, edge weights, and multiple edges between the vertices.
\end{definition}

\begin{definition}\label{def:random_graph}
    A random graph may be a directed graph or an undirected graph where existence of an edge is probabilistic.
\end{definition}

\begin{definition}\label{def:sub_graph}
    A graph $H = (V(H), E(H))$ is said to be a subgraph of a graph $G = (V(G), E(G))$ if $V(H) \subset V(G)$ and $E(H) \subset E(G)$.
\end{definition}

\begin{definition} \label{def:degmat} 
The degree matrix of an undirected simple graph $G$ with $N$ vertices $v_i: i = 1, 2, \cdots, N$ is a matrix $D(G) = (d_{ij})_{N \times N}$ where 
\begin{eqnarray}
d_{ij} = 
\begin{cases}
d(i) & \text{if}\ i=j; \\
0 & \text{otherwise}.
\end{cases}, 
\end{eqnarray}
Here, $d(i)$ is the degree of vertex $v_i$. For a weighted undirected graph, $d(i)$ is the sum of the edge-weights of the edges incident on the vertex $v_i$.
\end{definition}

\begin{definition} \label{def:adjmat}
The adjacency matrix, $Q(G) = (q_{ij})_{N \times N}$, of an undirected simple graph $G$ with $N$ vertices $v_i: i = 1, 2, \cdots, N$ is defined by 
\begin{eqnarray}
q_{ij} = 
\begin{cases}
1 & \text{if}~ (v_i, v_j)\in E; \\
0 & \text{otherwise}.
\end{cases}. 
\end{eqnarray}
In case of a weighted graph, if $w_{ij}$ is the weight of an edge $(v_i, v_j)$, we define
\begin{eqnarray}
q_{ij} = 
\begin{cases}
w_{ij} & \text{if}~ (v_i, v_j) \in E; \\
0 & \text{otherwise}.
\end{cases}
\end{eqnarray}
\end{definition} 

\section{Examples of NLSP} \label{sm:nlspexample}

\appsubsection{Determining effective resistance in complex networks} 

\begin{example} \label{ex1}
Electrical networks can be modeled as graphs, where each resistor is replaced by an edge, with the weight of the edge $(v_i, v_j) \in E$ being $w_{ij} = 1/r_{ij}$, where $r_{ij}$ is the resistance of the edge. Consider the graph, $G$, given in the top right panel of Figure 1 of the main manuscript, which corresponds to an electrical network. Let $u_1,\ u_2,\ u_3, \  \text{and} \ u_4$ represent the voltages at vertices $v_1,\ v_2,\ v_3$, and $v_4$ respectively. We represent the currents entering the network at vertices $ v_1, v_2, v_3, \  \text{and} \ v_4 $ as $I_1, I_2, I_3, \  \text{and} \ I_4$ respectively. We assume there is no leakage of currents in the network. Hence, the amount of current entering the vertex is equal to the amount of current exiting the vertex. We choose the direction of flow of current as $v_1 \to v_2$, $v_2 \to v_3$, $v_3 \to v_4$ and $v_4 \to v_1$. We remark here that the choice of directions are arbitrary, and it does not affect the final equations. From the conservation of currents at every vertex of the network, we have the following set of equations: 

\begin{align}
\text{At } v_1:\quad & (w_{12} + w_{41})u_1 - w_{12}u_2 - w_{41}u_4 = I_1 \\
\text{At } v_2:\quad & (w_{12} + w_{23})u_2 - w_{23}u_3 - w_{12}u_1 = I_2 \\
\text{At } v_3:\quad & (w_{34} + w_{23})u_3 - w_{34}u_4 - w_{23}u_2 = I_3 \\
\text{At } v_4:\quad & (w_{41} + w_{34})u_4 - w_{41}u_1 - w_{34}u_3 = I_4. 
\end{align} 

These equations can then be written in the form of a matrix as follows:
\begin{equation} \label{eqn:laplacian}
\underbrace{ 
\begin{bmatrix}
     w_{12}+w_{41} & -w_{12} & 0 & -w_{41} \\
     -w_{12} & w_{12}+w_{23} & -w_{23} & 0 \\
     0 & -w_{23} & w_{23} + w_{34} & -w_{34} \\
     -w_{41} & 0 & -w_{34} & w_{34+41}
 \end{bmatrix} }_{\text{L}}
\underbrace{\begin{bmatrix}
        u_1 \\
        u_2 \\
        u_3 \\
        u_4
    \end{bmatrix}}_{\vec{u}}
    =  \underbrace{\begin{bmatrix}
        I_1 \\
        I_2 \\
        I_3 \\
        I_4
    \end{bmatrix}}_{\vec{I}}.
\end{equation}

The above matrix represents the Laplacian matrix $L$ of graph $G$. Here, $\vec{u}$ is the vector of potentials, with each element of the vector providing the potential at a corresponding vertex of the network. $\vec{I}$ is the vector of currents entering each vertex of the network. In general, for any electrical network modeled as an undirected graph, we have 
\begin{equation}
    L\vec{u} =\vec{I}.
\end{equation}

Here, $\vec{u}$ is the unknown vector $\vec{x}$ and $\vec{I}$ takes the place of $\vec{b}$ in equation $L \vec{x} = \vec{b}$. This equation assists us to solve for the vector of potentials. We use $\vec{u} = L^{+}\vec{I}$ to compute the effective resistance between two vertices, say $v_1$ and $v_2$ of the graph, where $L^{+}$ is pseudo-inverse of $L$. This is obtained by allowing a unit current to flow into the network at $v_1$ and out of the network at $v_2$. We define $\vec{\delta}_i=(0, 0, \cdots,\ 1, \cdots, 0)^T$, where $1$ occurs only at $i^{th}$ index. To find effective resistance between $v_1$ and $v_2$, we set $I_1= 1$ unit of current and $I_2 =-1$ unit of current in Eq. \ref{eqn:laplacian}. Thus, $\vec{I}=\vec{\delta}_1 - \vec{\delta}_2$. As we allow only a unit current to enter and exit from vertices $v_1$ and $v_2$ respectively, finding effective resistance is as simple as finding the potential difference between vertices $v_1$ and $v_2$. Hence, $r_\text{eff}^{12}=u_1-u_2=(\vec{\delta}_1-\vec{\delta}_2)^T L^{+}(\vec{\delta}_1-\vec{\delta}_2)$. \hfill $\square$ 
\end{example}

\appsubsection{Determining number of vehicles on a lane: traffic flow congestion detection} 

\begin{example} \label{ex2} 
Consider the directed graph $\overrightarrow{G}$ depicted in the top right panel of Figure 1 of the main manuscript where all the vertices $v_1, v_2, v_3$, and $v_4$ preserve the flow conservation property. Viewed as a traffic flow problem, the edges are the lanes, while the weights on the edges denote the number of vehicles on that lane. The vehicles along a lane move in a particular direction, and is given by the direction marked on an edge. The vertices are the junctions. The vertices $v_1$ and $v_4$ take $c_1$ and $c_4$ vehicles respectively into the system of lanes. On the other hand, $c_2$ and $c_3$ are the number of vehicles that go out of the system via $v_2$ and $v_3$ respectively. The values of $c_1, c_2, c_3$, and $c_4$ are known. Let the flows via $\overrightarrow{(v_1, v_2)}, \overrightarrow{(v_2, v_3)}, \overrightarrow{(v_3, v_4)}$ be $\overrightarrow{(v_4, v_1)}$ are $y_1, y_2, y_3$, and $y_4$, respectively, which are unknowns. As the flow conservation property holds at all the vertices, we can construct four linear equations for the vertices, which are as follows 
\begin{equation}
    \begin{split}
    \text{At\ } v_1: & -y_1+y_4+c_1=0 \Rightarrow y_4-y_1 =-c_1, \\
    \text{At\ } v_2: & -y_2+y_1 - c_2= 0 \Rightarrow -y_2+y_1=c_2, \\
    \text{At\ } v_3: & -y_3+y_2 - c_3= 0 \Rightarrow -y_3+y_2=c_3,\\
     \text{At\ } v_4: & -y_4+y_3 + c_4= 0 \Rightarrow -y_4+y_3=-c_4.\\
     \end{split}
\end{equation}
These set of linear equations can be written in terms of a matrix equation 
\begin{equation}
\underbrace{ 
\begin{pmatrix}
     -1 & 0 & 0 & 1 \\
     1 & -1 & 0 & 0 \\
     0 & 1 & -1 & 0 \\
     0 & 0 & 1 & -1
 \end{pmatrix} }_{\text{B}}
\underbrace{\begin{pmatrix}
        y_1 \\
        y_2 \\
        y_3 \\
        y_4
    \end{pmatrix}}_{\vec{y}}
    =  \underbrace{\begin{pmatrix}
        -c_1 \\
        c_2 \\
        c_3 \\
        -c_4
    \end{pmatrix}}_{\vec{c}}.
\end{equation}
The system of linear equations can be expressed as $B\vec{y}=\vec{c}\ $, where $B$ is the incidence matrix of the graph $\overrightarrow{G}$, $\vec{y}$ is the vector of flows on each edge, which is unknown in the equation and $\vec{c}$ is the vector of flows entering the network at each vertex, which is a known quantity. To solve for $\vec{y}$ using a quantum algorithm, we must transform our matrix $B$ according to Equation (6) of the main manuscript. Our new matrix equation becomes 
\begin{equation} 
\underbrace{   \begin{pmatrix}
        0 & B \\
        B^{\dagger} & 0
    \end{pmatrix}}_{\text{H}}
\underbrace{{\begin{pmatrix}
        0 \\
        \vec{y}
    \end{pmatrix}}}_{\vec{Y}}
    = \underbrace{{\begin{pmatrix}
        \vec{c} \\
        0
    \end{pmatrix}}}_{\vec{C}}. 
\end{equation}  

We compute $\vec{Y}$ by solving $\vec{Y} = H^{+}\vec{C}$, where $H^+$ is the pseudo-inverse of $H$. To get the flow over any edge, say $y_i$, we evaluate the overlap between $\vec{\delta_i} \cdot \vec{y}$. We note that in the case of traffic flow congestion detection problem, this overlap yields the number of vehicles on a lane. \hfill $\square$ 
\end{example} 

\section{Scope and assumptions of our numerical analyses} \label{app:scope}

It is important to stress that taking an analytical route to obtain the information on spectra of these graph families is very non-trivial. Thus, we rather resort to numerical analysis, which is not only computationally taxing at large system sizes but also naturally comes with non-trivial assumptions. In this section, we provide the scope and assumptions made in our calculations, accompanied by error analysis on representative graph families whenever possible. \\

\begin{figure}[t]
    \centering
    \begin{tabular}{c}
    \includegraphics[width=0.68\linewidth]{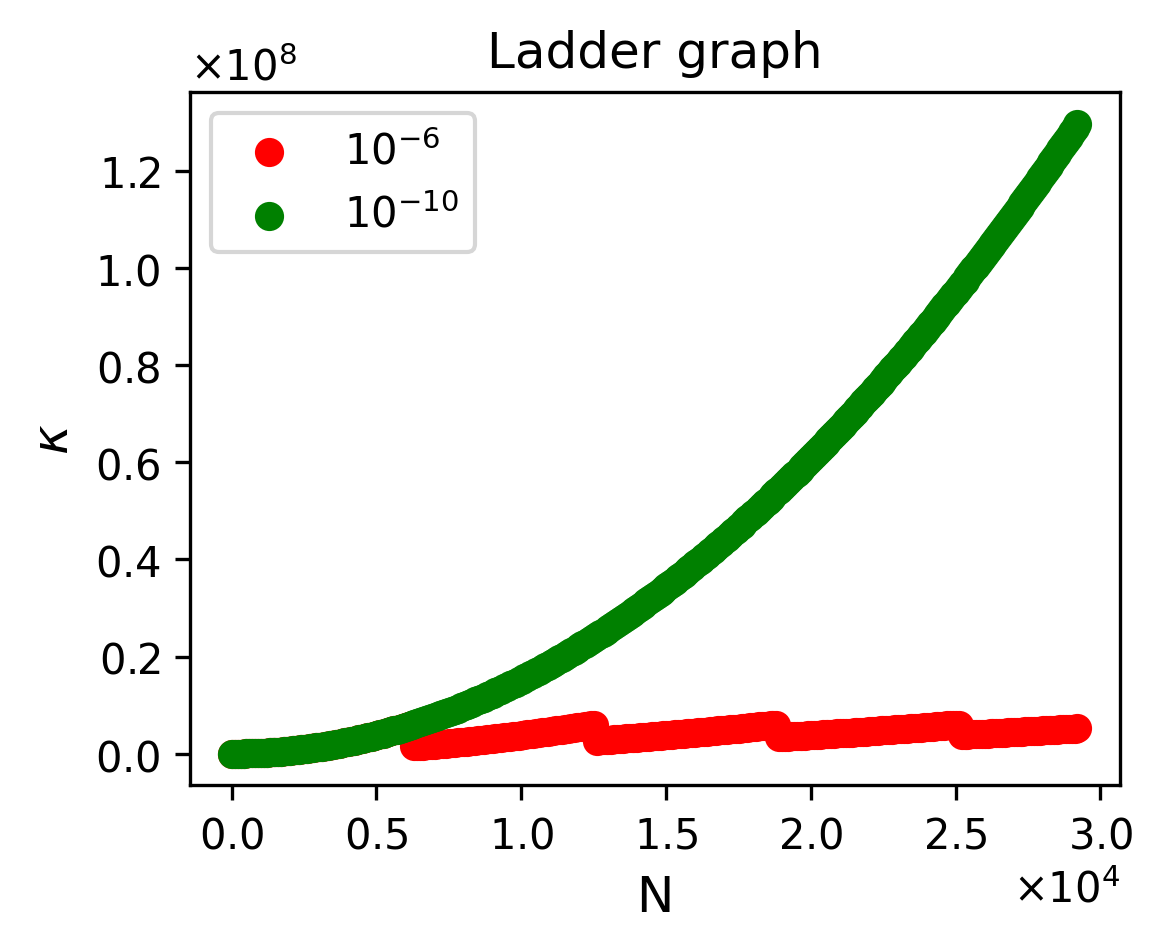} \\ 
    (a) \\ 
    \includegraphics[width=0.68\linewidth]{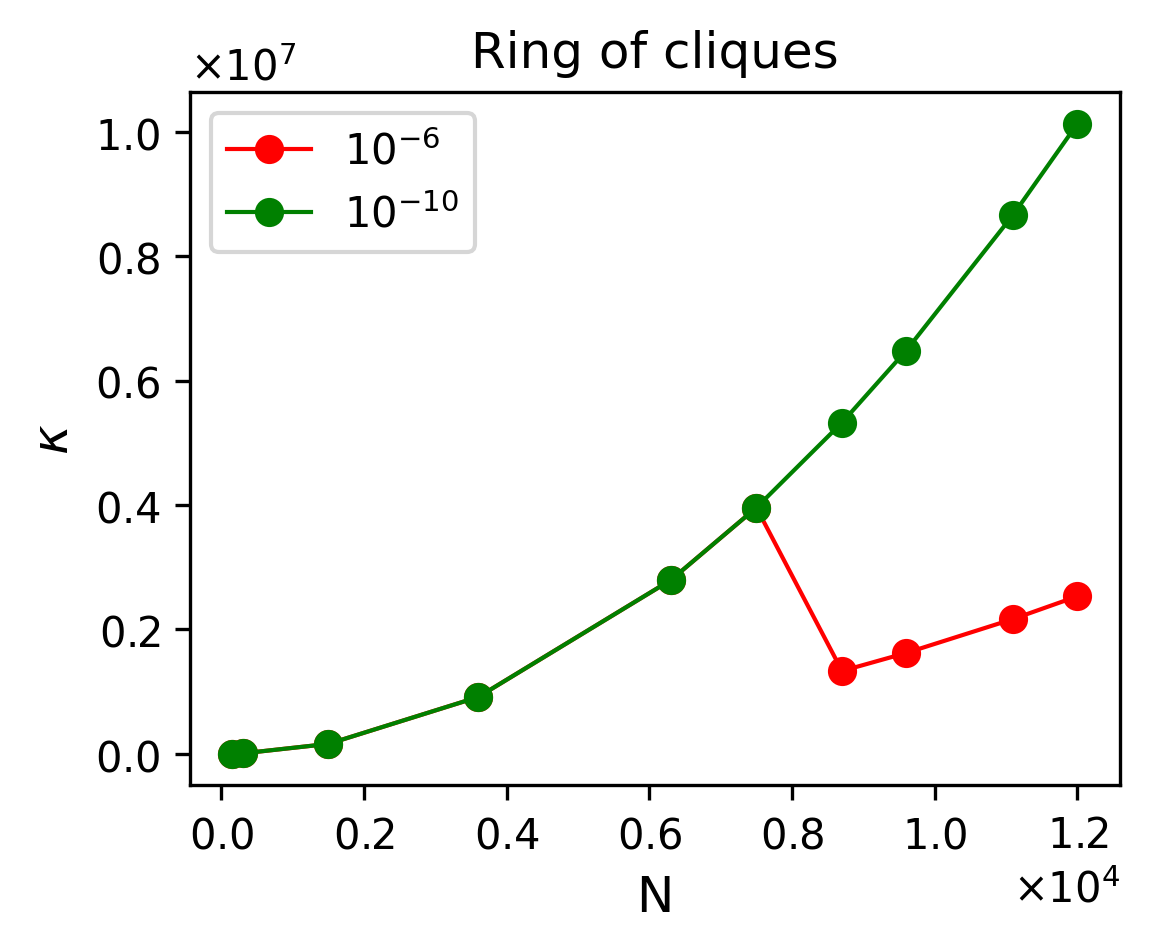} \\ 
    (b) \\ 
    \includegraphics[width=0.68\linewidth]{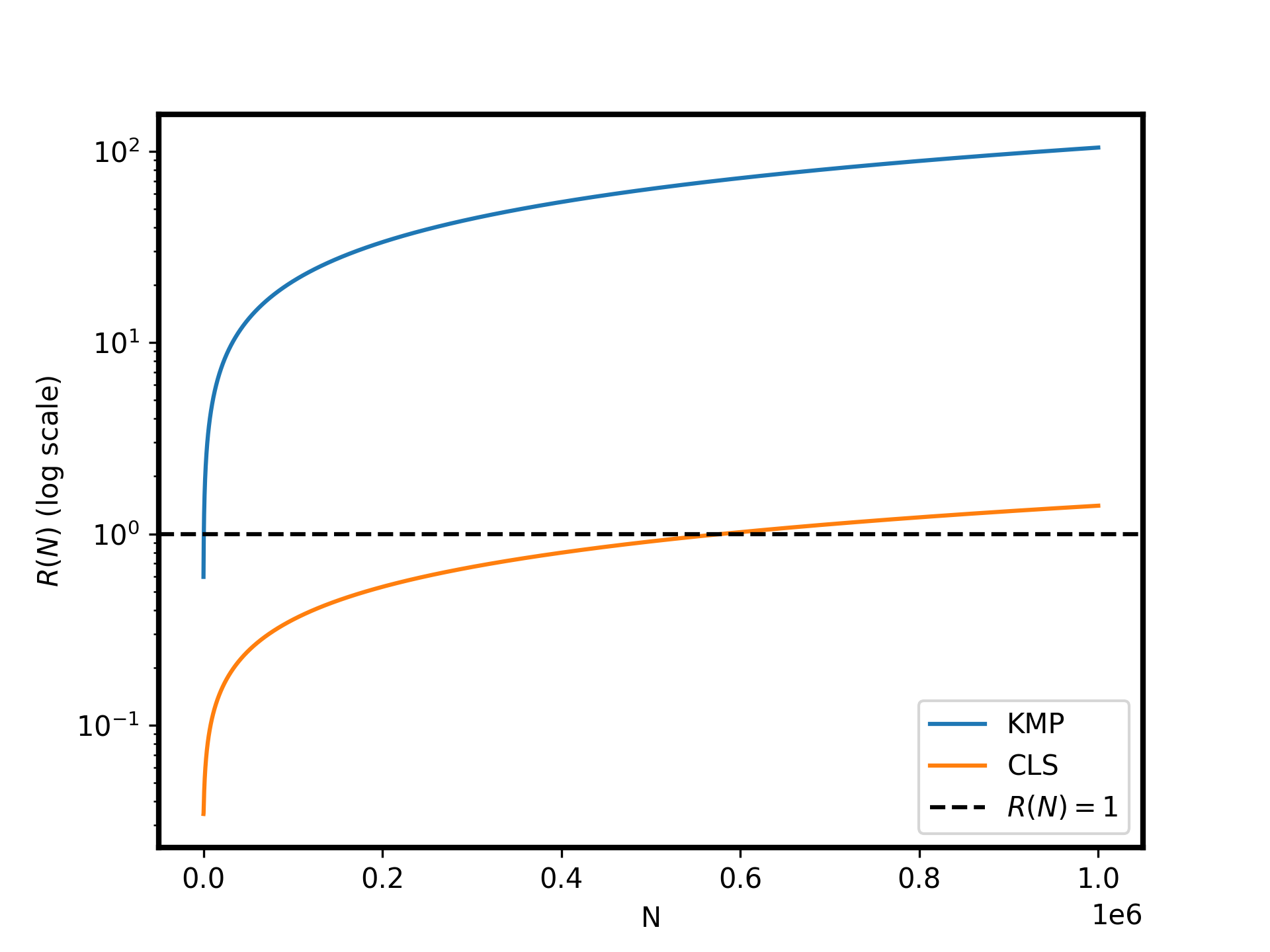}\\
    (c) \\ 
    \end{tabular}
    \caption{An illustration of cut-off problem with (a) ladder graph family, and (b) with ring of cliques graph family. Sub-figure (c) provides a comparison of $\tilde{R}(n)$ versus $n$ for CLS and KMP algorithms. The y-axis is on the log scale. }
  \label{fig:ks-vs-cg} 
\end{figure} 

\appsubsection{Assumptions} 

\begin{itemize}

\item \textbf{$1/\epsilon$ scaling: }As noted earlier, we assume for simplicity that $\epsilon^{-1} \sim \mathrm{log}(\mathcal{N})$. This assumption is reasonable, since for example, in HHL, this translates to $n_r$ growing as $\log(\log(\mathcal{N}))$, where $n_r$ is the number of clock register qubits in the HHL quantum circuit Section \ref{sm:hhl}). This assumption can be relaxed depending on the specific application and target precision of interest in future studies. 

\item We only consider graph Laplacians with \textbf{edge weight of each edge set to $1$}, for simplicity. A more general version of the analysis can be performed by relaxing this constraint for a future study, but for the purposes of this work, we perform preliminary analysis to inspect the effect of this assumption in obtaining advantage. We note that changing the edge weight of 1 for all of the edges to any other fixed real scalar for all of the edges does not affect $\kappa$ or $s$. Thus, we consider three cases of non-uniform edge weights, where they grow linearly, polynomially (second order), and logarithmically, in $y$, where an edge occurs between two vertices carrying indices $x$ and $y$. We carry out the analysis on two of the good graphs (refer Table II of the main manuscript): (i) the hypercube graph, and (ii) modified Margulis-Gabber-Galil graph. Our results are presented in Figs. \ref{fig:weighted_HCG} and \ref{fig:weighted_mgg}. We make the important observation that edge weight values do impact prospects of advantage; the hypercube graph retains its exponential advantage with HHL, CKS(1)/AQC(1) and dream QLS (refer Table I of main manuscript) for logarithmically growing edge weights, but becomes a bad graph (i.e., no advantage with HHL) when edge weights grow as a linear or a polynomial (second order) function of $y$. However, we do get exponential advantage with CKS(1)/AQC(1) for linearly or polynomially growing edge weights, while these scenarios show no advantage with dream QLS. For weighted modified Margulis-Gabber-Galil graph family, when the edge weight grows logarithmically as well as linearly, the graph family still remains in the good category and exhibits exponential advantage for HHL, CKS(1)/AQC(1) and the dream QLS. If the the edge weights grow polynomially (second order), the graph family is demoted to a bad graph family, which implies no advantage with HHL algorithm. We also find that, with polynomially growing edge weights, CKS(1)/AQC(1) and dream QLS show exponential advantage. 

\item \textbf{Cut-off problem: }To calculate the condition number, $\kappa$, of a matrix, one must be able to distinguish the minimum eigenvalue with zero. Since we carry out numerical analyses, one can immediately see that eigenvalues smaller than a certain value may not be capturable, thus leading to an incorrect estimate for $\kappa$. For instance, if we are not able to capture the minimum eigenvalue but instead end up capturing a larger one, we underestimate $\kappa$, and thus overestimate the degree of advantage obtained from that graph family. For all of our numerical calculations, we found that a cut-off of $10^{-6}$ for minimum eigenvalue is sufficient for the $\mathcal{N}$ values up to which we calculate $\kappa$. The pertinent data, where we compare the minimum eigenvalues at a few system sizes for two different thresholds, $10^{-6}$ and $10^{-10}$, and find them to be the same, is presented in Figs. \ref{app:Lcutoff} and \ref{app:Icutoff}. We note that the figures are only plotted for those graphs whose minimum eigenvalues show a downward trend with $\mathcal{N}$. Furthermore, when we move to $\mathcal{N}$ values larger than those considered for our study, we find two exceptions: the ladder graph family and the ring of cliques graph family. Figs. \ref{fig:ks-vs-cg}(a) and (b) highlight this `cut-off' problem for these two graph families. 

\item \textbf{The horizon problem: }Since for practical reasons, a numerical analysis can only go up to some finite value of $\mathcal{N}$, our assessment for advantage may be inaccurate, since there is always a possibility for the fit functional form to change as we reach much larger system sizes. Such a change would be hidden beyond the largest $\mathcal{N}$ that we pick for our numerical studies. We found three graph families: grid 2d graph, hexagonal lattice graph and triangular lattice graph, which exhibited different functional forms for condition number with the increase in system size. We also add at this juncture that the horizon problem is not avoidable and is fundamental, since we are using a classical computer (to compute $\kappa$) to assess quantum advantage. That is, to find if a QLS offers an advantage for a graph family, we must restrict our $N$ or $N'$ values (up to which we compute $\kappa$ subject to classical computing limitations) and then extrapolate our findings, subject now to the horizon problem, to larger system sizes in order to predict a potential onset of quantum advantage. 

\item \textbf{Choice of other graph input parameters: }We saw earlier the effect of edge weight choice on the prospect of realizing an advantage. Other graph input parameters such as choice of seed value may also influence our conclusions. For our analyses presented in the main text, we fix the seed value for reproducibility of our results. To understand the influence of seed values, we pick three representative candidates: Barabási-Albert graph family from Laplacian matrix based system of linear equations, Gaussian random partition graph (no source and sink) family, and planted partition graph (no source and sink) family from incidence matrix based system of linear equations. The results of our analyses are presented in Fig. \ref{fig:error_analysis}. For each graph family, we pick three different seed values and analyze the advantage offered by the respective graph families. We find that for the candidates that we consider, there is no change in their category upon changing the seed values. For example, the Barabási-Albert graph, which is a good graph, remains a good graph. Although the categories do not change for the graph families for the seed values that we considered for them, our observations cannot be generalized, and this aspect requires further analysis in a future study. The effect of tuning other parameters on advantage is also beyond the scope of our current work, and we defer it for a future study. \\
\end{itemize} 

\appsubsection{Scope of our numerical study} 

\begin{itemize}

\item \textbf{Number of quantum linear solvers considered: }Although we consider several efficient QLSs for our survey, we do exclude some, such as the quantum singular value transformation (QSVT)-based QLS. While QSVT is important due to its potential to offer near-optimal scaling in $\kappa$ and $\epsilon$, we do not consider it due to the non-trivial nature of $\mathcal{N}$ dependence entering via block encoding costs in its complexity. Other near-optimally scaling approaches built using block encoding have been excluded for the same reason (for example, the discrete adiabatic method \cite{costa2022optimal} and the Zeno eigenstate filtering method \cite{Lin2020optimalpolynomial}). Furthermore, algorithms such as AQC(exp) and CKS approaches that we consider in our study also scale near-optimally in $\kappa$ and $s$, and thus they may provide a reasonable idea of what to expect in this context from the approaches that we have excluded. Lastly, our analysis includes a `dream QLS' whose scaling is more favourable than any of the quantum linear solvers that we do not consider for our survey. 

\item \textbf{Number of graph families considered: }Since we carry out a numerical study, we are limited in the choice of graph families that we consider. Thus, our conclusions on the percentage of graph families that are best or better also are reliant heavily on the choice of graph families for our survey. Future work needs to be done in surveying broader classes of graphs for QLSs, especially given the dearth of literature on this front. 

\item In our numerical analyses, we substitute $\kappa$ and $s$ obtained from our numerical dataset (which usually do not exceed $\mathcal{N}$ of 10000) into runtime ratios, even though the runtime expressions themselves are for large $\mathcal{N}$ regimes. We adopt this simplification not only because deriving spectra of graph families analytically is very hard, but also predicting quantum advantage using classical resources is at least as challenging. Consequently, our extensive numerical analyses are to only be viewed as providing a reasonable understanding of the potential advantage that a graph family offers, and further future work is required in this direction. 

\item \textbf{The classical benchmark algorithm: }The runtime complexity for the fictitious CLS that we choose is set to be the same as that of the well-known conjugate gradient method purely due to favourable scaling in $\mathcal{N}$, $\kappa$, $s$, and $\epsilon^{-1}$, but the CLS itself should not be thought of as possessing the limitations of the conjugate gradient method such as being restricted only for positive definite matrices. In this context, we acknowledge that other important works exist in literature and by choosing a fictitious CLS as a representative benchmark classical solver to compare QLSs against, we do not do justice to the vast body of literature on classical algorithms for linear systems of equations. We defer a comparison of QLSs against leading classical linear solvers to a future study. Having said this, we pick two realistic specialized classical linear solvers for Laplacians as representative examples to qualitatively comment on our solver as a reasonable candidate. To this end, we pick a graph family from our candidate graph families and carry out a comparison: 
\begin{enumerate} 
\item The authors in Ref.\cite{koutis2014approaching} discuss an efficient method for solving systems of linear equations arising from a graph Laplacian ($(N\times N)$, symmetric, and diagonally dominant matrix), whose associated graph has $N$ vertices and $M$ edges, and which scales as $\mathcal{\tilde{O}}(M \mathrm{log}^2(N)\mathrm{log}(1/\epsilon))$. Our goal is to compare this solver, which we shall term the KMP solver (based on the authors' initials), against our CLS, and we do so by analyzing a situation where the former excels, thereby stacking the odds against CLS. We set $\epsilon^{-1}=\mathrm{log}(N)$. Noting that a graph that has no isolated vertices needs to have at least $N-1$ edges (best case scenario; the worst case would be $\sim N^2$ edges) and also simplifying $\widetilde{O}(f(n)) = O(f(n)  \log^k n)
$ by setting $k=1$ (best case scenario), the expression for complexity is $\mathcal{O}(N \mathrm{log}^2(N)\mathrm{log}(1/\epsilon) \ \mathrm{log}(N \mathrm{log}^2(N)\mathrm{log}(1/\epsilon)))$. For illustration, we now consider the specific case of a graph for which $\kappa$ and $s$ grow as log($N$). We see that the runtime complexity ratio, $R(N)=t_{\mathrm{CLS}}/t_{\mathrm{HHL}}$, simplifies to $N\log(\log(N))/\log^{5.5}(N)$, which is an exponential advantage. Instead, if we use the KMP algorithm in place of CLS, the runtime ratio is $\frac{(N\log(\log(N))\log[N \mathrm{log}^2(N) \log(\log(N))])}{\log^{5}(N)}$. Fig. \ref{fig:ks-vs-cg}(c) presents the curves for ${R}(N)$ versus $N$ for the cases of CLS and the KMP solver. We see that using either our CLS or the KMP solver still yields an exponential advantage with HHL. 

\item We consider another example, where the authors propose algorithms for solving systems of linear equations that arise in the specific case of graph Laplacians of planar graphs \cite{koutis2007linear, spielman2010algorithms}. Their algorithms scale as $\mathcal{O}(M \mathrm{log}(1/\epsilon))$. Assuming $M \sim N$, $1/\epsilon \sim \mathrm{log}(N)$, we obtain $R(N)$ to be $(N \mathrm{log}(\mathrm{log}(N))/\mathrm{log}^{7}(N)$, which gives an exponential advantage, just as in the case of CLS being used in place of these approaches. Furthermore, these algorithms too reach the crossover point at smaller system sizes. 
\end{enumerate} 
Thus, we see that qualitatively, the CLS we consider is a reasonable representative example at least for the purposes of our preliminary study to all of the classical linear solvers. 
For completeness, we comment on a special case of LSP applied to effective resistance determination where the conjugate gradient method works. Since the input vector, $\vec{b}$ has in it one $1$, one $-1$, and the rest of its elements as $0$, the vector is orthogonal to the all-1 vector that lies in the null space of $L$, and thus conjugate gradient suffices in spite of its limitation in being able to handle only positive definite matrices (for example, see Section 3.2 of Ref. \cite{spielman2010algorithms}). In fact, this holds when the condition $\sum_i b_i = 0$ is satisfied. 

\item \textbf{Padding the $A$ matrix: }We note that if the $A$ matrix from a problem/application is not of dimension $2^{n_b} \times 2^{n_b}$, we pad the matrix by adding a scalar along the diagonals, that is, $\begin{pmatrix}
A_{\mathrm{problem}} & 0 \\
0 & \mathcal{D}
\end{pmatrix}$, where $\mathcal{D}$ is an appropriately chosen diagonal matrix with all the diagonal entries being the same (a scalar matrix), such that we then solve by using a QLS for the equation $\begin{pmatrix}
A_{\mathrm{problem}} & 0 \\
0 & \mathcal{D}
\end{pmatrix} \begin{pmatrix}
    \vec{x} \\ 0
\end{pmatrix} = \begin{pmatrix}
    \vec{b} \\ 0
\end{pmatrix}$. We note that the sparsity of the non-padded and the padded matrices remain the same. Furthermore, one needs to ensure that $\kappa$ remains unaffected with a suitable choice for the scalar in the scalar matrix. We assume that there exists an oracle that supplies an appropriate scalar in such situations. Therefore, for our numerical analyses, we only compute $s$ and $\kappa$ of the graph Laplacian or graph incidence matrices we consider, and not those of the padded matrices. 

\item \textbf{Function fitting considerations: }When we seek fitting functions for our data points for $\kappa$ and $s$ of some considered graph, we mainly check the goodness of fit by simple visual inspection as well as tests that check if the fitted functional form gives absurd values for the quantities (for instance, $\kappa<1$ and $s<0$). For a not-so-obvious distribution of data points in a graph family, we fit the upper envelope as a reasonable approximation to the quantity in question. We also note that for the case of $s$, since it is constrained to be a positive integer, we round off the fit sparsity values to their nearest integer value. Furthermore, we only consider for fitting the following candidate functions: exponential, polylogarithmic (at most third order), and polynomial (at most third order). We exclude functions like $1/N^{p}; p \in \mathbb{Z}^{+}$, since they can lead to ill-behaved situations. For example, the complexity of the phase randomisation algorithm contains in it a $\mathrm{log}(\kappa)$ factor, and in the event that $\kappa \sim 1/N$, the complexity expression is not defined anymore. We also add that wherever analytical expressions are available, we fit that quantity (for example, s) to that expected function. For other cases, which constitute most of our candidates, we restrict ourselves to exponential/ polylogarithmic/ polynomial for simplicity. A caveat is that in spite of all these considerations for the fit function, it may or may not exactly reproduce the actual $\kappa$ or $s$ behaviour. For instance, we found that for the directed hypercube graph, we could get equally convincing fits within both log$^2(N'$) and log$^3(N')$. Although we expect the fit to do a reasonable job in general and not change the graph family category, this will remain an important fundamental limitation in such analyses. 
\end{itemize} 

\section{Details of the graph families surveyed: construction and input parameters} \label{app:params} 

In this section, we list all the classes of graph families which we consider for our numerical investigations. We recall that we consider both random and non-random graphs for our study. A random graph is a graph where the existence of an edge is probabilistic. The probability depends on a seed value. For different seed values, we get different probability functions that result in different graphs belonging to the same class. In a non-random graph, the existence of an edge is deterministic. It is unrealistic to consider all the graphs of the same class of random graphs for a numerical investigation. But fixing the seed value makes a random graph unique, which we consider as a representative of the family of random graphs. Other than the seed values, there can be other parameters that influence the structural properties of the graphs. Therefore, for each of the graphs, we present the seed values and parameters that we set for the numerical calculations in this work. We follow NetworkX (version 3.4.2) for the construction of all random and non-random graphs except the hypercube graphs. In the following two sections, we describe the graphs used for our studies on the system of linear equations generated by the Laplacian matrices and the incidence matrices. 

\appsubsection{Graphs whose Laplacian matrices are considered in our study} 

    All the graphs that we utilize in our work are undirected and simple. Our constructions ensure that no self-loop or multiple edges are generated in the graph. 
    \subsubsection{Non-Random graphs}
    \label{Non_Random_graphs_Laplacian}
	\begin{enumerate}
		\item 
		\textbf{Hypercube graphs:} The hypercube graphs are generalizations of cubes in graph theory. The $n$-dimensional hypercube graph $G_2^n$ has vertex set $V(G_2^n) = \{0, 1\}^{\times n}$ which contains $N=2^n$ vertices. The Hamming distance $\mathfrak{H}(x, y)$ between two tuples $x = (x_1, x_2, \cdots, x_n)$ and $y = (y_1, y_2, \cdots, y_n) \in \{0, 1\}^{\times n}$ is the number of indices $i$ for which $x_i \neq y_i$. There is an edge between vertices $x$ and $y$ if the Hamming distance, $\mathfrak{H}(x, y) = 1$. Note that $x_i$ and $y_i$ are either $0$ or $1$ for $i = 1, 2, \cdots, n$. The dimension $n$ is the only parameter for this class of graphs. In this article, we consider all the hypercube graphs with number of vertices $N$ ranging from $4$ to $16384$. The illustrations of the hypercube graph can be found in Figure 2(a) of the main manuscript.  
		\item 
		\textbf{Margulis-Gabber-Galil graphs:} In this article, we consider the vertex sets of the Margulis-Gabber-Galil graphs as  $\mathbb{Z}_n \times \mathbb{Z}_n$, where $\mathbb{Z}_n = \{\overline{0}, \overline{1}, \overline{2}, \cdots, \overline{(n - 1)}\}$, $\overline{k}$ is the set of all integers $k$ modulo $n$, and $n$ is a natural number. Each vertex $(x,y) \in \mathbb{Z}_n \times \mathbb{Z}_n $ is connected to $(x \oplus 2y, y), (x \oplus 2y \oplus 1, y), (x, y \oplus 2x)$ and $(x, y \oplus 2x \oplus 1)$, where $ \oplus $ denotes the addition modulo $n$. This construction would lead to self loops and parallel edges. We modify the graph to exclude self loops and parallel edges from this family. We note that the number of vertices in this family grows as $N=n^2$. We consider all such possible modified Margulis-Gabber-Galil graphs with $25 \leq N \leq 11881$. We depict modified Margulis-Gabber-Galil  graphs for $n = 3$ and $n= 6$ in Figure 2(c) of the main manuscript. 
		
		\item 
		\textbf{Sudoku graphs} \cite{herzberg2007sudoku}: The Sudoku graphs are graphs that consist of $N=n^4$ vertices, which are distributed into the cells of an $n^2 \times n^2$ grid. Every cell contains $n^2$ vertices. Two distinct vertices are adjacent if they belong to the same row, column, or cell. Thus, every vertex of Sudoku graph is connected to $3n^2-2n-1$ vertices. The quantity $d_{max}$ grows as $n^2$, which is also reflected in sparsity of the Laplacian matrix. In our computations, we consider all the Sudoku graphs whose vertices lie in the range of  $ 16\leq N \leq 104976$. We draw Sudoku graphs in Fig. \ref{fig:sm_goodgraph_lap}(a) for $n =2 $ and $n=3$.
		
		\item 
		\textbf{Grid 2d  graphs}: To construct a grid 2d graph, we arrange the $N= (r+1)(c+1)$ vertices into $(r+1)$ rows and $(c+1)$ columns, where $r, c \in \mathbb{N}$. Two distinct vertices are adjacent if they are consecutive in a row or a column. These graphs have two parameters $r$ and $c$. In our work we consider two families of grid 2d graphs, which are as follows: 
		\begin{enumerate}
			\item 
			We fix $r = 101$ and vary $c$ from $5$ to $197$ to generate a family of graphs whose vertices lie in the range $602 \leq N \leq 20196$. Two grid 2d graphs belonging to this family are depicted in Figure 3(a) of the main manuscript.
			\item 
			For the other family of graphs we consider $r + 1 = c+ 1 = 2^n$. With this constraint on the number of vertices, we generate graphs with $16 \leq N \leq 65536$. Grid 2d graphs belonging to this family for $n=2$ and $n=3$ are depicted in Fig. \ref{fig:sm_badgraph_lap}(ak).
		\end{enumerate}
		Since every vertex is connected to its immediate neighbouring vertices in rows or columns, its degree is at most $4$. Therefore, the sparsity of the Laplacian matrix of all the grid 2d graphs is $5$.
		
		\item 
		\textbf{Hexagonal lattice graphs}: 
		The graph is constructed on a two dimensional plane having $r$ rows and $c$ columns of hexagons. The vertices of the hexagons constitute the vertices of the graph and the edges of the graph are the sides of the hexagons (see Figure 3(c) of the main manuscript). For our analysis, we set $c=101$ and generate the family of graphs with $r=1$ to $r=108$. Thus, the number of vertices lie in the range $406 \leq N \leq 22234$.
		
		\item 
		\textbf{Triangular lattice graphs:} 
		The construction involves arranging triangles in $r$ rows and $c$ columns on a two-dimensional plane. The corners of the triangles are the vertices of the graph.  Also, the sides of the triangles are the edges of the graph. The number of rows $r$ is varied from $1$ to $100$ while keeping the number of columns $c$ fixed at $101$. The graphs are generated for $N$ in the range $103 \leq N \leq 5202$. Two instances of triangular lattice graphs are presented in Fig. \ref{fig:sm_badgraph_lap}(a).
		
		\item 
		\textbf{Complete graphs:} A complete graph is a graph, such that an edge between any two distinct vertices exists. The degree of all the vertices is $N-1$ in a complete graph with $N$ vertices. We consider all the complete graphs with number of vertices  $2 \leq N \leq 5004$. Two instances of this graph can be found in Fig. \ref{fig:sm_badgraph_lap}(c).
		
		\item 
		\textbf{Turan graphs:} 
		In general, Turan graphs are constructed from $p$ partitions or disjoint sets of different sizes. For our calculations, we set $p=2$, making it a bipartite graph. The number of vertices in those two partitions are $\big\lfloor\frac{N}{2}\big\rfloor$ and $N - \big\lfloor\frac{N}{2}\big\rfloor$. We analyze graphs whose vertices are in the range $5 \leq N \leq 5003$. Two illustrations of the graph family can be found in Fig. \ref{fig:sm_badgraph_lap}(e). 
		
		\item 
		\textbf{$H_{k,n}$ Harary graphs}~\cite{harary1962maximum}: For this family of graphs, we set the vertex connectivity, which refers to the minimum number of edges to be removed from the graph to make it a disconnected graph, to $3$. The number of vertices $N=n$ is varied from $5$ to $5009$. The construction procedure tries to minimize the number of edges in the graph. Two instances of the graph family are presented in Fig. \ref{fig:sm_badgraph_lap}(u).
		
		\item 
		\textbf{$H_{m,n}$ Harary graphs}~\cite{harary1962maximum} : Given $m$ edges and $N=n$ vertices, this construction yields a graph that ensures maximum vertex connectivity. The number of edges in this family of graphs is set to be one greater than the number of vertices that the graph contains. Here, we generate graphs for $5 \leq N \leq 5009$. We present two graphs from this family in Fig. \ref{fig:sm_badgraph_lap}(w).
		
		\item 
		\textbf{Ladder graphs:} 
		A path graph $P = (V, E)$ of order $n$ is a graph with vertex set $V = \{v_1, v_2, \cdots, v_n\}$ and edge set $E = \{(v_i, v_{i + 1}): ~\text{for}~ i = 1, 2, \cdots, (n - 1)\}$. A ladder graph with $N = 2n$ vertices is constructed with all the vertices and edges of two path graphs $P_1$ and $P_2$ each of order $n$ as well as the edges $\{(u_i, v_i): ~\text{for}~ i = 1, 2, \cdots, n ~\text{and}~ u_i \in V(P_1), v_i \in V(P_2)\}$. We consider all the ladder graphs with $ 10 \leq N \leq 5000$. Two graphs of this family are presented in Fig. \ref{fig:sm_badgraph_lap}(aa).
		
		\item 
		\textbf{Circular ladder graphs}: 
		A cycle graph $C = (V, E)$ of order $n$ is a graph with vertex set $V = \{v_1, v_2, \cdots, v_n\}$ and edge set $E = \{(v_i, v_{i + 1}): ~\text{for}~ i = 1, 2, \cdots, (n - 1)\} \cup \{(v_n, v_1)\}$. A circular ladder graph with $N = 2n$ vertices is constructed with all the vertices and edges of two cycle graphs $C_1$ and $C_2$ each of order $n$ as well as the edges $\{(u_i, v_i): ~\text{for}~ i = 1, 2, \cdots, n ~\text{and}~ u_i \in V(P_1), v_i \in V(P_2)\}$. In this case also, we consider all the circular ladder graphs with $10 \leq N \leq 5000$. Two graphs of this family are given in Fig. \ref{fig:sm_badgraph_lap}(y).
		  
		\item 
		\textbf{Ring of cliques:}
		A clique is a complete subgraph of a graph. A ring of clique is constructed by connecting each clique by a single edge. In this work, we set the number of vertices in each clique to be $3$, and thus all cliques are of equal size. Hence, the number of vertices in the graph grows as $N=3n$, for $n$ cliques. We consider all possible graphs with vertices in the range $15 \leq N \leq 5001$. Fig. \ref{fig:sm_badgraph_lap}(ac). 
		
		\item 
		\textbf{Balanced trees:}
		A tree is a graph without any cyclic subgraph. A rooted tree is a tree with a marked vertex called root. Usually, we label it with $0$. A leaf is a vertex in a tree with degree $1$. We say a tree is a balanced tree if the distance from the root to all the leaf vertices are equal. We call it the height of the balanced tree, which may be considered as a parameter to construct the families of the balanced trees. We denote it by $n$. If the degree of the root vertex is $r$ then the degree of all other non-leaf vertices is $(r + 1)$. We consider the following two families of balanced trees for our investigations: 
		\begin{enumerate}
			\item 
			\textbf{Balanced binary tree:} 
			In a balanced binary tree the degree of the root vertex $r = 2$. If the height of the tree is $n$, then the number of vertices, $N = 2^{n+1}-1$. For our calculations, we consider $n$ from $2$ to $14$. Two graphs from the graph family are presented in Fig. \ref{fig:sm_badgraph_lap}(ae).
			
			\item 
			\textbf{Balanced ternary tree:}
			Similarly, in a balanced ternary tree the degree of root vertex $r = 3$. If the height of the tree is $n$, then the number of vertices, $N = (3^{n+1} - 1)/2$. For our calculations, we consider $n$ from $2$ to $9$. Two instances are presented in Fig. \ref{fig:sm_badgraph_lap}(ag).
		\end{enumerate}
		\item 
		\textbf{Binomial tree:}
		A binomial tree of order $n$ is created by linking the root vertices of two identical copies of binomial trees of order $n-1$. A binomial tree of order $0$ has a single vertex and this vertex acts as a root vertex in successive graphs. Suppose that the root vertex of binomial tree of order $n-1$ is $v_1$. We create its identical copy, whose root vertex is now marked as $v_c$, and join $v_1$ and $v_c$ to construct a binomial tree of order $n$. The root vertex $v_c$ can be viewed as the leftmost child of $v_1$. The number of vertices, $N$, in the tree grows as $N=2^n$. For our calculations, we consider $n$ from $2$ to $14$. Two instances of the graph family are given in Fig. \ref{fig:sm_badgraph_lap}(ai).
	\end{enumerate}
	
	\subsubsection{Random graphs}
    \label{Random_graphs_Laplacian}
    
	For all the below random graphs, we fix the seed value to be $23$, unless specified.
	\begin{enumerate} 
		\item 
		\textbf{Random regular expander graphs:} Our procedure for constructing the random regular expander graph is as follows. We generate a random regular graph with regularity $k$ and $N$ vertices using the construction prescribed in Ref. \cite{friedman2003proof}. Next, we check if the second largest eigenvalue of the adjacency matrix $\lambda_2 \le 2\sqrt{k-1}$, where $k$ is an even number. In our construction, we set $k=6$. Recall that all the Ramanujan graphs ~\cite{murty2020ramanujan}, which is a prominent family of expander graphs, must fulfill this property. We repeat the procedure at most $200$ times for getting a Ramanujan graph. We experience that this process could not produce a Ramanujan graph in $3$ percent of the considered values of $N$, which is a negligible percentage. Varying the value of $N$ between $9$ and $5009$ we construct a family of Ramanujan graphs, used in our investigation. We present two graphs of this family in Fig. \ref{fig:sm_goodgraph_lap}(g). 
		
		\item 
		\textbf{Barabási-Albert graphs}\cite{barabasi1999emergence}: 
         The Barab\'{a}si-Albert graph is a class of random scale-free networks, which grows in size via the following mechanism. A new vertex is adjacent to $k$ old vertices such that the probability of choosing an old vertex depends on the degree of the old vertex. Therefore, the vertices with high degree are preferred while growing the graph, leading to fewer vertices with large vertex degree. In our analysis, we set $k=3$ and consider all graphs in this class with the number of vertices $5 \leq N \leq 5009$. Two of the graphs from this family are presented in Fig. \ref{fig:sm_goodgraph_lap}(e). 
		 
		\item 
        \textbf{Newman-Watts-Strogatz graphs} \cite{newman1999renormalization}: We first arrange $N$ vertices over a ring lattice. Each vertex in the ring is connected to its $g$ nearest neighbors or $(g - 1)$ nearest neighbors if $g$ is odd. Further we add new edges as follows. For each edge $(v_i, v_j)$ in the underlying ``$N$-ring with $g$ nearest neighbors", with probability $p$ add a new edge $(v_i, v_k)$ with randomly-chosen existing node $v_k$. 
        For our analysis, we set $g=3$ and $p=1$ with seed value $19$. We generate all Newman-Watts-Strogatz graphs with the number of vertices in the range $5 \leq N \leq 5005$. Fig. \ref{fig:sm_goodgraph_lap}(c) contains two instances of this family. 
        
		\item 
        \textbf{Random regular graphs}: We generate the random regular graphs using the algorithms discussed in \cite{steger1999generating, kim2003generating}. Here, the product of the number of vertices $N$ and the regularity $k$ must be an even number. For our calculations, we fix $k = 4$ and generate all graphs with the number of vertices $5 \leq N \leq 5013$. Two graphs from the family are given in Fig. \ref{fig:sm_goodgraph_lap}(i).
        
		\item 
        \textbf{G$_{n,p}$ random graphs:} We generate all G$_{n,p}$ random graphs, in other words the Erd\"{o}s-R\'{e}nyi random graphs \cite{erdds1959random}, with the number of vertices $5 \leq N \leq 5009$. The only parameter in this construction is $n$ which sets the number of vertices in the graph, i.e., $N=n$. Here, we set the probability of existence of an edge between two randomly selected vertices to be $p = 0.8$. Two graphs of this family are given in Fig. \ref{fig:sm_badgraph_lap}(ao).
        
		\item 
        \textbf{Gaussian random partition graphs} 
        \cite{brandes2003experiments}: A Gaussian random partition graph is constructed by generating $k$ partitions on the set of vertices $N$. The size of the partitions are drawn from a normal distribution with a fixed mean and variance. For our work, we consider the Gaussian random partition graphs with $k=5$ partitions whose size is determined by the Gaussian distribution with mean $5$ and variance $1$. The probability of establishing the edges between two vertices inside a partition is set to be $0.5$ and in between the partitions is fixed to be $0.4$. For our investigation, we generate all the Gaussian random partition graphs with vertices $5 \leq N \leq 5009$. Two of the graphs from this family are given in Fig. \ref{fig:sm_badgraph_lap}(g). 
        
		\item 
        \textbf{Geographical threshold graphs} 
        \cite{masuda2005geographical, bradonjic2007giant}: To construct a graph from this family, we place all the $N$ vertices on a Cartesian plane, with the distances between the vertices chosen at random. Each vertex is assigned a random weight drawn from an exponential distribution with rate parameter $k = 1$, in our case. Two vertices $p$ and $q$ with weights $x_p$ and $x_q$ are connected if $(x_p+x_q)r^{-2} \ge 10$, where $r$ is the Euclidean distance between two vertices. We generate graphs of this family in the range $5 \leq N \leq 5009$. Two of the graphs from this family are presented in Fig. \ref{fig:sm_badgraph_lap}(i).
		
		\item 
        \textbf{Soft random geometric graphs} 
        \cite{penrose2016connectivity}: To construct the soft random geometric graph, we place $N$ vertices randomly on a Cartesian plane. A probability function determines the connection between two vertices, which takes the Euclidean distance between two vertices as an input. If the Euclidean distance, $r$, between two vertices is at most $1$, then the probability of joining them by an edge is decided by the probability density function $e^{-r}$. Here, we generate graphs with number of vertices $5 \leq N \leq 5009$. Fig. \ref{fig:sm_badgraph_lap}(k) presents two graphs from this family.

        \item 
        \textbf{Thresholded random geometric graphs:} To construct this graph family, we place $N$ vertices randomly on the Cartesian plane, and each vertex is assigned a weight that is randomly chosen from an exponential distribution of rate parameter $k$. We set $k = 1$. Two vertices are connected by an edge if the sum of the their weights is greater than or equal to $2$ and the Euclidean distance between the vertices is less than or equal to $1$. We generate graphs with vertices $5 \leq N \leq 5009$. Fig. \ref{fig:sm_badgraph_lap}(m) displays two instances of this graph family. 
        
		\item 
        \textbf{Planted partition graphs} \cite{condon2001algorithms}: A planted partition graph has $N = l \times n$ vertices distributed into $l$ groups with $n$ vertices in each. Two vertices in the same group are linked with a probability $p$. Two vertices belonging to two different groups are linked with probability $q$. For our calculations, we consider $l = 2$, $p = 0.5$ and $q = 0.4$. To generate a graph family, we vary $n$, such that, the number of vertices $N$ varies from $10$ to $5014$. Fig. \ref{fig:sm_badgraph_lap}(o) presents two instances from this graph family.
        
		\item 
        \textbf{Random geometric graphs} 
        \cite{penrose2003random}: We place $N$ vertices randomly on a Cartesian plane join two vertices by an edge if their Euclidean distance, $r$ is at most $1$. For our calculations, we construct a graph family where the the number of vertices range from $7$ to $5008$. Fig. \ref{fig:sm_badgraph_lap}(q) contains two instances of this graph family.
        
		\item 
        \textbf{Uniform random intersection graphs}
        \cite{fill2000random}: These graphs are generated from random bipartite graphs. First, we construct a bipartite graph with two partite sets containing $N$ and $N - 3$ vertices, respectively. The probability of existence of an edge between the partite sets is set to be $0.6$. Then, the bipartite graph is projected onto the partite sets having $N$ vertices. There is an edge between two vertices in the resultant graph, if those two vertices share a common neighbor in the first bipartite graph. We vary the number of vertices $N$ from $5$ to $2999$, due to computational limitations. Fig. \ref{fig:sm_badgraph_lap}(s) represents two graphs belonging to this family.
		
		\item 
		\textbf{Random lobster graphs:} 
        This graph is a tree which generates a caterpillar graph when all the leaf vertices are removed. We vary the number of vertices present in the backbone of the graph to generate a family. An edge between two vertices belonging to the backbone of the graph is drawn with probability $0.6$, while the probability of edge creation beyond backbone level is set to be $0.5$. We set the seed value to $19$. The number of vertices in the graph varies in the range $20 \leq N \leq 5014$. Fig. \ref{fig:sm_badgraph_lap}(am) contains two instances of random lobster graph family.
	\end{enumerate}

\appsubsection{Graphs whose incidence matrices are considered in our studies} 

    For our work, we have random and non-random directed graphs. These graphs may have a source or a sink vertex by the virtue of their construction. Therefore, we follow two approaches for random directed graphs. In the first approach, we carry out our analysis by retaining the source and sink vertices of the graph, whereas in the second approach, we minimally modify the edge set to not have any sink or source vertices in the graph. In the latter case, the number of vertices remain the same as the original graph. We recall that, for incidence matrices, system size $\mathcal{N}$ refers to the sum of number of vertices and edges, $\mathcal{N} = N' = N+M$. We also do not allow multiple edges and self-loops. 

\subsubsection{Non-random directed graphs}
We consider two non-random directed graph-families in our work. Below, we mention them briefly.
\begin{enumerate} 
    \item 
    \textbf{Paley graphs:} 
    For an odd prime number $N$, the vertex set of the Paley graph is $\mathbb{Z}/N\mathbb{Z}$. A directed edge exists from vertex $u$ to $w$ if $u - w \equiv{x^2 (\text{mod}\ N)}$, where $x^2 \in \mathbb{Z}/N\mathbb{Z}$ and $N$ is the number of vertices. Also, to avoid bi-directed edges, $N \equiv 3  (\text{mod}\ 4)$, so that $-1 \not \equiv x^2(\text{mod}\ N)$. We follow the construction given in \cite{networkx}. Our system size $N'$ varies from $6$ to $9730$. Fig. \ref{fig:sm_good_inc}(o) presents two instances of this graph family.
     
    \item 
    \textbf{Directed hypercube graphs:}
    Similar to the undirected hypercube graphs of dimension $n$, we consider the vertex set $V(G_2^n) = \{0, 1\}^{\times n}$ which contains $N=2^n$ vertices. There is a directed edge from vertex $v_1$ to $v_2$ if the Hamming distance between the vertices is $1$ and the Hamming weight of $v_1 < v_2$. Here, Hamming weight of a vertex refers to the number of $1$s in the corresponding $n$-tuple. This arrangement of directed edges indicates that there is only one source and sink in a directed hypercube graph. The number of edges in a hypercube graph grows as $2^{n-1}n$. Hence, the system size in a directed hypercube graph grows as $N' = 2^n(1+n/2) \sim \mathcal{O}(2^nn)$. The number of vertices in the graph assists the graph-family to grow in terms of $2^n$, where $n=2, 3, \cdots$. Our system size varies $N'$ from $8$ to $131072$. Two instances of this graph family can be found in  Figure 4(a) of the main manuscript. 

\end{enumerate}

\subsubsection{Random directed graphs}
We consider nine random directed graphs in this article. For the sake of reproducibility, we fix seed value of every random directed graph family to be $19$. Further, we followed the constructions from Ref.~\cite{networkx} with slight modifications to avoid bi-directed edges and self-loops in the graphs. 

\begin{algorithm} 
    \caption{To modify source/sink vertices to non-source/non-sink vertices}\label {alg:source_sink}
    
    \begin{algorithmic}[1]
    \Statex \textbf{Input:} Directed graph, $G$, with source and/or sink vertices.
    \Statex \textbf{Output:} Directed graph without source and sink vertices. 
    \For{every vertex $v \in V(G)$}
        \If{in-degree$(v)$ + out-degree$(v) = 1$}
            \If{in-degree$(v) = 1$}
                \State Find the vertex $u$ for which $\overrightarrow{(u, v)} \in E(G)$.
                \State Add edge $\overrightarrow{(v, y)}$ and $y \notin \{v, u\}$ chosen at random.

            \Else 
                \State Find the vertex $u$ for which $\overrightarrow{(v, u)} \in E(G)$.
                \State Add edge $\overrightarrow{(y, v)}$ and $y \notin \{v, u\}$ chosen at random.
            \EndIf
        \EndIf
    \EndFor

    \State Create two lists containing the source and sink vertices: $source, \ sink$.

    \If {$source$ and $sink$ are empty}
        \State print "There are no source and sink vertices in the graph".
        
    \ElsIf {$source$ is empty}
        \State Create $list1 = \{u | u \notin sink\}$.
        
        \State Add edges from vertices in $sink$ to vertices in $list1$ chosen at random. If there are edges toward $sink$ from vertices in $list1$ whose out-degree$~>1$, reverse them. 
    \ElsIf {$sink$ is empty}
            \State Create $list2 = \{v | v \notin source\}$.
            \State Add edges towards vertices in $source$ from vertices in $list2$ chosen at random. If there are edges from $source$ to vertices in $list2$ whose in-degree$~>1$, reverse them.
    \ElsIf{$\text{length} (source)  \le \text{length}(sink)$}
            \State Add an edge from sink vertex to corresponding source vertex without creating a bi-directed edge.
            \State For the remaining sink vertices, choose vertices from $source$ at random and build an edge from sink vertex to a randomly chosen source vertex.
    \Else 
                \State Repeat step 22.
                \State For the remaining source vertices, add edges from randomly chosen sink vertices to these source vertices.         
    \EndIf
  \end{algorithmic}
\end{algorithm}

\begin{enumerate}

    \item 
    \textbf{GN graph:} A GN graph grows the network by linking a newly added vertex to one of the existing vertices based on their degree. We define a linear function $f(d)=d$, where $d$ is the degree of the vertex. This function determines the probability of connecting the new vertex to one of the existing vertices. In our work, the value of system size $N'$ goes from $9$ to $9997$ for GN graphs with source and sink vertices. Also, the system size varies from $11$ to $9998$ for GN graphs without source and sink vertices. Figures 4(e) and 4(g) of the main manuscript present two instances each for without and with source and sink vertices, respectively.
    
    \item 
    \textbf{GNC graph} 
    \cite{krapivsky2005network}: A GNC graph is grown by adding a directed edge from a new vertex to an existing vertex chosen randomly. Further, directed edges go from the new vertex to the successors of the existing vertex to which it is connected. A vertex, $v_u$, is said to be the successor to a vertex $v_w$, if there is a directed edge going from $v_w$ to $v_u$. For GNC graphs with source and sink vertices, we consider $N'$ in the range of $11$ to $8498$. Also, for the family of GNC graphs without any source and sink vertices, the system size $N'$ varies from $11$ to $6759$. Two GNC graphs without and with source and sink vertices are depicted in Figs. \ref{fig:sm_bad_inc}(a) and \ref{fig:sm_bad_inc}(c) respectively.
    
    \item 
    \textbf{GNR graph}~\cite{krapivsky2001organization}: 
    The construction of the GNR graphs is similar to the GN graphs. Here, when a new vertex is added, a directed edge appears from the new vertex to a randomly chosen existing vertex, called the target vertex. There is a finite chance for this new directed edge to get redirected to first successor of target vertex. For our work, we set the probability of reconnection to $0.5$. For the GNR graphs with source and sink vertices, the system size grows from $9$ to $9997$. Also, for GNR graphs without source and sink vertices, $N'$ varies from $12$ to $10000$. Figs. \ref{fig:sm_bad_inc}(e) and \ref{fig:sm_bad_inc}(g) respectively present the figures for graphs without and with source and sink vertices belonging to this family.
    
    \item 
    \textbf{Gaussian random partition graphs:} The underlying principles for the directed Gaussian random partition graph is similar to its undirected version as described in sub-section \ref{Random_graphs_Laplacian}. We follow the construction from Ref. \cite{networkx} to build directed Gaussian random partition graphs. The values of the parameters to build this graph family are same as the values set in the undirected version. We vary $N'$ from $12$ to $10897$ for both the cases. Figure 4(c) contains two instances of this family for without source and sink vertices and Figure \ref{fig:sm_good_inc}(a) has illustrations of two graphs for with source and sink vertices of this family. 
    
    \item 
    \textbf{Planted partition graphs:} 
    Our construction for the planted partition graph is also similar to their undirected mentioned in the sub-section \ref{Random_graphs_Laplacian}. Here, we fix the number of vertices present in each partition to be $5$ and vary the number of partitions to grow the graph. We set the probability of linking vertices inside a partition to $0.8$ and between the partitions to $0.4$. For graphs with sink and source vertices, we vary $N'$ from $228$ to $24405$. For graphs without sink and source vertices, we vary $N'$ from $228$ to $20208$. Two graphs each for without and with source and sink vertices are presented in Figs. \ref{fig:sm_good_inc}(c) and \ref{fig:sm_good_inc}(e) respectively.
    
    \item  
    \textbf{Navigable small world graphs} 
    \cite{kleinberg2000small}: The graph is defined on a two-dimensional $n \times n$ grid with $N = n^2$ vertices. The distance between two vertices $(a, b)$ and $(c, d)$ is defined as $\mathcal{R} = |a-c| +|b-d|$. Based on the distance between vertices, we define short-range and long-range connections. A vertex $(a, b)$ is connected to all vertices which are at distance $1$ to it. The number of long-range connections of a vertex is set to $1$. The probability of connecting the vertex to a target vertex at a distance $\mathcal{R}$ is proportional to $1/\mathcal{R}^2$. For graphs without source and sink vertices, $N'$ is varied from $84$ to $20537$, whereas for graphs with source and sink vertices, $N'$ goes from $83$ to $14827$. Figs. \ref{fig:sm_good_inc}(g) and \ref{fig:sm_good_inc}(i) show two graphs each for without and with source and sink vertices corresponding to this family. 

    \item 
    \textbf{G$_{n, p}$ random graphs:} 
    Our construction for the G$_{n, p}$ random graph is also similar to their counterparts mentioned in the sub-section \ref{Random_graphs_Laplacian} with same values for the parameters. For graphs without sink and source vertices, $N'$ is varied from $44$ to $42534$, while for graphs with sink and source vertices, $N'$ ranges from $44$ to $21721$. Two graphs each for without and with source and sink vertices can be found in Figs. \ref{fig:sm_good_inc}(k) and \ref{fig:sm_good_inc}(m) respectively.
    
    \item 
    \textbf{Random uniform $k$-out graphs:} 
    In this graph, $k$ directed edges go from every vertex of the graph to $k$ vertices chosen randomly without replacement. For our calculations, we set $k = 2$. For graphs with source and sink, $N'$ takes the values from $12$ to $10943$ and graphs without source and sink, $N'$ goes from $13$ to $10817$. Figs. \ref{fig:sm_good_inc}(q) and \ref{fig:sm_good_inc}(s) present two graphs each for without and with source and sink vertices respectively.
    
    \item 
    \textbf{Scale-free graphs} 
    \cite{bollobas2003directed}: In this graph family, the in-degree and out-degree distributions of vertices follow power laws. We start with a cycle graph with 3 edges. A directed edge from a new vertex to an existing vertex, chosen according to its in-degree distribution, is added with probability $0.41$. An edge between two existing vertices are added with probability $0.54$ and the vertices are chosen according to their in-degree and out-degree distributions. The probability of adding a directed edge from an existing vertex, chosen according to its out-degree distribution, to a new vertex is $0.05$. The probability of choosing an existing vertex would depend on the in-degree (or out-degree) in addition to a constant bias. The constant bias for in-degree distribution is set to be $\delta_{\mathrm{in}} = 0.2$ and for out-degree distribution is set to be $\delta_{\mathrm{out}}=0$. Here, we vary the system size $\mathcal{N}$ from $12$ to $10997$ for graphs without sink and source vertices. For graphs with sink and source vertices, $\mathcal{N}$ is varied from $11$ to $9995$. Figs. \ref{fig:sm_good_inc}(u) and \ref{fig:sm_good_inc}(w) present two instances each for graph without and with source and sink vertices respectively.
\end{enumerate} 

\section{Beyond graph theoretic considerations: Other requirements and challenges} \label{sec:req}

\appsubsection{Algorithmic challenges} 

\begin{enumerate} 
\item \textbf{State preparation for $\ket{b}$:} This scales exponentially in system size in a general setting, and is one of those important open problems not only for QLSs, but also whenever one employs quantum phase estimation (QPE) as a primitive, and is a major impediment to quantum advantage even if other factors align in our favour. Depending on the specific application of interest, one may be able to identify problem-specific symmetries to reduce the scaling. 

\item \textbf{Extracting a feature of the solution vector:} HHL, for example, loses its exponential advantage if we seek to extract the solution vector, $\vec{x}$. Therefore, in practice, one needs to extract a feature of the vector and not the vector itself. For instance, if one considers the traffic flow problem in Fig. \ref{fig:gblsp}, and asks what a particular $x_i$ is, that is, the number of vehicles on that particular lane, then one could append a SWAP test module at the end of the HHL circuitry in order to obtain an overlap between the solution vector and a computational basis vector to obtain $x_i$. This is reasonable for a traffic flow congestion detection problem, as often, one would want to only understand the traffic congestion on specific lanes and not all the exponentially many lanes in a network. However, the ability to extract a feature is application-dependent, and it is unclear whether or not an efficient protocol exists for feature extraction for any application. 

\item NLSPs often occur in the context of optimization problems, where one envisages an \textbf{iterative procedure} that involves, for example, executing several HHLs in tandem. Such an execution is hardly straightforward, especially given the post-selection step at the end of each iteration. 

\item \textbf{Constraints imposed by specific applications: }If we consider the traffic flow congestion detection problem, the elements of the solution vector are constrained to be positive integers. Furthermore, the application can also naturally admit overdetermined or underdetermined systems, and interpreting the solutions themselves is no longer straightforward. In such cases, modifications to the quantum algorithm and/or the preceding classical or quantum pre-processing steps and their net effect on advantage offered is unclear. 
\end{enumerate}

\appsubsection{Quantum hardware challenges} 

\begin{enumerate} 
\item \textbf{NISQ era requirements: } Although HHL and other QLSs discussed in this work are outside of the scope of NISQ computers and would ideally require a fault-tolerant implementation, it is important to test these solvers on smaller system sizes for various applications even in the NISQ era, continuously pushing the boundaries as quantum hardware advances. A typical theme in carrying out such computations would be to rely upon classical resources (which do not necessarily scale well) to reduce quantum circuit depth and number of qubits, we term these gamut of techniques as resource reduction. In Section VII.1 of the main manuscript, we illustrated these ideas through toy graph instances. Our results show that with current quantum hardware capabilities, we can only probe $(4 \times 4)$ size Laplacian matrices (unless we find special graphs and leverage techniques like all-qubit fixing, as explained in Section VII.2 of the main manuscript), indicating the actual gap between what is possible and what is our goal for the long-term. For perspective, it is routine to solve systems of linear equations involving matrices that are about $(10^9 \times 10^9)$ in size, for example, in quantum chemical computations on traditional computers \cite{Timo}. 

\item \textbf{Fault-tolerant quantum computing era requirements: }Of most relevance in the late NISQ/ early fault-tolerant quantum computing eras would be graph families that provide exponential speed-up and polynomial speed-up (with a larger-than-quadratic degree polynomial), in view of overheads associated with quantum error corrected implementations eating up a sub-linear or a lower degree polynomial speed-up offered by good graph families and some better graph families (for example, see Ref. \cite{prx2021}, where the authors show that quadratic speedups will not enable quantum advantage on early generation fault-tolerant computers that use surface codes). It is hard to predict the extent to which such issues will plague good and better graph families in the fault-tolerant quantum computing era, where one can use codes with lower thresholds. However, in all of these cases, one aims at reducing either the runtime or the number of physical qubits or in an ideal case, both of them, in fault-tolerant QLS implementations. This would in turn depend on achieving lower physical error rates and faster gate times in the quantum hardware that is used, as well as choice of suitable quantum error correction codes. Each of these is formidable, and until these obstacles are cleared, truly realizing an advantage even with better and best graphs is hard. 
\end{enumerate} 

\section{Some theorems relevant to all-qubit fixing}\label{sec:aqf} 

\begin{theorem} \label{th:lapmat} 
Given a Laplacian matrix, $L$, and a vector, $\vec{b}$, such that it has for its entries exactly one $1$ and one $-1$ with the rest of its entries being $0$, any problem to be solved using the HHL algorithm reduces to a one-qubit calculation via all-qubit fixing as long as the input $\vec{b}$ is an eigenvector of $L$. 
\end{theorem}

\begin{proof}
If $\vec{b}$ is an eigenvector of $L$, a QPE calculation would yield only the eigenvalue corresponding to $\vec{b}$, that is, only one bitstring. Thus, all the qubits get fixed in the subsequent HHL calculation. Therefore, every clock register qubit of the QPE module of HHL is set to either $\ket{0}$ or $\ket{1}$, and hence the state register has correspondingly either identity ($I_{2^{n_b}\times2^{n_b}}$) or a unitary ($U^{2^{n_r}}=e^{iAt 2^{n_r}}$) respectively. Each gate cancels out with its counterpart in the following QPE$^\dag$ module. Thus, we get back $\ket{b}$ at the end of the circuit, that is, $|\langle x|b \rangle|=1$. The controlled-rotation module that occurs between QPE and QPE$^\dag$ only now has $RY(\theta_i)$ gates on the HHL ancillary qubit register. In fact, since we recover the eigenvalue corresponding to the input eigenstate, only one $\theta_i$ is non-zero. We measure $Z$ on this trivial one-qubit computation to obtain $P(1) = || \ket{x}_{un}||^2$. 
\end{proof}

\begin{theorem} \label{th:condonlapmat}
The conditions on $L$ so that $\vec{b}=\vec{\delta_i}-\vec{\delta_j}$ is an eigenvector of the matrix are given by $l_{pi}=l_{pj}, \forall p \not \in \{i, j\}, \ \mathrm{and} \ l_{ii}=l_{jj}$. 
\end{theorem}

\begin{proof}
Let $L\vec{b}=\beta\vec{b}$. Since, $\vec{b}=\vec{\delta_i}-\vec{\delta_j}$, we have $b_i=1$ and $b_j=-1$. The action of $L$, whose matrix elements are denoted by $l_{pq}$, on the vector $\vec{b}$ is given by 
\begin{eqnarray}
L \begin{bmatrix} 0 \\  0 \\ \cdots \\ 0 \\ b_i \\ 0 \\ \cdots \\  0 \\ b_j \\ 0 \\ \cdots \\ 0 \end{bmatrix} = \begin{bmatrix} l_{1i}-l_{1j} \\  l_{2i}-l_{2j} \\ \cdots \\ l_{i-1\ i}-l_{i-1\ j} \\ l_{ii}-l_{ij} \\ l_{i+1\ i}-l_{i+1\ j} \\ \cdots \\  l_{j-1\ i}-l_{j-1\ j} \\ l_{ji}-l_{jj} \\ l_{j+1\ i}-l_{j+1\ j} \\ \cdots \\ l_{Ni}-l_{Nj} \end{bmatrix}. 
\end{eqnarray}
Noting that the right hand side should equal $\beta\vec{b}$, we arrive at the conditions 
\begin{equation}
\boxed{l_{pi}=l_{pj}, \forall p \neq i, j, \ \mathrm{and} \ l_{ii}=l_{jj}. } 
\end{equation}

\end{proof} 

Theorem \ref{th:condonlapmat} indicates that the degree of vertices $i$ and $j$ must be equal. Also, for any vertex $p$ in the graph, either both $i$ and $j$ are adjacent to $p$ or not adjacent to $p$. Therefore, both $i$ and $j$ have same sets of neighbors. Also, if there is no edge between $i$ and $j$, the distance between them is $2$. Given any graph, $G$, we can construct a new graph, $H$, by adding two new vertices and a set of edges, such that $H$ will satisfy the conditions of Theorem \ref{th:condonlapmat}. 

\begin{theorem}\label{th:new_graph}
Let $G = (V(G), E(G))$ be any graph with the set of vertices $V(G) = \{v_i:i=1,\ 2,\ 3, \cdots, N\}$ and set of edges $E(G)$. Let $H = (V(H), E(H))$ be a graph such that $V(H) = V(G) \cup \{v_{(N + 1)}, v_{(N + 2)}\}$ and set of edges $E(H) = E(G) \cup \{(v_{(N + 1)}, u_1), (v_{(N + 1)}, u_2), \cdots, (v_{(N + 1)}, u_k)\} \cup \{(v_{(N + 2)}, u_1), (v_{(N + 2)}, u_2), \cdots, (v_{(N + 2)}, u_k)\}$,
where $u_1, u_2, \dots u_k \in V(G)$ in some order. Then, the Laplacian matrix $L(H)$ has eigenvalue $k$ with an eigenvector $\vec{b} = (\underbrace{0,0,\ldots,0}_{N \text{ times}}, 1, -1)^{\mathsf{T}}$.
\end{theorem}

\begin{proof}
    We know that $L(H) = D(H) - Q(H)$, where $D(H)$ and $Q(H)$ are the degree matrix and the adjacency matrix of the graph $H$, respectively. Denote $L(H) = (l_{ij})_{(N + 2) \times (N + 2)}$ where
	$$l_{ij} = \begin{cases} d_i^{(H)} &\text{if}~ i = j; \\ -1 &\text{if}~ (i, j) \in E(H); \\ 0 &\text{otherwise};
	\end{cases}$$
	where $d_i^{(H)}$ is the degree of vertex $v_i$ in the graph $H$. Now multiply $L(H)$ with the vector $\vec{b}$. The $i$-th row of $L(H)\times \vec{b}$ is
	\begin{equation}
		(L(H) \vec{b})_{i \bullet} = \sum_{j = 1}^{N + 2} l_{ij} b_j,
	\end{equation}
	where $i = 1, 2, \cdots, (N + 2)$. Let $b_1 = b_2 = \cdots = b_N = 0$. Now, consider the following cases: 
	\begin{enumerate}[label = {Case \arabic*:}]
		\item 
		Let $v_i$ be a vertex other than $v_1, v_2, \cdots, v_k, v_{(N + 1)}$, and $v_{(N + 2)}$. Now, $l_{i (N + 1)} = l_{i (N + 2)} = 0$. Also, $b_1 = b_2 = \cdots = b_N = 0$. Therefore, $(L(H) \vec{b})_{i \bullet} = 0$.
		\item 
		Let $u_q$ be any one of $u_1, u_2, \cdots, u_k$. Note that the rows of $L(H)$ corresponding to $u_q$ always have two non-zero entries $l_{q (N + 1)} = l_{q (N + 2)} = -1$. As $b_1 = b_2 = \cdots = b_N = 0$, we have $\sum_{j = 1}^{N} l_{q j} b_j = 0$. As $b_{(N + 1)}$ and $b_{(N + 2)}$ have opposite signs, we have $l_{q (N + 1)} b_{(N + 1)} + l_{q (N + 2)} b_{(N + 2)} = 0$. 
		\item 
		Let $i = (N + 1)$ or $i = (N + 2)$. Note that the degrees of $(N + 1)$ and $(N + 2)$ are $d_{(N + 1)}^{(H)} = d_{(N + 2)}^{(H)} = k$. Therefore $l_{(N + 1)(N + 1)} = l_{(N + 2)(N + 2)} = k$. Now, $(L(H) b)_{(N + 1), \bullet} = k$ and $(L(H) b)_{(N + 2), \bullet} = -k$. 
	\end{enumerate}
	Combining all the cases, we observe that 
	\begin{equation}
		L(H)\vec{b} = (\underbrace{0,0,\ldots,0}_{n \text{ times}}, k, -k)^{\mathsf{T}} = k \vec{b}. 
	\end{equation}
	Therefore, $L(H)$ has an eigenvalue $k$ with eigenvector $\vec{b}$.
\end{proof}

\begin{theorem} \label{th:aqf}
A complete graph accommodates a one-qubit HHL via all-qubit fixing. 
\end{theorem}

\begin{proof}
The degree matrix, $D$, is a diagonal matrix with all $N$ entries taking the value $(N-1)$, and thus the condition $l_{ii}=l_{jj}$ is satisfied. Furthermore, since a complete graph is all-to-all connected by construction, $l_{ij}=-1, \forall i\neq j$. Thus, the second condition is satisfied too. Therefore, the Laplacian of a complete graph satisfies Theorem \ref{th:condonlapmat}, and hence Theorem \ref{th:lapmat}. 

An important implication for the complete graph is that one need not \textit{find} an eigenvector; the $1$ and $-1$ entries of $\vec{b}$ can be placed anywhere and such a resulting $\vec{b}$ is still an eigenvector. 
\end{proof}


\begin{figure*}[t]
\begin{tabular}{c}
\includegraphics[width=14cm]{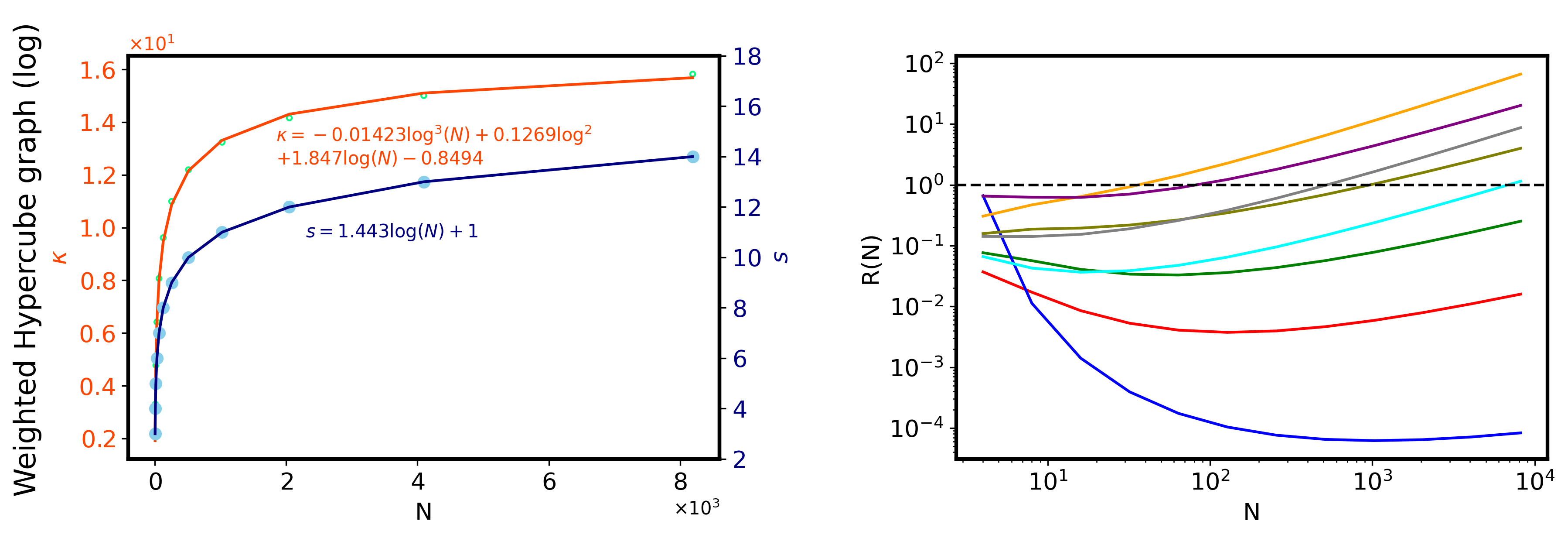} \\ 
\vspace{0.5em}
\makebox[0.45\textwidth]{(a)}
\makebox[0.36\textwidth]{(b)} \\
    
\includegraphics[width=14cm]{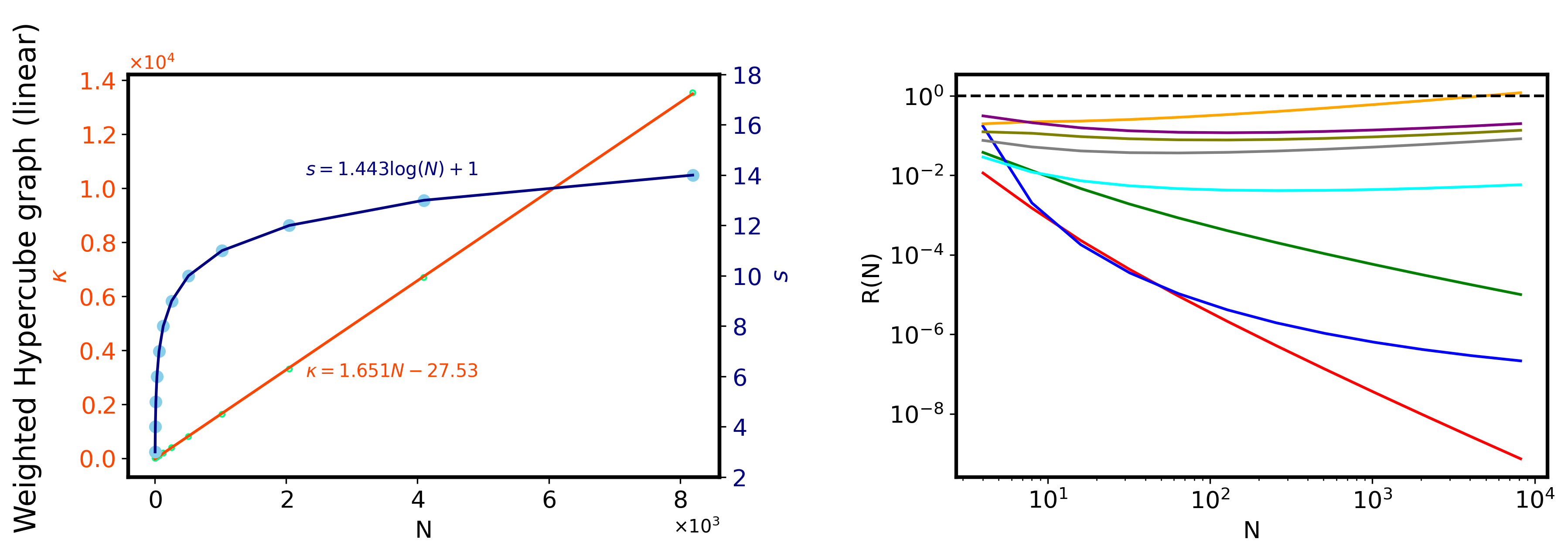} \\
\vspace{0.5em}
\makebox[0.45\textwidth]{(c)}
\makebox[0.36\textwidth]{(d)} \\

\includegraphics[width=14cm]{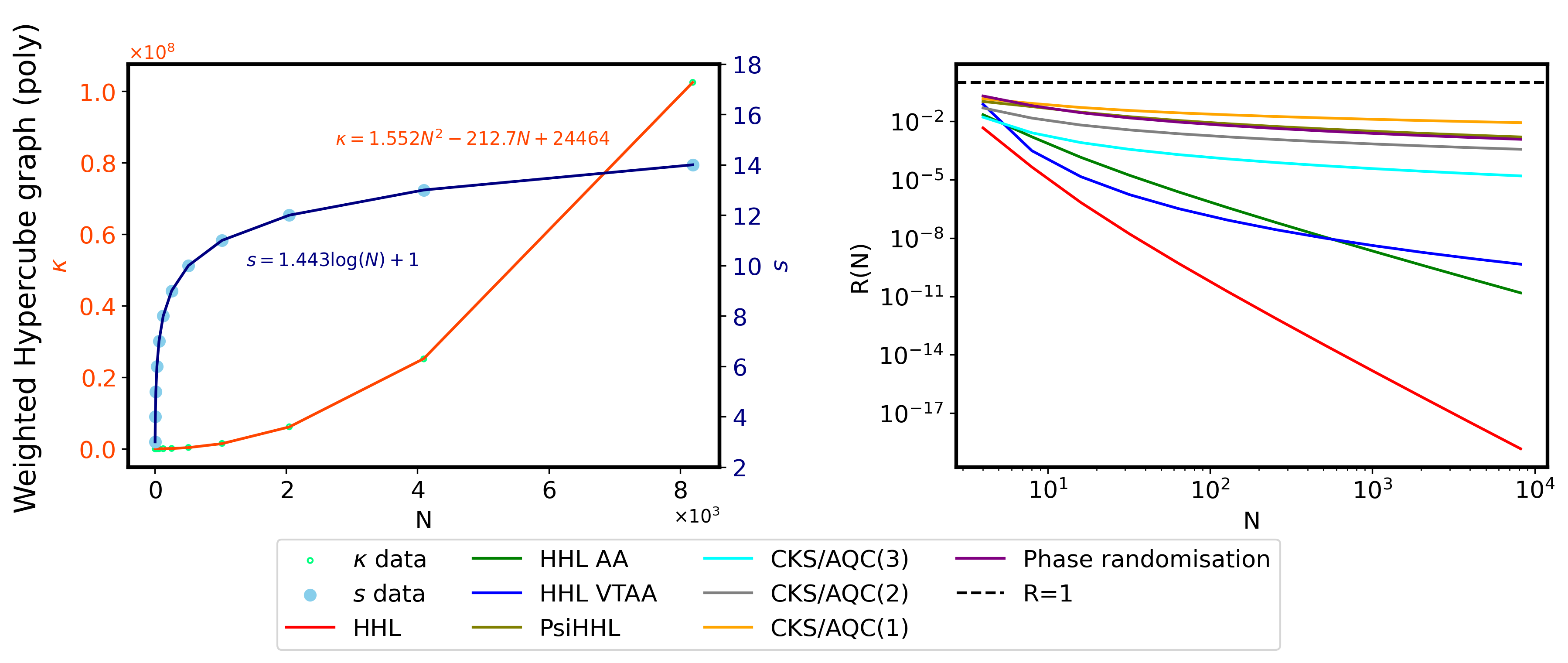} \\
\vspace{0.5em}
\makebox[0.45\textwidth]{(e)}
\makebox[0.36\textwidth]{(f)} \\

\end{tabular} 
\caption{Figures providing the edge weight analysis of the hypercube graph family for three different edge weight functions: logarithmic, given in sub-figures (a) and (b), where the edge weight function $w_{ij} = \mathrm{log}(j+5)$, linear, given in sub-figures (c) and (d) with $w_{ij} = j+1$, and polynomial, shown in sub-figures (e) and (f) for the polynomial function defined by $w_{ij} = j^2+1$. }\label{fig:weighted_HCG}
\end{figure*}

\begin{figure*}[t]
\begin{tabular}{c}
\includegraphics[width=14cm]{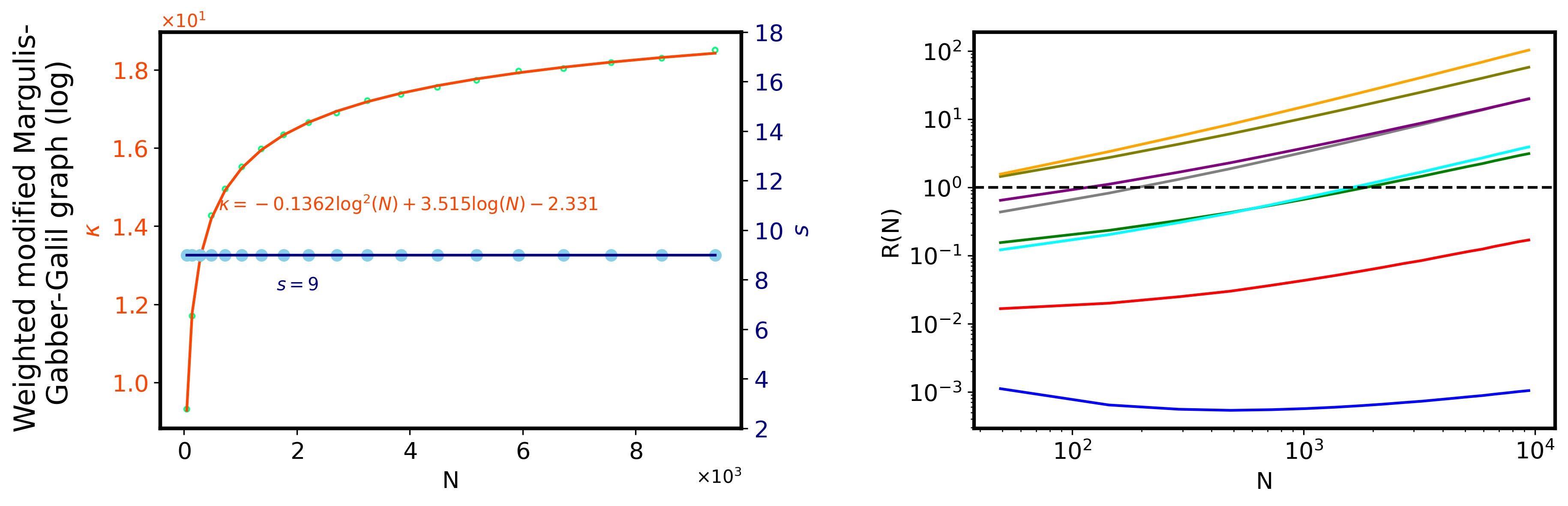} \\ 
\vspace{0.5em}
\makebox[0.45\textwidth]{(a)}
\makebox[0.36\textwidth]{(b)} \\
    
\includegraphics[width=14cm]{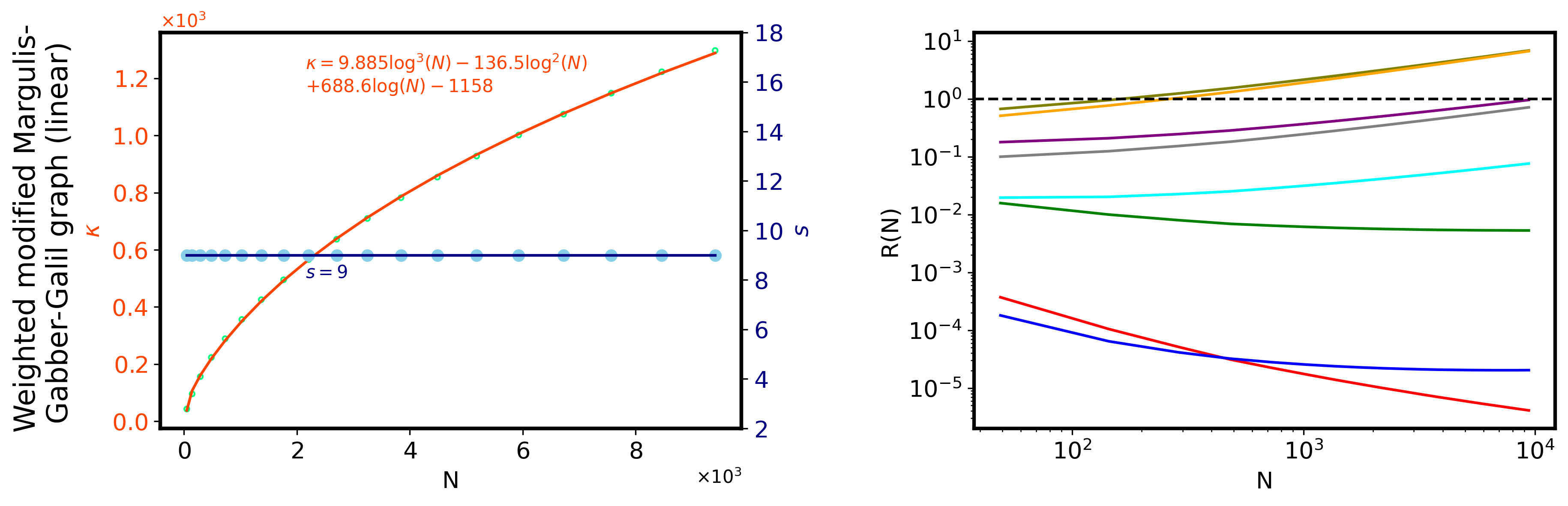} \\
\vspace{0.5em}
\makebox[0.45\textwidth]{(c)}
\makebox[0.36\textwidth]{(d)} \\

\includegraphics[width=14cm]{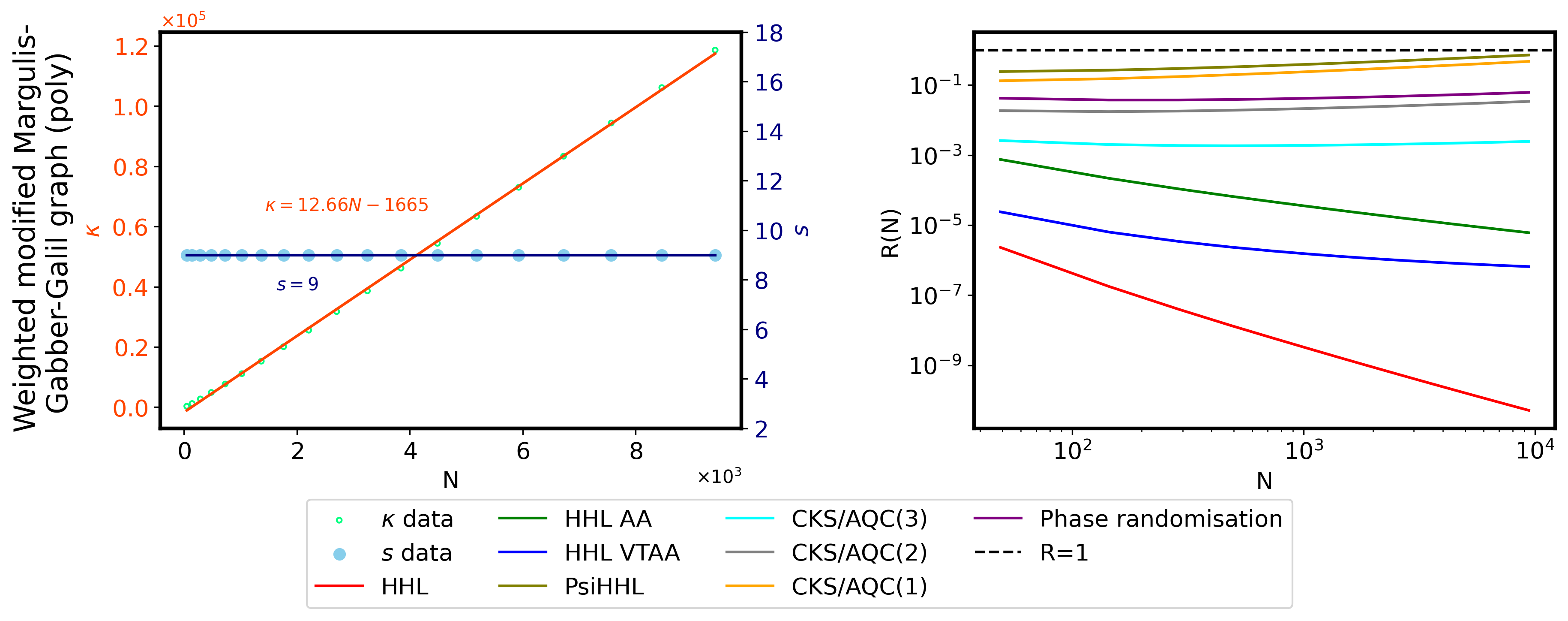} \\
\vspace{0.5em}
\makebox[0.45\textwidth]{(e)}
\makebox[0.36\textwidth]{(f)} \\

\end{tabular} 
\caption{Figures presenting the analysis for modified Margulis-Gabber-Galil graph family, for which each vertex is  denoted by a tuple $(p, q)$. Plots (a) and (b) contain the analysis for graphs with logarithmic edge weights defined as $w_{(p,q)(r,s)} = \mathrm{log}(nr + s+1) + 1$, while sub-figures (c) and (d) presents our results for linear edge weights, $w_{(p,q)(r,s)} = nr + s +1$, and (e) and (f) show the results for edge weight function $w_{(p,q)(r,s)} =(nr+s+1)^2$. Here, the number of vertices in the graph is $N=n^2$.}  \label{fig:weighted_mgg}
\end{figure*} 

\begin{figure*}[h!]
\begin{tabular}{c}
\includegraphics[width=17cm]{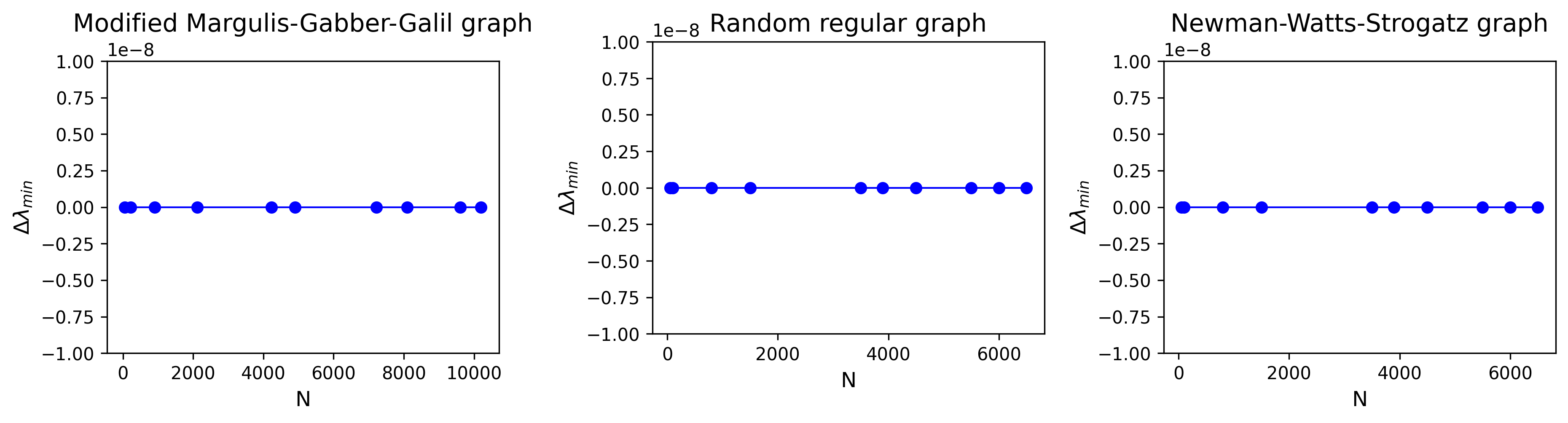 } \\
    \vspace{0.5em}
    \makebox[0.31\textwidth]{(a)}
    \makebox[0.31\textwidth]{(b)}
    \makebox[0.31\textwidth]{(c)}\\ 

\includegraphics[width=17cm]{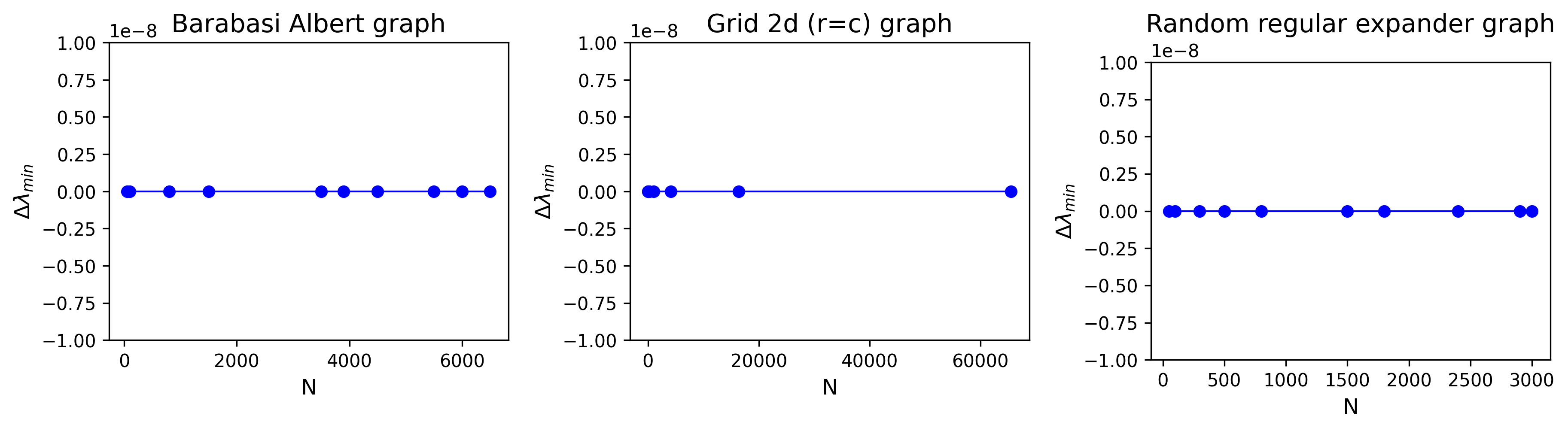} \\
    \vspace{0.5em}
    \makebox[0.31\textwidth]{(d)}
    \makebox[0.31\textwidth]{(e)}
    \makebox[0.31\textwidth]{(f)}\\ 

\includegraphics[width=17cm]{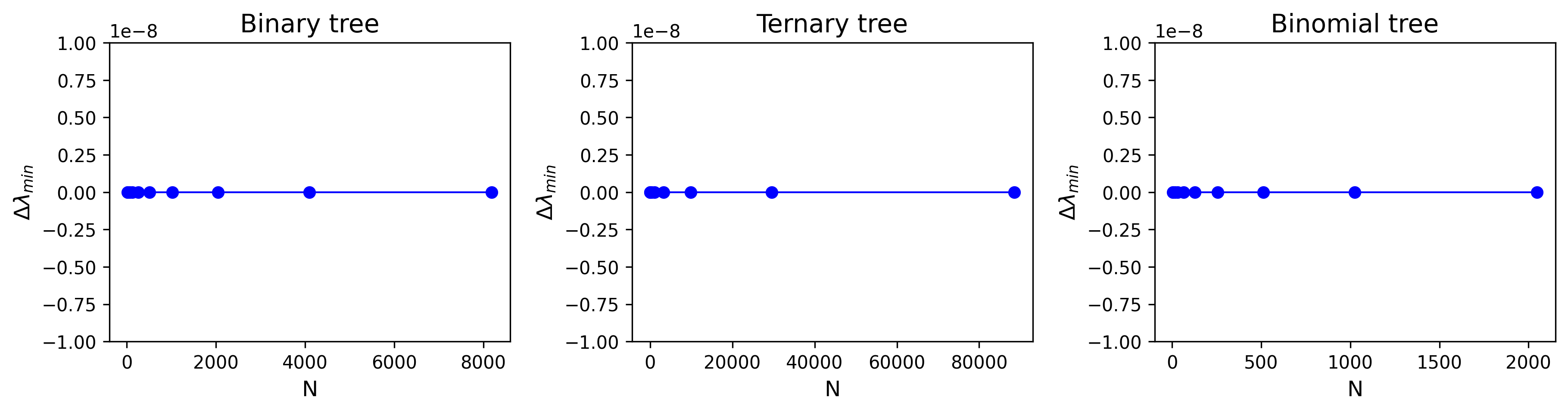} \\
    \vspace{0.5em}
    \makebox[0.31\textwidth]{(g)}
    \makebox[0.31\textwidth]{(h)}
    \makebox[0.31\textwidth]{(i)}\\ 
\includegraphics[width=17cm]{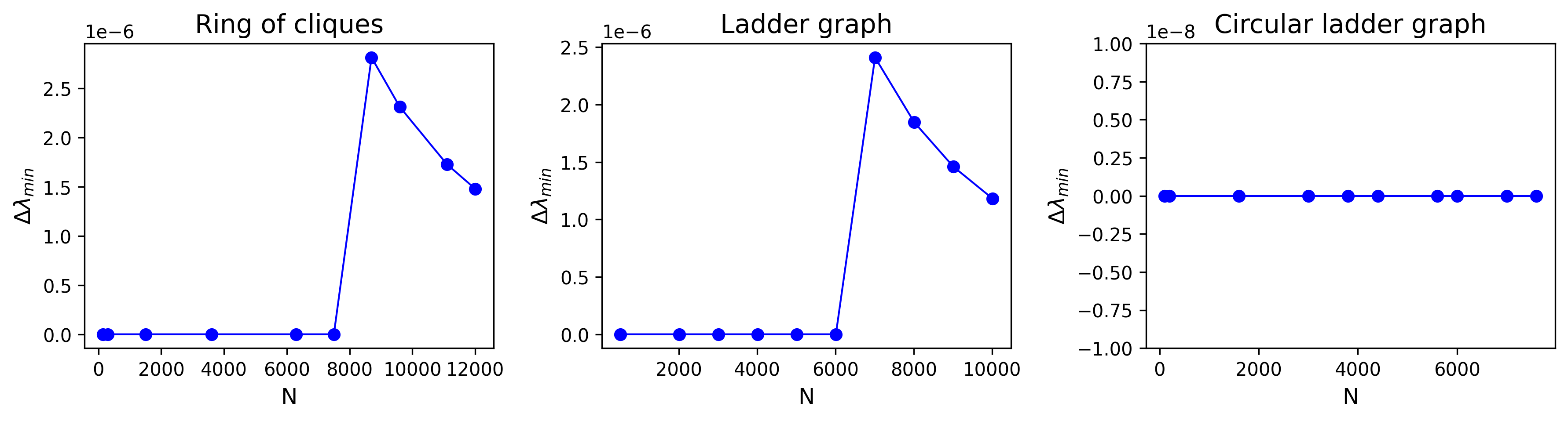} \\
    \vspace{0.5em}
    \makebox[0.31\textwidth]{(j)}
    \makebox[0.31\textwidth]{(k)}
    \makebox[0.31\textwidth]{(l)}\\
\multicolumn{1}{r}{(\textit{Continued})} \\ 
\end{tabular}
\end{figure*}

\begin{figure*}[h!]
\begin{tabular}{c}
\includegraphics[width=17cm]{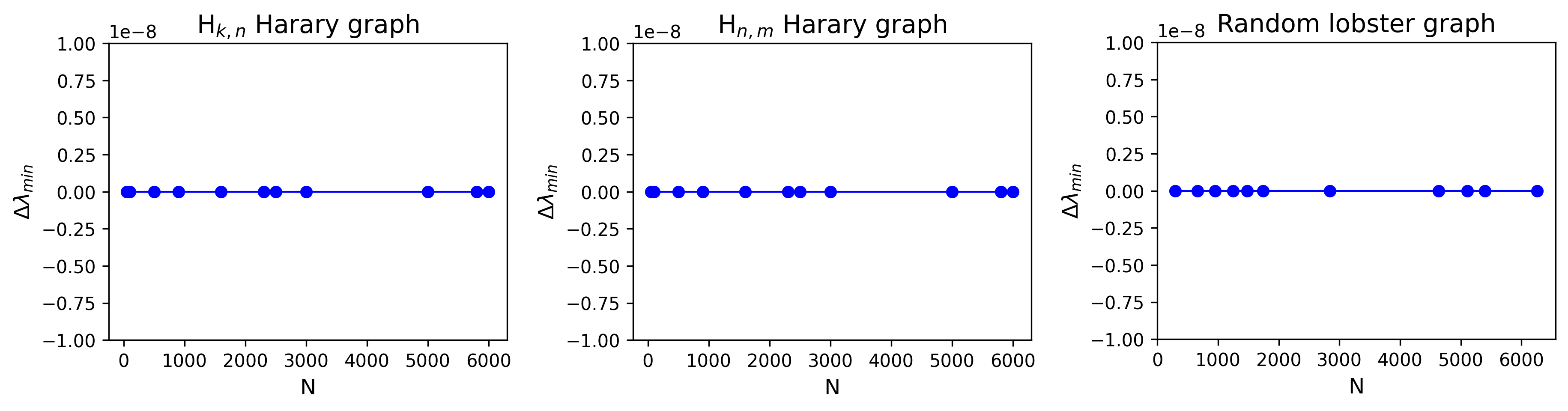} \\
    \vspace{0.5em}
    \makebox[0.31\textwidth]{(m)}
    \makebox[0.31\textwidth]{(n)}
    \makebox[0.31\textwidth]{(o)}\\
\end{tabular} 
\caption{Figures representing the difference in minimum eigenvalues between two threshold choices, $10^{-6}$ (our main results) and $10^{-10}$, for those considered graph Laplacians whose minimum eigenvalues show a downward trend with system size, $N$. }\label{app:Lcutoff}
\end{figure*} 

\begin{figure*}[h!]
\hspace*{-2cm}
\begin{tabular}{cc}
\includegraphics[width=10cm]{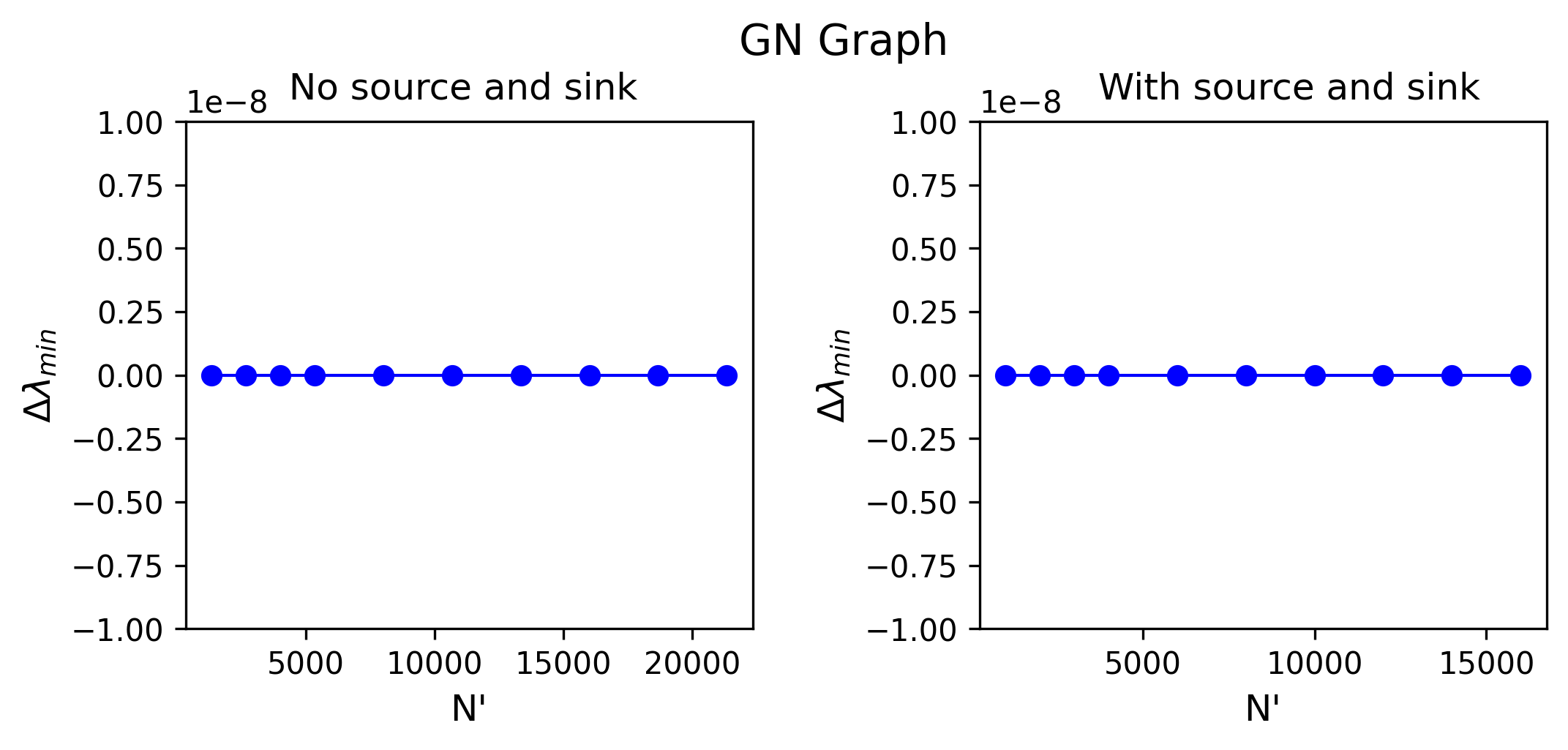 } &\includegraphics[width=10cm]{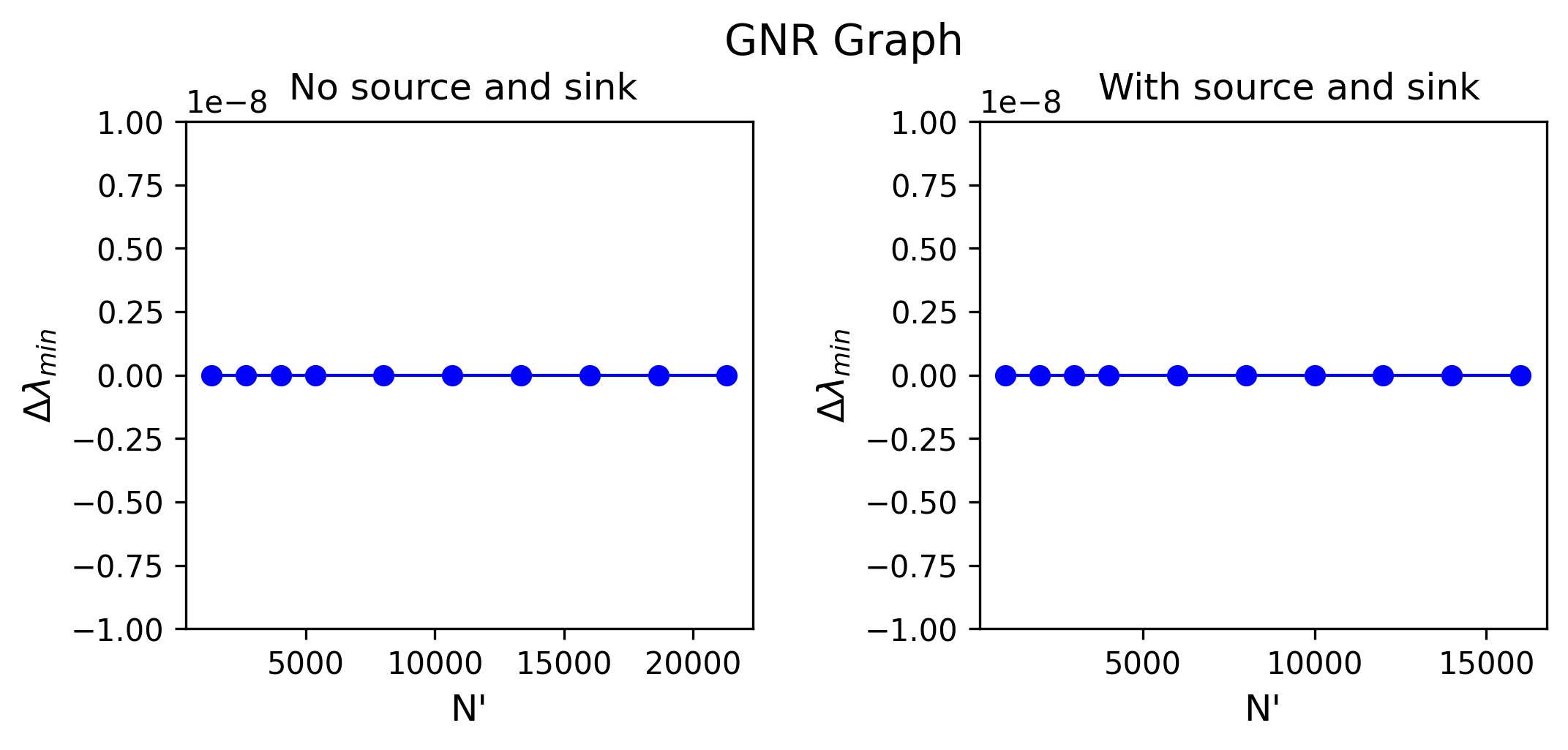} \\ 
\hspace*{0.7cm}\makebox[10cm][c]{\makebox[5cm][c]{(a)}\makebox[5cm][c]{(b)}} &
\hspace*{0.7cm}\makebox[10cm][c]{\makebox[5cm][c]{(c)}\makebox[5cm][c]{(d)}} \\
\includegraphics[width=10cm]{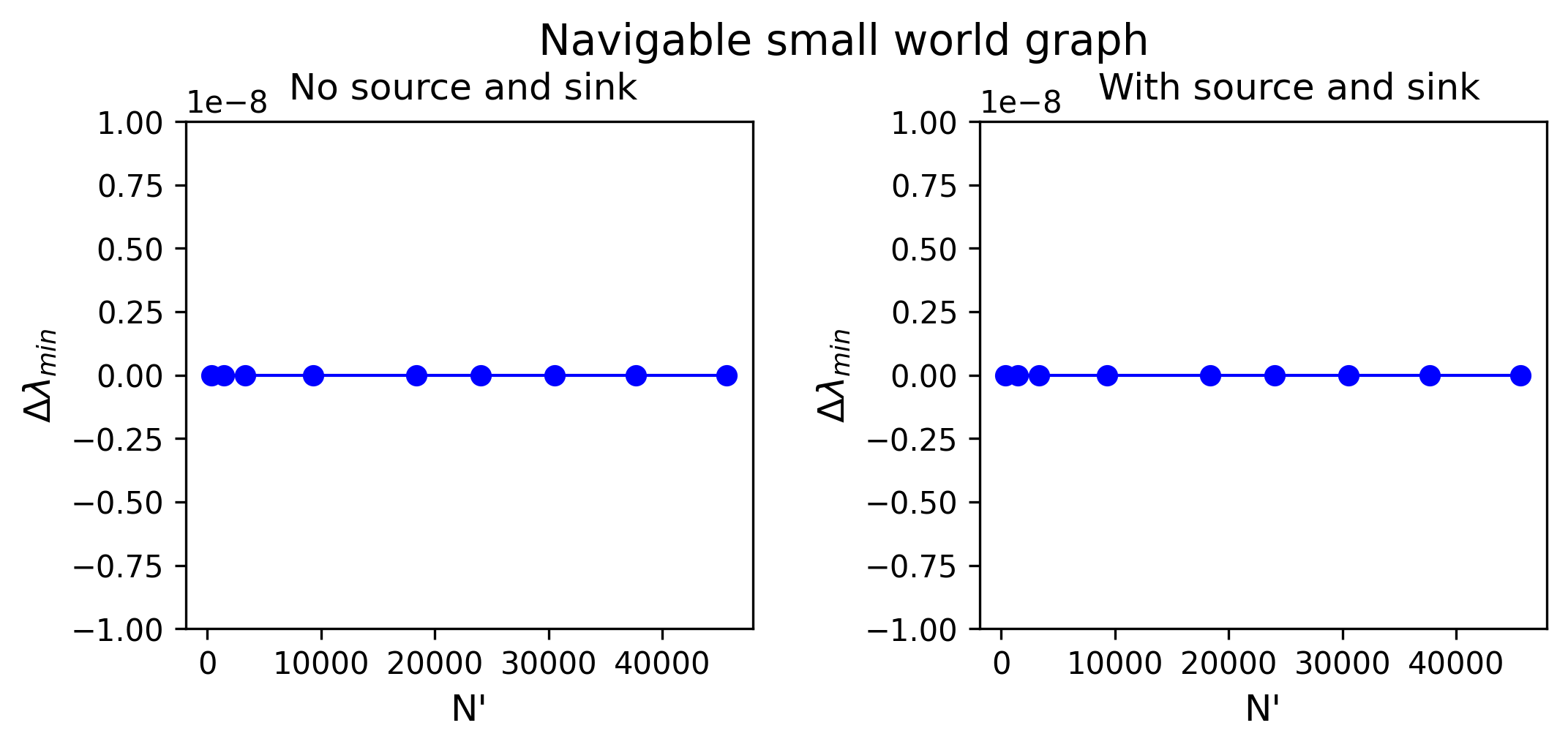} &\includegraphics[width=10cm]{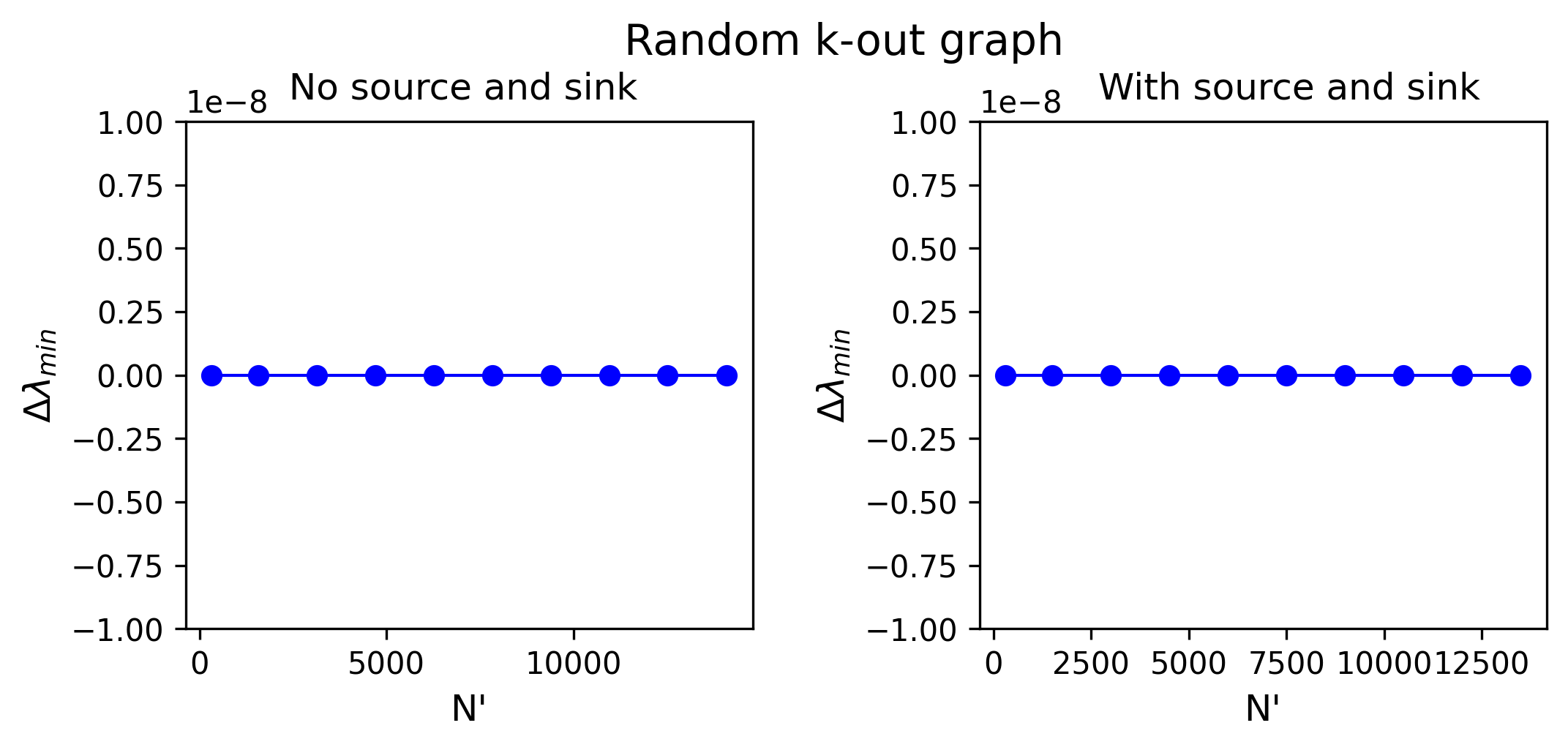} \\
\hspace*{0.7cm}\makebox[10cm][c]{\makebox[5cm][c]{(e)}\makebox[5cm][c]{(f)}} &
\hspace*{0.7cm}\makebox[10cm][c]{\makebox[5cm][c]{(g)}\makebox[5cm][c]{(h)}} \\ 
\multicolumn{1}{r}{(\textit{Continued})} \\ 
\end{tabular}
\end{figure*}

\begin{figure*}[h!]
\begin{tabular}{c}
\includegraphics[width=10cm]{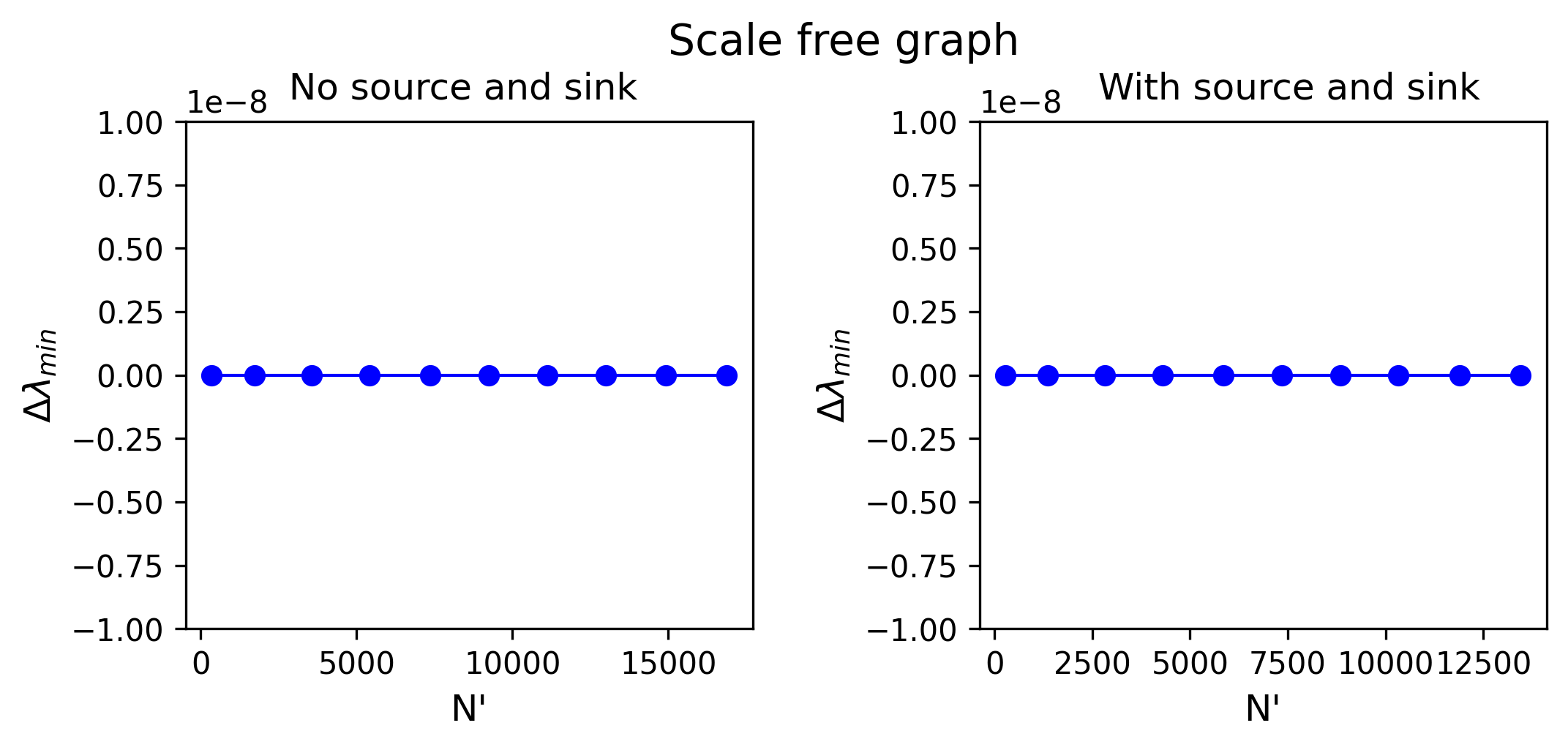} \\
    \vspace{0.5em}
    \makebox[10cm][c]{\makebox[5cm][c]{(i)}\makebox[5cm][c]{(j)}} \\
\end{tabular}
\caption{Figures representing the difference in minimum eigenvalues between two threshold choices, $10^{-6}$ (our main results) and $10^{-10}$, for those considered graph incidence matrices whose minimum eigenvalues show a downward trend with system size, $N'$. } \label{app:Icutoff} 
\end{figure*}

\begin{figure*}[t]
\begin{tabular}{c}
\includegraphics[width=14cm]{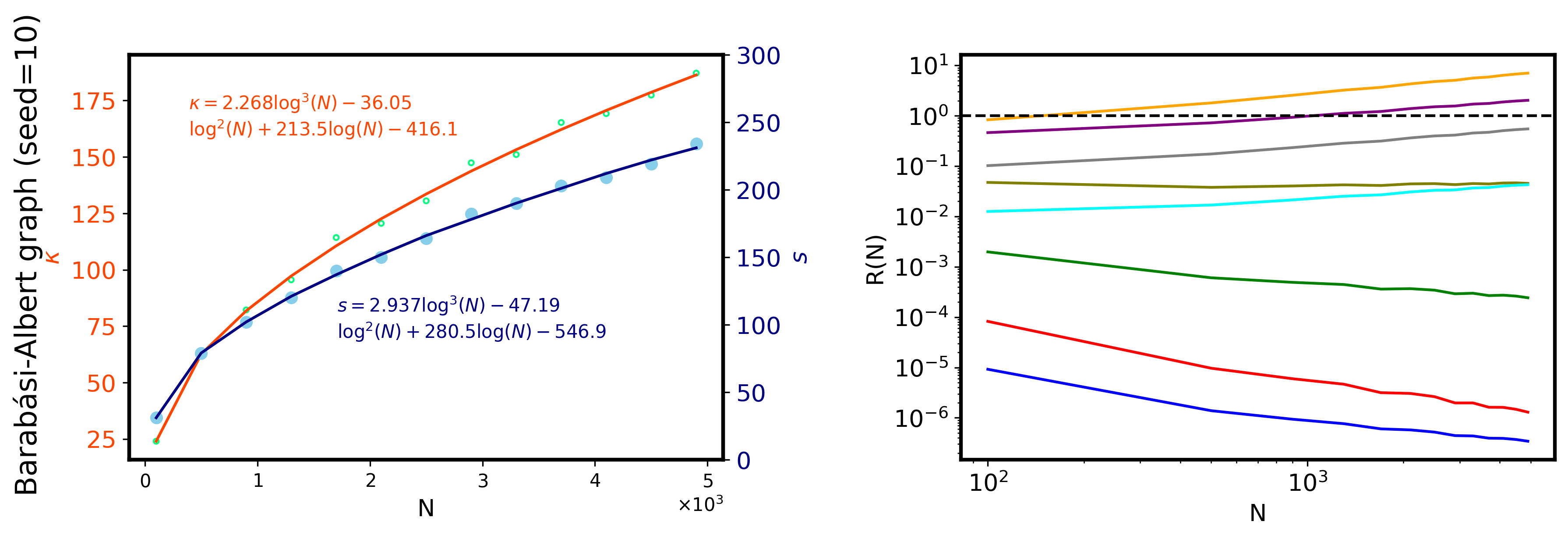} \\ 
\vspace{0.5em}
\makebox[0.45\textwidth]{(a)}
\makebox[0.36\textwidth]{(b)} \\

\includegraphics[width=14cm]{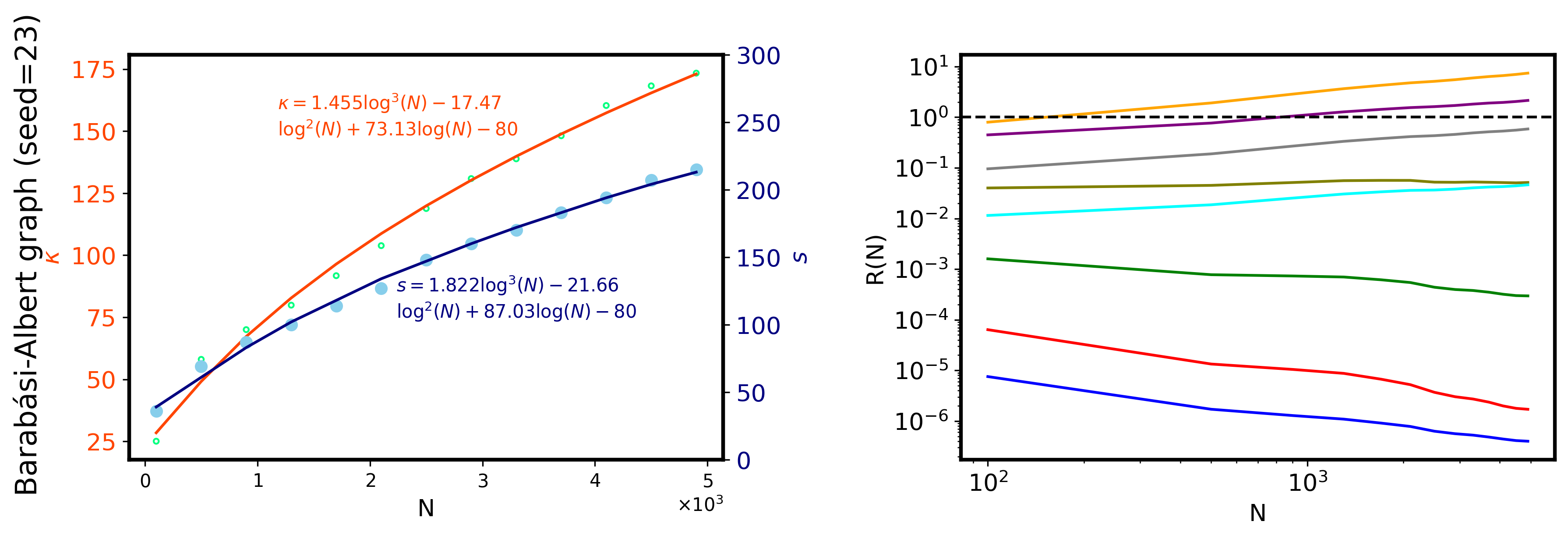} \\
\vspace{0.5em}
\makebox[0.45\textwidth]{(c)}
\makebox[0.36\textwidth]{(d)} \\

\includegraphics[width=14cm]{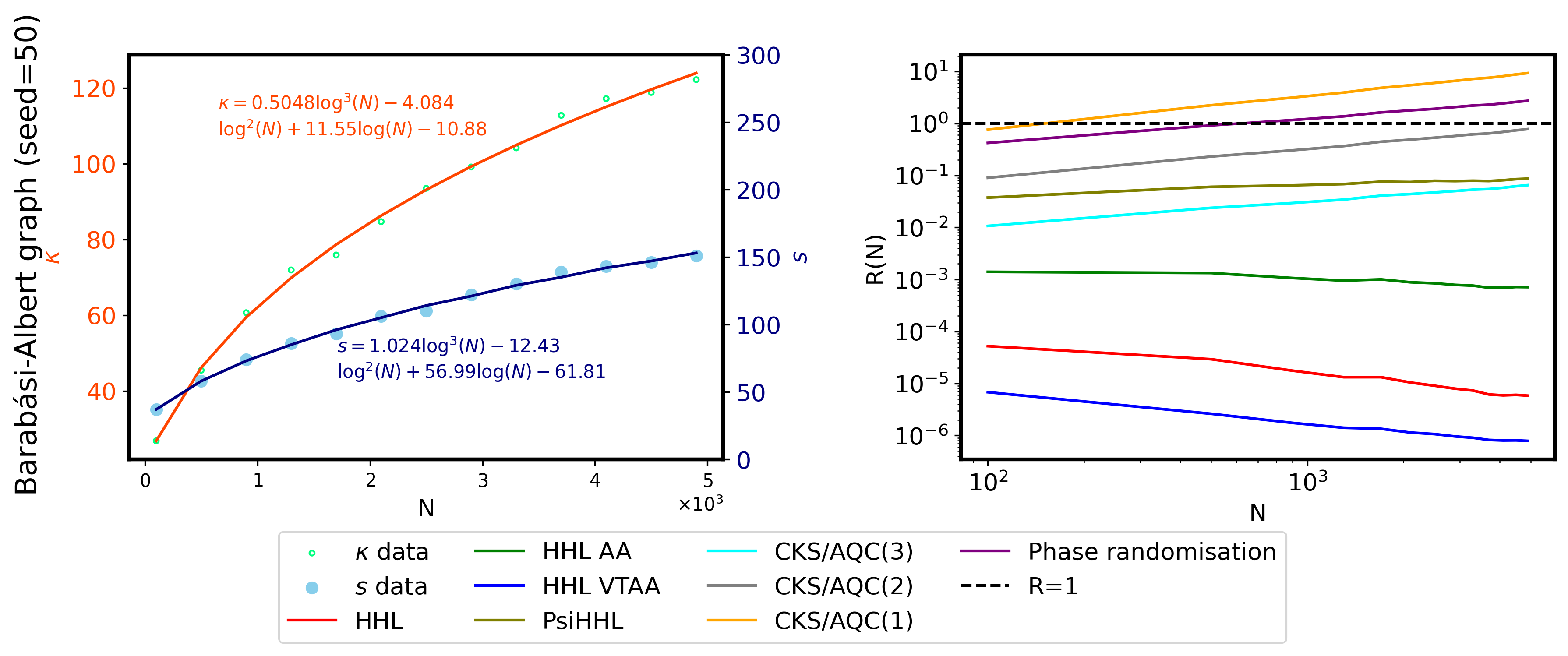} \\
\vspace{0.5em}
\makebox[0.45\textwidth]{(e)}
\makebox[0.36\textwidth]{(f)} \\

\end{tabular} 
\end{figure*}

\begin{figure*}[h!]
\begin{tabular}{c}

\includegraphics[width=14cm]{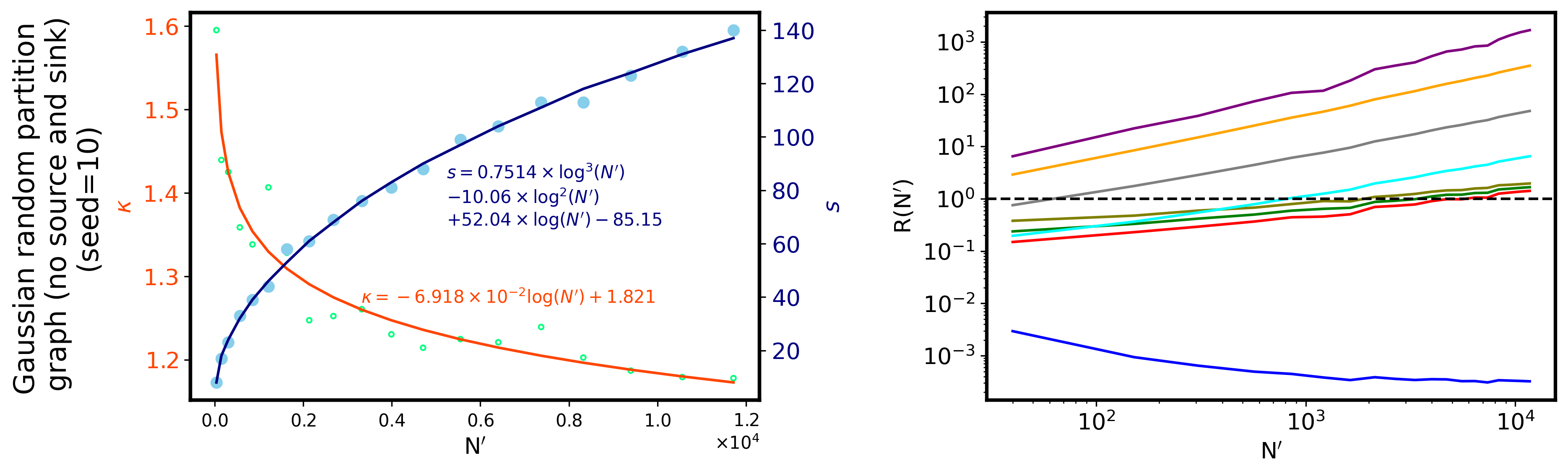} \\
\vspace{0.5em}
\makebox[0.45\textwidth]{(g)}
\makebox[0.36\textwidth]{(h)} \\

\includegraphics[width=14cm]{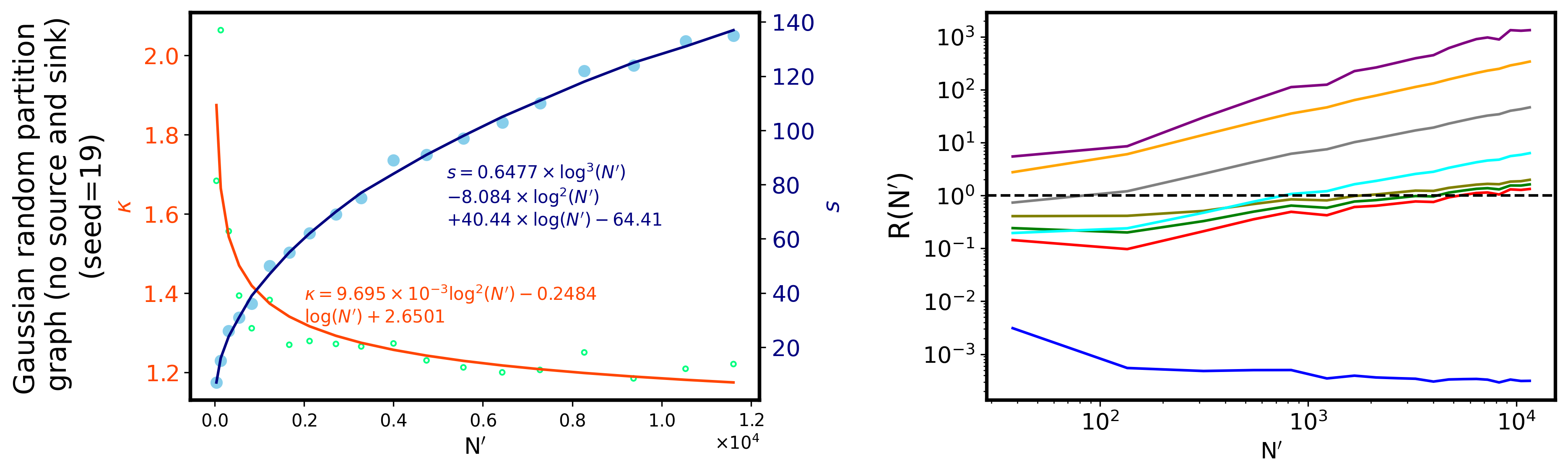} \\
\vspace{0.5em}
\makebox[0.45\textwidth]{(i)}
\makebox[0.36\textwidth]{(j)} \\

\includegraphics[width=14cm]{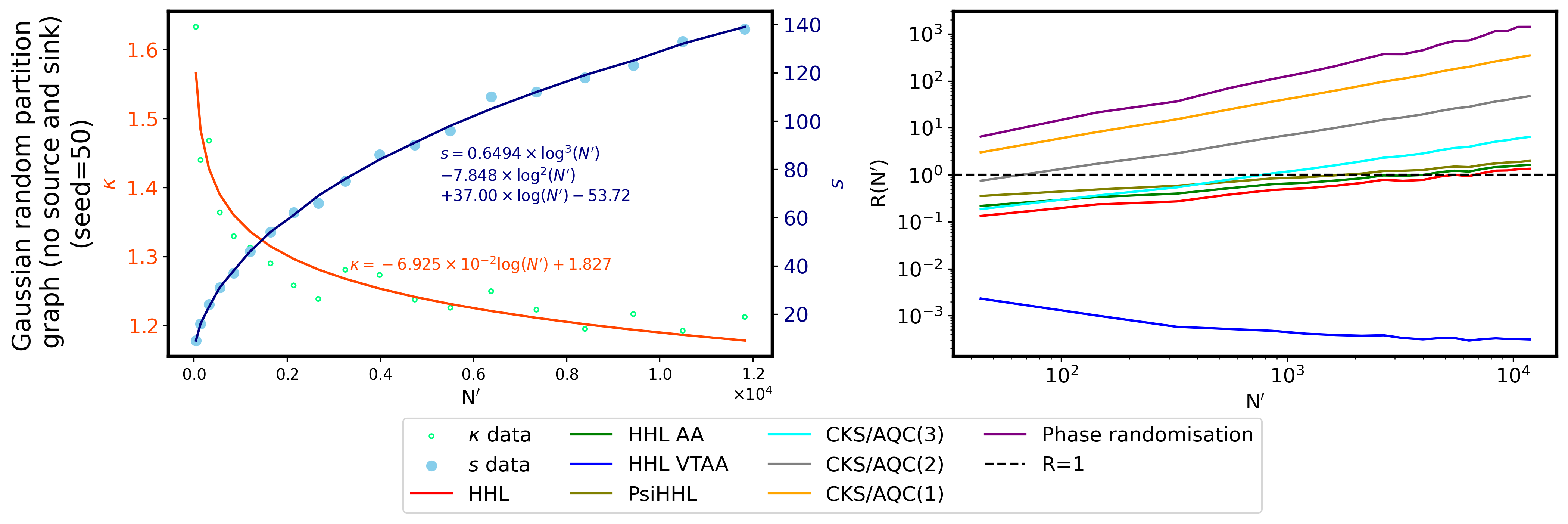} \\
\vspace{0.5em}
\makebox[0.45\textwidth]{(k)}
\makebox[0.36\textwidth]{(l)} \\

\multicolumn{1}{r}{(\textit{Continued})} \\   
\end{tabular} 
\end{figure*} 

\begin{figure*}[h!]
\begin{tabular}{c}
\includegraphics[width=14cm]{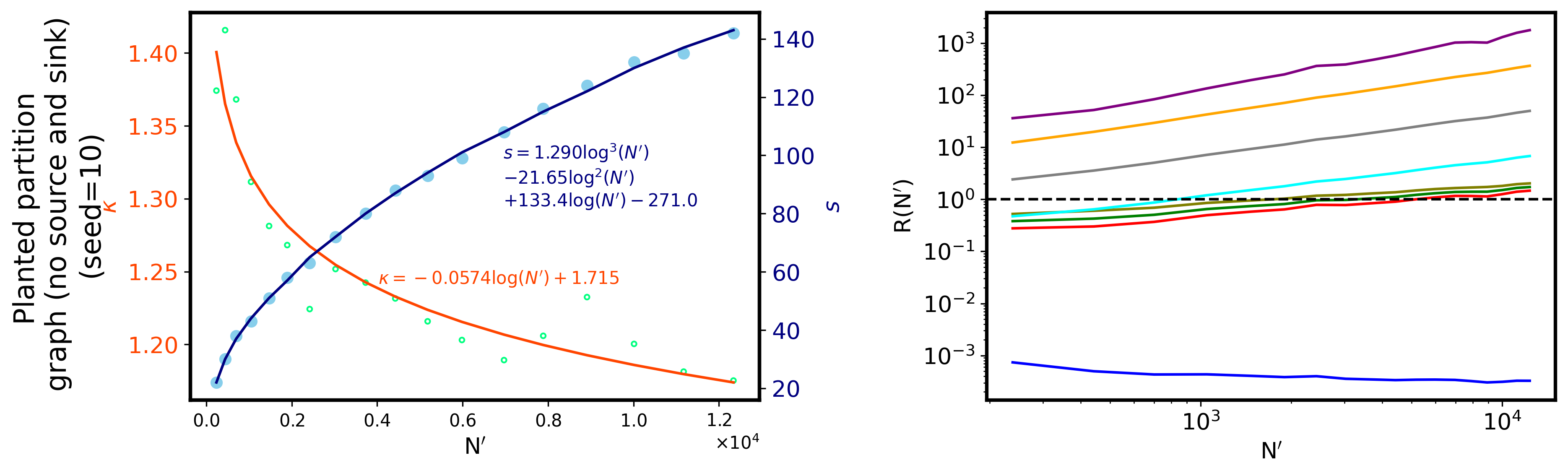} \\
\vspace{0.5em}
\makebox[0.45\textwidth]{(m)}
\makebox[0.36\textwidth]{(n)} \\   
\includegraphics[width=14cm]{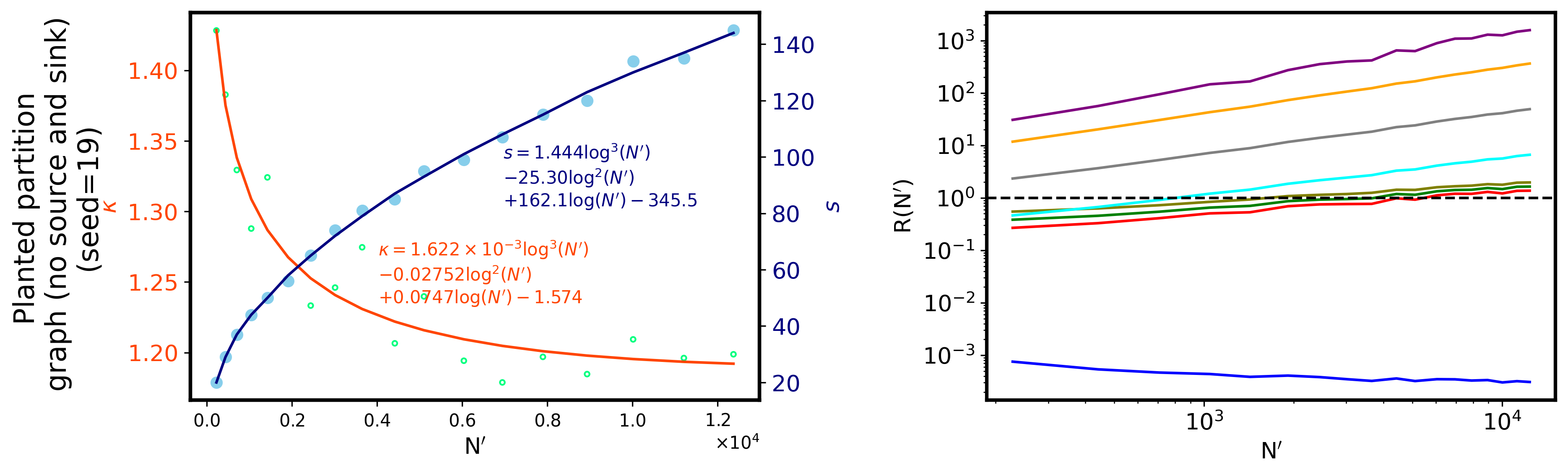} \\
\vspace{0.5em}
\makebox[0.45\textwidth]{(o)}
\makebox[0.36\textwidth]{(p)} \\

\includegraphics[width=14cm]{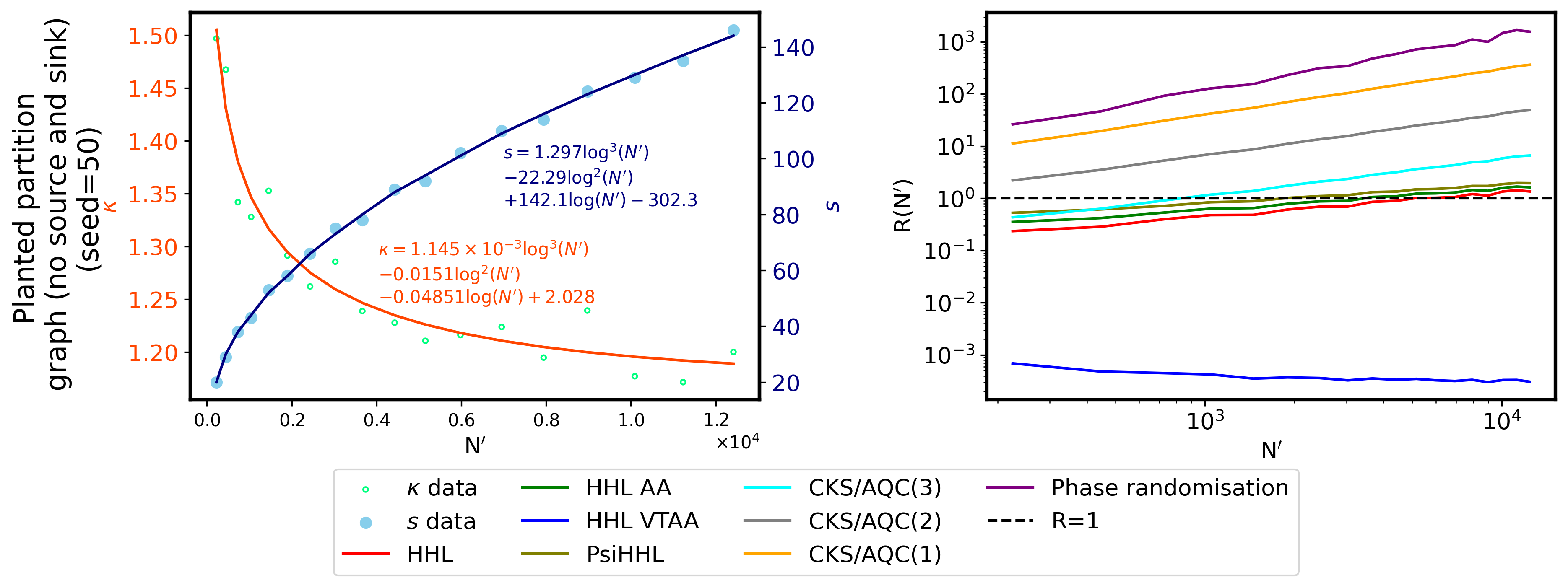} \\
\vspace{0.5em}
\makebox[0.45\textwidth]{(q)}
\makebox[0.36\textwidth]{(r)} \\
\end{tabular} 
\caption{Subfigures (a)-(f) show the advantage offered by Barabási-Albert graph family for three different seed values: 10, 23 and 50. Subfigures (g)-(l) represent advantages offered by directed Gaussian random partition graph family for seed values 10, 19 and 50. From (m) - (r), figures demonstrate the advantages offered by planted partition graph for seed values: 10, 12 and 50.}\label{fig:error_analysis}
\end{figure*}

\begin{figure*}[t]
\begin{tabular}{c}
\includegraphics[width=14cm]{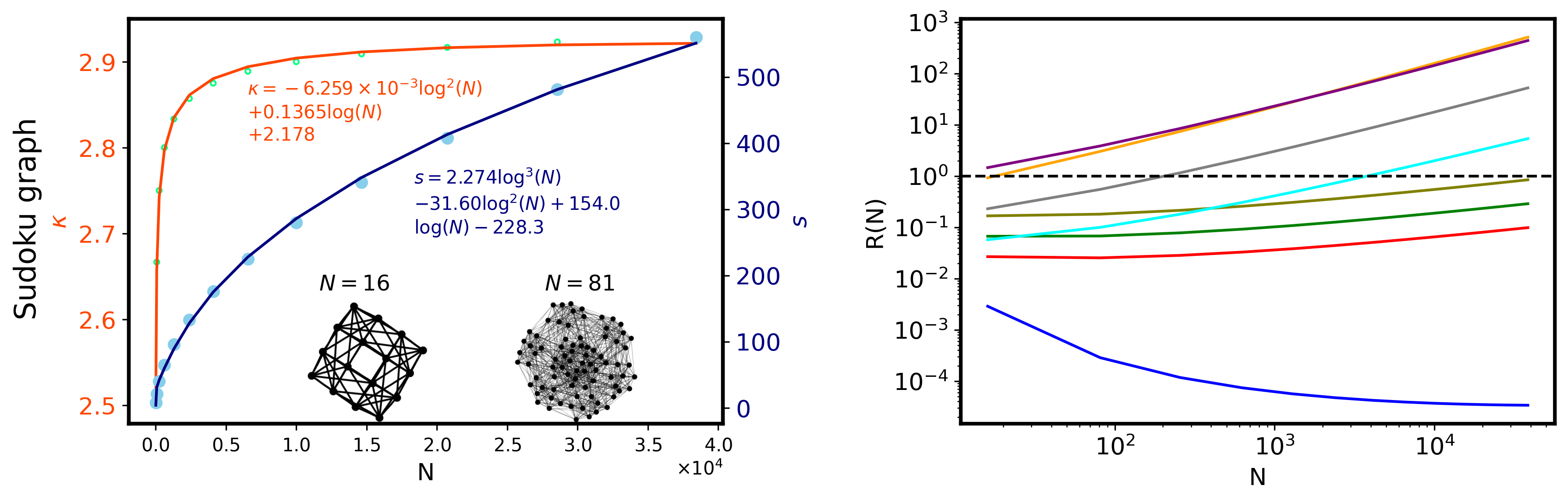} \\
\vspace{0.5em}
\makebox[0.45\textwidth]{(a)}
\makebox[0.36\textwidth]{(b)} \\

\includegraphics[width=14cm]{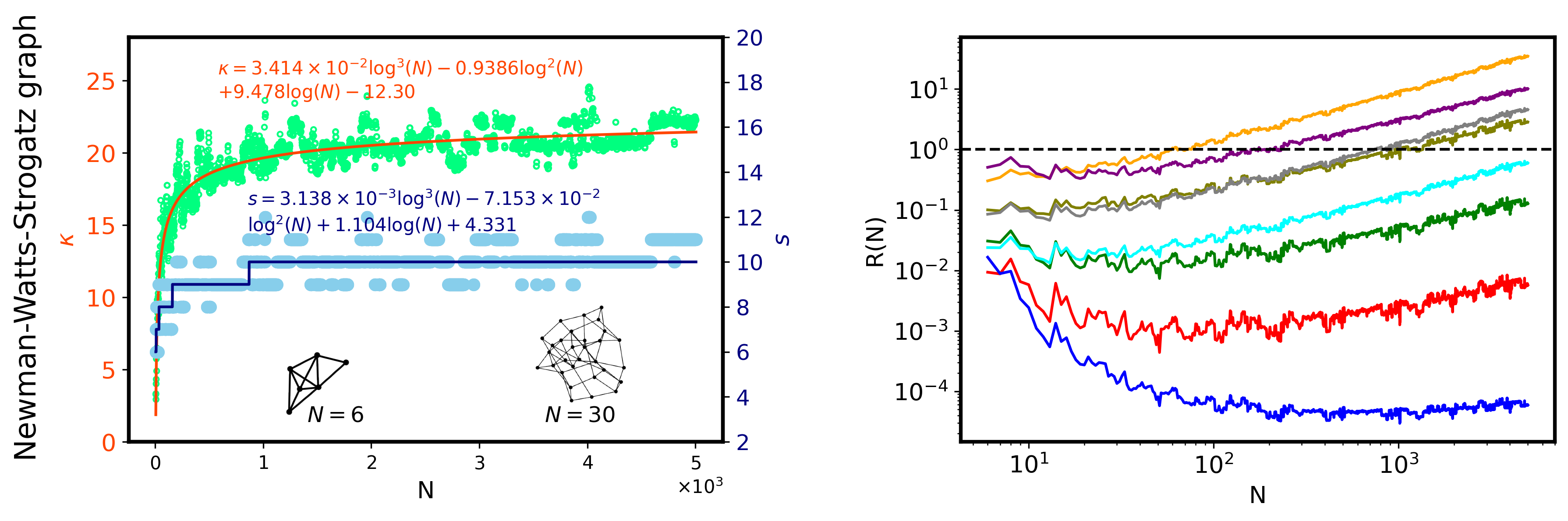} \\
\vspace{0.5em}
\makebox[0.45\textwidth]{(c)}
\makebox[0.36\textwidth]{(b)} \\

\includegraphics[width=14cm]{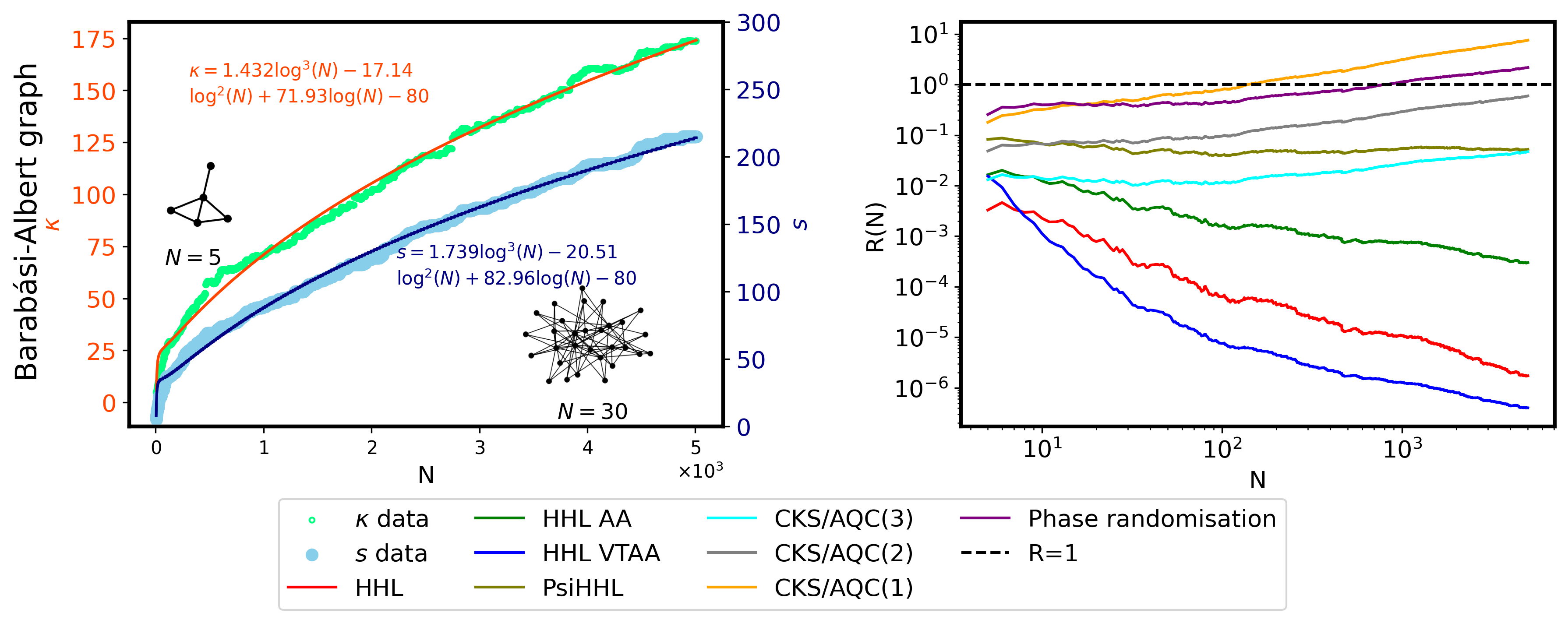} \\ 
\vspace{0.5em}
\makebox[0.45\textwidth]{(e)}
\makebox[0.36\textwidth]{(f)} \\

\multicolumn{1}{r}{(\textit{Continued})} \\ 
\end{tabular} 
\end{figure*}

\begin{figure*}[h!]
\begin{tabular}{c}

\includegraphics[width=14cm]{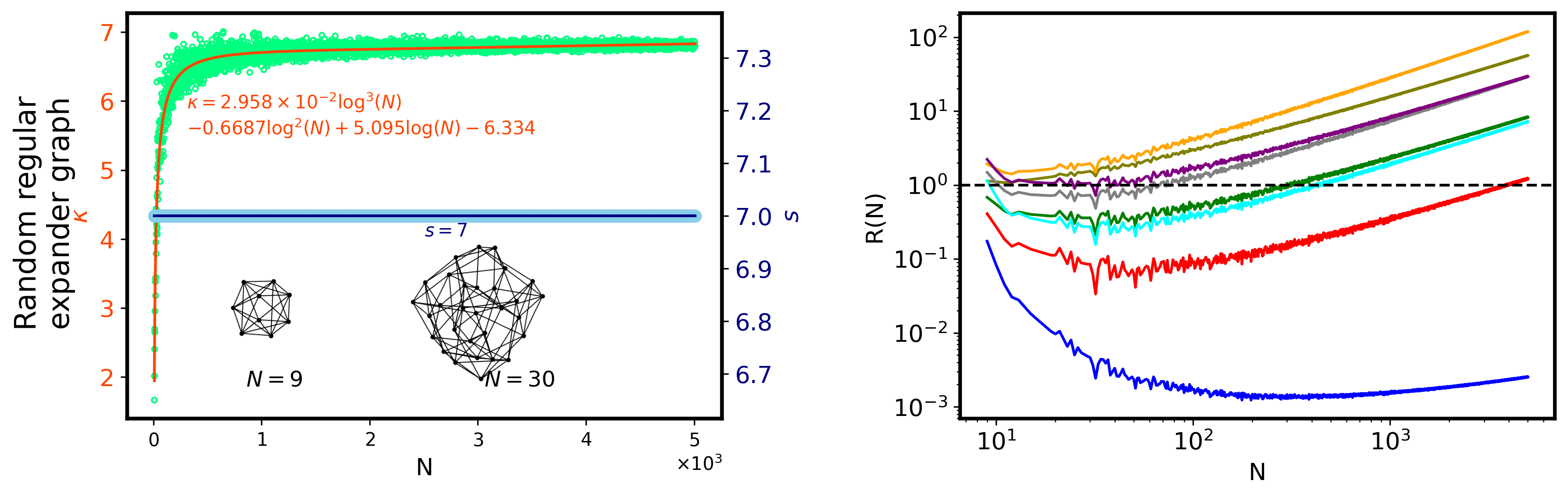} \\
\vspace{0.5em}
\makebox[0.45\textwidth]{(g)}
\makebox[0.36\textwidth]{(h)} \\

\includegraphics[width=14cm]{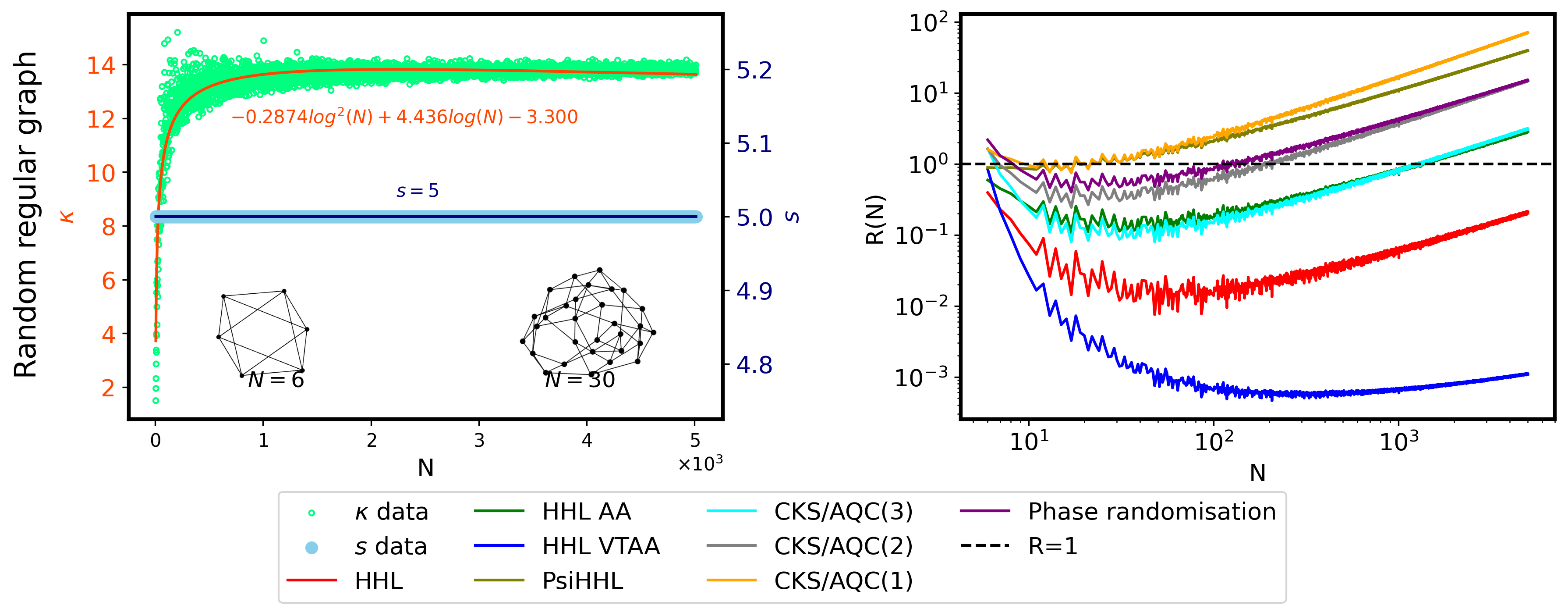} \\
\vspace{0.5em}
\makebox[0.45\textwidth]{(i)}
\makebox[0.36\textwidth]{(j)} \\

\end{tabular} 
\caption{Sub-figures showing $\kappa$ and $s$ behaviour of good graphs listed in Table II of the main manuscript. }\label{fig:sm_goodgraph_lap}
\end{figure*} 

\begin{figure*}[t]
\begin{tabular}{c}

\includegraphics[width=14cm]{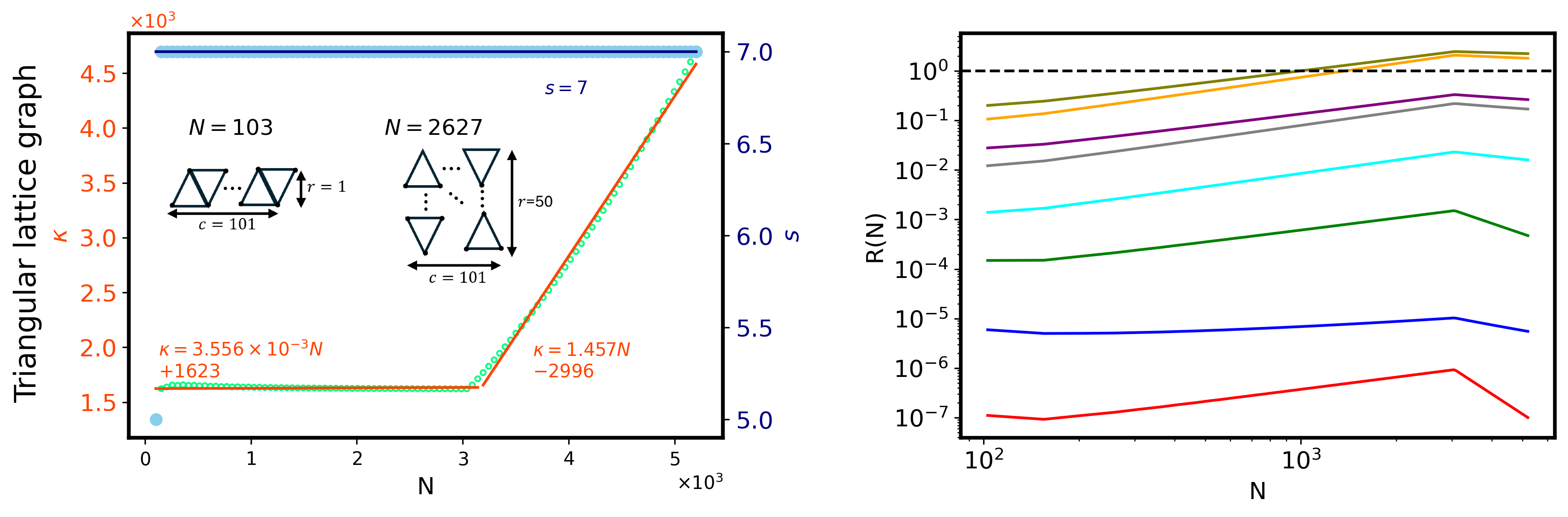} \\ 
\makebox[0.45\textwidth]{(a)}
\makebox[0.36\textwidth]{(b)} \\

\includegraphics[width=14cm]{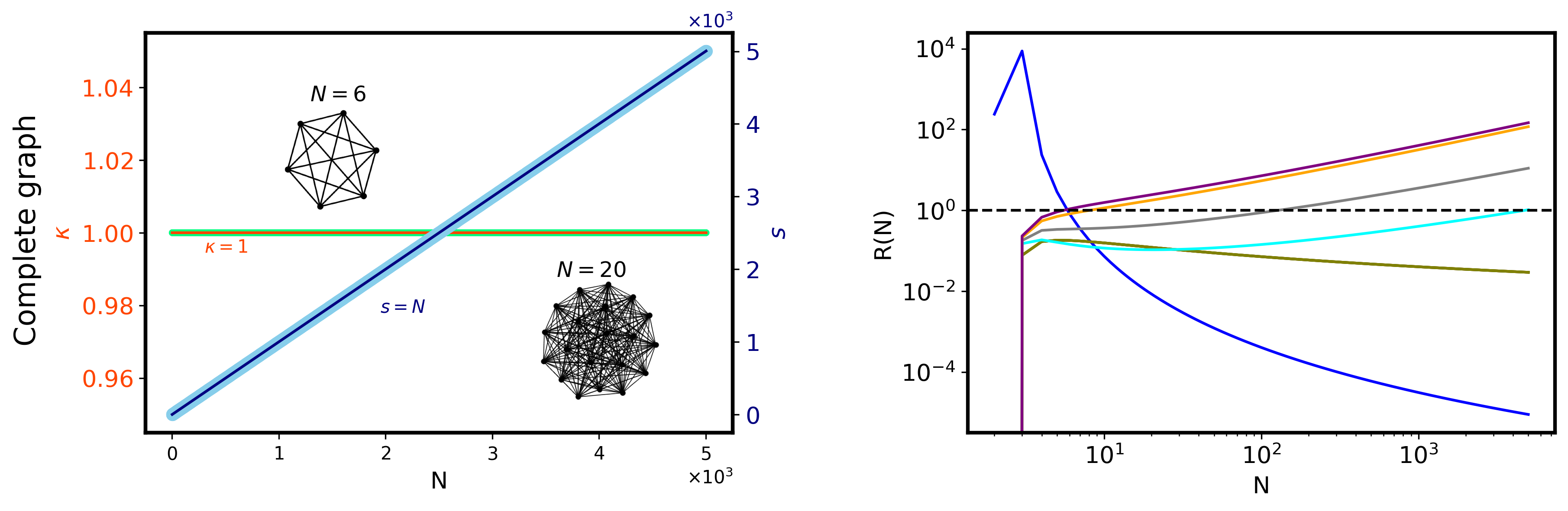} \\
\makebox[0.45\textwidth]{(c)}
\makebox[0.36\textwidth]{(d)} \\

\includegraphics[width=14cm]{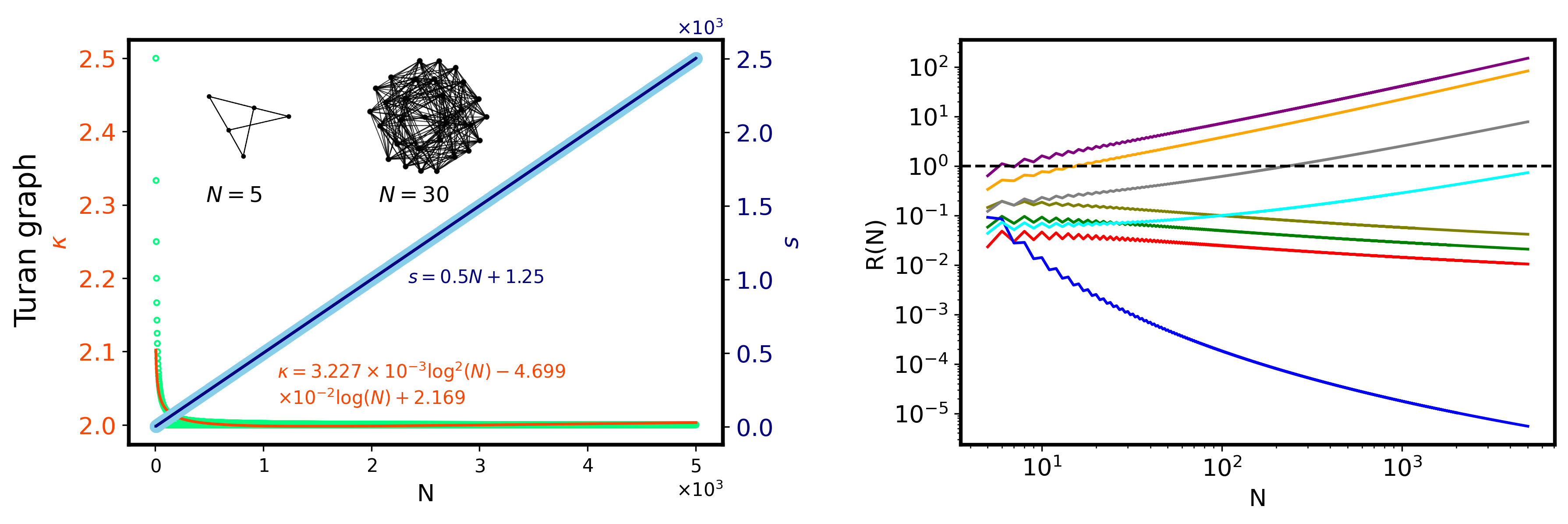} \\
\makebox[0.45\textwidth]{(e)}
\makebox[0.36\textwidth]{(f)} \\

\includegraphics[width=14cm]{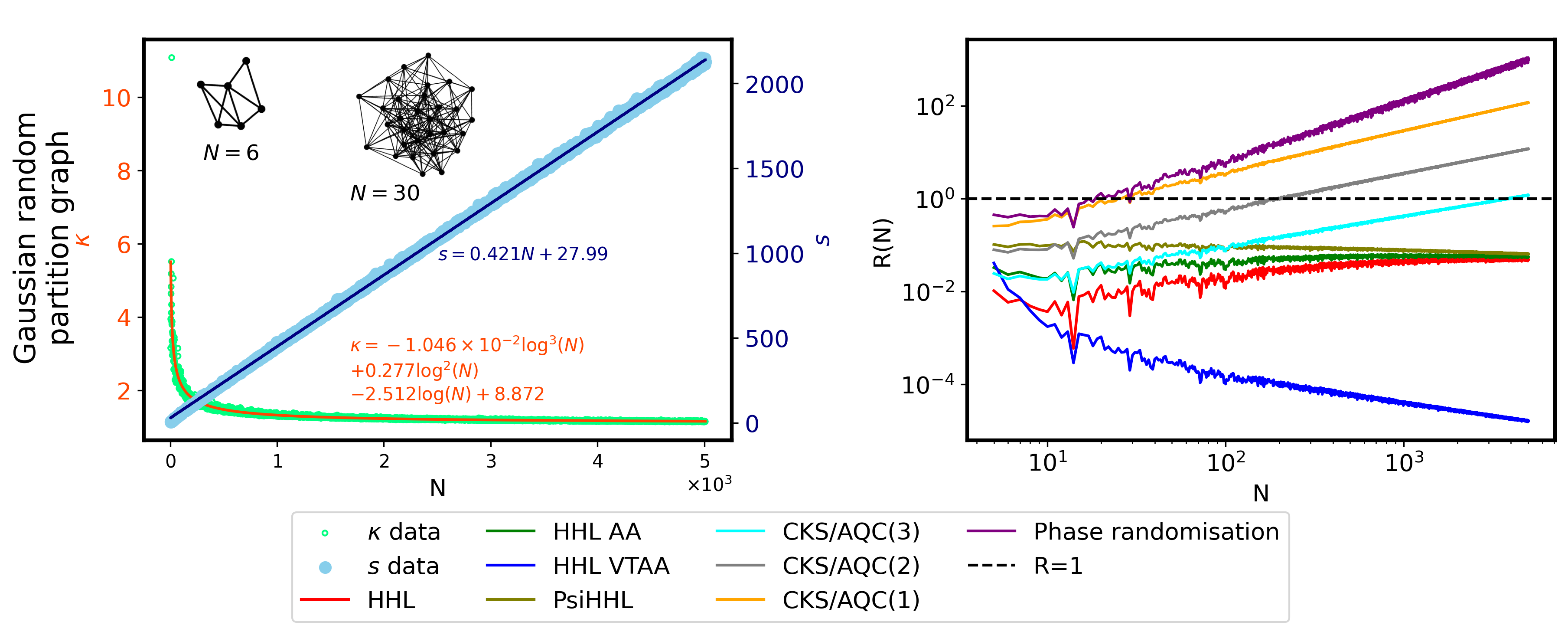} \\
\makebox[0.45\textwidth]{(g)}
\makebox[0.36\textwidth]{(h)} \\

\multicolumn{1}{r}{(\textit{Continued})} \\  
\end{tabular} 
\end{figure*}

\begin{figure*}[h!]
\begin{tabular}{c}

\includegraphics[width=14cm]{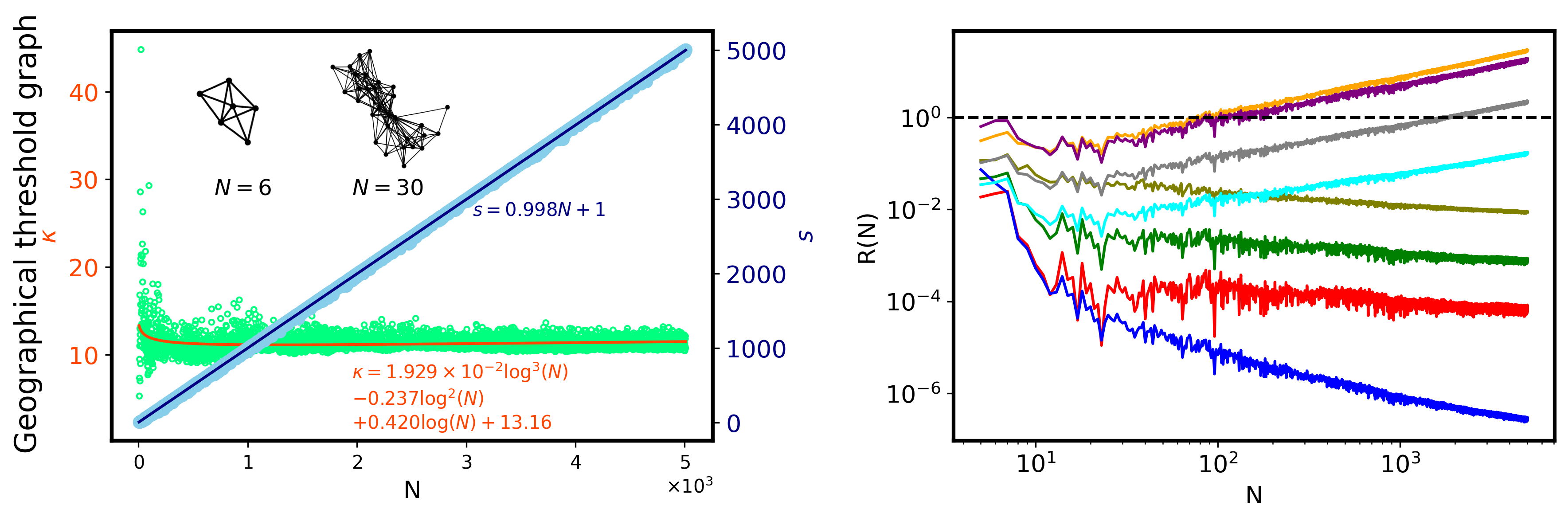} \\
\makebox[0.45\textwidth]{(i)}
\makebox[0.36\textwidth]{(j)} \\

\includegraphics[width=14cm]{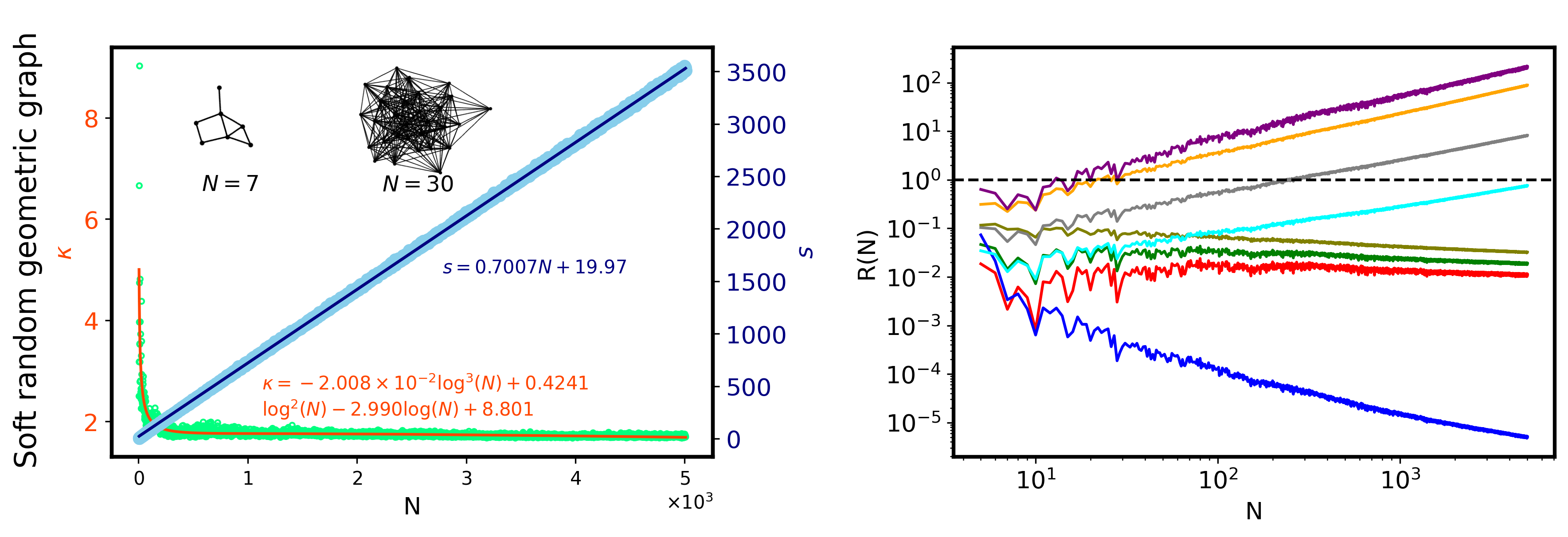} \\
\makebox[0.45\textwidth]{(k)}
\makebox[0.36\textwidth]{(l)} \\

\includegraphics[width=14cm]{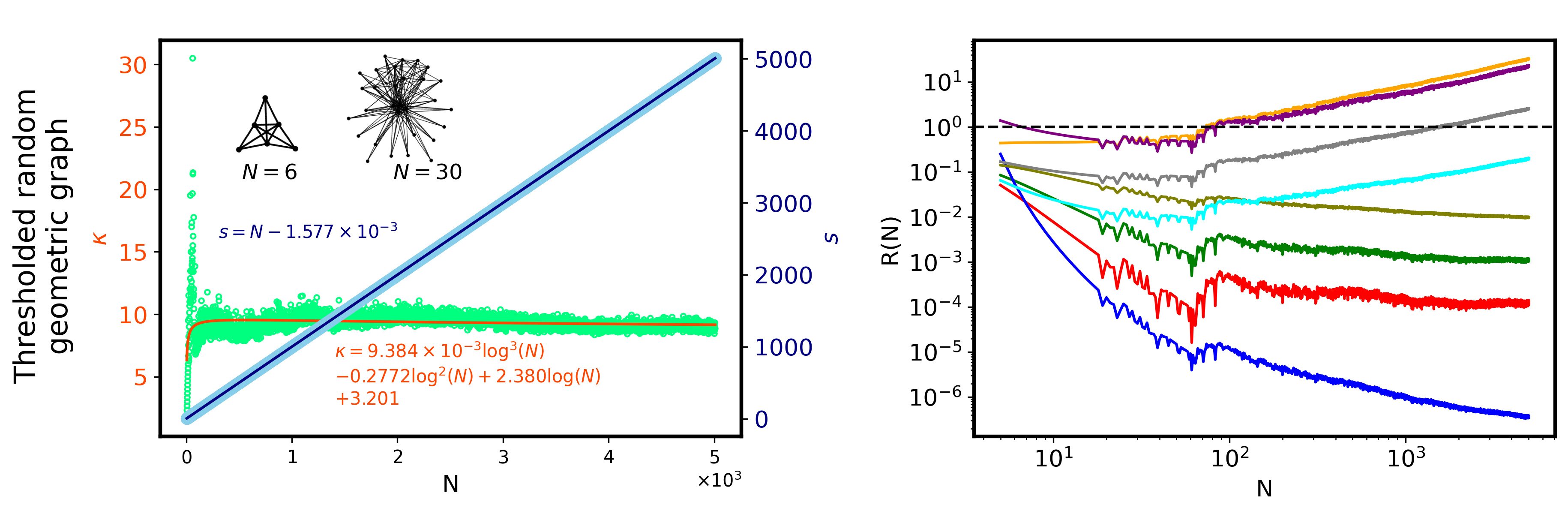} \\
\makebox[0.45\textwidth]{(m)}
\makebox[0.36\textwidth]{(n)} \\

\includegraphics[width=14cm]{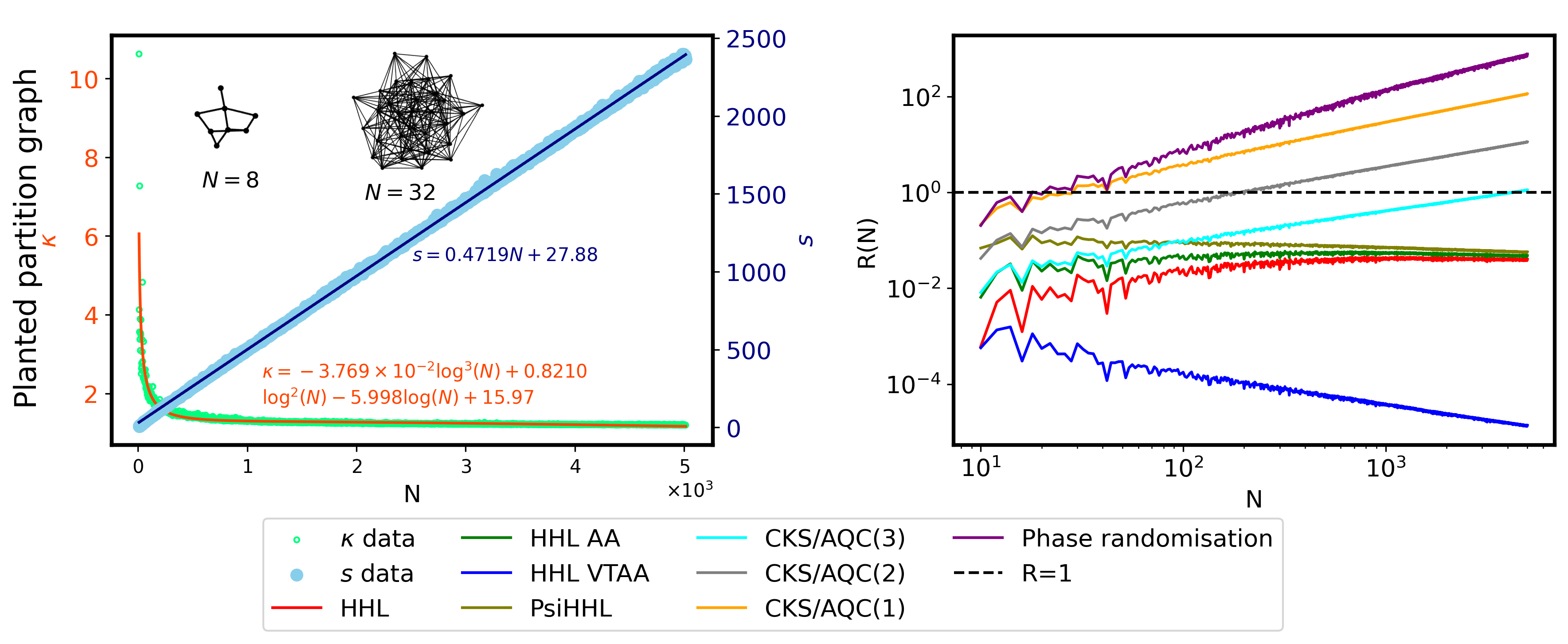} \\
\makebox[0.45\textwidth]{(o)}
\makebox[0.36\textwidth]{(p)} \\

\multicolumn{1}{r}{(\textit{Continued})} \\ 
\end{tabular} 
\end{figure*}

\begin{figure*}[h!]
\begin{tabular}{c}

\includegraphics[width=14cm]{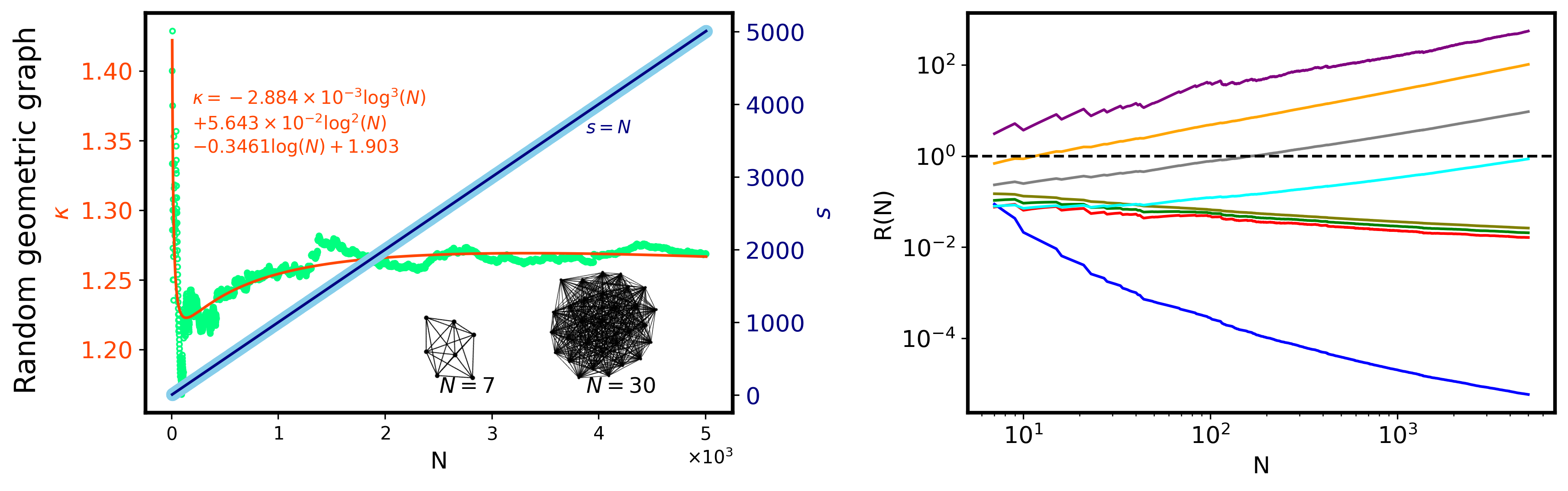} \\
\makebox[0.45\textwidth]{(q)}
\makebox[0.36\textwidth]{(r)} \\


\includegraphics[width=14cm]{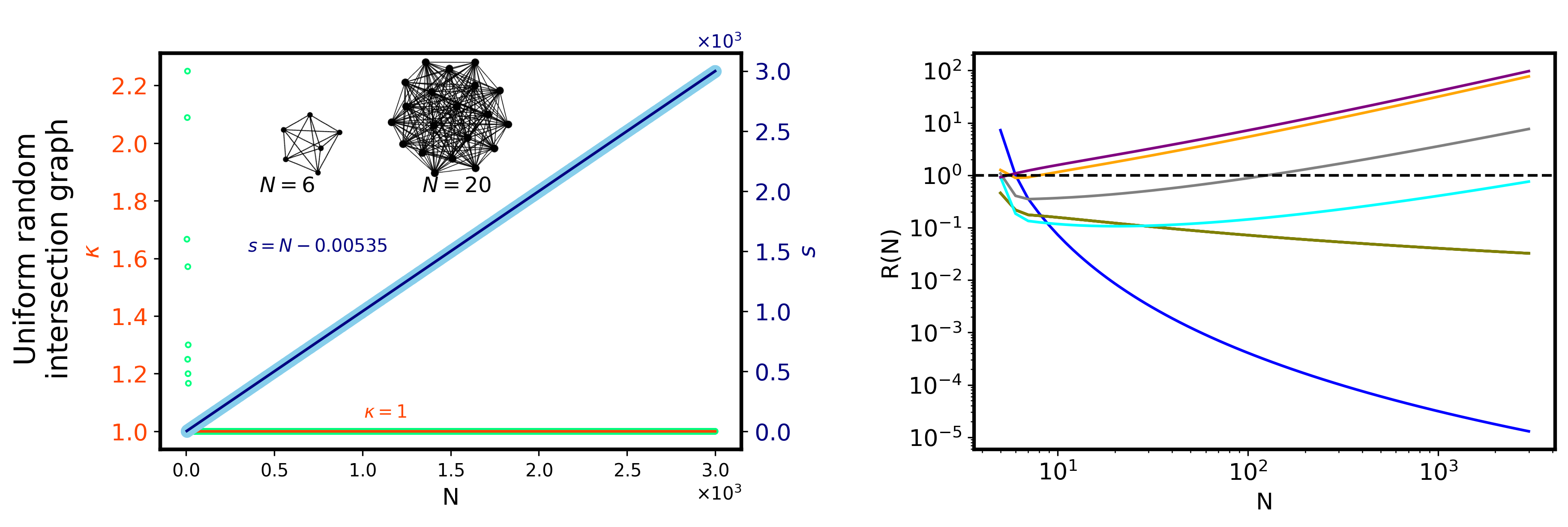} \\
\makebox[0.45\textwidth]{(s)}
\makebox[0.36\textwidth]{(t)} \\

\includegraphics[width=14cm]{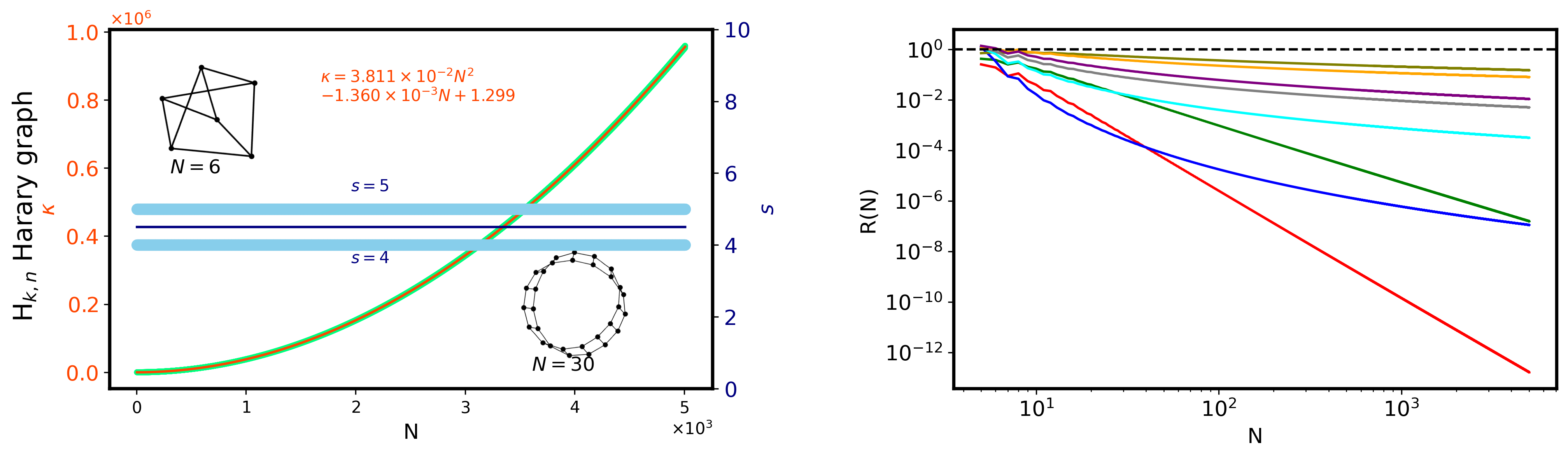} \\
\makebox[0.45\textwidth]{(u)}
\makebox[0.36\textwidth]{(v)} \\ 
    
\includegraphics[width=14cm]{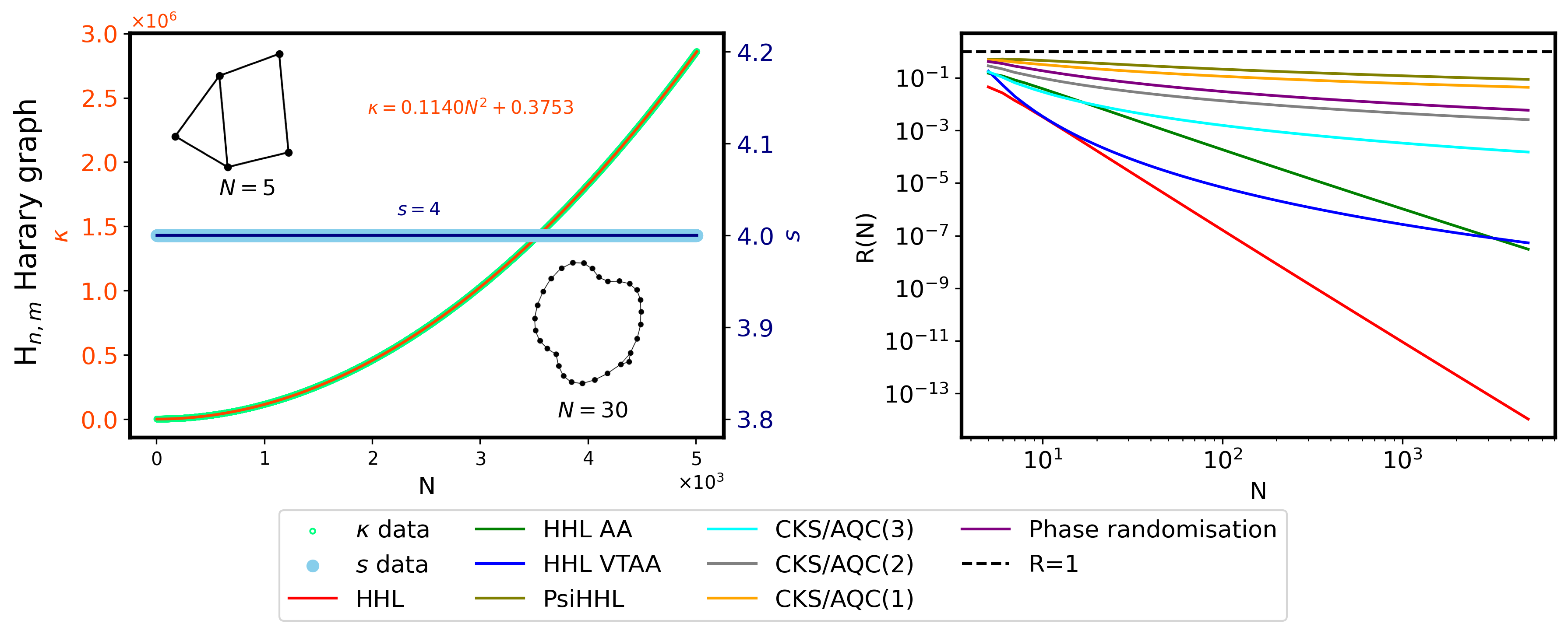} \\
\makebox[0.45\textwidth]{(w)}
\makebox[0.36\textwidth]{(x)} \\

\multicolumn{1}{r}{(\textit{Continued})} \\ 
\end{tabular} 
\end{figure*}

\begin{figure*}[h!]
\begin{tabular}{c}

\includegraphics[width=14cm]{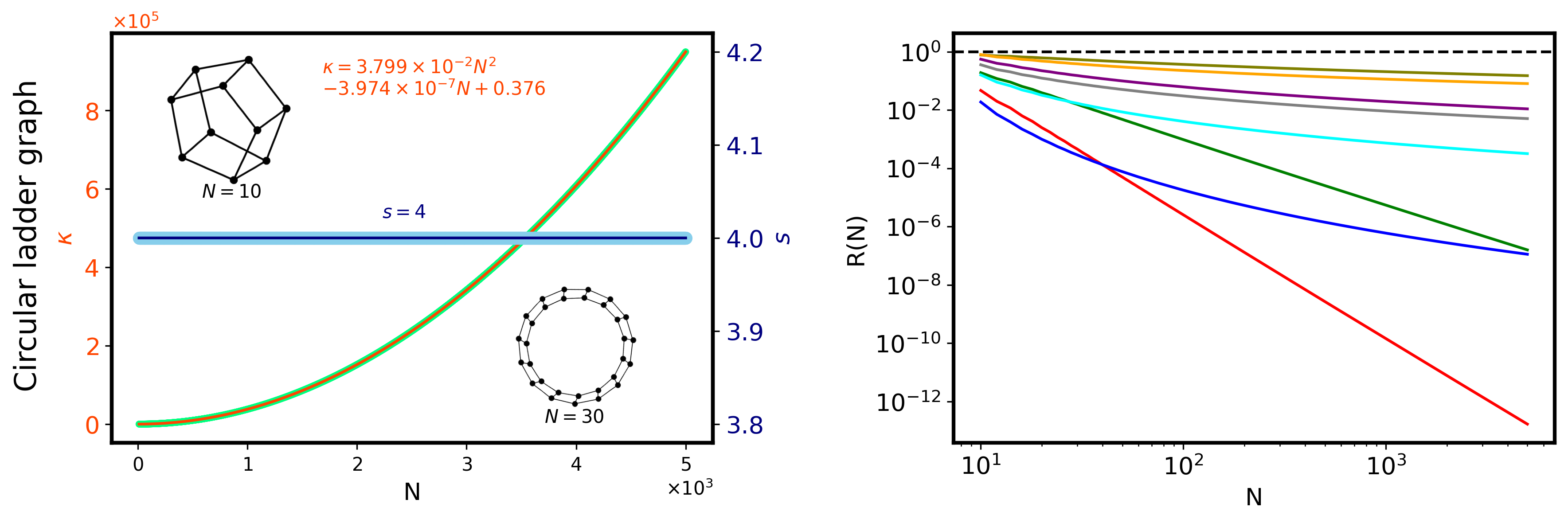} \\
\makebox[0.45\textwidth]{(y)}
\makebox[0.36\textwidth]{(z)} \\ 

\includegraphics[width=14cm]{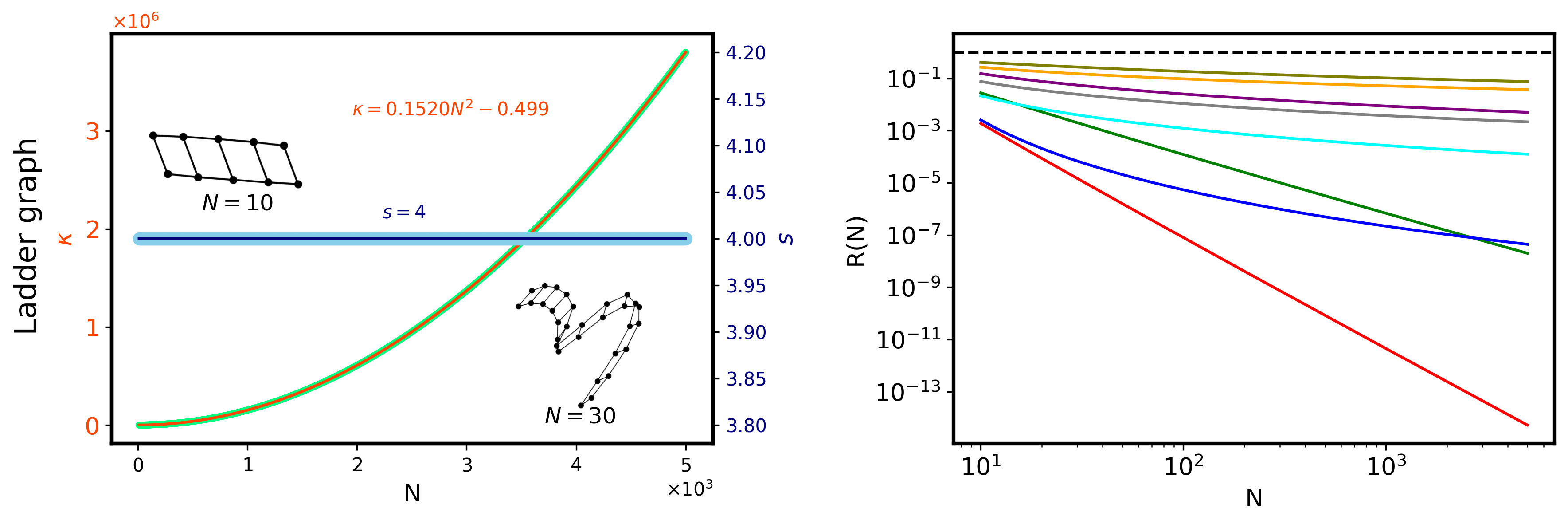} \\
\makebox[0.45\textwidth]{(aa)}
\makebox[0.36\textwidth]{(ab)} \\

\includegraphics[width=14cm]{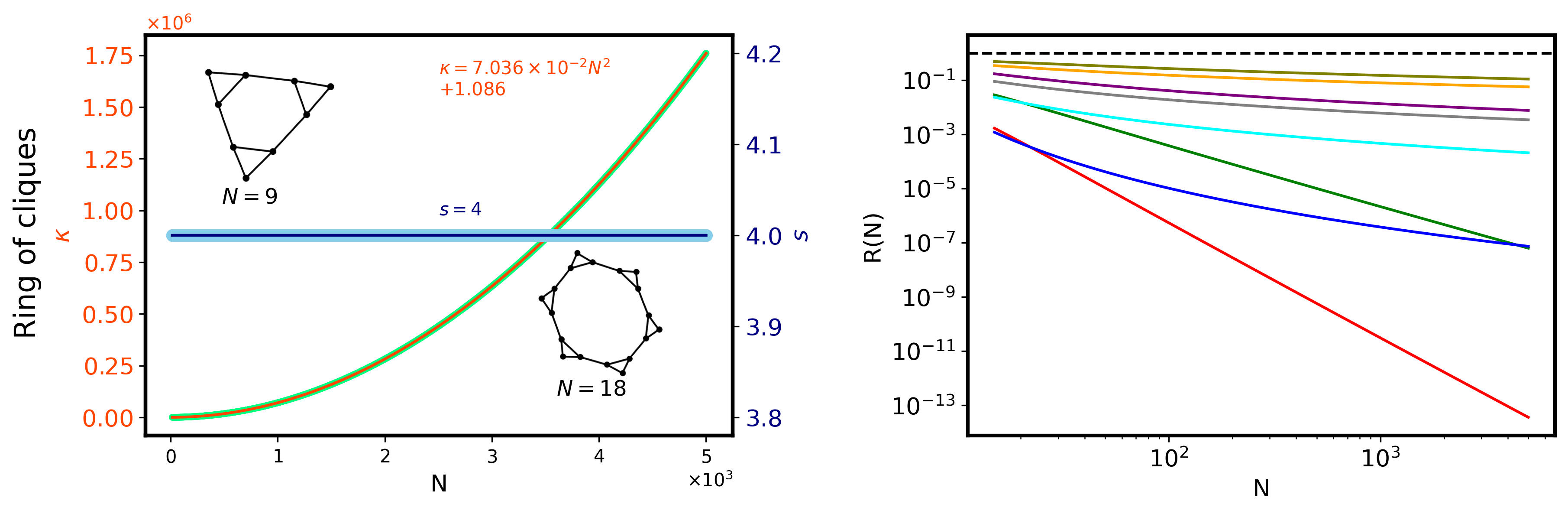} \\
\makebox[0.45\textwidth]{(ac)}
\makebox[0.36\textwidth]{(ad)} \\

\includegraphics[width=14cm]{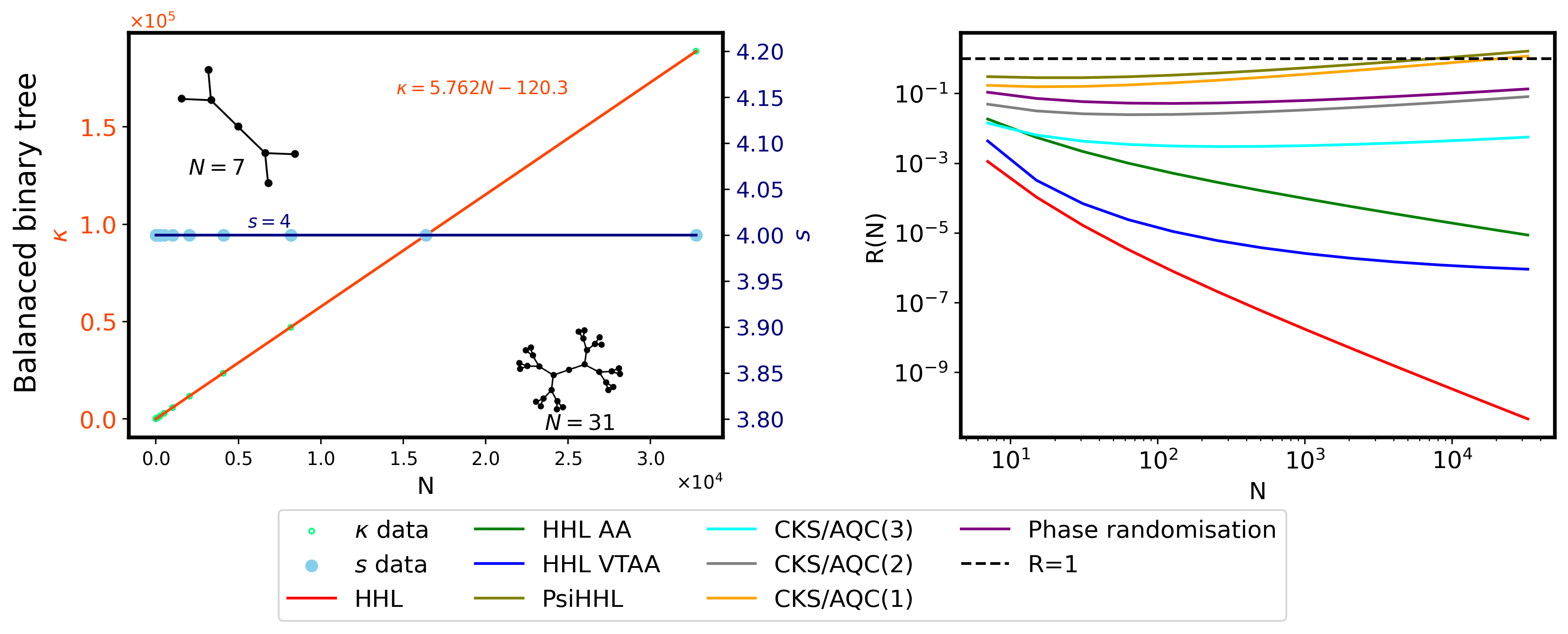} \\
\makebox[0.45\textwidth]{(ae)}
\makebox[0.36\textwidth]{(af)} \\

\multicolumn{1}{r}{(\textit{Continued})} \\ 
\end{tabular} 
\end{figure*}

\begin{figure*}[h!]
\begin{tabular}{c}

\includegraphics[width=14cm]{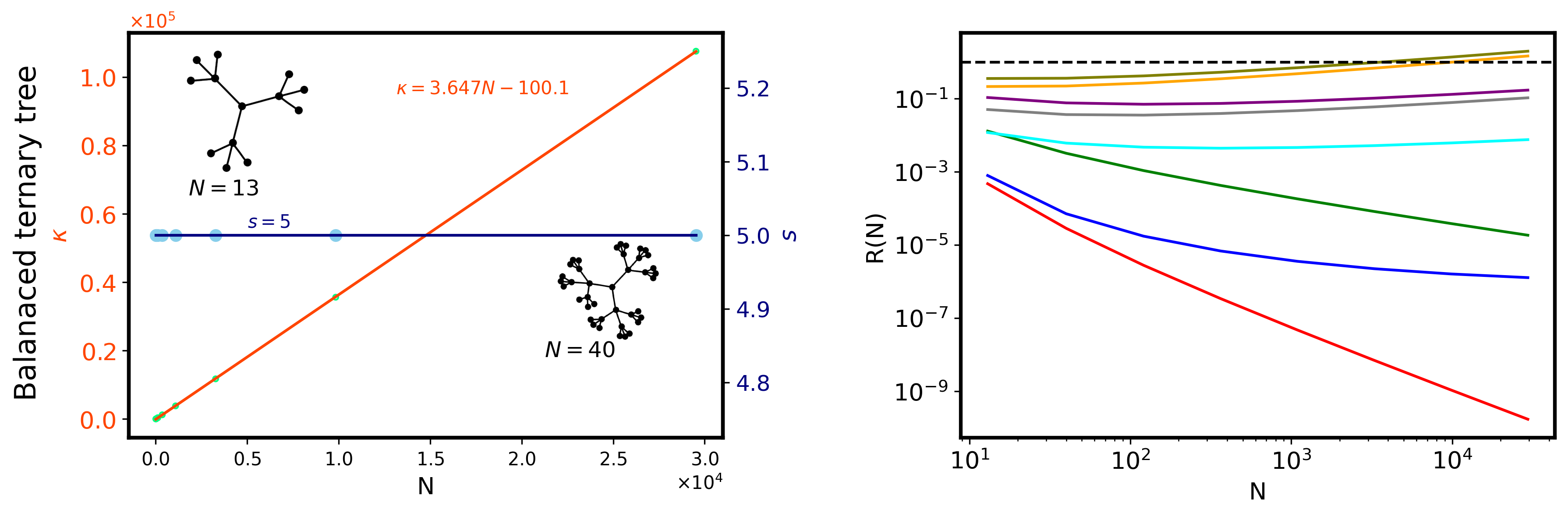} \\
\makebox[0.45\textwidth]{(ag)}
\makebox[0.36\textwidth]{(ah)} \\ 
    
\includegraphics[width=14cm]{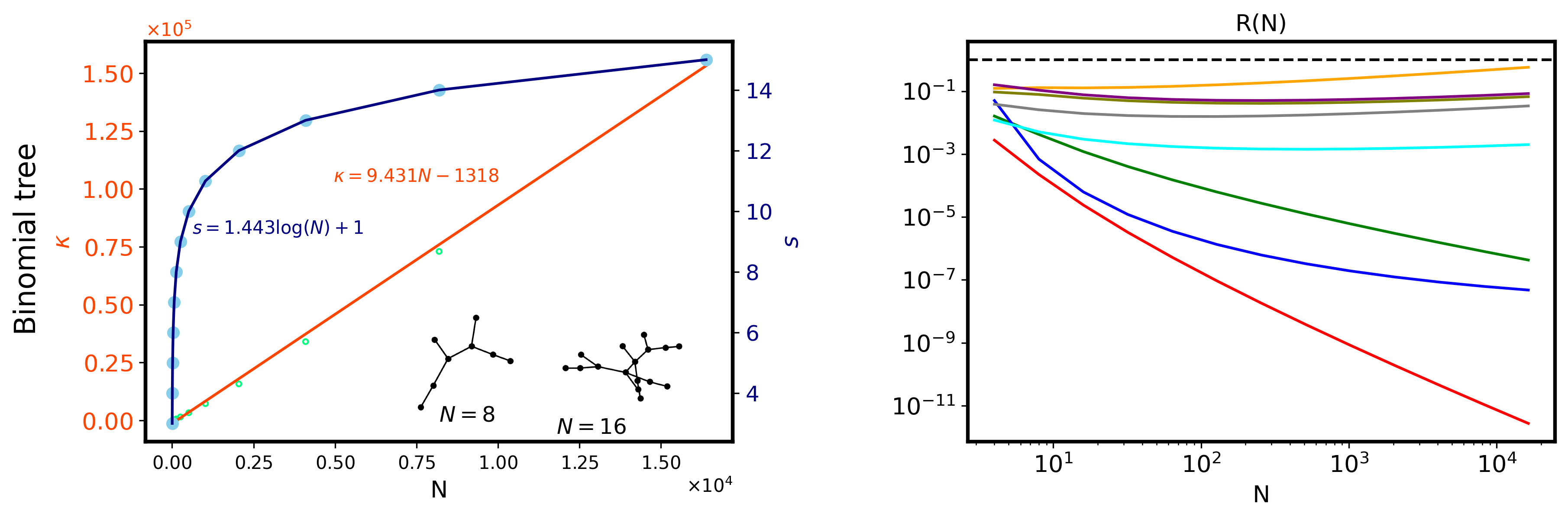} \\
\makebox[0.45\textwidth]{(ai)}
\makebox[0.36\textwidth]{(aj)} \\

\includegraphics[width=14cm]{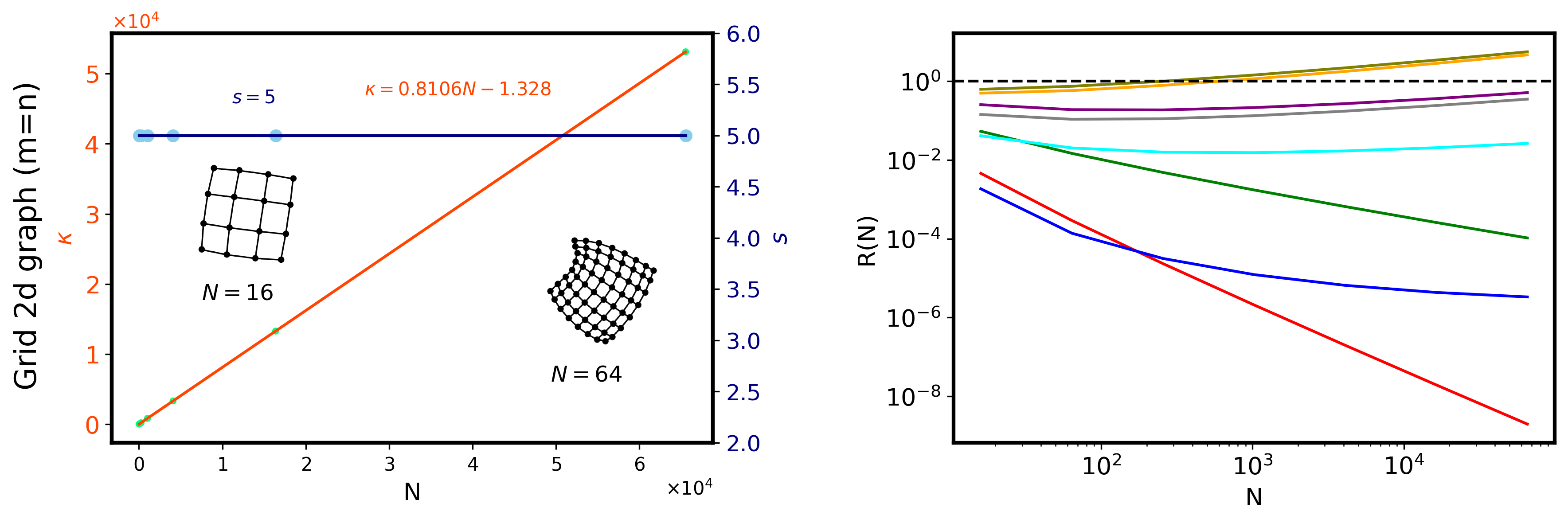} \\
\makebox[0.45\textwidth]{(ak)}
\makebox[0.36\textwidth]{(al)} \\ 

\includegraphics[width=14cm]{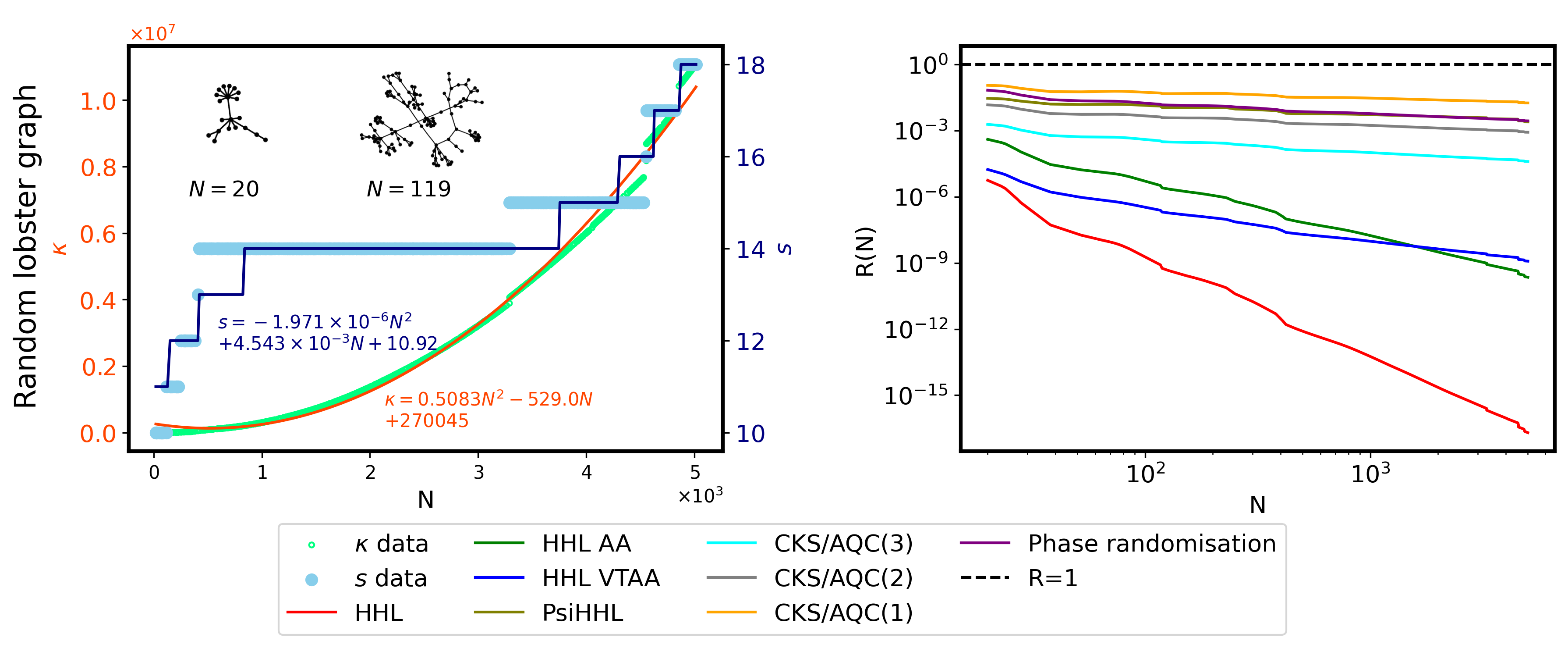} \\
\makebox[0.45\textwidth]{(am)}
\makebox[0.36\textwidth]{(an)} \\ 
    
\multicolumn{1}{r}{(\textit{Continued})} \\
\end{tabular} 

\end{figure*}

\begin{figure*}[t]
\begin{tabular}{c}

\includegraphics[width=14cm]{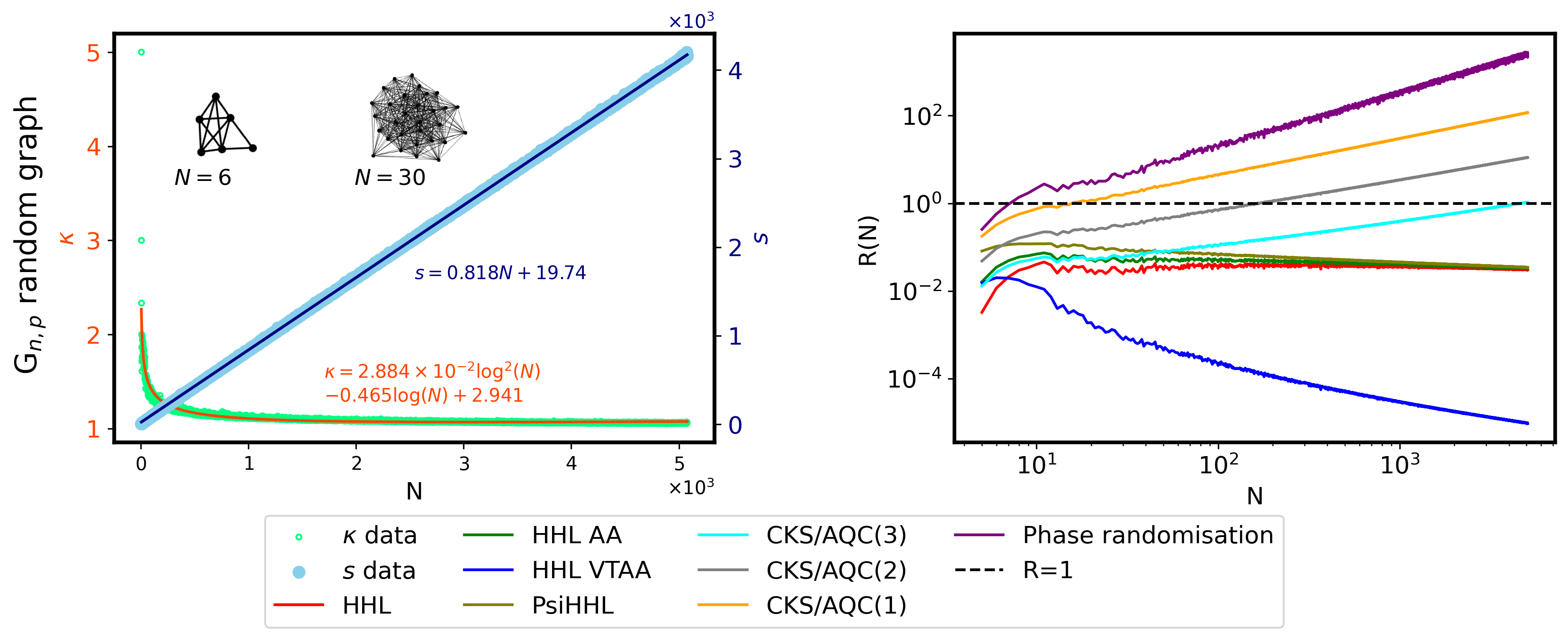} \\
\makebox[0.45\textwidth]{(ao)}
\makebox[0.36\textwidth]{(ap)} \\ 
    
\end{tabular} 
\caption{Sub-figures showing $\kappa$ and $s$ behaviour of bad graphs listed in Table III of main manuscript. We also plot the runtime ratios for different QLSs, considered in our study, with the numerical data. }\label{fig:sm_badgraph_lap}
\end{figure*}

\begin{figure*}[t]
\begin{tabular}{c}
    
\includegraphics[width=14cm]{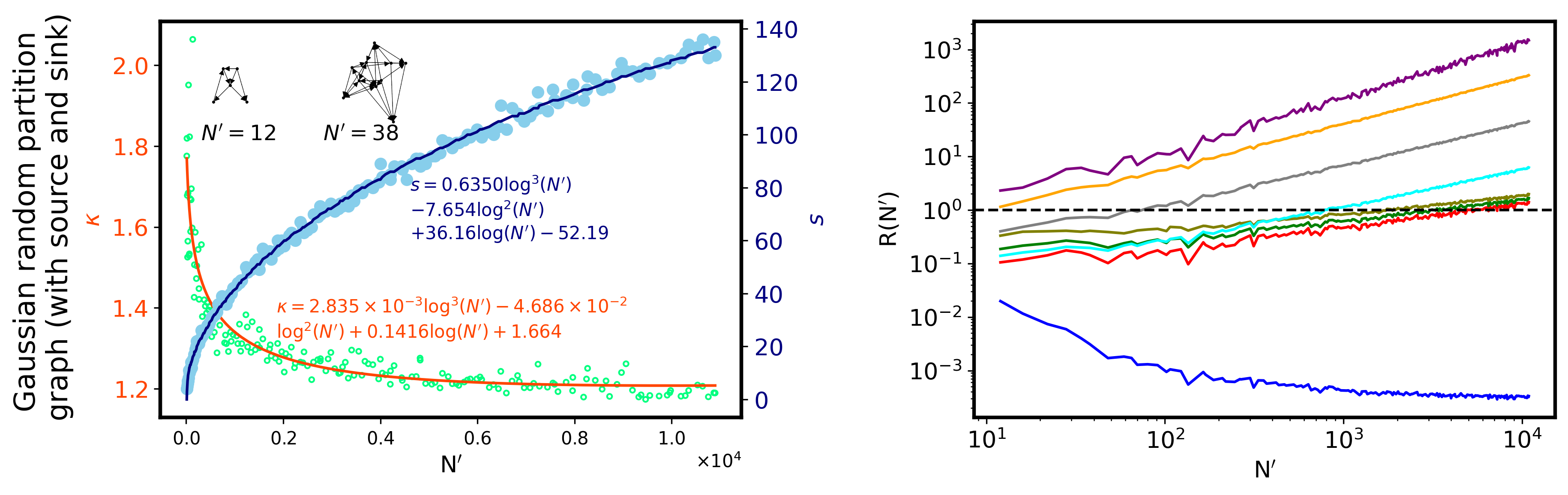} \\
\makebox[0.45\textwidth]{(a)}
\makebox[0.36\textwidth]{(b)} \\ 

\includegraphics[width=14cm]{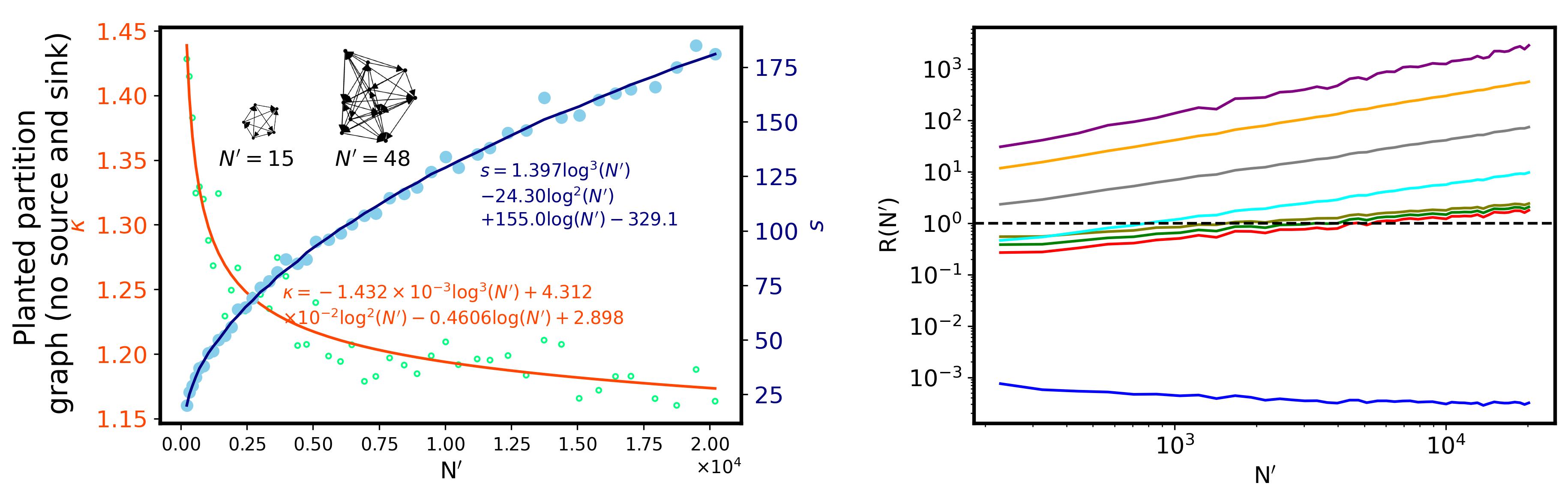} \\
\makebox[0.45\textwidth]{(c)}
\makebox[0.36\textwidth]{(d)} \\ 

\includegraphics[width=14cm]{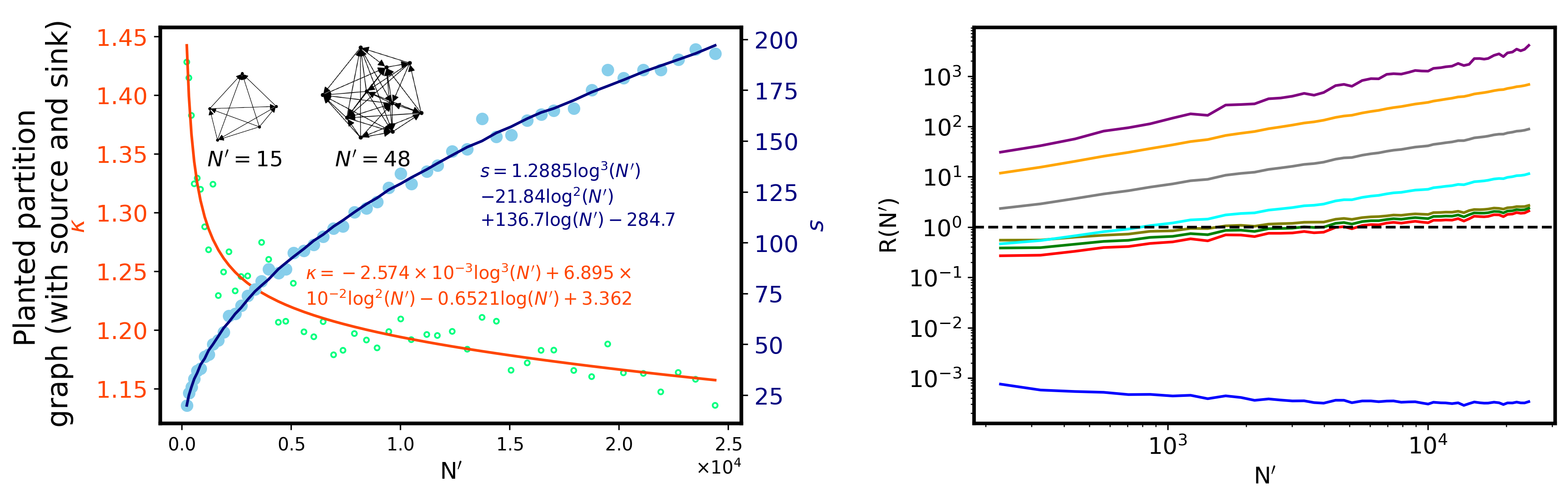} \\
\makebox[0.45\textwidth]{(e)}
\makebox[0.36\textwidth]{(f)} \\ 

\includegraphics[width=14cm]{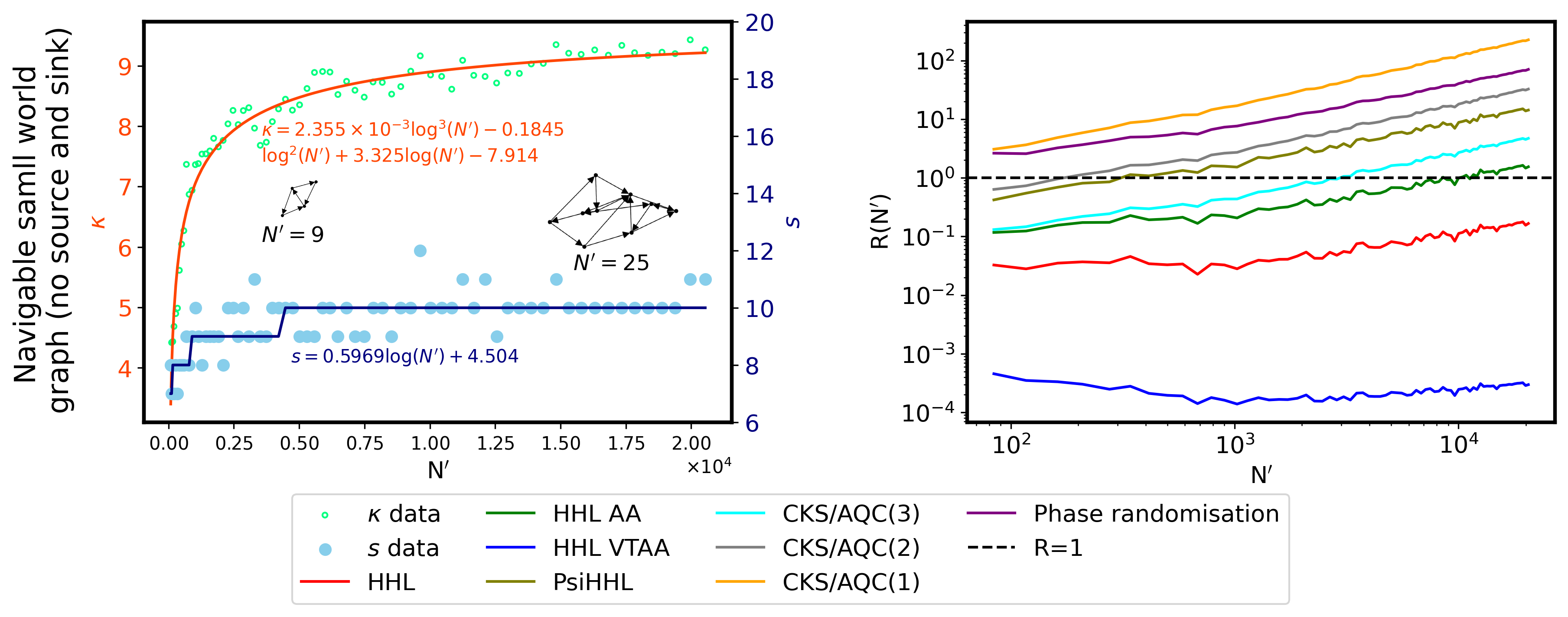} \\
\makebox[0.45\textwidth]{(g)}
\makebox[0.36\textwidth]{(h)} \\ 

\multicolumn{1}{r}{(\textit{Continued})} \\
\end{tabular} 

\end{figure*}

\begin{figure*}[t]
\begin{tabular}{c}

\includegraphics[width=14cm]{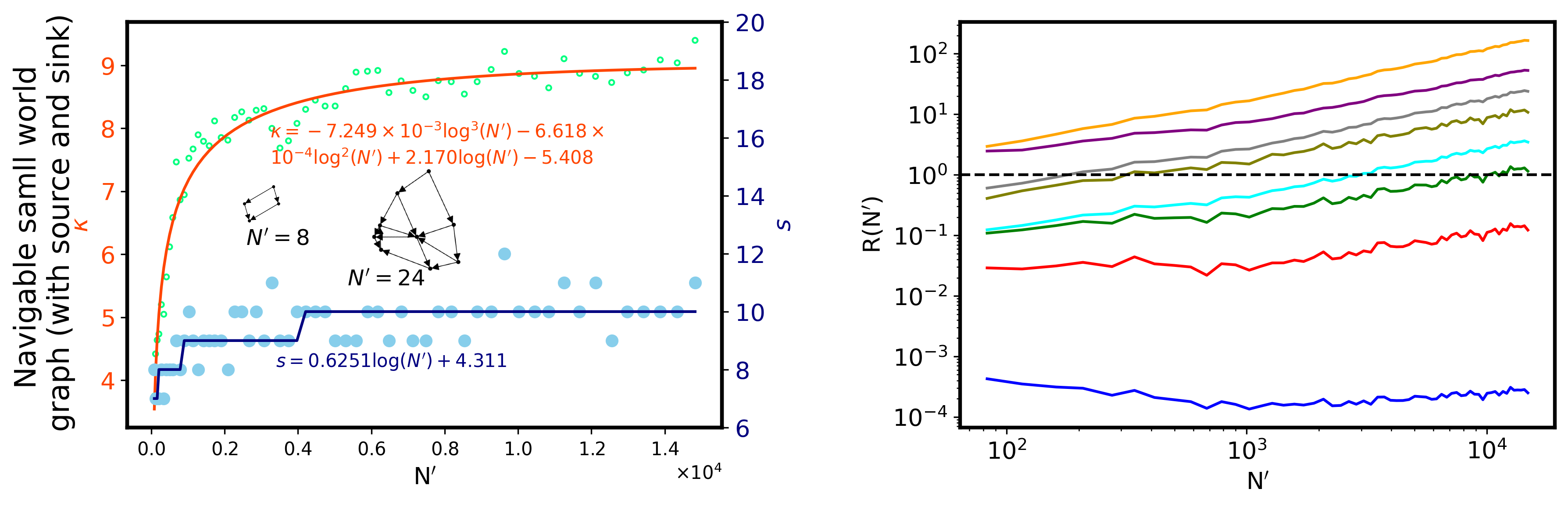} \\
\makebox[0.45\textwidth]{(i)}
\makebox[0.36\textwidth]{(j)} \\ 

\includegraphics[width=14cm]{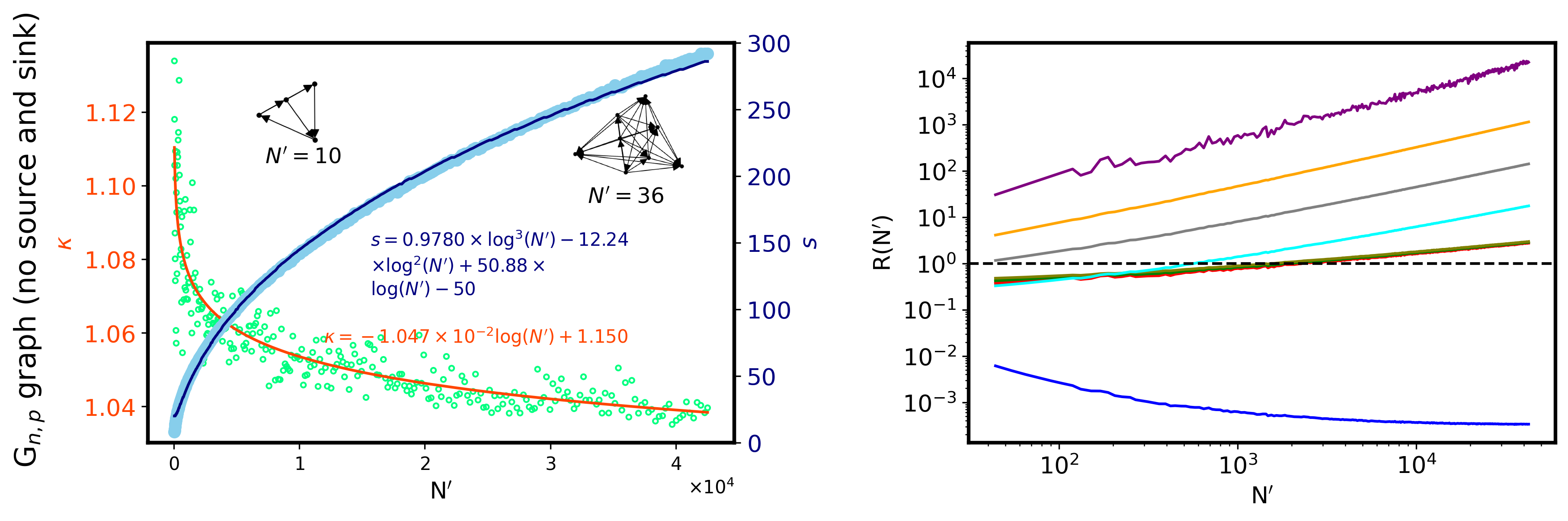} \\
\makebox[0.45\textwidth]{(k)}
\makebox[0.36\textwidth]{(l)} \\ 

\includegraphics[width=14cm]{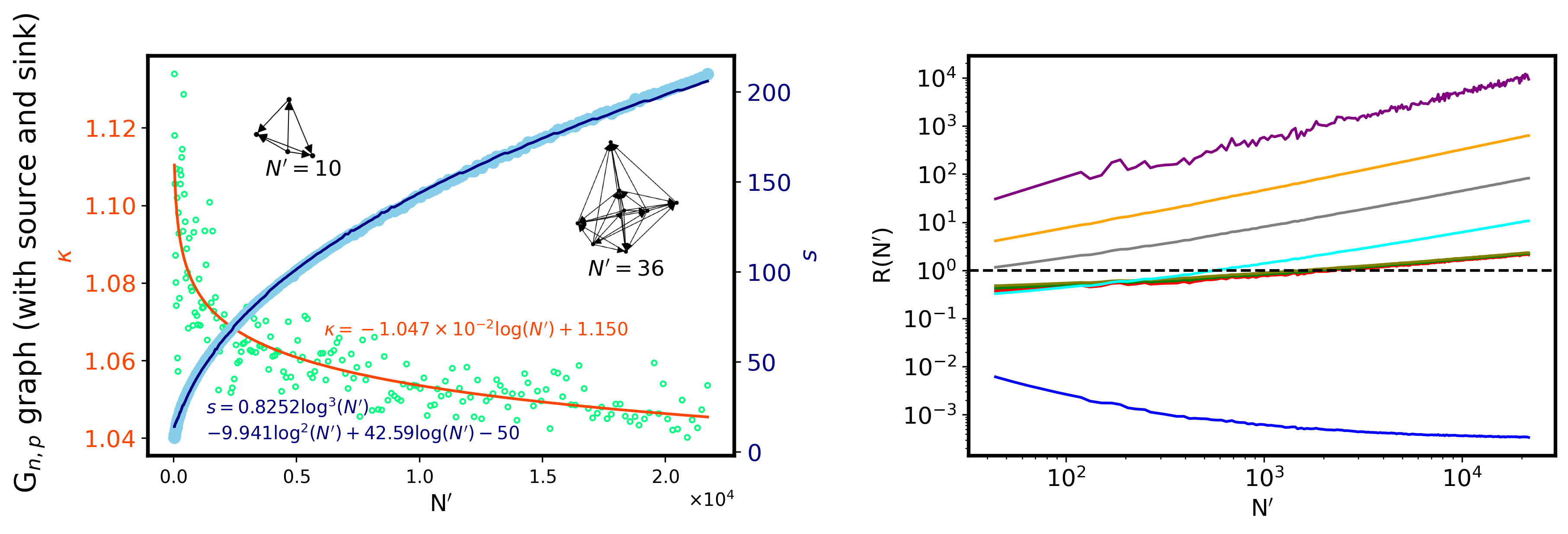} \\
\makebox[0.45\textwidth]{(m)}
\makebox[0.36\textwidth]{(n)} \\ 

\includegraphics[width=14cm]{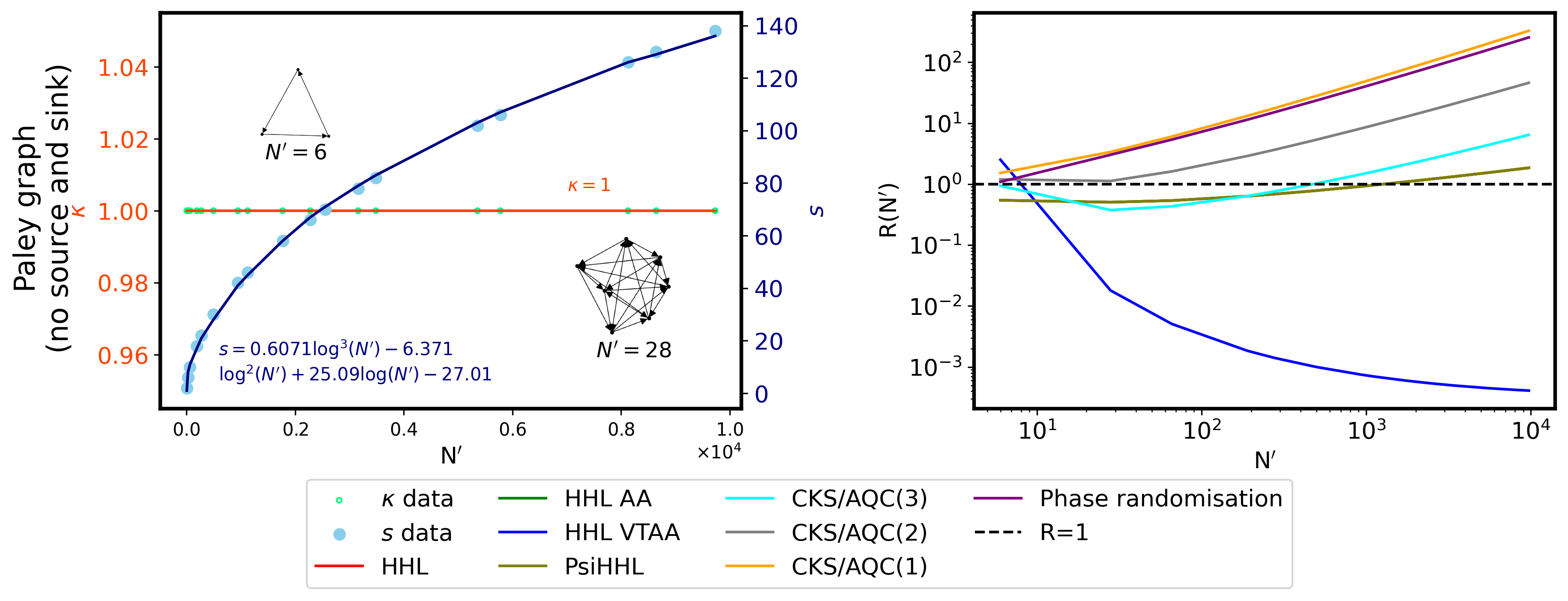} \\
\makebox[0.45\textwidth]{(o)}
\makebox[0.36\textwidth]{(p)} \\ 

\multicolumn{1}{r}{(\textit{Continued})} \\
\end{tabular} 

\end{figure*}

\begin{figure*}[h!]
\begin{tabular}{c}

\includegraphics[width=14cm]{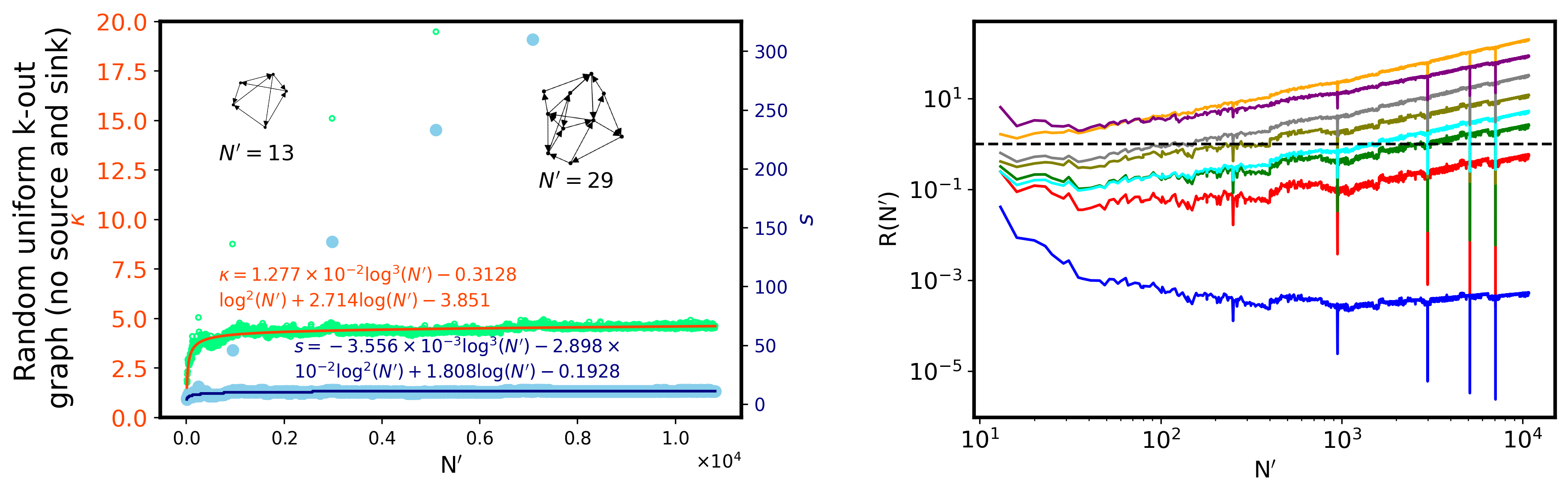} \\
\makebox[0.45\textwidth]{(q)}
\makebox[0.36\textwidth]{(r)} \\ 

\includegraphics[width=14cm]{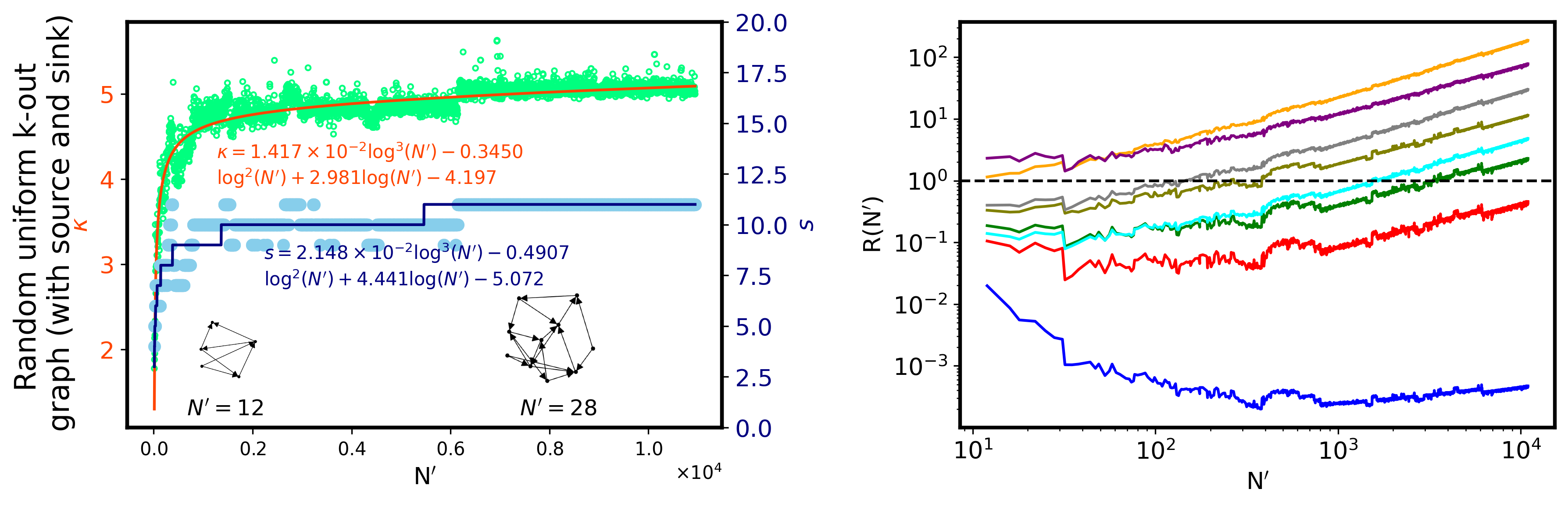} \\
\makebox[0.45\textwidth]{(s)}
\makebox[0.36\textwidth]{(t)} \\ 

\includegraphics[width=14cm]{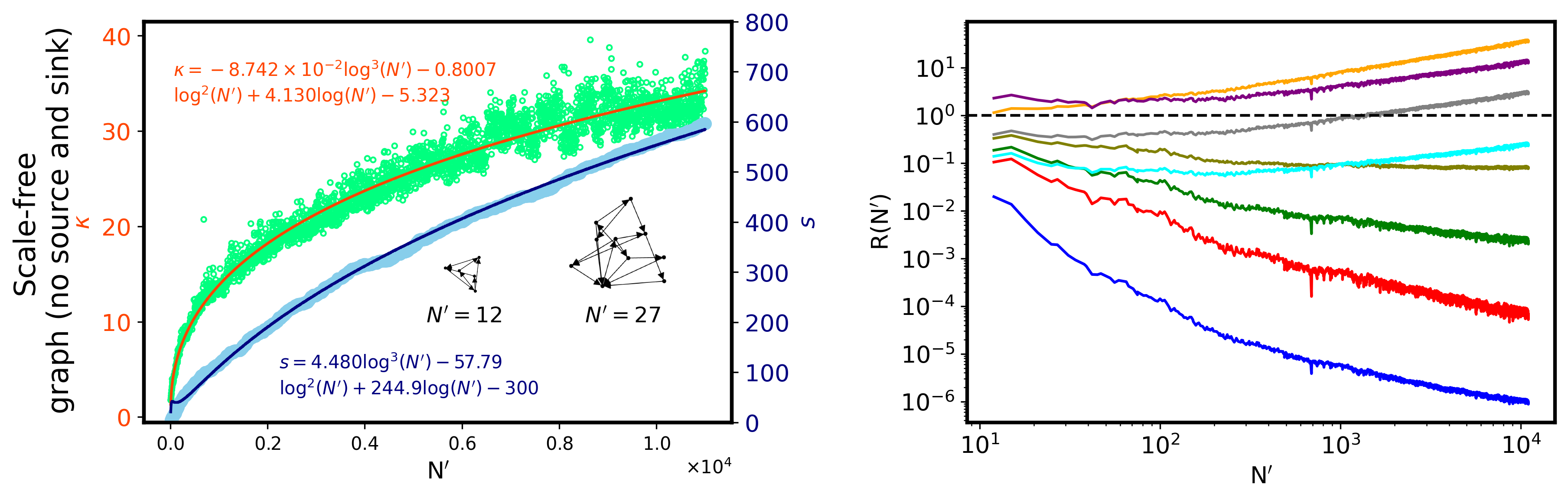} \\
\makebox[0.45\textwidth]{(u)}
\makebox[0.36\textwidth]{(v)} \\ 

\includegraphics[width=14cm]{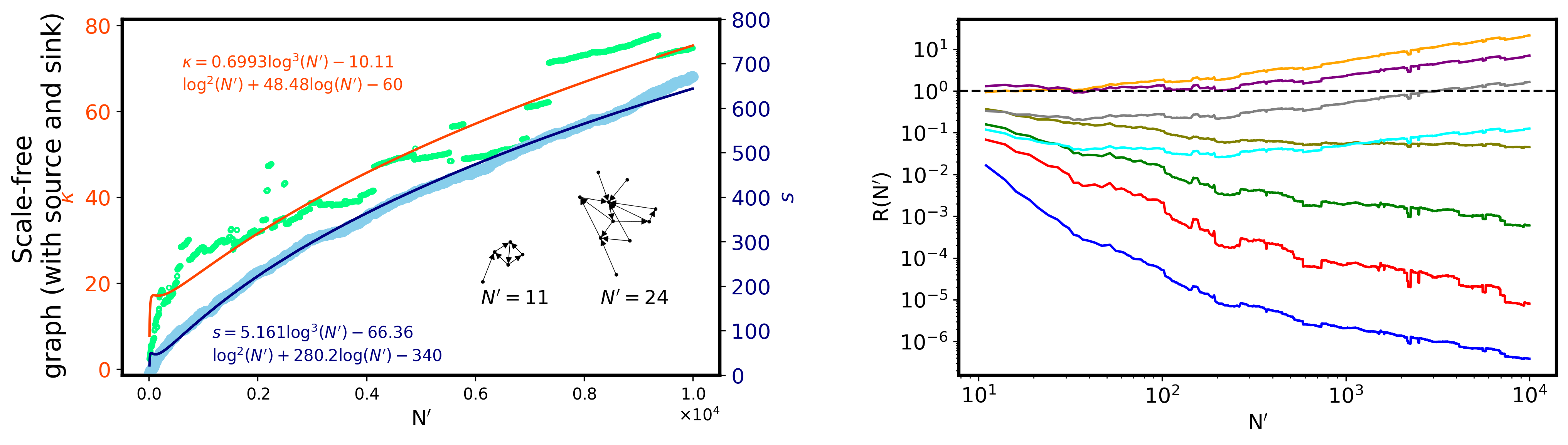} \\
\makebox[0.45\textwidth]{(w)}
\makebox[0.36\textwidth]{(x)} \\
\end{tabular}

\caption{Sub-figures showing $\kappa$ and $s$ behaviour of different good graphs listed in Table IV of the main manuscripts as well as as well as runtime ratios of different QLSs obtained from our numerical data. Sub-figures of first two graph families from Table IV can be found in the main manuscript.   }\label{fig:sm_good_inc}

\end{figure*}


\begin{figure*}[h!]
\begin{tabular}{c}

\includegraphics[width=14cm]{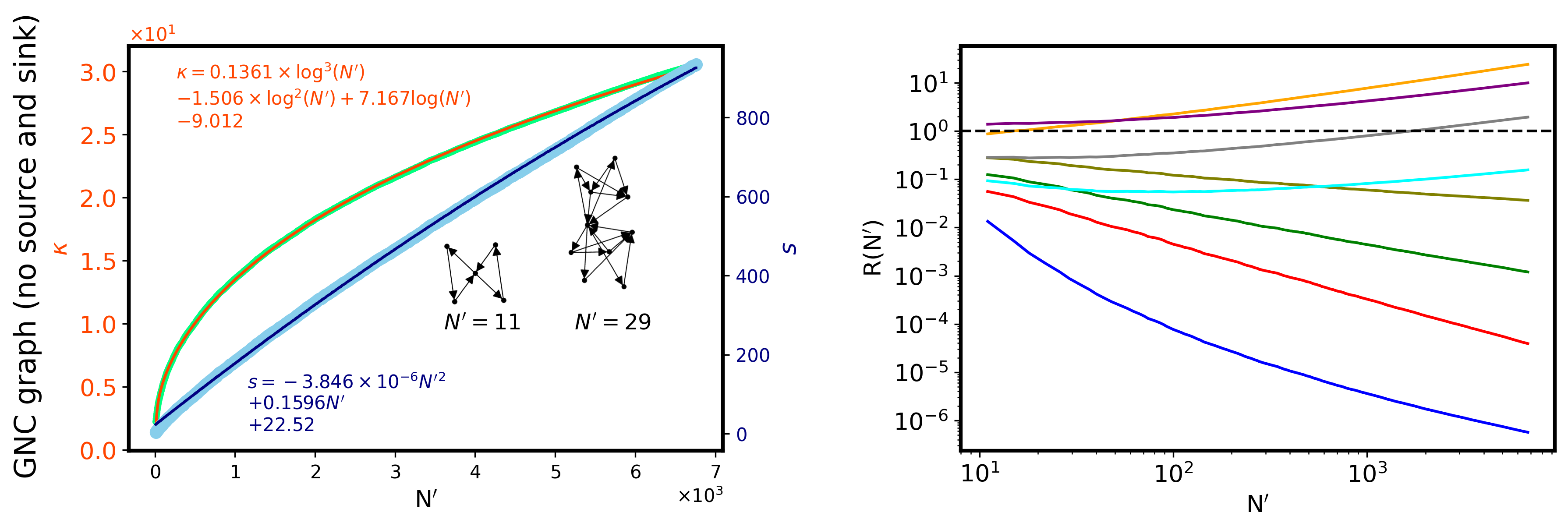} \\
\makebox[0.45\textwidth]{(a)}
\makebox[0.36\textwidth]{(b)} \\

\includegraphics[width=14cm]{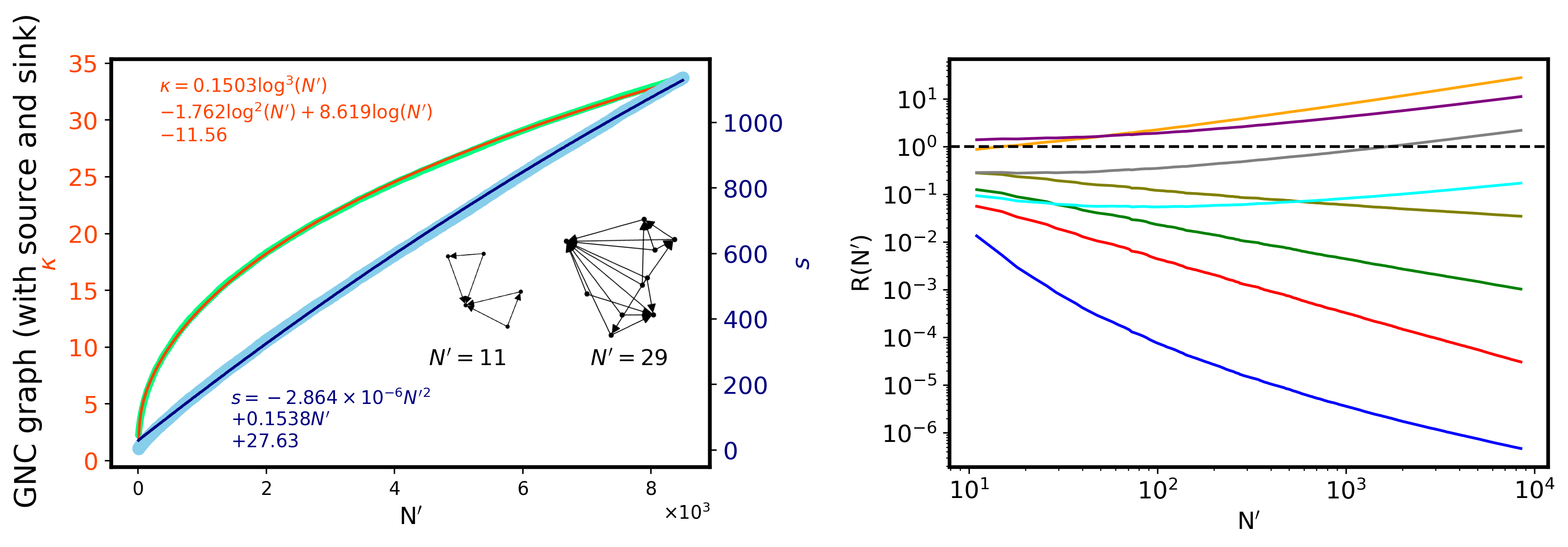} \\
\makebox[0.45\textwidth]{(c)}
\makebox[0.36\textwidth]{(d)} \\

\includegraphics[width=14cm]{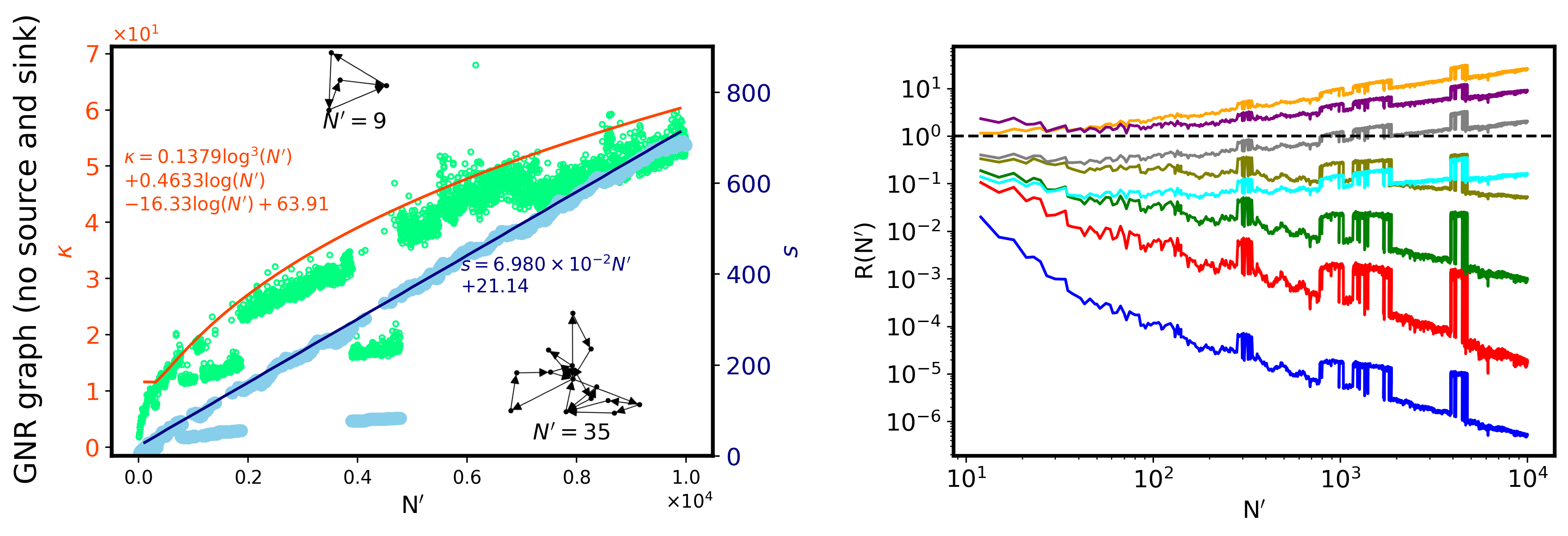} \\
\makebox[0.45\textwidth]{(e)}
\makebox[0.36\textwidth]{(f)} \\

\includegraphics[width=14cm]{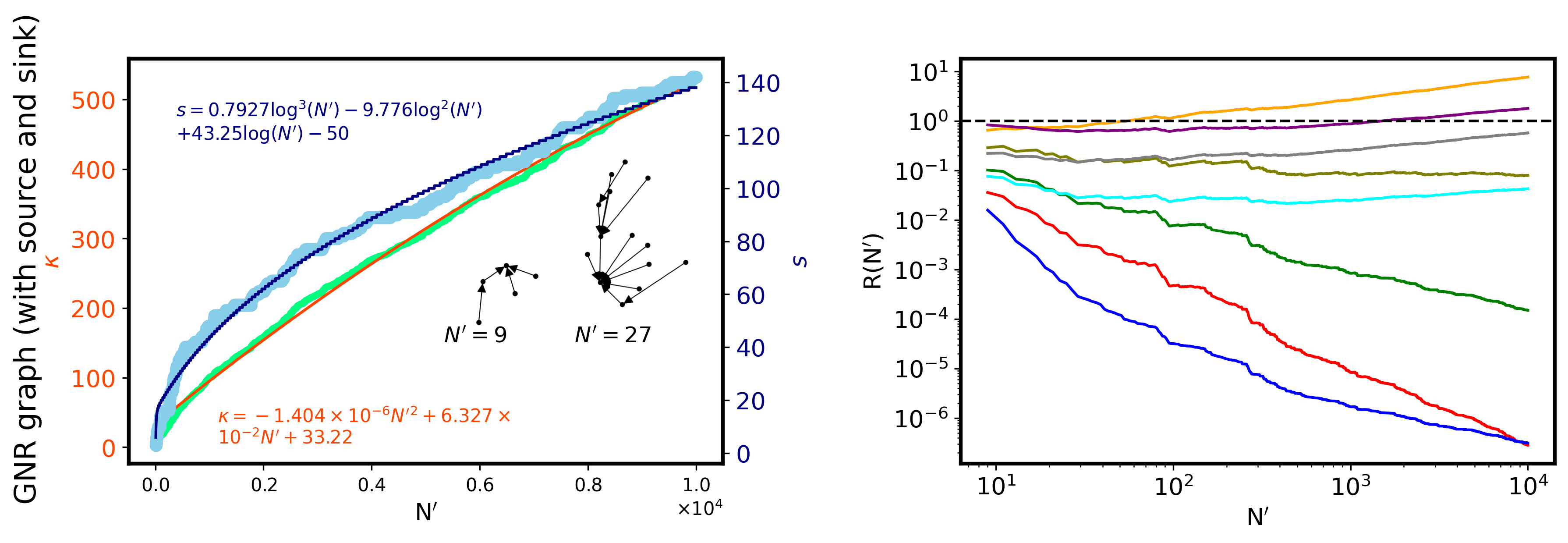} \\
\makebox[0.45\textwidth]{(g)}
\makebox[0.36\textwidth]{(h)} \\
\end{tabular}

\caption{Sub-figures showing $\kappa$ and $s$ behaviour of different bad graphs listed in Table V of the main manuscript and plots of runtime ratios with different QLSs considered in our study.}\label{fig:sm_bad_inc}

\end{figure*}

\begin{figure*}[h!]
\begin{tabular}{c}


\includegraphics[width=17cm]{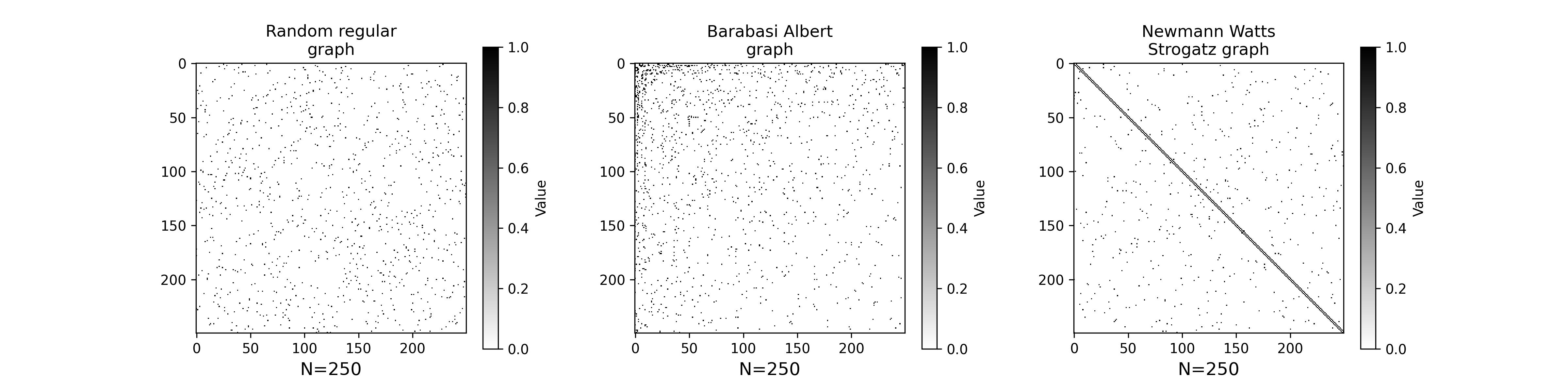} \\
    \vspace{0.5em}
    \makebox[0.31\textwidth]{(a)}
    \makebox[0.31\textwidth]{(b)}
    \makebox[0.31\textwidth]{(c)}\\ 
\includegraphics[width=17cm]{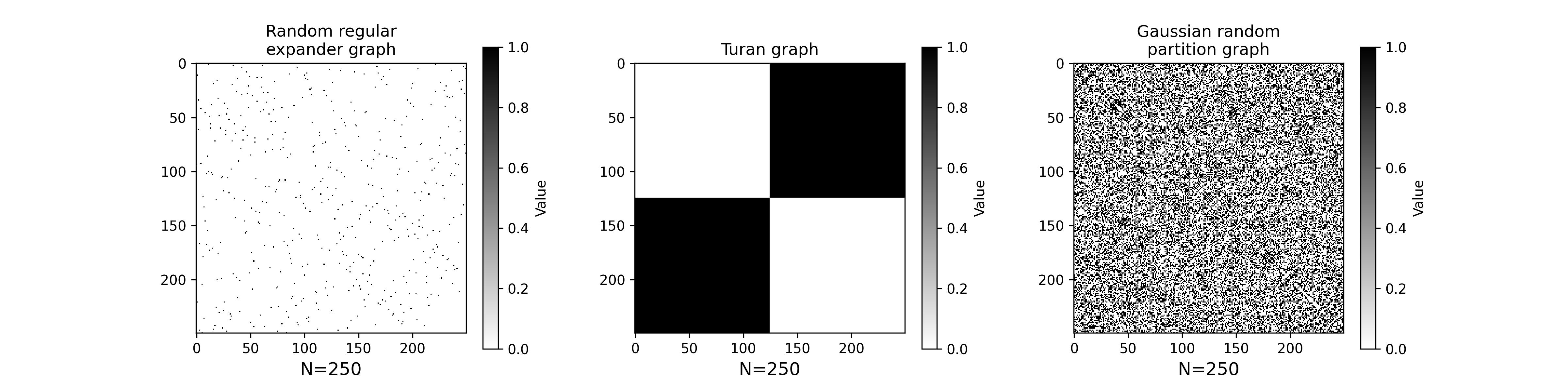} \\
    \vspace{0.5em}
    \makebox[0.31\textwidth]{(d)}
    \makebox[0.31\textwidth]{(e)}
    \makebox[0.31\textwidth]{(f)}\\

\includegraphics[width=17cm]{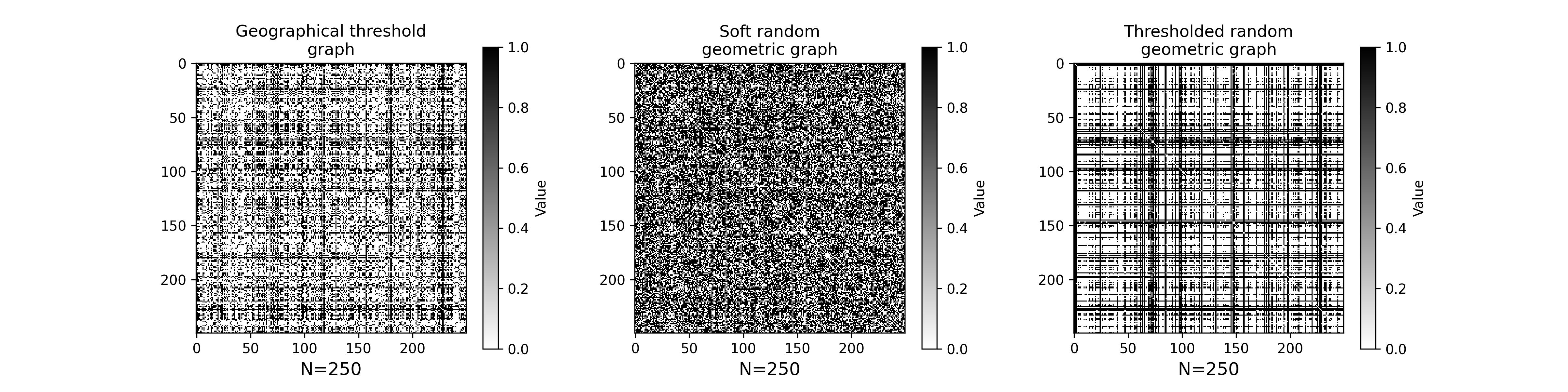} \\
    \vspace{0.5em}
    \makebox[0.31\textwidth]{(g)}
    \makebox[0.31\textwidth]{(h)}
    \makebox[0.31\textwidth]{(i)}\\ 
    
\multicolumn{1}{r}{(\textit{Continued})} \\   
\end{tabular} 
\end{figure*} 

\begin{figure*}[h!]
\begin{tabular}{c}

\includegraphics[width=17cm]{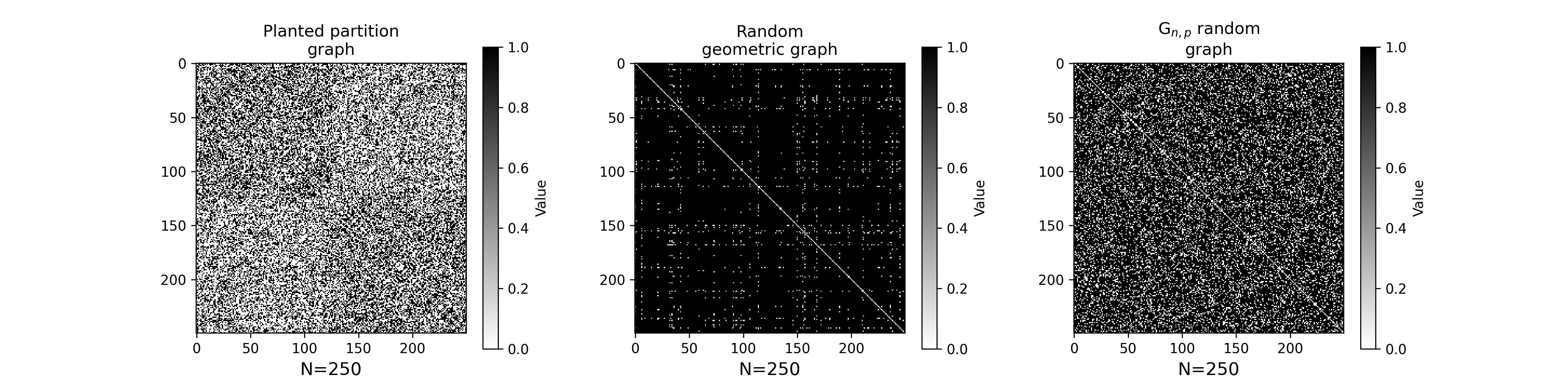} \\
    \vspace{0.5em}
    \makebox[0.31\textwidth]{(j)}
    \makebox[0.31\textwidth]{(k)} 
    \makebox[0.31\textwidth]{(l)}\\

\includegraphics[width=17cm]{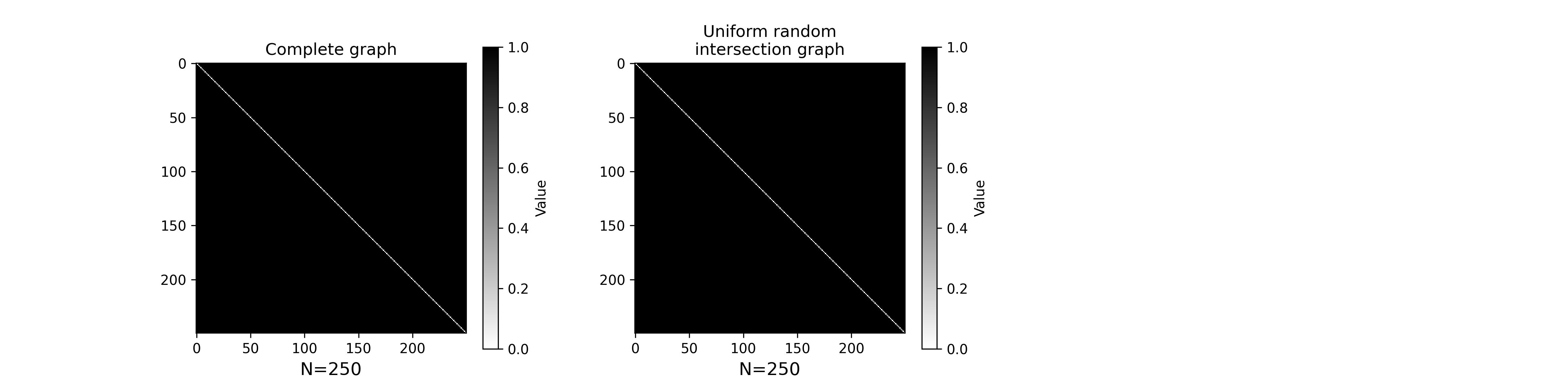} \\
    \vspace{0.5em}
    \makebox[0.31\textwidth]{(m)}
    \makebox[0.31\textwidth]{(n)} 
    \makebox[0.31\textwidth]{}\\

\end{tabular} 
\caption{We have presented one instance from those graph families whose condition number was polylog function of $\mathcal{N}$. We observe that for all these graph families, the entries in the adjacency matrix, thus the Laplacian matrix appear diffused throughout the matrix. }\label{fig:dispersed_graph}

\end{figure*} 

\begin{figure*}[h!]
\begin{tabular}{c}


\includegraphics[width=17cm]{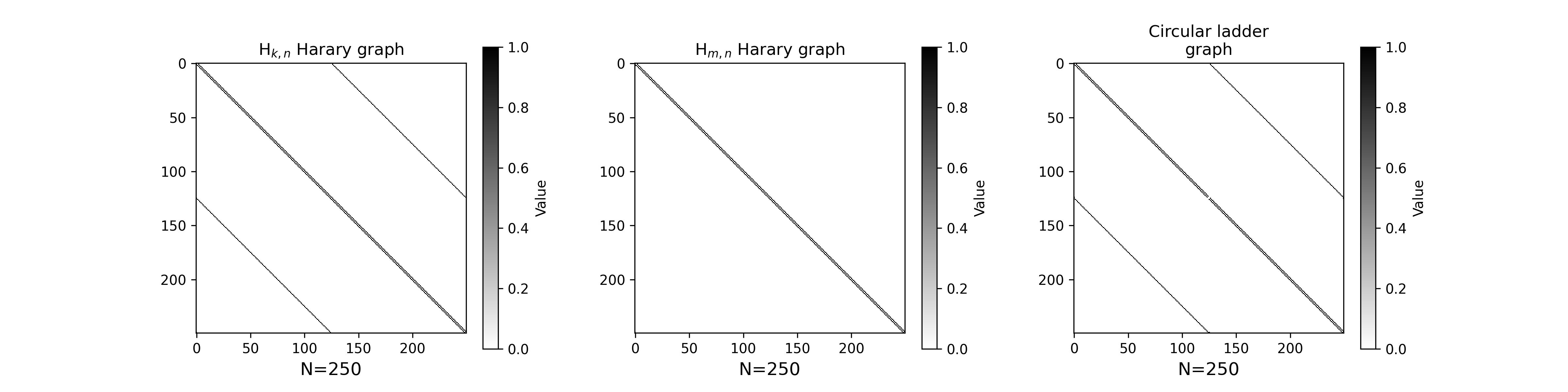} \\
    \vspace{0.5em}
    \makebox[0.31\textwidth]{(a)}
    \makebox[0.31\textwidth]{(b)}
    \makebox[0.31\textwidth]{(c)}\\ 
\includegraphics[width=17cm]{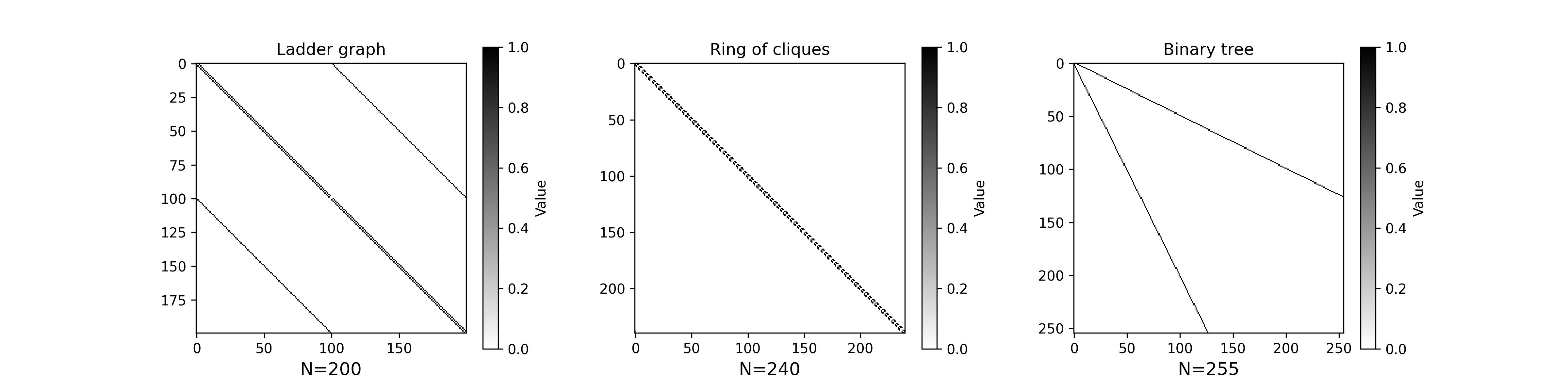} \\
    \vspace{0.5em}
    \makebox[0.31\textwidth]{(d)}
    \makebox[0.31\textwidth]{(e)}
    \makebox[0.31\textwidth]{(f)}\\ 

    \includegraphics[width=17cm]{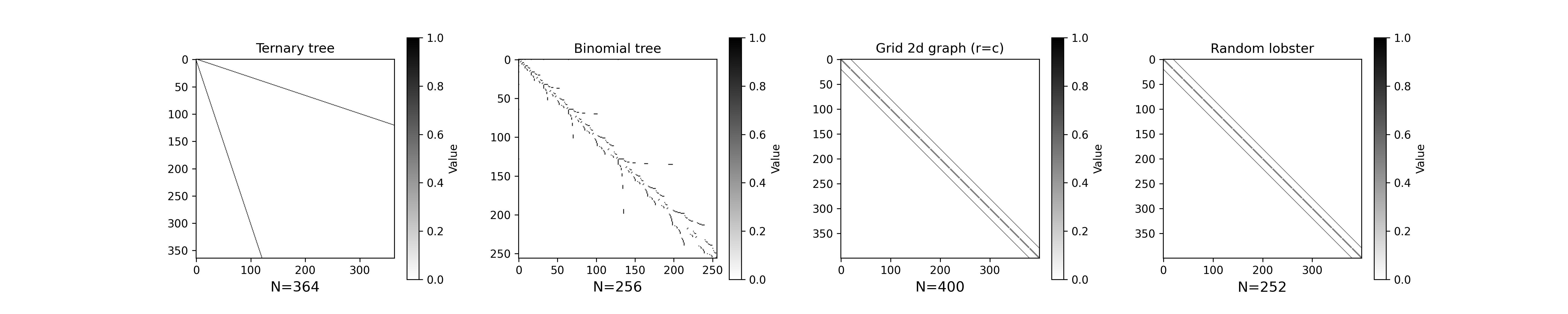} \\
    \vspace{0.5em}
    \makebox[0.31\textwidth]{(g)}
    \makebox[0.31\textwidth]{(h)}
    \makebox[0.31\textwidth]{(i)}\\ 
 
\end{tabular} 

\end{figure*} 

\begin{figure*}[h!]
\begin{tabular}{c}

\end{tabular} 
\caption{We have presented one instance from those graph families whose condition number was a polynomial function of $\mathcal{N}$. For all the graph families mentioned here, we witness a sharp structure in their adjacency matrix and so in its Laplacian matrix.  }\label{fig:band_graph}

\end{figure*} 

\end{appendix}

\clearpage 

\bibliography{apssamp} 

\end{document}